\documentclass[a4paper,11pt,twoside]{memoir}

\usepackage{amsmath,amssymb,amsthm}
\usepackage[latin1]{inputenc}
\usepackage[T1]{fontenc}
\usepackage{epsfig}
\usepackage{makeidx}
\usepackage{wrapfig}
\usepackage[mathscr]{euscript}
\usepackage{algorithmic}

\usepackage{eurosym} 

\usepackage{xspace}
\usepackage{url}
\usepackage{enumerate}
\usepackage{multicol}

\usepackage[usenames,dvipsnames]{color} 

\usepackage[linktocpage=true,colorlinks=true,citecolor=red,linkcolor=blue]{hyperref}	

\allowdisplaybreaks[1]

\newcommand{\commentComete}[1]{} 
\newcommand{\commentMS}[1]{\commentComete{{\textbf{MS: }} #1}} 

\def\Aset{\mathcal{A}}
\def\Bset{\mathcal{B}}

\def\Dset{\mathcal{D}}
\def\Fset{\mathcal{F}}

\def\Lset{\mathcal{L}}
\def\Oset{\mathcal{O}}

\def\Qset{\mathcal{Q}}
\def\Sset{\mathcal{S}}
\def\Tset{\mathcal{T}}

\def\Xset{\mathcal{X}}
\def\Yset{\mathcal{Y}}
\def\Wset{\mathcal{W}}

\newcommand{\Asym}[2]{a_{#1}^{#2}}
\newcommand{\Bsym}[2]{b_{#1}^{#2}}
\newcommand{\Fsym}[2]{f_{#1}^{#2}}

\newcommand{\Aseq}[2]{{\alpha_{#1}^{#2}}}
\newcommand{\Bseq}[2]{{\beta_{#1}^{#2}}}
\newcommand{\Fseq}[2]{{\phi_{#1}^{#2}}}

\newcommand{\prob}[2]{#1(#2)}
\newcommand{\Cprob}[3]{#1(#2 | #3)}

\newcommand{\IIHS}{\text{\textrm{IIHS}}\xspace}
\newcommand{\IIHSs}{\text{\textrm{IIHSs}}\xspace}

\def\Isys{{\mathscr I}} 
\newcommand{\si}{\hat{s}}

\def\Lfun{\mathit{labels}}

\theoremstyle{plain}
\newtheorem{theorem}{Theorem}
\newtheorem{lemma}[theorem]{Lemma}
\newtheorem{corollary}[theorem]{Corollary}
\newtheorem{proposition}[theorem]{Proposition}
\newtheorem{definition}[theorem]{Definition}
\newtheorem{convention}[theorem]{Convention}

\newtheorem{remark}[theorem]{Remark}
\newtheorem{example}{Example}

\let\phi\varphi

\newcommand{\PP}{{\mathbf P}}

\newcommand{\ihs}{\operatorname{{\mathcal I}}}

\newcommand{\val}{\operatorname{\mathit{val}}}

\newcommand{\last}{\operatorname{\mathit{last}}\xspace}

\newcommand{\distr}{{\operatorname{{\mathcal{D}}}}}

\newcommand{\cone}[1]{\langle {#1} \rangle}

\newcommand{\qi}{\hat{q}}

\newcommand{\fpaths}{{\operatorname{{Paths}^\star}}}
\newcommand{\paths}{{\operatorname{{Paths}}}}
\newcommand{\cpaths}{{\operatorname{{CPaths}}}}

\newcommand{\CM}{{\mathcal M}}

\def\bigsqcap{\mathop{\rule[-0.2ex]{.07em}{2.17ex}
                \rule[1.8ex]{0.5em}{.17ex}
                \rule[-0.2ex]{.07em}{2.17ex}}}

\DeclareMathOperator{\trace}{{\mathit{trace}}}

\newcommand{\Traces}{\mathit{Traces}}

\newcommand{\psp}{{\operatorname{\mathbf{P}}}}

\newcommand{\pow}{\mathcal{P}} 

\newcommand{\avail}{\operatorname{\mathit{enab}}}
\newcommand{\sift}{\operatorname{\mathit{sift}}}
\newcommand{\view}{\operatorname{\mathit{view}}}

\newcommand{\smallprobsum}[1]{\mkern4mu{\textstyle\circ\mkern-15.5mu\sum_{#1}\:}}
\newcommand{\bigprobsum}[1]{{\;\displaystyle\odot\mkern-21mu\sum_{#1}}}
\newcommand{\probsum}[1]{\mathchoice{\bigprobsum{#1}}{\smallprobsum{#1}}{\smallprobsum{#1}}{\smallprobsum{#1}}}

\newcommand{\calx}{\mathcal{X}}
\newcommand{\caly}{\mathcal{Y}}
\newcommand{\calz}{\mathcal{Z}}
\newcommand{\calk}{\mathcal{K}}
\newcommand{\calh}{\mathcal{H}}
\newcommand{\call}{\mathcal{L}}
\newcommand{\calu}{\mathcal{U}}

\newcommand{\maxj}[2]{\max^{#1}_{#2}}

\newcommand{\qm}[1]{``#1''} 

\newsubfloat{figure} 
\newsubfloat{table} 

\newcommand{\gramor}{\ \,|\ \,}
\newcommand{\paral}{\;|\;}
\newcommand{\labarr}[1]{\overset{#1}{\longrightarrow}}
\newcommand{\smallsum}[1]{\textstyle{\sum_{#1}\:}}

\newcommand{\bigfrac}[2]{\frac{\displaystyle #1}{\displaystyle #2}}

\newcommand{\outp}[1]{\overline{#1}}
\newcommand{\ccsp}{CCS$_p$}

\def\defsym{\stackrel{\textrm{def}}{=}}

\newcommand{\mscite}[2]{\begin{flushright} \emph{``#1''} \\ #2 ~ \end{flushright} \vspace{0.5cm} } 

\def\diameter{\delta}

\newcommand{\border}[3]{#1_{\left\langle #2 \right\rangle}(#3)}

\def\vtt{VT^{+}}

\def\val{\mathit{Val}}
\def\ind{\mathit{Ind}}

\def\bound{\mathit{Bnd}}

\def\gain{\mathit{gain}}
\def\bingain{\mathit{gain}_{\mathit{bin}}}
\def\guess{\mathit{guess}}

\newcommand{\revision}[1]{#1}             

\begin{document}


\thispagestyle{empty}
	\begin{multicols}{2}
		\begin{large}
			\noindent
			\textsc{\'Ecole Polytechnique} \\ \\
			\textsc{PhD. Thesis - \emph{Th\`ese de Doctorat}} \\ 
			\textsc{Sp\'ecialit\'e Informatique}
		\end{large}
		\begin{flushright}
			\includegraphics[width=0.75\linewidth]{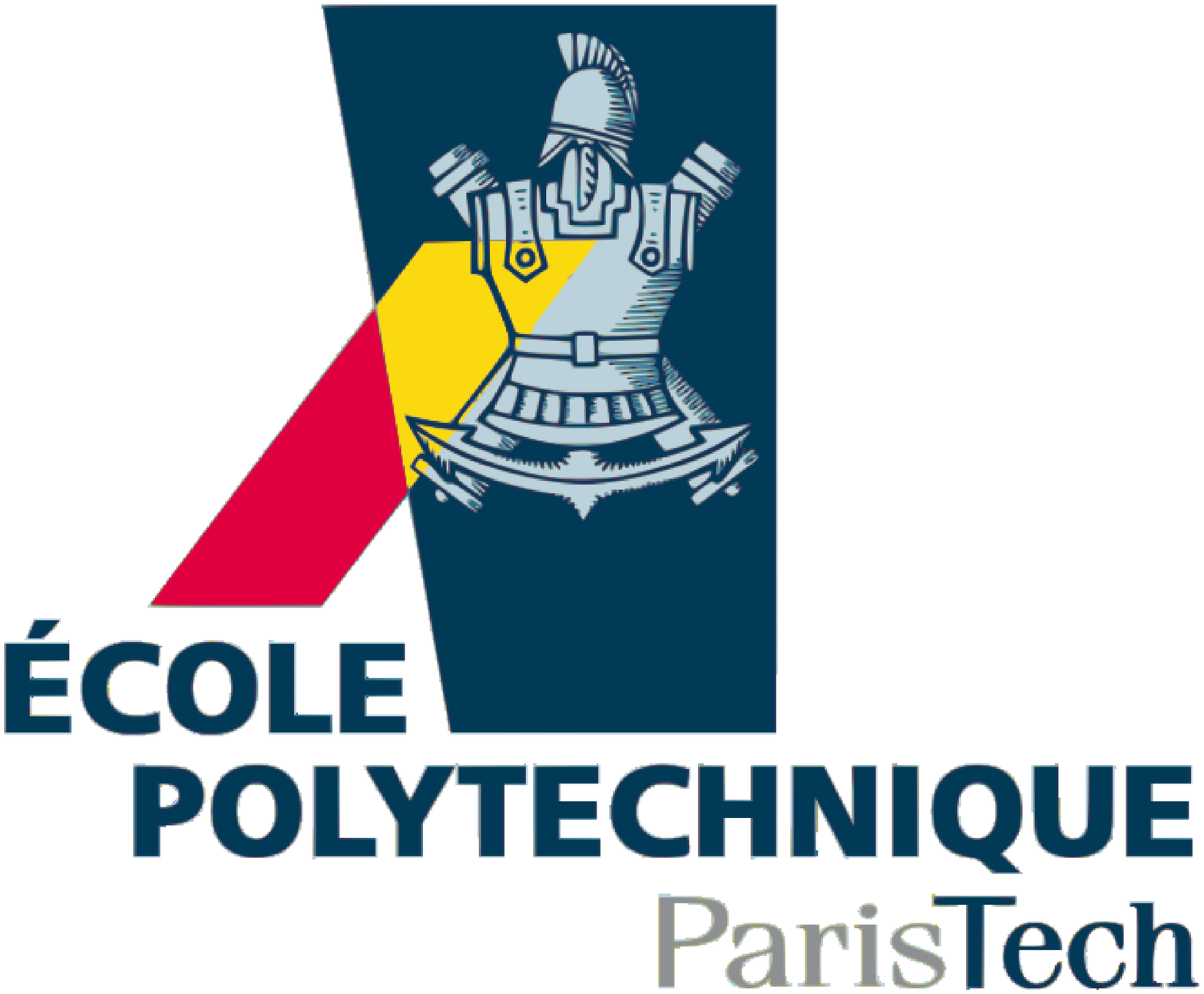}
		\end{flushright}
	\end{multicols}
	\ \\
\begin{center}
	\begin{LARGE}
			\textsc{{Formal approaches to information hiding:}}
	\end{LARGE}
	\begin{Large} \ \\
		\textsc{An analysis of interactive systems, statistical disclosure control, and refinement of specifications}
	\end{Large} \\
	\ \\ \ \\ \ \\
	\begin{large}
		\textsc{M\'ario S. Alvim} \\ 
	\end{large}
	\textsc{LIX, \'{E}cole Polytechnique} \\
	\textsc{Palaiseau, France}
	\ \\ \ \\ \ \\
	\emph{Supervisor}\\
	\textsc{Catuscia Palamidessi} \\
	\ \\ \ \\
	\emph{Rapporteurs}\ \\
	\textsc{Gilles Barthe} \\
	\textsc{Michael Mislove} \\
	\ \\
	\emph{Examinateurs} \ \\
	\textsc{B\'eatrice B\'erard} \\
	\textsc{St\'ephanie Delaune} \\
	\textsc{Lo\"{i}c H\'{e}lou\"{e}t} \\
	\textsc{Daniel Le M\'etayer} \\
	\textsc{Geoffrey Smith}  \\
	\ \\ \ \\ 
	\textsc{$12^{th}$ of October 2011}
\end{center}

\thispagestyle{empty}

\cleardoublepage
\thispagestyle{empty}


	\begin{multicols}{2}
		\begin{center}
			\includegraphics[width=0.75\linewidth]{figures/polytechnique_logo.eps} \\
			\'Ecole Polytechnique \\
			\includegraphics[width=0.75\linewidth]{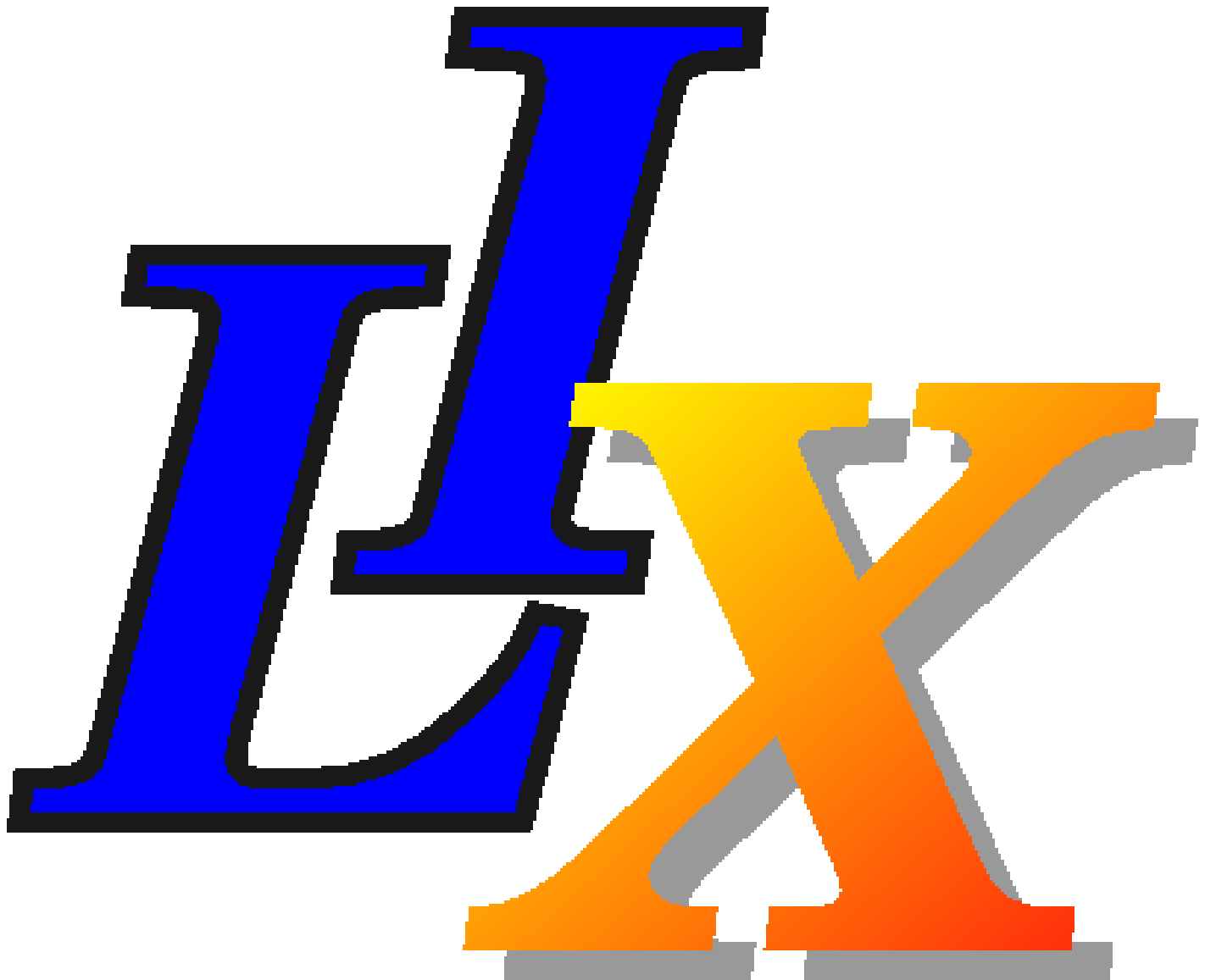} \\
			Laboratoire d'Informatique \\			
		\end{center}
	\end{multicols}
	
	\ \\ \ \\ \ \\ \ \\
	
	\begin{center}
		\includegraphics[width=0.5\linewidth]{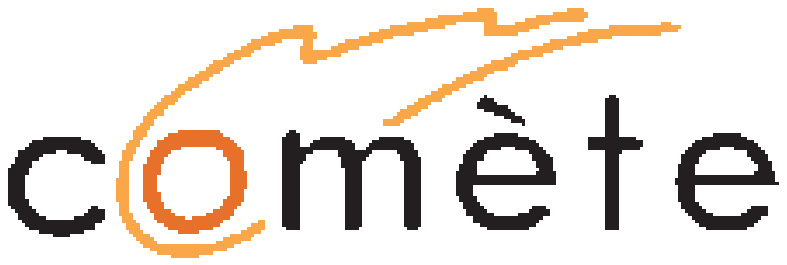} \\
		\'Equipe Com\`ete \\
	\end{center}
		
	\ \\ \ \\ \ \\ \ \\
	
	\begin{multicols}{2}
		\begin{center}
			\includegraphics[width=0.75\linewidth]{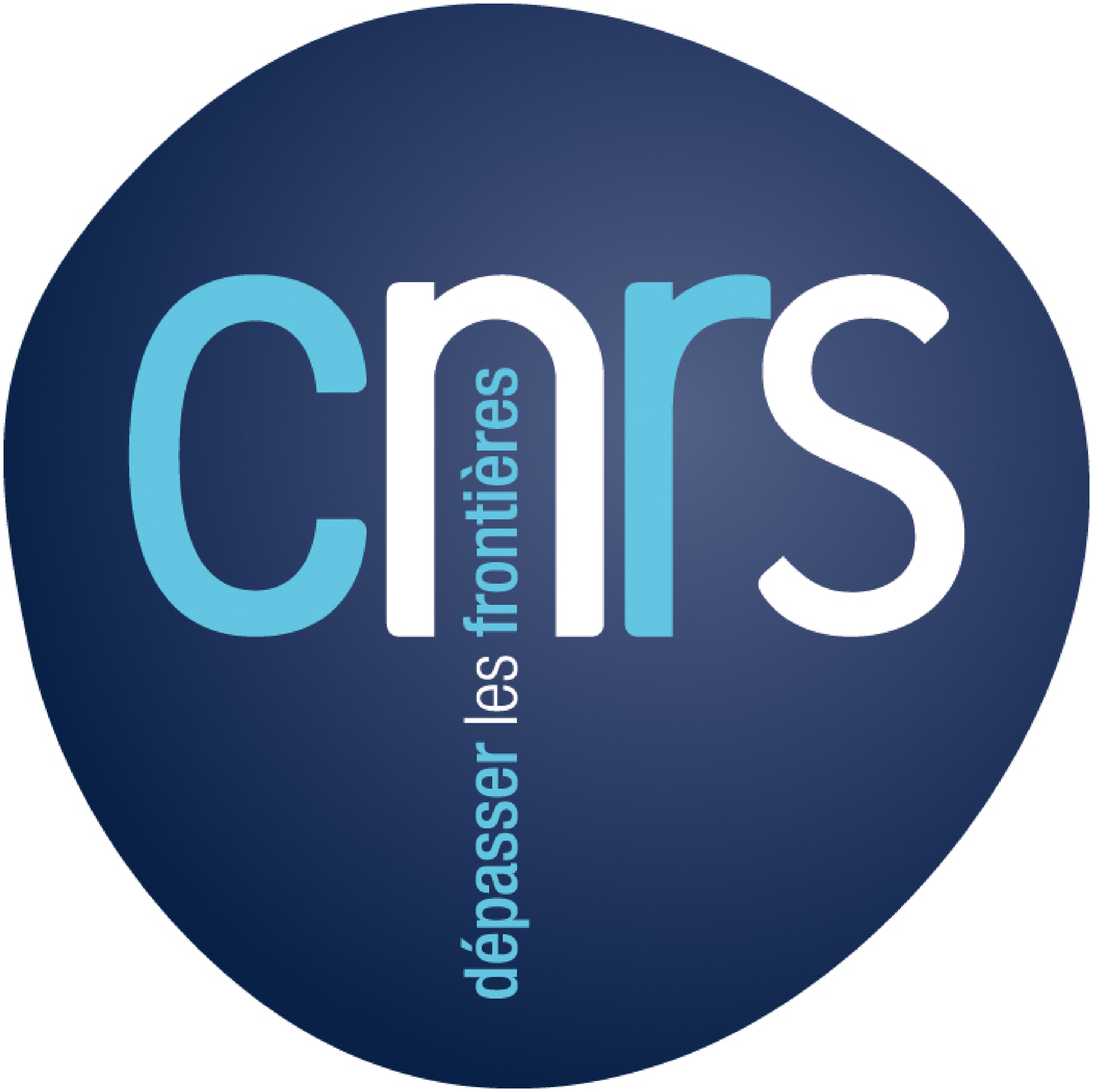} \\
			Centre National de la Recherche Scientifique \\
			\includegraphics[width=0.75\linewidth]{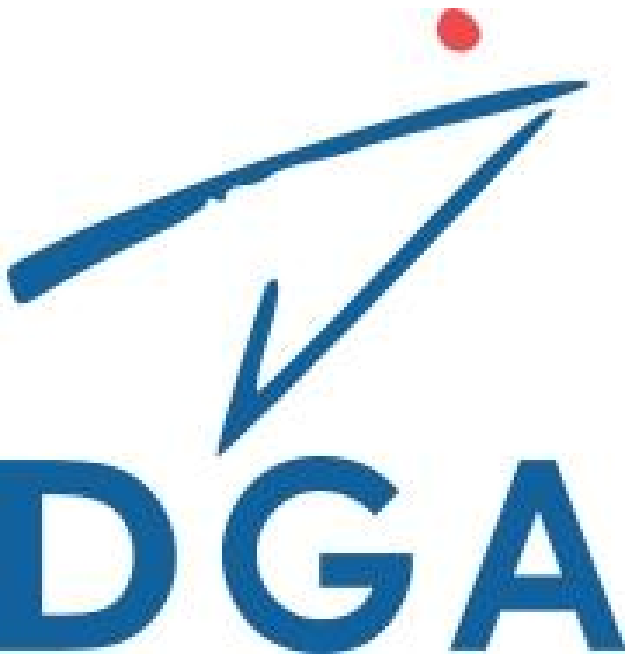} \\
			D\'el\'egation G\'en\'erale de l'Armement \\
		\end{center}
	\end{multicols}
	
\thispagestyle{empty}

\frontmatter
\setsecnumdepth{subsection}
\settocdepth{subsection}

\tableofcontents

\newpage
\listoffigures

\newpage
\listoftables

\newpage
\chapter{Acknowledgements}
\mscite{Praise the bridge that carried you over.}{George Colman}
Every piece of work is produced within a context, and naturally this thesis is no exception. I want to dedicate this space to express my gratitude to some people that have helped to create an environment of scientific, material, and emotional support, which was crucial to the development of my work over the past three years. I am deeply grateful to all these people, and the influence they have had in this work is only a small part of the influence and importance they have in my life.

First of all, I will always be deeply grateful to Catuscia Palamidessi for her outstanding work as my thesis supervisor. During these three years she has provided a stimulating and exciting scientific environment, endowed with all the material and logistic support a student could ever need. Her passion for science is contagious, and her brilliance and persistence are qualities that I can only hope to be fortunate enough to achieve someday. And not only is she a widely recognized researcher, but she is also a remarkable human being, whose kindness and ethics have set an example that I will always keep with me in academia and for life. It is with sincere joy that I can say that, besides the fruitful scientific cooperation, we were able to create a deep link of friendship, and I will do my best so that both can last for life. 

Another person of fundamental importance in my path to this day is Elaine Pimentel. As the first scientific tutor I have ever had, and later as my Master's program supervisor, she was the one welcoming me to the fascinating world of academia. She guided my first steps in research, and her dedication and intelligence are remarkable. Elaine was the strongest supporter I have ever had for doing a doctoral program abroad, especially in the early times when not even my family was convinced yet it was a good idea. More than once Elaine was a thoughtful friend and a wise advisor, who helped me figure out solutions for practical problems that, in some moments, made me doubt I could get to the end of this program. Thank you, Elaine, very much for it all.

I would also like to thank the CNRS (\emph{Centre National de la Recherche Scientifique}) and the DGA (\emph{Direction G\'en\'erale de l'Armement}) for providing the funds for these three years of research in France. I also thank INRIA for all the financial and logistical support with respect to scientific conferences, events, and work trips. 

I am grateful to the members of my jury, who kindly gave their time to go through my work and evaluate it. Thanks to B\'eatrice B\'erard, St\'ephanie Delaune, Lo\"{i}c H\'{e}lou\"{e}t, Daniel Le M\'etayer, and Geoffrey Smith. And special thanks to my \textit{rapporteurs} Michael Mislove and Gilles Barthe, who produced the evaluation report for my thesis. I am honored to have had the opportunity to have such a high qualified jury.

I would also like to thank all the people from the Graduate School (\emph{\'Ecole Doctorale}) of \'Ecole Polytechnique, especially Audrey L\'emarechal for her help with the documentation regarding my stay in France, Fabrice Baronnet for his administrative work, and Christine Ferret for everything involving the thesis defense.

I feel especially fortunate for having had the opportunity to work in such a stimulating environment as is the LIX laboratory (\emph{Laboratoire d'Informatique de l'\'Ecole Polytechnique}), and in particular the Com\`ete team. It is with a weight in my heart that I leave all these amazing people. I am deeply grateful to Frank Valencia, who gave me one of the warmest welcomes I got in my new life in Europe. Frank was not only a teacher, but a colleague, a gym companion, and a good friend. I am also grateful to his wife, Sara S\"odergren, and to their son, Felipe Valencia, for all the good moments shared. Thanks also to Andr\'es Aristiz\'abal for his kindness and for always being ready to help; to Carlos Olarte for the help, friendship and good moments shared together (I will never forget that it was Carlos who took me on my first walk in Paris, and introduced me to the Eiffel Tower); to Sophia Knight for the shared laughs, food, jokes and complaints that make our \qm{love-hate} friendship unique; and to Justin Dean, Sophia's husband, who is a remarkably kind and smart guy with an interesting view of life. Thanks to Dale Miller, for always having wise advice to offer when I needed it, and thanks as well to Catuscia and Dale's kids, Nadia and Alexis, for the good moments shared. I am also grateful to Christele Braun, Ehab El Salamouny, J\'eremy Dubrueil, Jesus Aranda, Lili Xu, Luis Pino, Marco Giunti, Marco Stronati, Nicol\'as Bordenabe, Raluca Diaconu, Romain Beauxis, and Sylvain Pradalier, who, even if I did not have the opportunity to work with them directly, helped make LIX such a great environment.

I would like to thank my co-authors, with whom I have had the opportunity not only to cooperate scientifically, but also to create friendships. Thanks to Miguel E. Andr\'es for our fruitful collaboration, the constant good mood, and the always stimulating \qm{joke-fights}. Thanks to Konstantinos (Kostas) Chatzikokolakis for all the work we developed together, the enlightening discussions about so many subjects, and the good moments shared. Thanks to Pierpaolo Degano, with whom I had the pleasure of collaborating and learning from.

The team of administrative support at LIX was also fundamental for my work. I would like to thank Marie-Jeanne Gaffard for her remarkable competence and dedication, which have frequently saved me from a great deal of trouble. Her professional behavior is a model to be followed, and I wish I could encounter people like her everywhere I will ever work. I am also grateful to Val\'erie Lecomte, for the countless times she helped me, even when it was not her duty, always with the characteristic competence and sympathy. I cannot forget Corinne Poulain, who guided me through the endless administrative maze when I arrived in France. Thanks also to James R\'egis for the technical support; and to Isabelle Biercewicz and Lydie Fontaine for the assistance in my first years at LIX. I would also like to say a couple of words about Ryna Lam Pech, whose cheerful smile and always good mood made each coffee time in the cafeteria an even more enjoyable moment.

I am also grateful to the experienced scientists who have shared part of their vast knowledge with me, either in conferences, workshops or informal meetings, and reinforced my view that people in academia are not only brilliant, but usually good human beings as well. Thanks especially to Geoffrey Smith for sharing his expertise with me in so many insightful conversations, and for organizing the exciting workshop on information flow at Florida International University. I will not forget the hospitality he, his wife Elena, his sons Daniel and David, and the adorable Yoshi offered in Miami. Thanks also to Prakash Pananganden for the lectures at the SFM-10:QAPL summer school in Bertinoro, and also for the opportunity to participate in the workshop on quantum and classic information flow at the Bellairs Research Institute. 

I cannot proceed without mentioning all the amazing friends I have in Brazil, who were fundamental in the background that brought me here. Even being far away, they are constantly in my mind, and I always count down the days to the next time I will see them again. Thanks to Aline Miranda, whom I have had the privilege of knowing and whose friendship I enjoy very much; to Aline Resende, an incredible friend, on whom I know I can always count on at any time of day or night, and with whom I have had some of the most joyful and memorable moments of all my life; to An\'isio Lacerda, the talented and sensible guy whom I always enjoyed talking to about any subject (serious or not); to Deznie Lopes, who always has a smile to offer; to Katia Lage, the sweet and kind friend who is always there to help others; to Lara Coelho, the funny and practical girl, whose visit to Paris was one of the highlights of my time in the city; and to Marina Cruz, my childhood friend, the one I have known for the longest in my life and whose love always warms up my heart. I am also deeply grateful to Adriani Quatrini, who played such an important role of support and understanding during one of the darkest moments in my last three years; and to Giselle Moura, who has cared so much for me and was the main force driving the process that literally changed my face and, therefore, my life (for much better).

I had never left Brazil at all until the day I moved to Paris, and when I arrived in Europe I did not know a single person on this side of the Atlantic Ocean. It was a big turning point in my life, and I am so glad that I decided to come, for these three years in Paris were not only a period of professional growth, but also of incredible personal learning. I have had the pleasure of meeting here some of the most remarkable human beings I have ever met, both at the professional and personal levels. In particular, our \qm{sweet, sweet Maunoury}, the building shared as home by so many foreign students at \'Ecole Polytechnique, has been the stage of countless adventures, memorable moments, and deep learning. Without the companionship of the people I met there, I would not have been able to enjoy my stay in France as much, and therefore my work would not have been as productive. I would like to thank each and every one of the people I met in Maunoury for the friendship that has marked me so deeply. Also, I want to thank each one for particular things that I will keep in my memory forever. Thanks to Saddaf Shabbir for all the philosophical discussions by the lake during summer (or until late night otherwise), that have enlightened me so much in so many subjects; to Andreas Engelhardt for the constant companionship and mutual-understanding which have so many times lightened the weight of being abroad; to Nadia Vertti for the happiness and cheerfulness that could always make me smile at any time; to Keesjan de Vries for all the awesome trips shared together (\emph{Do you wanna know why? Well...}); to Ricardo Kawahara for sharing the fun of nights out, and also the frustration of the way back home by the Noctilien 122; to Micha\l{} Zydor for the uncountable movies seen together in Paris; to Fabien Immler for being my \qm{German little brother}; to Alex Rinke for the hospitality during the winter holidays in Berlin in 2009/2010; to Oliver Valencia for the fun moments at B\^obar; to Kalle Backlund, Anna Folke Larsen, and Uli Schneider for all the unforgettable evenings at their place in Rue Guisarde and at Chez Georges; to Steffen Lohrey and Marie Le Mouel for the nice evenings watching Audrey Hepburn movies in my room; to Chiara Altomare, Manuele Aufiero, Paolo Carozzo and Lorenzo Sponza (the \qm{Italian mafia}) for the constant cheerfulness in our beloved international kitchen; to Benjamin Mosk for the energy to never say no to a night out dancing; to Maria Rosario (Charo) Mestre for the company not only in Paris, but also in Frankfurt; to \'Alvaro Izquierdo for the constant company in the gym, and the fun trips together; to Amy Gilson, Anton Karrman, Davi Vasconcellos, Leland Ellison, Lysandra Alves, and Michael Martin for the unforgettable Summer of 2009; to Citlali Cabrera for her kindness in every moment, and the nice dinners she offered to me; to Igor Reshetnyak for always being ready to help in anything; to Th\'eo Touvet for the rare example of confident and unique life choices; to Tom\'as Lungenstrass for the constant smile and good mood; and to Fran\c{c}ois Wirion and Julia Duras for the first moments shared in the doctoral program. Thanks also to Alex Lang, Alfredo Parra, Daniel Ruiz, Federico C\'ardenas, Benjamin Uekermann, Fredrik Hallgren, Henri de Belsunce, Herbert Mangesius, Ivan Moschevitin, Joe Gault, Nikita Kazarinov, Pedro Vit\'oria, Przemys\l{}aw Chojecki, Sara Rome, and Seydou Traor\'e for all the unforgettable moments. I cannot forget Hannah Schneider and Sofia Karlsson, who have not lived in Maunoury but are part of the family, and I would like to thank them for the friendship and hospitality when I visited both Cologne and Stockholm. 

It was not only on campus, however, that I met friends. Among the many amazing people I met in Paris, and all over the world, are Alexandra Silva, good company in several conferences and summer schools, whom I hope to meet often, both as a friend and as a colleague; Diogo Arbigaus, the kind and good friend who, even though he is Brazilian, I have met only in Paris; Maria Poulaki, whose refreshing company and kindness always make me feel good; Nicol\'as Lopez and all the Spanish crowd, whose parties in Rue Souflot will be always in my memory; and Izabel Rezende, a family member away from home, who was an essential and kind support during my stay in France.  

I often say that we do not have much control over our lives, and that the best we can do is to try to be prepared enough to catch a good opportunity when it shows up. Today I can look back and be glad to say that I caught at least two life-time opportunities in the past three years. The first one was on the $1^{st}$ of October 2008, when I landed in Paris to start my doctoral program at \'Ecole Polytechnique. The second one was on the $19^{th}$ of March 2010, when I met Trevor Ray Tisler. Meeting him was a turning point in my life, and his emotional support has paved the road so I could work with a lighter spirit. I am grateful for the patience with which he has revised my English writing so many times, the dedication he has shown to me even being overseas for over a year now, and for his love, support and presence in my life.

Finally, I would like to thank my family, of whom I am so proud, for the love and support during my whole life, and especially during the challenges these past three years have imposed on me. Thanks to my mother, Maria Ang\'elica, who has always been a model human being for me, as a strong yet sweet woman, and who gives me strength in hard moments and shares my joy in the good ones; to my brother Marco Ant\^onio, who has set an example for me with his dedication, ethical behavior and kindness that are a constant in everything he does; to my brother Marcus Vin\'icius, whose particular sense of humor and \qm{tough} behavior are not enough to hide a kind heart and a person one can always count on; to my step-father Mario Montoya, who is a remarkable human being, and who has given me more support, understanding and love than my biological father has ever done; to my sisters in law Luciana Salom\~ao and D\'ebora Pires, for being like real sisters, and for the countless joyful moments shared; and to my cousin Adriana de Lima, for always being by my side and supporting me.

I apologize to the several people that played an important role in my way and who have not found their name mentioned here: I am sorry if my memory played a trick on me.
\begin{flushright}
	\vspace{-0.23cm}
	\emph{M\'ario S. Alvim\\ Paris, December 2011}
\end{flushright}

\newpage
\thispagestyle{empty}
\begin{abstract}
	In this thesis we consider the problem of information hiding in the scenarios of interactive systems, statistical disclosure control, and refinement of specifications. We apply quantitative approaches to information flow in the first two cases, and we propose improvements for the usual solutions based on process equivalences for the third case.

In the first scenario we consider the problem of \emph{defining the information leakage in interactive systems} where secrets and observables can alternate during the computation and influence each other. We show that the information-theoretic approach which interprets such systems as (simple) noisy channels is not valid. The principle can be recovered, however, if we consider channels of a more complicated kind, that in information theory are known as channels with memory and feedback. We show that there is a complete correspondence between interactive systems and these channels, and we propose the use of directed information from input to output as the real measure of leakage in interactive systems. We also show that our model is a proper extension of the classical one, i.e. in the absence of interactivity the model of channels with memory and feedback collapses into the model of memoryless channels without feedback.

In the second scenario we consider the problem of \emph{statistical disclosure control}, which concerns how to reveal accurate statistics about a set of respondents while preserving the privacy of individuals. We focus on the concept of \emph{differential privacy}, a notion that has become very popular in the database community. Roughly, the idea is that a randomized query mechanism provides sufficient privacy protection if the ratio between the probabilities that two adjacent datasets give \revision{a certain} answer is bound by a constant. We observe the similarity of this goal with the main concern in the field of information flow, namely limiting the possibility of inferring the secret information from the observables. We show how to model the query system in terms of an information-theoretic channel, and we compare the notion of differential privacy with that of min-entropy leakage. We show that differential privacy implies a bound on the min-entropy leakage, and we also consider the utility of the randomization mechanism, which represents how close the randomized answers are, in average, to the real ones. Finally we show that the notion of differential privacy implies a tight bound on utility, and we propose a method that under certain conditions builds an optimal randomization mechanism.

Moving the focus away from quantitative approaches, in the third scenario we address the problem of using \emph{process equivalences to characterize information-hiding properties} (for instance secrecy, anonymity and non-interference). In \revision{the} literature, some works have used this approach, based on the principle that a protocol $P$ with a variable $x$ satisfies such property if and only if, for every pair of secrets $s_1$ and $s_2$, $P[ ^{s_1} / _x]$ is equivalent to $P[ ^{s_2} / _x]$. We show that, in the presence of nondeterminism, the above principle \revision{may rely} on the assumption that the scheduler \qm{works for the benefit of the protocol}, and this is usually not a safe assumption. Non-safe equivalences, in this sense, include complete-trace equivalence and bisimulation. This problem arises naturally when \emph{refining a specification into an implementation}, since usually the former is more abstract than the latter, and the refinement process involves reducing the nondeterminism. The scheduler is, in this sense, a final product of the refinement process, after all the nondeterminism is ruled out. We present a formalism in which we can specify admissible schedulers and, correspondingly, safe versions of complete-trace equivalence and bisimulation. We prove that safe bisimulation is still a congruence. Finally, we show that safe equivalences can be used to establish information-hiding properties.
\end{abstract}
\thispagestyle{empty}

\mainmatter
\setsecnumdepth{subsection}

\chapterstyle{demo}
\pagestyle{ruled}

\chapter{Introduction}
\label{chapter:introduction}
\mscite{There are two mistakes one can make along the road to truth: \\ not going all the way, and not starting.}{Gautama Siddharta}
\section{Information hiding}
\label{section:information-hiding}

In the last few decades the amount of information flowing through computational systems has increased dramatically. Never before in history has a society been so dependent on such a huge amount of information being generated, transmitted and processed. It is expected that this \revision{solid} trend of increase will continue in the near future, if not virtually indefinitely, reinforcing the need for efficient and safe ways to cope with this reality.

Although the efficient and broad dissemination of information is a goal in many situations, there are instances where the disclosure of information is undesirable or even unacceptable. The field of \emph{information hiding} concerns the problem of guaranteeing that part of the information relative to an event is kept secret. In computer science, the term information hiding encompasses a large spectrum of fields. Different fields have distinct historical motivations and the resulting research followed a unique path. The variation of the subfields of information hiding depends on three main factors: (i) \emph{what} one wants to keep secret; (ii) from \emph{which adversary or attacker} does one want to keep it secret; and (iii) \emph{how powerful} the adversary or attacker is.

The field of \emph{confidentiality} (or \emph{secrecy}) refers to the problem of keeping an action secret. One application of confidentiality \revision{is} \emph{cryptographic protocols}, where the sender and the receiver of a message can be known, but the contents of the message itself \revision{are} considered to be sensitive information. Generally, we can say that confidentiality concerns \emph{data}, while the field of \emph{privacy} concerns \emph{people's personal information}. When dealing with privacy, we may be interested in protecting the information about someone (a credit card number, for instance) or the person's identity itself. \emph{Anonymity} is the field that concerns the protection of the identities of agents involved in events. In principle, anonymity can be related to both the \emph{active agent} (often the \emph{sender} of a message), or to the \emph{passive agent} (often the \emph{receiver} of a message). For instance, in the case of a journalist receiving information from a confidential source, the identity of the sender is intended to be secret. As for the case of an intelligence agency sending a coded message to a spy, the identity of the receiver is confidential information. There is yet another kind of anonymity, sometimes referred to as \emph{unlinkability}, where the identity of agents and actions performed are public information, but the linkage between agents and the actions performed should not be determined. One example of unlinkability is a confidential voting system, where both the voters and the final vote count are in the public domain, but the relationship between the voters' identities and the ballots cast is protected.

One application of privacy that has drawn a lot of attention in recent years is the problem of statistical databases. A statistic is a quantity computed from a sample, and the goal of \emph{statistical disclosure control} is to enable the user of the database to learn properties of the population as a whole, while maintaining the privacy of individuals in the sample. The field of statistical databases highlights the delicate equilibrium between the benefits and the drawbacks of the spread of information. A practical example occurs in medical research, where it is desirable that a great number of individuals agree to \revision{give} their personal medical information. With the information acquired, researchers or public authorities can calculate a series of statistics from the sample (such as the average age of people with a particular condition) and decide, say, how much money the health care system should spend next year in the treatment of a specific disease. It is in the interest of each individual, however, that her participation in the sample will not harm her privacy. In our example, the individuals usually do not want to have disclosed their specific status with relation to a given disease, not even to the users querying the database. Some studies, e.g. \cite{Joison:01:EJSP}, suggest that when individuals are guaranteed anonymity and privacy they tend to be more cooperative in \revision{giving} personal information. 

Another important field of information hiding is \emph{information flow}, which concerns the leakage of classified information via public outputs in programs and systems. Consider a system that asks the users a password to grant their access to some functionality. Naturally, the password itself is intended to be secret, however an attacker trying to guess it will always get an observable reaction from the system, whether the response is an acceptance or a rejection of the entered code. In either case, the observable behavior of the system reveals some information about the password, because even if it is not guessed correctly, at least the search space is narrowed (even if, in this case, only slightly).

It is important to note that the subdivisions of information hiding are not mutually exclusive. In a system where public outputs can reveal the identity of agents, for instance, both the problems of information flow and of anonymity are present. The classification is usually based more on the contextual motivation for the problem than on a rigid taxonomy of subfields. In fact, in recent years there has been an active line of research exploring the similarities between problems such as the foundations of anonymity and information flow, and also privacy and information flow. The result has been an increasing convergence between these fields. In this thesis we explore the similarities between information flow, statistical databases, and anonymity.

In a broader context, the importance of information hiding goes far beyond the realm of computer science, and there are a lot of subtle questions that need to be considered carefully. From a political and even philosophical perspective, the unrestricted use of privacy protection can be controversial. Even though it is broadly accepted that people should have the right to exchange e-mails privately, to vote in democratic elections anonymously, and to express their ideas on the Internet freely, there are situations where information protection policies can be argued to have serious drawbacks. The same mechanism that grants a political activist anonymity and free speech on the Internet, while living under a repressive government, also grants a pedophile anonymity to broadcast harmful material. This balance between freedom and control in the virtual media has been the subject of passionate discussion. Independently of whether one's goal is to maximize or to minimize the degree of information protection in a given situation, it is anyway desirable to measure \emph{the extent to which} the information is protected, to define which specific \emph{definition of protection} the information falls under, and \emph{from whom} the information is protected.

In this thesis we avoid the controversy of deciding in which cases the application and extent of information hiding methods are justifiable. Rather, our focus is on measuring the degree of information protection offered by a system, thus making evaluation and comparison of different systems possible . Specifically, we are interested in using concepts of information theory to quantify the leakage of information.
 
\section{Qualitative and quantitative approaches to information hiding: a brief history} 
\label{section:history}

Historically, the research on information hiding has evolved from the simple but imprecise \emph{qualitative approach} toward the more refined, but at the same time more complex, \emph{quantitative approach}. In the following sections we will briefly overview both. We do not intend to provide here an exhaustive study of the subject, but rather to highlight some of the most important contributions of each of these lines of research to the field of information hiding. 

\subsection{The qualitative approach}

The qualitative approach emerged first in the literature of information hiding. The central idea is that, by observing the output of a system, the adversary cannot be completely sure of what the secret information is. The \emph{principle of confusion} says that for every observable output generated by a secret input, there is another secret that could also have generated the same output. In anonymity, for instance, this corresponds to the concept of \emph{possible innocence}, i.e. the impossibility of identifying the culprit with certainty by only observing the system's output. The principle of confusion does not take into consideration the adversary's certainty \revision{about} the value of the secret: it is enough that there be an alternative hypothesis, no matter how unlikely it is. This is also known as the \emph{possibilistic approach}.

One of the first developments in this field dates from 1976, when Bell and La Padula defined the model of \emph{multilevel security systems} \cite{Bell:76:Misc}. In this model the components of a system are classified as either \emph{subjects}, i.e. active entities such as users or processes, or as \emph{objects}, i.e. passive entities such as files. The subjects are divided into \emph{trusted} and \emph{untrusted} entities, and the authors define restrictions on how to manage untrusted objects. The rule \revision{\qm{no read up or write down}} states that untrusted entities can read only from objects of the same or lower levels, and that they can only write into objects of the same or higher levels. This model was developed to support different levels of security, and aimed to ensure that information only flows from lower to higher levels and never in the opposite direction. Each input into and output from the system is labeled with a security level. Any pair of an input and its \revision{corresponding} output is called an \emph{event}. A \emph{view} of a security level $l$ corresponds to the events at level $l$ or lower, and all the events of a higher level are \emph{hidden} to level $l$.

Usually in this model only two levels are distinguished: \emph{high} and \emph{low}. The high level corresponds to sensitive information, which should only be available to some users with special privileges, while the low level corresponds to public information accessible to everyone. The goal of \emph{secure information flow analysis} is, in this context, to avoid leakage from the high level to the low level.

Bell and La Padula's model, however, did not address the problem of leakage of information due to \emph{covert channels}. A covert channel is a way of transmitting information from the high to the low environment by means not designed or intended for this purpose. Consider, for instance, a system where a low user $\ell$ can send a file to a high user $h$, and $h$ has the power to redefine the access rights to the file. The user $h$ can either maintain the permission of $\ell$ to write in the file, or she can change the policy so $\ell$ no longer has access to it. In this scenario, a covert channel between a corrupted high user $h$ and low user $\ell$ can be established as follows. The low user sends a file to the high user, who then uses her power of deciding whether to grant or to deny $\ell$ further access to it to encode a message. In a later stage, $\ell$ tries to write in the file, and an access failure can be interpreted as the bit $0$, while a success can be interpreted as the bit $1$. In this way any message can eventually be sent through the covert channel from the corrupted high user to the low one.

To cope with the threat of covert channels, Goguen and Meseguer developed the concept of \emph{noninterference}\cite{Goguen:82:SAP}. A system is \emph{noninterfering} when the actions of high users do not alter what can be seen by low users. In other words, the low outputs of the system will only reflect the values of the low inputs, independently of what the high inputs are (if any). The authors proposed a model of noninterference that separated the system from the security policies. Their model, nevertheless, was only appropriate for deterministic systems. 

Noninterference, however, may be a too restrictive concept for several practical applications. It does not allow, for instance, the \emph{summarization of data}. It is often the case where a system allows statistical (or summarizing) functions (e.g. mean, total number) to be calculated on its high inputs and then disclosed to low users, even if the high inputs themselves are supposed to be kept secret. These systems are typical in the area of statistical databases, and we will discuss this issue in more detail in Section~\ref{section:privacy-db}. Clearly, a system that allows the summarization of high data for the low environment violates noninterference, since a change on the high input may affect the low output.

Considering this problem, in 1986 Sutherland \cite{Sutherland:86:NCSC} proposed the concept of \emph{nondeducibility on inputs}, which focuses not on whether the output is affected according to a change in the input, but on whether it is possible \emph{to deduce} the input from the output. Under this definition, a system may allow summarization of data and still be secure, since the output of a statistical function does not necessarily allow the adversary to deduce what the inputs are. One drawback of the concept of nondeducibility on inputs is that it assumes that the strongest form of the principle of confusion is enough to ensure security. Notably, it relies on the assumption that \qm{no high value can be ruled out after observing a low value}. This is not a strong enough security guarantee in many real systems. In some cases, even if no high value can be ruled out as a possibility, a single value (or a small set of values) can be much more likely than the others, and in practice it makes little sense to consider the alternatives. This criticism can be seen as an early attempt \revision{to} consider a quantitative approach for information flow, where it is taken into consideration \qm{how much} an attacker learns (or does not learn) about the secret matters.

Another important issue in security systems is the problem of \emph{compositionality}. In \cite{McCullough:87:SAP}, McCullough pointed out the importance of \emph{hook-up security}, i.e. the compositionality of multi-user systems. Usually, real systems are far too complex to be analyzed as a whole, especially because the task of designing and implementing a system is normally divided between teams. Each team is responsible for a number of components that, in a later stage, will be put to work together. It is highly desirable that security properties be verified in each component separately, and that this verification guarantee that the final composite system is also secure. McCullough showed that the concepts of multilevel security systems, noninterference, and nondeducibility on inputs are not composable. As a replacement, he proposed the concept of \emph{restrictiveness}, according to which no high level information should affect the \emph{behavior} of the system, as seen by a low user. 

In \cite{Wittbold:90:SAP} Wittbold \revision{and Johnson} addressed the question of nondeducibility on inputs under a different perspective, showing that it is not a guarantee of absence of leakage. Consider the following algorithm, where $H$ and $L$ stand for the high and the low environments, respectively. Here we assume the variables $x$ and $y$ are binary, and the randomized command $x \gets 0 \oplus_{0.5} 1$ assigns to $x$ either the value $0$ or the value $1$ with $0.5$ probability each.
\ \\
\begin{algorithmic}
	\WHILE{true} 
		\STATE $x \gets 0 \oplus_{0.5} 1$;
		\STATE output $x$ to $H$;
			\STATE input $y$ from $H$;
		\STATE output ($x$ XOR $y$) to $L$;
	\ENDWHILE
\end{algorithmic}
\ \\
In the above algorithm, the low environment only has access to the value ($x$ XOR $y$). Note, however, that the high environment $H$ learns the value of $x$ before having to choose the value of $y$, and therefore it can use this knowledge to encode a message: To transmit the bit $0$, $H$ chooses $y = x$, and to transmit the bit $1$, $H$ chooses $y = 1 - x$. It is clear that there is some flow of information from the high to the low environment, even though $L$ cannot deduce the high input $y$ from the low output ($x$ XOR $y$). Hence, satisfying nondeducibility on inputs does not guarantee a system to be secure. Wittbold \revision{and Johnson} defined, then, the concept of \emph{nondeducibility on strategies}, which means that regardless of what view $L$ has of the machine, no strategy is excluded from being used by $H$.

\subsection{The quantitative approach}

The qualitative approach, although simple and easy to apply, does not reflect reality in many practical situations. In many cases some information leakage is tolerable or even intentional. Take an election protocol. After the final vote count is released, there are fewer possible \revision{hypotheses} concerning who voted for whom than the \revision{hypotheses} available before the votes were cast. In this example there is a natural leakage of information, since the uncertainty about the sensitive information decreases after the observation of the protocol's output. This leakage occurs, however, as a necessary functionality of the protocol. 

In fact, in \revision{most real systems} noninterference cannot be achieved, as typical systems will always leak some information. This does not mean, however, that all systems are equally good or bad, because the amount of leakage usually varies from system to system. Therefore it is important to quantify \emph{how much} leakage a system allows. Quantitative methods are useful to evaluate the extent to which a system is secure, and to compare it to other systems. 

One of the first attempts to quantify information leakage was made by Denning in 1982. In \cite{Denning:82:Misc} she defined the leakage from a state $s$ to a state $s'$ as the decrease in uncertainty about the high information in $s$ resulting from the low information in $s'$. She used the concept of conditional entropy\footnote{The concepts of \emph{entropy}, \emph{conditional entropy} and \emph{mutual information} will be defined formally in Chapter~\ref{chapter:probabilistic-info-flow}. For the moment it is enough to know that \emph{entropy} is a measure of the uncertainty of a random variable; \emph{conditional entropy} is a measure of the uncertainty of one random variable given another random variable; and \emph{mutual information} is a measure of how much information two random variables share.} $H(h_{s}|\ell_{s'})$, where $h_{s}$ is the high information in $s$ and $\ell_{s'}$ is the low information in $s'$. Her definition of leakage was:
\begin{equation*}
	\revision{M_{1} = H(h_{s}|\ell_{s}) - H(h_{s}|\ell_{s'})}
\end{equation*}

If the quantity $M_{1}$ is positive, then it is considered to be the leakage of information. This measure of leakage, however, does not consider the history of low inputs, a problem pointed out by Clark, Hunt and Malacaria in \cite{Clark:07:JCS}. Without the history one cannot summate the increase in knowledge (or decrease in uncertainty) that accumulates between the low states $s$ and $s'$. They proposed, instead, the following measure of leakage:
\begin{equation*}
	\revision{M_{2} = H(h_{s}|\ell_{s}) - H(h_{s}|\ell_{s'},\ell_{s})}
\end{equation*}

Since $H(X|Y,Z) \leq H(X|Y)$ for all random variables $X$, $Y$ and $Z$, we have $M_{1} \leq M_{2}$. The quantity $M_{2}$ corresponds to the Shannon conditional mutual information $I(h_{s};\ell_{s'}|\ell_{s})$.

In 1987, Millen made a formal connection between information flow and Shannon information theory by relating noninterference and mutual information \cite{Millen:87:SAP}. In Millen's model, a computer system is seen as a channel whose input is a sequence $W$, possibly generated by a set of users, and whose output (after the computation is completed) is $Y$. The random variable $X$ represents a subsequence of $W$ generated by a user $U$, while $\overline{X}$ represents the high inputs generated by users other than $U$. Millen showed that in deterministic systems if $X$ and $\overline{X}$ are independent and $X$ is not interfering with $Y$, then the Shannon mutual information $I(X;Y)$ between $X$ and $Y$ is zero. In other words, noninterference is a sufficient condition for absence of information flow.

In 1990, Massey gave an important contribution to the field of information theory, which influenced the further development of quantitative information flow. In~\cite{Massey:90:SITA} he showed that the usual definition of discrete memoryless (i.e. history-independent) channels used at that time in fact did not take into account the possibility for the use of feedback. He highlighted the conceptual difference between causality and statistical dependence, and presented an accurate mathematical description of discrete memoryless channels that allowed feedback. Then he introduced the concept of \emph{directed information}, which captures the idea of causality between the input and the output of a channel, and argued that in the presence of feedback, directed information is a more appropriate measure of the flow of information from input to output than mutual information.

In the same year, McLean also considered the concept of time in the description of systems by proposing his \emph{Flow Model} \cite{McLean:90:SSP}. According to this model, there is a flow of information only when a high user $H$ assigns values to objects in a state that precedes the state in which a low user $L$ makes her assignment. In this situation only part of the correlation between high and low information is considered as leakage. This addressed the problem of causality, but this model was too general, and relatively difficult to apply.

In \cite{Gray:91:SSP} Gray worked on bridging the gap between the overly complicated Flow Model and the more practical, yet restricted, approach of Millen. Gray used a general-purpose probabilistic (as opposed to nondeterministic) state machine that resembled Millen's model. In Gray's model, the value \revision{$\mathcal{T}(s,I,s',O)$} represents the probability of a given state $s$ evolving into another state $s'$, under the input $I$, and producing output $O$. The channels are partitioned into two sets, $H$ and $L$, representing the channels connected to high and low processes, respectively. The high and the low environments can communicate only through their interactions with the system, as no other form of communication between them is allowed. Gray wanted to take time and causality into consideration in his definition of leakage, and he did so by allowing feedback and memory in his model. His formulation of a security guarantee was the following:
\begin{equation}
	\label{eq:gray-flow}
	\begin{split}
		P(L^{I} \cap L^{O} \cap H^{I} \cap H^{O}) > 0 \quad \quad \implies \\
		P(\ell | L^{I} \cap L^{O} \cap H^{I} \cap H^{O}) = P(\ell | L^{I} \cap L^{O})		
	\end{split}
\end{equation}

\noindent where $L^{I}$ and $L^{O}$ represent the history of low inputs and outputs, respectively, and $H^{I}$ and $H^{O}$ represent the history of high inputs and outputs, respectively. The symbol $\ell$ represents the final output event channels in the low environment. The formulation \eqref{eq:gray-flow} states that the probability of a low output may depend on the previous history of the low environment, but not on the previous history of the high environment.

Gray also tried to generalize the concept of capacity to the case of channels with memory and feedback. He provided a formula expressing the flow of information from the whole history of inputs and outputs (\revision{during the time period $0 \ldots t-1$}) to the the low output (at time $t$), and conjectured that the capacity of the channel would be:
\begin{equation}
	\label{eq:gray-capacity-1}
	C \defsym \lim_{n \rightarrow \infty} C_n
\end{equation}

\noindent where
\begin{equation}
	\label{eq:gray-capacity-2}
		\begin{aligned}
			C_{n} \defsym & \max_{H,L} \frac{1}{n} \sum_{i=1}^{n} I( \mathit{In\_Seq\_Event}_{H,t}, \mathit{Out\_Seq\_Event}_{H,t} ; \\
			 & \mathit{Final\_Out\_Event}_{L,t} | \mathit{In\_Seq\_Event}_{L,t}, \mathit{Out\_Seq\_Event}_{L,t} ) \\
		\end{aligned}
\end{equation}

\noindent and $\mathit{In\_Seq\_Event}_{A,t}$ is the input history at channel $A$ (where $A$ stands for $L$ or $H$) up to time $t-1$, $\mathit{Out\_Seq\_Event}_{A,t}$ is the output history at channel $A$ up to time $t-1$, and $\mathit{Final\_Out\_Event}_{L,t}$ is the low output event at time $t$. Gray showed that the absence of information flow implies that \revision{capacity as formulated in} \eqref{eq:gray-capacity-1} is zero. He also conjectured that this definition of capacity would correspond to the notion of maximum transmission rate supported by the channel. As pointed out in \cite{Alvim:11:JCS}, however, the problem with Gray's conjecture is the following. For an output at time $t$, the only causal relation considered is the one with the history of inputs up to time $t-1$, while the effect that the input at time $t$ itself may have on the output is ignored. In this way, \eqref{eq:gray-capacity-1} does not express the complete causal relation between input and output. The correct notion of capacity in the presence of memory and feedback, which corresponds to the maximum transmission rate for the channel, was proposed in 2009 by Tatikonda and Mitter \cite{Tatikonda:09:TIT}, and it will be discussed later on in Chapter~\ref{chapter:interactive-systems}.

A similar formal approach, although with different motivations, was presented by McIver and Morgan in \cite{McIver:03:book}. They focused on the problem of preserving security guarantees while refining specifications into implementations. The authors used an equation similar to \eqref{eq:gray-capacity-2}, but in the context of sequential programing languages enriched with probabilities. Their aim was to protect the high values during the whole execution of the program, instead of the initial high values only. In other words, they wanted to assure that if the high information is not known by the low environment at the beginning of the computation, then it cannot be inferred at any later stage. They proved that, for deterministic programs, if the final values of the high objects are protected, then the initial values are protected as well. McIver and Morgan also defined the concept of \emph{information escape} as:
\begin{equation*}
	H(h|\ell) - H(h'|\ell')
\end{equation*}

\noindent where $H(h|\ell)$ represents the uncertainty (conditional entropy) of the high information given the low information at the beginning of the computation, and $H(h'|\ell')$ represents the same uncertainty at the end of the computation. They defined the channel capacity as the least upper bound of information escape over all possible input distributions. In this context, a system is considered secure if it has capacity equal to zero. One advantage of this model is that it is not necessary to keep track of the whole history of the computation, but on the other hand it can be applied only in scenarios where the adversary does not have memory. 

In Chapter~\ref{chapter:probabilistic-info-flow} we will take up again the discussion of quantitative approaches to information flow based on information theory. For the moment we will focus on some topics related to information hiding that are of special relevance for this thesis.

\section{Case studies of information hiding}
\label{section:case-studies}

In this section we present three case studies of information hiding that we address in this thesis.

\begin{enumerate}
	\item The case of \emph{quantitative information flow}, i.e. how much about the secret information an adversary can learn by observing the system's output, and by knowing how the system works. We give special attention to the broadly studied problem of anonymity, which can be seen as \revision{a} particular case of the more general problem of information flow where the secret information is the identity of the agents.
	
	\item The question of \emph{statistical disclosure control}, which concerns the problem of allowing users of a database to obtain meaningful answers to statistical queries, while protecting the privacy of the individuals participating in the database. We focus on differential privacy, an approach to this problem that has drawn a lot of attention in recent years.
	
	\item The problem of \emph{preserving security guarantees \revision{while} deriving implementations from specifications}. Usually specifications are more abstract than implementations, i.e. they present more nondeterminism. The task of implementing a system reduces the nondeterminism of the specification, and if it is not done carefully, an implementation may rule out possibilities allowed by specification that are essential for the security guarantees. 
\end{enumerate}
 
\subsection{Quantitative information flow and anonymity}
\label{section:anonymity}

Anonymity is one of the most studied subjects of information hiding. The research in this area has been active in the past several years, and the advances made can be extended to the more general scenario of information flow. As briefly introduced in Section~\ref{section:information-hiding}, anonymity concerns the protection of the identities of the agents involved in the events. 

With the advent of the Internet, the protection of anonymity has become an issue in the daily life of millions of people around the world. The importance of anonymity is even more evident concerning the protection of freedom of speech, a situation that is particularly delicate in countries under repressive regimes.

Pfitzmann, Dresden and Hansen \cite{Pfitzmann:08:Misc} have proposed a standard terminology for anonymity concepts. In their work there are three different notions of anonymity based on the agents involved:

\begin{itemize}
	\item \emph{Sender anonymity}: when the identity of the originator should be protected;
	\item \emph{Receiver anonymity}: when the identity of the recipient should be protected;
	\item \emph{Unlinkability}: when it might be known that an agent $A$ originated a message and an agent $B$ received a message, yet it should not be known whether the message sent by $A$ was actually the one received by $B$.
\end{itemize}

Reiter and Rubin also gave a classification of the types of adversary in an anonymity system in \cite{Reiter:98:TISS}, where they also proposed the anonymity protocol Crowds (see Section~\ref{section:examples-anonymity-protocols}). In their work, they considered that the adversary can be an eavesdropper simply observing the traffic of messages on the network, or she can be an active attacker (i.e. a collaboration between senders, between receivers, or between others taking part in the system), or even a combination of the previous two types. The authors also defined a hierarchy of anonymity degrees that a system can provide. In decreasing order of strength, the proposed scale is listed below. In this list, let $s, s'$ denote secrets and $o$ an observable, i.e. a particular action or output of the system that is distinguishable from the point of view of the attacker.

\begin{description}
	\item[Strong anonymity] From the attacker's point of view, the observables produced by the system do not increase her knowledge about the secret information, i.e. the identity of the individual involved in an event. Chaum also described the concept of strong anonymity in his work on the Dining Cryptographers protocol \cite{Chaum:88:JC}. It represents the ideal situation where the execution of the protocol does not give to the adversary any extra information about the secrets. The concept was formalized as follows.
\begin{equation}
	\label{eq:strong-anonymity-chaum}
	\forall s,o \quad p(s|o) = p(s)
\end{equation}
	
This definition is the equivalent of \qm{probabilistic noninterference}. In \cite{Chatzikokolakis:06:TCS}, Chatzikokolakis and Palamidessi showed that the condition expressed by \eqref{eq:strong-anonymity-chaum} is equivalent to:
\begin{equation}
	\label{eq:strong-anonymity-cp05}
	\forall s,s',o \quad p(o|s) = p(o|s')
\end{equation}
	
\noindent	i.e. the probability of the system producing an observable is the same, no matter what the secret information is. This definition is known as \emph{equality of likelihoods} and is advantageous as it does not \revision{depend on} the probability distribution on secrets. 

Another definition of strong anonymity, more restrictive, was proposed by Halpern and O'Neill \cite{Halpern:03:CSF,Halpern:05:JLMCS}. It is equivalent to each of the previous definitions (\eqref{eq:strong-anonymity-chaum} or \eqref{eq:strong-anonymity-cp05}) plus the assumption that the input probability is uniform. Halpern and O'Neill focused on the adversary's lack of confidence in her guess about the secret, and defined strong anonymity as:
\begin{equation}
	\label{eq:strong-anonymity-ho}
	\forall s,s',o \quad p(s|o) = p(s'|o)
\end{equation}
	
The formulation \eqref{eq:strong-anonymity-ho} is also known as \emph{conditional anonymity} and corresponds to the level of anonymity called \emph{beyond suspicion} in Reiter and Rubin's classification.
	
	\item[Beyond suspicion] From the attacker's point of view, an agent is no more likely to be the culprit than any other agent in the system. It can be formalized as in \eqref{eq:strong-anonymity-ho}.
	
	\item[Probable innocence] From the attacker's point of view, an agent does not appear more likely to be involved in an event than not to be involved. Formally:
	\begin{equation}
		\label{eq:probable-innocence-1}
		\forall s,o \quad p(s|o) \leq 0.5
	\end{equation}
	
	The formulation \eqref{eq:probable-innocence-1}, however, is not broadly accepted as the definition of probable innocence. In~\cite{Chatzikokolakis:06:TCS}, Chatzikokolakis and Palamidessi showed that the property that Reiter and Rubin indeed proved for the Crowds protocol in~\cite{Reiter:98:TISS} was:
	\begin{equation}
		\label{eq:probable-innocence-2}
		\forall s,o \quad p(o|s) \leq 0.5
	\end{equation}
	
	\item[Possible innocence] From the attacker's point of view, there is always a \revision{non-zero} probability that the agent involved in the event is someone else. Formally:
	\begin{equation*}
		\forall s,o.\left(p(s|o) > 0 \implies \exists s'.p(s'|o) > 0 \right)
	\end{equation*}
	
\end{description}

The above hierarchy gives a richer classification of the degree of protection offered by a system than would be possible with simpler possibilistic models. 

Among the quantitative approaches to anonymity, two are of our special interest: the ones based on information-theoretic concepts and the ones based on the Bayes risk. In the following section we give a brief overview of these two approaches. These concepts will be revisited in more detail in Chapter~\ref{chapter:probabilistic-info-flow}.

\subsubsection{Anonymity protocols as noisy channels} 

Information theoretic approaches to anonymity, and \revision{more generally} to information flow, rely on concepts such as entropy and mutual information to measure the adversary's lack of information about the secret before and after observing the system's output. Typically the system is seen as a noisy channel and the concept of noninterference corresponds to the converse of the channel capacity.

There are several works in the literature that have proposed measures of degrees of anonymity in terms of the entropy and mutual information, for instance \cite{Serjantov:02:PET, Diaz:02:PET, Zhu:05:ICDCS, Deng:06:FAST}. In \cite{Chatzikokolakis:08:IC} Chatzikokolakis, Palamidessi and Pananganden proposed the concept of \emph{conditional capacity} to cope with the situation where some leakage of information is intended by the system. Consider again the election protocol example. By design, the final vote counting needs to be announced and it usually increases the attacker's knowledge about the secret. In this situation, the leakage should be calculated modulo the information that is supposed to be disclosed, i.e. the vote count. In this work the authors also proposed methods to calculate the channel capacity exploiting some symmetries present in several practical systems. 

\subsubsection{Hypothesis testing and Bayes risk} 

In some real world situations an individual faces the following situation: she is interested in the value of some random variable $A \in \Aset$ but she has access only to the values of another random variable $O \in \Oset$. She knows that $A$ and $O$ are correlated by a known conditional probability distribution. This situation occurs in several fields, for instance in medicine (to make a diagnosis, the physician has access to a list of symptoms, but not to the disease itself). The attempt to infer $A$ from $O$ is known as the problem of \emph{hypothesis testing}. Here we are interested in the use of hypothesis testing in the context of anonymity (and information flow). More specifically, the adversary tries to infer the secret $A$ given that she has access to the observables $O$ and she knows how the system works, i.e. how the probabilities of $O$ are conditioned with relation to $A$. 

A commonly studied approach to the problem is based on the Bayesian method and consists of assuming the a priori probability distribution on $A$ as known, and then deriving from that and from the knowledge about how the system works, an a posteriori probability distribution after some fact has been observed. It is well known that the best strategy for the adversary is to apply the MAP rule (Maximum A posteriori Probability rule), which as the name suggests, chooses the hypothesis with the maximum probability for the given observation. Here, by \qm{best} strategy we mean the one that induces the smallest probability of error in guessing the hypothesis, that in this case corresponds to the \emph{Bayes risk}.

In \cite{Chatzikokolakis:08:JCS} Chatzikokolakis, Palamidessi and Pananganden explored the hypothesis testing approach to anonymity, in a scenario where the adversary has one single try to guess the secret (after exactly one observation). They associated the level of anonymity to the probability of error, i.e. the probability of an attacker making a wrong guess about the secret. In order to consider the worst case scenario and to give upper bounds for the level of anonymity provided, the adversary is assumed to use the MAP rule strategy. In this case, the probability of error corresponds to the Bayes risk, and the degree of protection offered by a protocol corresponds to the Bayes risk associated with the channel matrix.

In \cite{Smith:07:TGC,Smith:09:FOSSACS} Smith also considered the scenario of one-try attacks and proposed the notion of \emph{vulnerability}, which takes into consideration the probability that the adversary can guess the secret correctly after observing the behavior of the system only once. Smith proposed the framework of \emph{min-entropy leakage}, which is closely related to the Bayes risk, but is different as it uses the concept of entropy (more precisely min-entropy) and formalizes leakage in information theoretic terms.

In Chapter~\ref{chapter:probabilistic-info-flow} we will present a deeper discussion about the use of information theory for the formalization of information flow, including the notions of Shannon entropy, mutual information and the framework of min-entropy leakage for one-try attacks. First, however, we will review some fundamental anonymity protocols in literature.

\subsubsection{Examples of anonymity protocols}
\label{section:examples-anonymity-protocols}

On the Internet, every computer has a unique IP address which specifies the computer's logical location in the topology of the network. This IP address is usually sent along with any request originating from the computer. Even if the computer uses an IP address for a single session via an ISP (Internet Service Provider), the identification can be logged and retrieved later with the ISP's compliance. One common way to try to preserve anonymity is to use a \emph{proxy}, i.e. an intermediary computer that gathers all the requests of a group of computers and serves as a unique gate for any communication with the world outside of the network. For practical purposes, it is as if all the requests originated from the proxy, and the members of the group are indistinguishable from the point of view of an outside observer. One drawback presented by the use of proxies is that it creates single points of failures, decreasing the network's robustness.

The problem illustrated above is one of the motivations for the use of communication protocols specifically designed to protect anonymity. In this section we review two of the most fundamental, and probably most famous, examples of anonymity protocols in literature: \emph{the dining cryptographers} protocol, and the \emph{Crowds} protocol.

\paragraph{The dining cryptographers}

The dining cryptographers protocol was proposed by Chaum in \cite{Chaum:88:JC}. It is one of the first anonymity protocols in the literature, and it is one of the few protocols that can assure strong anonymity.

The protocol is usually presented in a simplified scenario, where three cryptographers employed by the NSA (The National Security Agency of the United States) are having dinner in a restaurant. At the end of the dinner, the NSA decides whether it will pay the bill itself or whether it will assign the duty of paying to one of the cryptographers at the table. In the case the NSA decides that one of the cryptographers will pay, it announces the decision secretly to the chosen one. The goal of the protocol is to reveal whether one cryptographer will pay the bill or not, without revealing the identity of the payer. In other words, to an external observer (and to the non-paying cryptographers as well), the only accessible information is whether the NSA is paying or not, but not the identity of the cryptographer paying (if any). We assume that the NSA does not disclose its decision to anyone but to the cryptographer it chooses (again, if any), and that the solution should be distributed, i.e. only message passing between agents is allowed, and no centralized agent coordinates the process.

The dining cryptographers protocol solves this problem as shown schematically in Figure~\ref{fig:dc-protocol}. Each cryptographer ($\mathit{Crypt_0}$, $\mathit{Crypt_1}$ and $\mathit{Crypt_2}$) tosses a coin that is visible only to himself and to his right-hand neighbor. In this way every cryptographer has a shared coin with each of the other two. After all three coins ($c_0$, $c_1$ and $c_2$) are tossed, each cryptographer checks whether the two coins visible \revision{to him} agree (both are heads or both are tails) or disagree (one is head and the other is tails). Then they announce publicly \emph{agree} or \emph{disagree}, according to the result they obtained with their coins. The only exception is that, if a cryptographer is paying, he will announce the opposite of what he sees, i.e. he will announce \emph{disagree} in the case that his coins agree and \emph{agree} if they do not. It can be proven that if the number of \emph{disagrees} is even, then the NSA is paying, and if the number of \emph{disagrees} is odd, then one of the cryptographers is paying. Moreover, if the coins are all fair, the protocol offers strong anonymity in the following sense: The execution of the protocol does not provide to an external observer enough evidence to change her knowledge about which cryptographer is the payer, if any. In other words the probability of any cryptographer being the payer, under the adversary's point of view, does not change after the observation of the protocol's execution.

\begin{figure}[!tb]
	\centering
	\includegraphics[width=0.6\columnwidth]{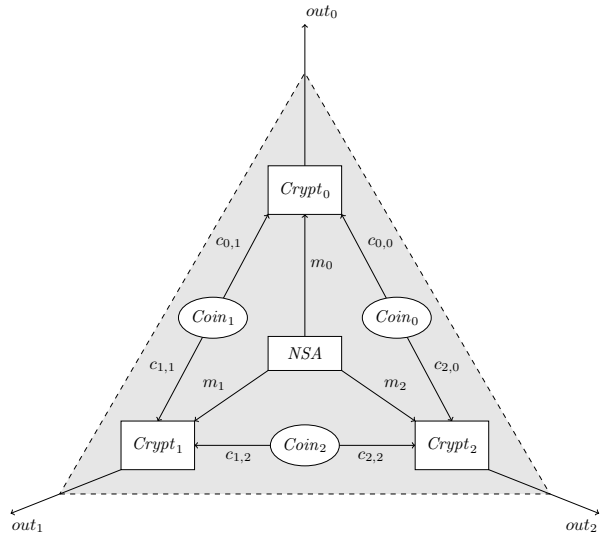}
	\caption{An example of the dining cryptographers protocol}
	\label{fig:dc-protocol}
\end{figure}

The dining cryptographers protocol can be generalized to any number of graph nodes (i.e. cryptographers) and any type of graph connectivity (i.e. the shared coins between pairs of cryptographers). Then the same solution can be used for anonymous communication as follows. Each pair of nodes share a common secret (the value of the coin) of length $n$, equal to the length of the transmitted data. It is assumed that the coins are drawn uniformly from the set of possible secrets. Each node then computes the binary sum (XOR operation) of all its shared secrets and announces the result. The only exception is that the node that wants to transmit adds the datum, also of length $n$, to the sum it announces. It can be shown that the total sum of the announcements of all nodes is equals to the data to be transmitted, since each secret is counted twice (once by each node that can see it) and, therefore, is canceled out by the XOR operation. The protocol works under the assumption that only one node at a time tries to transmit, and if it is the case that more than one sender wants to transmit at the same time, the conflict needs to be solved by some sort of coordinator.


One drawback of the dining cryptographers protocol is its inefficiency: whenever a single node wants to transmit, all the nodes in the graph need to collaborate to make it happen, at the cost of a large number of message exchanges. Moreover, as previously stated, in the case where more than one node wants to transmit at the same time, a coordinator is necessary to solve the conflict.

\paragraph{Crowds}

The Crowds protocol was first presented in \cite{Reiter:98:TISS} and it allows Internet users to perform web transactions without revealing their identity. Usually, on the Internet, when a user communicates with a server the latter can discover the IP address of the originator. The idea behind Crowds is to gather users into a crowd and randomly redirect the request multiple times inside the group before finally letting it reach the server. In this situation, it is impossible for the server, and for any other user, to identify the initiator of the request once it receives the message: whenever someone sends a message there is a considerable probability that she is only a forwarder for someone else.

To be more precise, a \emph{crowd} is a group of $m$ users who participate in the protocol. It is possible that a subgroup of $c$ users are corrupted and collaborate to disclose the identity of the original sender. Also, we assume that the protocol has a parameter $p_{f} \in \left(0,1\right]$. We call \emph{originator} or \emph{initiator} the user who wants to make a request to the server. The originator needs to create a \emph{path} between herself and the server in order to have her request reach the final destination, as shown in Figure~\ref{fig:crowds}.

\begin{figure}[!tb]
	\centering
	\includegraphics[width=0.75\columnwidth]{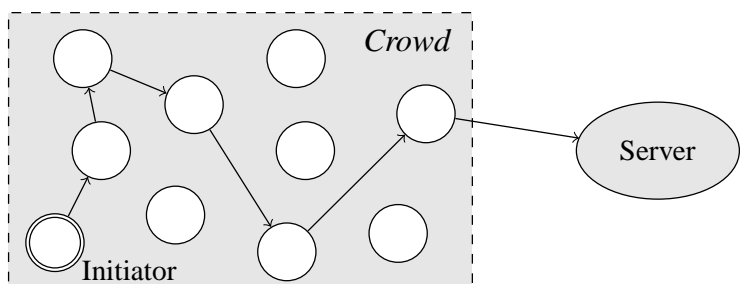}
	\caption{The Crowds protocol at work}
	\label{fig:crowds}
\end{figure}

The protocol works as follows:

\begin{itemize}
	\item At the first step the initiator chooses, according to a uniform probability distribution, another user in the crowd (possibly herself) and forwards the request to this user;
	\item The user who receives the message then makes a random choice. With probability $p_{f}$ she forwards the message to the server, and with probability $1-p_{f}$ she decides to forward the message to some user in the crowd. If this is the case, she chooses a user (possibly herself) according to a uniform probability distribution, and forwards the message to this user. This step is then repeated by the new message holder.
\end{itemize}

The response from the server to the originator follows the same path, in the opposite direction. Moreover, all the communications in a path are encrypted using a \emph{path key}, which protects the path from threats posed by local eavesdroppers. Each user has access to the communications in which she participates, but it is assumed that a user cannot intercept messages exchanged between other users. It can be proven that the protocol is strongly anonymous with respect to the web server. Intuitively this is the case because at least one forward step is always performed, and after this step any user can be the holder of the message with equal probability. Therefore, from the server's point of view any user is equally likely to be the originator of the request.

A more interesting case is to analyze the level of anonymity ensured with respect to a corrupted user. If in the very first step of the execution of the protocol the message is forwarded to a corrupted user, she can gain more information about the possible originator than the server. A user, whether the originator or not, is said to be \emph{detected} if she sends a message to a corrupted user. Since the originator always appears in a path, she is more likely to be detected than the rest of the users. Detecting a user (at least for the first time in a path) increases the probability that this user is the originator. Therefore, strong anonymity cannot hold with relation to corrupted users. 

In \cite{Reiter:98:TISS} it is proven that if the number $c$ of corrupted users is not too large, the protocol can at least ensure the level protection of probable innocence. More precisely, if the number $m$ of users in the crowd satisfies 
\begin{equation*}
	m \geq \frac{p_{f}}{p_{f} - \frac{1}{2}}(c+1)
\end{equation*}

\noindent then the protocol ensures probable innocence in the sense of \eqref{eq:probable-innocence-2}.

\subsection{Statistical disclosure control}
\label{section:privacy-db}

The field of \emph{statistical disclosure control} concerns the problem of revealing accurate statistics about a set of respondents while preserving the privacy of individuals. In statistical databases, the data of a (large) number of participants is compiled, and users are allowed to pose statistical queries (such as average or total counting) about the sample. This kind of database is of special importance in many areas. For instance, medical databases can provide information about how a disease spreads, and a census database can help authorities to decide how to spend the next year's budget. 

The data in a statistical database can be obtained in different ways. It can be collected in a census, for instance, it can be obtained opportunistically by monitoring the traffic in a network, or it can even be \revision{given} by the participants by their own choice. No matter how the data is obtained, however, it is still important to ensure that the individual's participation in the database will not harm her privacy. This is not a trivial goal to achieve: the main purpose of a statistical database, in the first place, is to reveal some information about the population as a whole, i.e. to let users infer \qm{general truths} about this population. As an example, suppose that a statistical database of individuals of a certain country indicates that, in this population, the life expectancy for women is $5$ years longer than for men. Clearly this piece of information reveals something about the whole population, \emph{even about individuals not present in the database}.

There are several approaches to dealing with the problem of preserving privacy in statistical databases. One of them is based on ensuring \emph{large query sets}, i.e. that no query can be posed for a small set of individuals. The problem with this approach is that, even if two query sets are \qm{large enough}, their combination may not be. Consider the following two queries: \qm{How many people have disease $y$?} and \qm{How many people, not named $X$, have disease $y$?}. Both queries operate on large sets, but clearly the superposition of the two queries immediately reveals sensitive information about the individual named $X$. Another attempt to achieve privacy is based on the \emph{encryption of the data} in the dataset. This is not a general solution since, as we have seen, the privacy threats do not concern only the individuals in the database and, therefore, the encryption of the data will not address this issue.

Another possible solution is to apply some sort of \emph{query auditing}: the curator of the database checks whether or not a query is possibly disclosing before deciding to provide an answer to it. This approach would cope with the problem of the two superposing queries mentioned above, yet it presents two serious drawbacks: first, automatic tools to check every query are practically infeasible; and, second, the \revision{refusal} to answer a query can be in itself a disclosing act. Another attempt to deal with the problem is by using \emph{subsampling} of the dataset. We normally view a dataset as a collection of rows, where each row contains the data of an particular participant. The idea of subsampling is to randomly choose a subset of the rows, compute the answer to the query based on this subsample, and then report it as the final answer. If the subset is large enough, it should reflect the statistical properties of the whole database. This approach, however, protects a participant only to the extent to which it is unlikely that she is in the subsample. If being in the subsample has catastrophic results, then someone will always be seriously harmed.

The \emph{input perturbation} approach is based on modifying either the data or the query in hope of confusing the adversary. For instance, a \emph{randomized response} mechanism can be used at the moment \revision{the data} is acquired. This modification is permanent and not even the curator knows what the original data was. The queries to the database are then made taking into consideration the randomized noise.

Yet another approach is to add randomized noise \emph{to the answer of the query}. The idea is to compute the answer on the complete set of (the original) values in the database, and then randomize the response before reporting it to the user. If this is done naively, however, it can easily be taken care of by the adversary. Suppose that the noise is chosen to be a Gaussian additive noise with mean zero. If the query is repeated a sufficient number of times, a statistical analysis of the answers can easily estimate with high accuracy what the real answer is. Even if the curator of the database opts to record the query and always report the same answer for it, it may not solve the problem: syntactically different queries can be semantically equivalent, and if the query language is rich enough the semantic equivalence is undecidable.

In this context, it is clear that the problem of statistical disclosure control is not trivial. Yet another issue to be considered is \emph{auxiliary (or side) information}. Auxiliary information is any piece of data about individuals that the attacker has and that does not come from the database itself. It may originate from priors, beliefs, newspapers or even other databases. Some decades ago, Dalenius~\cite{Dalenius:77:ST} considered the problem of auxiliary information and proposed a famous \qm{ad omnia} privacy desideratum: nothing about an individual should be learnable from the database that could not be learned without access to the database. In other words, if the adversary has some side information and gains some knowledge about the individuals using it, by learning the response from the database this knowledge about individuals should not increase. Dalenius' property is, however, too strong to be useful in practice: Dwork showed in \cite{Dwork:06:ICALP} that no useful database can \revision{satisfy} it. She then proposed the notion of \emph{differential privacy}, which is based on the idea that the presence or absence of an individual in the database, or the individual's particular value, should not significantly change the probability of obtaining a certain answer for a given query \cite{Dwork:06:ICALP,Dwork:10:SODA,Dwork:11:CACM,Dwork:09:STOC}. 

The concept of differential privacy can be formalized as follows. Let $\calx$ be the set of all possible databases, and $\calz$ be the set of possible answers to a query. Two databases $x,x' \in \calx$ are \emph{adjacent} (or \emph{neighbors}), written $x \sim x'$, if they differ in the value of exactly one individual. Then, for some $\epsilon > 0$:

\begin{definition}[\cite{Dwork:11:CACM}]
	\label{def:diff-privacy-intro}
	A randomized function $\mathcal{K}$ from $\calx$ to $\calz$ satisfies {$\epsilon$-differential privacy} if for all pairs $x,x'\in \calx$, with $x\sim x'$, and all $S \subseteq \calz$, we have:
	\begin{equation*}
		\mathit{Pr}[\mathcal{K}(x) \in S] \leq e^{\epsilon} \revision{\cdot} \mathit{Pr}[\mathcal{K}(x') \in S]		
	\end{equation*}	
	
\end{definition}

The concept of differential privacy has had an extraordinary impact in the database community, and we will discuss the meaning and implications of the above formulation in greater depth in Chapter~\ref{chapter:differential-privacy}. For the moment, it is enough to note that this definition intuitively ensures that individuals can opt in or out of the database without significantly changing the probability of any given answer to a query to be reported. In other words, it is \qm{safe} for an individual to join (or to leave) the database. Dwork also showed that in order to ensure differential privacy it is enough to consider a Laplacian mechanism of noise~\cite{Dwork:06:ICALP}.

Although differential privacy is a promising approach to the question of statistical disclosure control, the fact that it relies on the randomization of the query response poses some challenges with respect to the \emph{utility} of the query mechanism. If the noise is not added with sufficient care, the reported answer can be so \qm{different} from the real answer that the informative purpose of the database is compromised. In Chapter~\ref{chapter:differential-privacy} we will come back to the question of how to apply differential privacy and, at the same time, provide maximum utility to the query mechanism.

\subsection{Refining specifications into implementations}
\label{section:spec-implementations}

Deriving implementations of a system given its specification, while respecting security constraints, is a challenging problem in information hiding and, \revision{more generally}, in security. A specification $S$ is refined by an implementation $P$ if $P$ preserves all \revision{logically} expressible properties of $S$. One needs to be careful, however, when refining a specification in the realm of information hiding. According to Morgan \cite{Morgan:09:SCP}: 

\begin{quote} 
	A rigorous definition of how specifications relate to implementations, as part of reasoning, must ensure that implementations reveal no more than their specifications: they must, in effect, preserve ignorance.
\end{quote}

By \qm{ignorance}, the author means what the user does not know about what she cannot see. This notion is closely related to the problem of information flow, i.e. determining how much about the secret behavior of a system an adversary can infer from an observation and her knowledge about how the system works. 
	
To illustrate the problem, we will discuss the following example, adapted from the original one in \cite{Morgan:09:SCP}. Consider a partition of the program states into \emph{visible} ($v$) and \emph{hidden} ($h$). Assume that the two variables $v$ and $h$ have the same domain $\mathbb{N}$ (the natural numbers), and in a specification $S$, after the value of $h$ is assigned, the following is stated: \emph{choose $v$ from the domain $\mathbb{N}$}. Then we can ask \qm{from the final value of $v$, what can the observer deduce about the value of $h$, given that she knows how the system works?}. Of course the answer will depend on how the implementation $I$ of the specification is done. If $I$ is simply $v := 0$, then nothing is learned, since what the user knows about the value of $h$ is exactly what she already knew before. If the implementation is $v := h \ \texttt{mod} \ 2$, then she can learn $h$'s parity. If the implementation is $v := h$, then she learns the exact value of $h$. Intuitively, the three implementations are in increasing order according to the loss of ignorance they induce. 
	
It is desirable that the implementation of a specification be \qm{ignorance preserving}, in the sense that the implementation should not reveal more about the secrets than the specification does. Some works in the literature suggest that one should be careful when dealing with secure refinements if one wants to preserve information-flow security properties. In \cite{Jacob:89:SSP}, for instance, Jacob shows that even if an implementation is a consistent refinement with respect to a specification, it does not imply that the (information-flow) security properties of the specification are preserved in the implementation.

As pointed out in \cite{Chatzikokolakis:09:FOSSACS}, nondeterminism is often used in system specifications as a way of abstracting from implementation details (such as scheduler policy). Implementations are obtained from specifications by refinement algebras, which reduce nondeterminism. As we have seen in a previous example, if we assume $v$ and $h$ are both of type $\mathbb{N}$, then the specification \emph{choose $v$ from the domain $\mathbb{N}$} can be refined to $v := h$, which is simply a reduction of nondeterminism. This is known as the \qm{refinement paradox} \cite{Morgan:09:SCP}, because it does not preserve ignorance. While the specification does not tell anything about the value of $h$, the refinement completely reveals it.
	
The process of reducing nondeterminism by refinements is related to the notion of \emph{schedulers} in nondeterministic systems: \emph{designing an implementation of a specification involves choosing a scheduler to solve all the nondeterminism of the specification.} The scheduler is indeed a final result of the refinement process, after all the nondeterminism is ruled out.
	
According to this perspective, similar concerns about refinement algebras should be taken into consideration when dealing with schedulers. Indeed, it can be shown that, given a specification $S$ and a scheduler that leads to a consistent implementation $P$ with respect to $S$, it is not guaranteed that the security properties of $S$ are preserved in $P$. 
	
In the domain of refinement of specifications, the solution proposed in \cite{Morgan:09:SCP} is to apply some principles to the refinement algebra in order to assure the preservation of ignorance. These principles restrict the refinement relation, eliminating the cases that do not preserve ignorance. 

A similar problem arises in the context of concurrent systems, where the scheduler that \revision{resolves} the nondeterminism can violate security properties. In Chapter~\ref{chapter:safe-equivalences} we focus on this problem and we propose restrictions on the schedulers that also lead to ignorance-preserving refinements.

\section{Plan of the thesis and contribution}
\label{section:contribution}

In Chapter~\ref{chapter:preliminaries} we review some basic notions necessary for the development of this thesis, including the concepts of probability spaces, probabilistic automata and \ccsp{} (a probabilistic version of the process algebra of concurrent communicating processes).

In Chapter~\ref{chapter:probabilistic-info-flow} we review the main approaches that have been considered to quantify the notion of information leakage using concepts of information theory. We explain concepts such as entropy, conditional entropy, mutual information and capacity. We focus on how distinct notions of entropy can model attackers with different levels of power, and we introduce the mathematical background necessary for most of this thesis. Finally we compare the main notions of uncertainty and leakage in the literature.

In Chapter~\ref{chapter:interactive-systems} we consider the problem of defining the information leakage in interactive systems where secrets and observables can alternate during the computation. We show that the information-theoretic approach \revision{that} interprets such systems as classic channels is not valid. The principle can be recovered, however, if we consider channels of a more complicated kind, namely channels with memory and feedback. We show that there is a complete correspondence between interactive systems and such channels. We also propose the use of directed information, as opposed to mutual information, to represent leakage in interactive systems. This proposal is based on recent results in information theory that have shown that, in channels with memory and feedback, the transmission rate does not correspond to the maximum mutual information (the standard notion of capacity), but rather to the maximum (normalized) directed information. We show that our model is a proper extension of the classical one, i.e. in the absence of interactivity the model of channels with memory and feedback collapses into the model of memoryless channels without feedback. Finally, we show that the capacity of the channels associated with interactive systems is a continuous function with respect to a pseudometric based on the Kantorovich metric.

In Chapter~\ref{chapter:differential-privacy} we analyze critically the notion of differential privacy in the light of the conceptual framework provided by  min-entropy leakage. We show that there is a close relationship between differential privacy and leakage, due to the graph symmetries induced by the adjacency relation on databases. Furthermore, we consider the utility of the randomized answer, which measures its expected degree of accuracy. We focus on certain kinds of utility functions called \qm{binary}, which have a close correspondence with the notion of min-entropy leakage and the Bayes risk. Again, there can be a tight correspondence between differential privacy and utility, depending on the symmetries induced by the adjacency relation and by the query. Using these symmetries we can, in some cases, build an optimal-utility randomization mechanism while preserving the required level of differential privacy. We also provide a study of the kind of structures that can be induced by the adjacency relation and the query, and how to use them to derive bounds on the leakage and achieve the optimal utility.

In Chapter~\ref{chapter:safe-equivalences} we move away from the quantitative realm and focus on the problem of nondeterminism in systems specifications. In the field of security, process equivalences have been used to characterize various information-hiding properties (for instance secrecy, anonymity and noninterference) based on the principle that a protocol $P$ with a variable $x$ satisfies such a property if and only if, for every pair of secrets $s_1$ and $s_2$, $P[^{s_1}/ _x]$ is equivalent to $P[^{s_2}/ _x]$. We argue that, in the presence of nondeterminism, the above principle relies on the assumption that the scheduler \qm{works for the benefit of the protocol}, and this is usually not a safe assumption. Non-safe equivalences, in this sense, include complete-trace equivalence and bisimulation. We present a formalism in which we can specify admissible schedulers and, correspondingly, safe versions of these equivalences. We prove that safe bisimulation is still a congruence. Then we show that safe equivalences can be used to establish information-hiding properties.

Finally, in Chapter~\ref{chapter:conclusion} we make our final \revision{observations}.

\section{Publications}

Most of the results in this thesis have already been \revision{the} subject of scientific publications. More precisely:

\begin{itemize}
	
	\item \revision{Chapter~\ref{chapter:probabilistic-info-flow}} is based on the paper \textbf{Probabilistic Information Flow} \cite{Alvim:10:LICS} that appeared in the proceedings of \emph{$25^{th}$ Annual IEEE Symposium on Logic in Computer Science} (LICS 2010).
	
	\item Chapter~\ref{chapter:interactive-systems} is based on the papers:
		\begin{itemize} 
			\item \textbf{Information Flow in Interactive Systems} \cite{Alvim:10:CONCUR} that appeared in the proceedings of the \emph{$21^{st}$ International Conference on Concurrency Theory} (CONCUR 2010);
			\item \textbf{Quantitative Information Flow in Interactive Systems} \cite{Alvim:11:JCS} to appear in the \emph{Journal of Computer Security}.
		\end{itemize}  
	
	\item Chapter~\ref{chapter:differential-privacy} is based on two complementary works:
		\begin{itemize}
			\item The paper \textbf{On the relation between Differential Privacy and Quantitative Information Flow} \cite{Alvim:11:ICALP} to appear in the proceedings of the \emph{38th International Colloquium on Automata, Languages and Programming} (ICALP 2011);
			\item The technical report \textbf{Differential Privacy: on the trade-off between Utility and Information Leakage} \cite{Alvim:11:TechRep}.
		\end{itemize}
		
	\item Chapter~\ref{chapter:safe-equivalences} is based on the paper \textbf{Safe Equivalences for Security Properties} \cite{Alvim:10:IFIP-TCS} that appeared in the the proceedings of the \emph{6th IFIP International Conference on Theoretical Computer Science} (IFIP-TCS 2010).
	
\end{itemize}

\chapter{Preliminaries}
\label{chapter:preliminaries}
\mscite{I can make just such ones if I had tools, and I could make tools \\ if I had tools to make them with.}{Eli Whitney}

In this chapter we review some technical concepts from the literature that will be used throughout this thesis.

\section{Probability spaces}
\label{section:prob-spaces}

In this section we recall some concepts about probability spaces.

Let $\Omega$ be a set and $\pow(\Omega)$ represent its powerset, i.e. the collection of all subsets of $\Omega$. A \emph{$\sigma$-algebra} (also called \emph{$\sigma$-field}) over $\Omega$ is a non-empty collection of sets $\Fset \subseteq \pow(\Omega)$ that is closed under complementation and countable union. For any $\sigma$-field $\Fset$, the property $\Omega \in \Fset$ holds, and also that $\Fset$ is closed under countable intersection (by De Morgan's laws). 

A \emph{\revision{(positive)} measure} on $\Fset$ is a function $\mu: \Fset \rightarrow [0, \infty)$ such that

\begin{enumerate}
	\item $\mu(\emptyset) = 0$, and
	\item $\mu(\bigcup_{i} C_{i}) = \sum_{i} \mu(C_{i})$, where $\{C_{i}\}_{i}$ is a countable collection of pairwise disjoint sets in $\Fset$.
\end{enumerate}

A \emph{probability measure} on $\Fset$ is a measure $\mu$ on $\Fset$ such that $\mu(\Omega) = 1$. A \emph{probability space} is a tuple $(\Omega, \Fset, \mu)$ where $\Omega$ is a non-empty set called the \emph{sample space}, $\Fset$ is a $\sigma$-algebra on $\Omega$ called the \emph{event space}, and $\mu$ is a probability measure on $\Fset$. In the discrete case, we have
\begin{equation*}
	\forall C \in \Fset. \quad  \mu(C) = \sum_{x \in C} \mu(\{x\}) 
\end{equation*}

In this case we can construct $\mu$ from a function $p: \Omega \rightarrow [0,1]$ satisfying $\sum_{x \in \Omega}p(x) = 1$ by assigning $\mu(\{x\}) = p(x)$. The function $p$ is called a \emph{probability distribution} over $\Omega$.

The set of all probability measures with sample space $\Omega$ will be denoted by $\Dset(\Omega)$. We will also denote by \revision{$\delta_x(\cdot)$} (called the \emph{Dirac measure} on $x$ or also a \emph{point mass}) the probability distribution such that $\mu(\{x\}) = 1$.

If $A$ and $B$ are events, i.e. elements of a $\sigma$-field $\Fset$, then $A \cap B$ is also an event. If $\mu(A) > 0$ then we can define the \emph{conditional probability} $p(B|A)$ as
\begin{equation*}
	p(B|A) = \frac{\mu(A \cap B)}{\mu(A)}
\end{equation*}

\noindent representing the probability of B given that A holds. Note that $p(\cdot|A)$ is a new probability measure on $\Fset$. For the scope of this thesis we are interested only in the discrete case, so it is enough to use the definition above and make sure that we never condition on an event $A$ with zero probability.

Let $\mathcal{F}, \mathcal{F}'$ be two $\sigma$-fields on $\Omega,\Omega'$ respectively. A \emph{random variable} $X$ is a function $X : \Omega \mapsto \Omega'$ that is \emph{measurable}, meaning that the inverse of every element of $\mathcal{F}'$ belongs to $\mathcal{F}$:
\begin{equation*}
	\forall C \in \mathcal{F}'. \quad X^{-1}(C) \in \mathcal{F}
\end{equation*}

Then, given a probability measure $\mu$ on $\mathcal{F}$, $X$ induces a probability measure $\mu'$ on $\mathcal{F}'$ as   
\begin{equation*}
	\forall C \in \mathcal{F}'. \quad \mu'(C) = \mu(X^{-1}(C))
\end{equation*}

If $\mu'$ is a discrete probability measure then it can be constructed by a probability distribution over $\Omega'$, called \emph{probability mass function (pmf)}, defined as
\begin{equation*}
	P([X = x]) = \mu(X^{-1}(x))
\end{equation*}  

\noindent for each $x \in \Omega'$. The random variable in this case is called discrete. If $X,Y$ are discrete random variables then we can define a discrete random \revision{variable} $(X,Y)$ by its pmf 
\begin{equation*}
	P([X=x, Y=y]) = \mu(X^{-1}(x) \cap  X^{-1}(y))
\end{equation*}

If $X$ is a real-valued discrete random variable then its \emph{expected value} (or \emph{expectation}) is defined as
\begin{equation*} 
	E(X) = \sum_{i} x_i\,P([X=x_i]) 
\end{equation*}

A family  $\rho = \{\prob{p_v}{\cdot}\}_{v}$ of probability measures parametrized on $v$ \revision{(where $v$ can range over $\{0,\ldots,n\}$ for some natural $n$)} is called a \emph{stochastic kernel}.\footnote{The general definition of stochastic kernel is more complicated (cfr.~\cite{Tatikonda:09:TIT}), but it reduces to this one in the discrete case, which is what we use in this thesis.}.	

\paragraph{Notation:} We will use capital letters $A,B,X,Y,Z$ to denote random
variables and calligraphic letters $\mathcal{A},\mathcal{B},\mathcal{X},\mathcal{Y},\mathcal{Z}$ to denote their image.
With a slight abuse of notation we will use $p$ (and $p(x),p(y)$) to denote either

\begin{itemize}
	\item a probability distribution, when $x,y \in \Omega$, or
	\item a probability measure, when $x,y \in \mathcal{F}$ are events, or
	\item the probability mass function $P([X=x]),P([Y=y])$ of the random variables $X,Y$ respectively, when $x\in \calx, y\in \caly$.
\end{itemize}

\commentMS{Introduce definition of the Kantorovich metric}


\section{Probabilistic automata}
\label{section:prob-automata}

Let $\mu \colon \Sset \to [0,1]$ be a discrete probability distribution on a countable set~$\Sset$, and let the set of all discrete probability distributions on $\Sset$ be $\Dset(\Sset)$. 

A \emph{probabilistic automaton} \cite{Segala:95:PhD} is a quadruple $M = (\Sset, \Lset, \si, \vartheta )$ where $\Sset$ is a countable set of \emph{states}, $\Lset$ is a finite set of \emph{labels} or \emph{actions}, $\si$ is the \emph{initial} state, and $\vartheta$ is a \emph{transition function} $\vartheta: \Sset \to \pow(\distr(\Lset \times \Sset))$. If $\vartheta(s) = \emptyset$ then $s$ is a \emph{terminal} state. We write $s {\to} \mu$ for $\mu \in \vartheta(s), \ s\in \Sset$.  Moreover, we write $s \smash{\stackrel{\ell}{\to}} r$ for $s, r \in \Sset$ whenever $s {\to} \mu$ and $\mu(\ell,r) > 0$. A \emph{fully probabilistic automaton} is a probabilistic automaton satisfying $|\vartheta(s)|\leq 1$ for all states. In such an automaton, when $\vartheta(s)\not= \emptyset$, we overload the notation and denote by $\vartheta(s)$ the distribution outgoing from $s$.

A \emph{path} in a probabilistic automaton is a sequence $\sigma = s_0 \stackrel{\ell_1}{\to} s_1 \stackrel{\ell_2}{\to} \cdots$ where $s_i \in \Sset$, $\ell_i \in \Lset$ and $s_i \smash{\stackrel{\ell_{i+1}}{\to}} s_{i+1}$. A path can be \emph{finite} in which case it ends with a state. A path is \emph{complete} if it is either infinite, or finite ending in a terminal state. Given a finite path $\sigma$, $\last(\sigma)$ denotes its last state. Let $\paths_s(M)$ denote the set of all paths, $\fpaths_s(M)$ the set of all finite paths, and $\cpaths_s(M)$ the set of all complete paths of an automaton $M$, starting from the state $s$. We will omit $s$ if $s=\si$. Paths are ordered by the prefix relation, which we denote by $\leq$. The \emph{trace} of a path is the sequence of actions in $\Lset^{*} \cup \Lset^{\infty}$ obtained by removing the states, hence for the above $\sigma$ we have $\trace(\sigma) = l_1l_2\ldots$. If $\Lset'\subseteq \Lset$, then $\trace_{\Lset'}(\sigma)$ is the projection of $\trace(\sigma)$ on the elements of $\Lset'$. 

Let $M = (\Sset, \Lset, \si, \vartheta )$ be a (fully) probabilistic automaton, $s\in \Sset$ a state, and let $\sigma \in \fpaths_{\!\!\!s}(M)$ be a finite path starting in $s$. The \emph{cone} generated by $\sigma$ is the set of complete paths $\cone{\sigma} = \{ \sigma^\prime \in \cpaths_s(M) \mid \sigma \leq \sigma^\prime\}.$ Given a fully probabilistic automaton $M=(\Sset, \Lset, \si, \vartheta)$ and a state $s$, we can calculate the \emph{ probability value} $\PP_s(\sigma)$ of any finite path $\sigma$ starting in $s$ as follows: 
\begin{align*}
	\PP_s(s) = 1 & \text{, and} \\
	\PP_s(\sigma \, \stackrel{\ell}{\to}\, s') = \PP_s(\sigma)\ \mu(\ell,s') & \text{~where~} \last(\sigma) \to \mu
\end{align*}

Let $\Omega_{s} \defsym \cpaths_s(M)$ be the sample space, and let $\mathcal F_{s}$ be the smallest $\sigma$-algebra induced by the cones generated by all the finite paths of $M$. Then $\PP$ induces a unique \emph{probability measure} on $\mathcal F_{s}$ (which we will also denote by $\PP_s$) such that $\PP_s(\cone{\sigma}) = \PP_s(\sigma)$ for every finite path $\sigma$ starting in $s$.  For $s=\si$ we write $\PP$ instead of $\PP_{\si}$.

A \revision{(total)} \emph{scheduler} for a probabilistic automaton $M$ is a function \revision{defined as} $\zeta \colon \fpaths(M) \to (\Lset \times \distr(\Sset)\cup \lbrace\bot\rbrace)$ such that for all finite \revision{paths} $\sigma$, if $\vartheta(\last(\sigma))\not=\emptyset$ then $\zeta(\sigma)\in\vartheta(\last(\sigma))$, and $\zeta(\sigma)=\bot$ otherwise. Hence, a scheduler $\zeta$ selects one of the available transitions in each state, and determines therefore a fully probabilistic automaton, obtained by pruning from $M$ the alternatives that are not chosen by $\zeta$. A scheduler is history dependent since it takes into account the path and not only the current state. \revision{It is possible to define partial schedulers, i.e. schedulers that may halt the execution at any time.} In this thesis, however, we will consider only total schedulers, to be more in line with the standard semantics of CCS.


\section{CCS with internal probabilistic choice}
\label{section:ccs}

In this section we present an extension of standard CCS (\cite{Milner:89:BOOK}) obtained by adding internal probabilistic choice. The resulting calculus can be seen as a simplified version of the probabilistic $\pi$-calculus presented in \cite{Herescu:00:FOSSACS,Palamidessi:05:TCS} and it is similar to the one considered in \cite{Deng:05:BookJW}. The restriction to CCS and to internal choice is suitable for the scope of this thesis. 

Let $a$ range over a countable set of \emph{channel names}. 

The syntax of \ccsp{} is the following:
\[
\begin{array}[t]{@{\textrm{\hspace{20pt}}}l@{}l@{\textrm{\hspace{20pt}}}l}
	\multicolumn{2}{l}{\alpha ::= a \gramor \bar{a} \gramor \tau}		& \textrm{\textbf{prefixes}} \\[2pt]

	\multicolumn{2}{l}{P,Q ::= }	& \textrm{\textbf{processes}} \\[2pt]
			& \alpha.P				& \textrm{prefix} \\[2pt]
	\gramor & P \paral Q			& \textrm{parallel} \\[2pt]
	\gramor & P + Q					& \textrm{nondeterministic choice} \\[2pt]
	\gramor & \smallsum{i}p_iP_i	& \textrm{internal probabilistic choice} \\[2pt]
	\gramor & (\nu a)P				& \textrm{restriction} \\[2pt]
	\gramor & !P					& \textrm{replication} \\[2pt]
	\gramor & 0						& \textrm{nil}
\end{array}
\]

\noindent where the $p_i$'s in the probabilistic choice should be non-negative and their sum should be $1$. We will also use the notation $P_1 +_p P_2$ to represent a binary sum $\smallsum{i} p_iP_i$ with $p_1=p$ and $p_2=1-p$. 

\begin{figure}[tbp]
	$
	\begin{array}{ll@{\textrm{\hspace{15pt}}}ll}
		\textrm{ACT} &
			\bigfrac
				{}
				{\alpha.P \labarr{\alpha} \delta(P)} &

		\textrm{RES} &
			\bigfrac
				{P  \labarr{\alpha} \mu \qquad \alpha \neq a,\outp{a}}
				{(\nu a)P \labarr{\alpha} (\nu a) \mu } \\[20pt]

		\textrm{SUM1} &
			\bigfrac
				{P \labarr{\alpha} \mu }
				{P + Q \labarr{\alpha} \mu} &

		\textrm{SUM2} &
			\bigfrac
				{Q \labarr{\alpha} \mu }
				{P + Q \labarr{\alpha} \mu} \\[20pt]

		\textsc{PAR1} &
			\bigfrac
				{P \labarr{\alpha} \mu }
				{P \paral Q  \labarr{\alpha} \mu \paral Q} &

		\textsc{PAR2} &
			\bigfrac
				{Q \labarr{\alpha} \mu }
				{P \paral Q  \labarr{\alpha} P \paral \mu} \\[20pt]

		\textrm{COM} & 
			\bigfrac
				{P  \labarr{a} \delta(P') \quad Q \labarr{\outp{a}}{} \delta(Q')}
				{P \paral Q \labarr{\tau} \delta(P' \paral Q')} &
	
		\textrm{PROB} &
			\bigfrac
				{}
				{\smallsum{i} p_iP_i \labarr{\tau} \smallsum{i}p_i\,\delta(P_i)} \\[20pt]

		\textrm{REP1} &
			\bigfrac
				{P  \labarr{\alpha} \mu }
				{!P  \labarr{\alpha} \mu \paral !P} &

		\textrm{REP2} & 
			\bigfrac
				{P  \labarr{a} \delta(P_1) \quad P \labarr{\outp{a}}{} \delta(P_2)}
				{!P \labarr{\tau} \delta(P_1 \paral P_2 \paral !P)}
	\end{array}
	$
	\caption{The semantics of \ccsp}
	\label{fig:ccsp_sem}
\end{figure}

The semantics of a \ccsp{} term is a probabilistic automaton defined inductively on the basis of the syntax according to the rules in Figure \ref{fig:ccsp_sem}.  We write $s \labarr{a} \mu$ when $(s,a,\mu)$ is a transition of the probabilistic automaton. Given a process $Q$ and a measure $\mu$, we denote by $\mu\paral Q$ the measure $\mu'$ such that $\mu'(P\paral Q) = \mu(P)$ for all processes $P$ and $\mu'(R) = 0$ if $R$ is not of the form $P\paral Q$. Similarly $(\nu a)\mu = \mu'$ such that $\mu'( (\nu a)P ) = \mu(P)$.

A transition of the form $P \labarr{a} \delta(P')$, i.e. a transition having for target  a Dirac measure, corresponds to a transition of a non-probabilistic automaton (a standard labeled transition system). Note that each rule of \ccsp{} corresponds to one rule of CCS, except for PROB. The latter models the internal probabilistic choice: a silent $\tau$ transition is available from the sum to a measure containing all of its operands, with the corresponding probabilities. 

Note that in the produced probabilistic automaton, all transitions to non-Dirac measures are silent. This is similar to the \emph{alternating model} \cite{Hansson:89:SRTS}, however our case is more general because the silent and non-silent transitions are not necessarily alternated.  On the other hand, with respect to the simple probabilistic
automata the fact that the probabilistic transitions are silent looks like a restriction. It has been proved by Bandini and Segala \cite{Bandini:01:ICALP}, however, that the simple probabilistic automata and the alternating model are essentially equivalent, so, being in between, our model is equivalent as well.

\paragraph{Encoding message passing into \ccsp{}} Sometimes it is convenient to make message passing explicit in the notation of \ccsp. Namely, we enrich its syntax by \revision{allowing} the prefixes to be $c(a) \gramor c \langle x \rangle \gramor \tau$, where $c,a,x$ are names, and the semantic rule COM is substituted by:
$$	
	\textrm{COM'} \quad
		\bigfrac
			{P  \labarr{ c \langle a \rangle } \delta(P') \quad Q \labarr{ c(x) }{} \delta(Q')}
			{P \paral Q \labarr{\tau} \delta(P' \paral Q' \left[ ^{a} / _{x}  \right])}
$$

\noindent where $P  \labarr{ c \langle a \rangle } \delta(P')$ denotes a process that sends the name $a$ through channel $c$ and then evolves to $P'$, and $Q \labarr{ c(x) }{} \delta(Q')$ denotes a process that receives the name $x$ through channel $c$ and then evolves to $Q'$. Here $Q'\left[ ^{a} / _{x}  \right]$ is the process $Q'$ in which every occurrence of $x$ is replace by $a$. 

The expressive power of \ccsp{} with message passing and without it is the same\revision{~\cite{Milner:89:BOOK}}. In this thesis we will use this fact and consider explicit message passing as an alias for the \revision{corresponding} encoding into the presentation of \ccsp{} given in Figure~\ref{fig:ccsp_sem}.

\chapter{The rationale behind the use of information theory for leakage}
\label{chapter:probabilistic-info-flow}
\mscite{Why, only why?}{Nadia Vertti}
In this chapter we review the most important concepts related to the \emph{information theoretic} approach to quantitative information flow. We aim at presenting these concepts in a contextualized way, discussing the intuition behind them and interpreting what they mean in terms of security.

\paragraph{Plan of the Chapter} Section~\ref{section:info-theory-communication} gives a brief overview on information theory for communication. Section~\ref{section:info-theory-info-flow} introduces the information theoretic approach to information flow. Section~\ref{section:uncertainty-leakage} presents and compares several different notions based on information theory that have been used in \revision{the} literature to characterize uncertainty and leakage.

\section{Information theory and communication}
\label{section:info-theory-communication}

The study of information theory started with Claude E. Shannon's work on the problem of coding messages to be transmitted through unreliable (or noisy) channels. A \emph{communication channel} is a (physical) \revision{means} through which information can be transmitted. The input is fed \revision{into} the channel, but due to noise or any other \revision{problems} that can occur during the transmission, the output of the channel may not reflect with fidelity the input. It is usual to describe the unreliable behavior of the channel in a probabilistic way. In the discrete (finite) case, if $\Aset = \{ a_{1}, a_{2}, \ldots, a_{n} \}$ represent the possible inputs for the channel, and $\Bset = \{b_{1}, b_{2}, \ldots, b_{m} \}$ represent the possible outputs, the channel's probabilistic behavior can be represented as a channel matrix $M_{n \times m}$ where each element $M_{i,j}$ ($1\leq i \leq n$, $1 \leq j \leq m$) is defined as the probability of the channel outputting $b_j$ when the input is $a_i$. In this way, we can see the input and output as two correlated random variables linked by the channel's probabilistic behavior\revision{~\footnote{\revision{Note that we are assuming that channels are \emph{loseless}, since the rows are probability distributions instead of sub-probability distributions.}}}. 

A unique feature of information theory is its use of a numerical measure of the amount of information gained when the contents of a message are learned. More specifically, information theory reasons about the degree of uncertainty of a certain random variable, and the amount of information that it can reveal about another random variable. Among the tools provided by information theory there are concepts as \emph{entropy}, \emph{conditional entropy}, \emph{mutual information} and \emph{channel capacity}, which will be reviewed in Section~\ref{section:shannon-entropy}. We consider here only the discrete case, since this is enough for the scope of this thesis.

\section{Information theory and information flow}
\label{section:info-theory-info-flow}

Several works in the literature use an information theoretic approach to model the problem of information flow and define the leakage in a quantitative way, as for example \cite{Zhu:05:ICDCS,Clark:05:JLC,Malacaria:07:POPL,Malacaria:08:PLAS,Moskowitz:03:CNIS,Moskowitz:03:WPES,Chatzikokolakis:08:IC}. The idea is to model the computational system as \revision{an} \emph{information theoretic channel}. The input represents the secret, the output represents the observable, and the correlation between the input and output (\emph{mutual information}) represents the information leakage. The worst case leakage corresponds then to the \emph{capacity} of the channel, which is by definition the maximum mutual information that can be obtained by varying the input distribution. 

In the works mentioned above,  the notion of mutual information is based on \emph{Shannon entropy}, which (because of its mathematical properties) is the most  established measure of uncertainty. From the security point of view, this measure corresponds to a particular model of attack and a particular way of estimating the security  threat (vulnerability of the secret). Other notions have been considered, and argued to be more appropriate for security in certain scenarios. These include: \emph{min-entropy}~\cite{Renyi:61:Berkeley,Smith:09:FOSSACS}, \emph{Bayes risk}~\cite{Cover:91:BOOK,Chatzikokolakis:08:JCS}, \emph{guessing entropy}~\cite{Massey:94:IT}, and \emph{marginal guesswork}~\cite{Pliam:00:INDOCRYPT}. In Section~\ref{section:uncertainty-leakage} we will discuss their meaning and show how they relate (or do not relate) to each other and to Shannon entropy.  

Whatever definition of uncertainty (i.e. vulnerability) we want to adopt, the notion of leakage is inherent to the system and can be expressed in a uniform way as the difference between the initial uncertainty, i.e. the degree of ignorance about the secret \emph{before} we run the system, and the remaining uncertainty, i.e. the degree of ignorance about the secret \emph{after} we run the system and observe its outcome. Following the principle advocated by Smith \cite{Smith:09:FOSSACS}, and by many others: 
\begin{equation}
	\label{eqn:leakage}
	\revision{\mathit{information \ leakage}}    \ =   
	\begin{array}[t]{l}
	 \mathit{initial \ uncertainty}\\-\\\mathit{remaining \  uncertainty}
	 \end{array}
\end{equation}

In \eqref{eqn:leakage}, the initial uncertainty depends solely on the input distribution, aka \revision{the} \emph{a priori distribution} or \emph{prior}. Intuitively, the more uniform it is, the less we know about the secret (in the probabilistic sense). After we run the system, if there is a probabilistic correlation between input and output, then the observation of  the output should increase our knowledge of the secret. This is determined by the fact that the distribution on the input changes: in fact we can update the probability  of each input with the corresponding conditional probability of the same input, given the output. The new distribution is called \revision{the} \emph{a posteriori distribution}. In case \revision{the} input and output are independent, then the a priori and the a posteriori distributions coincide and the  knowledge should remain the same. We will use  the attributes \qm{a priori} (or \qm{prior})  and \qm{a posteriori} to refer to before and after the observation of the output, respectively. 

The above intuitions should be reflected by any reasonable notion of uncertainty: it should be higher on more uniform distributions, and it should decrease or remain equal with the observation of related events. 

If the uncertainty is expressed in terms of Shannon entropy, then the initial uncertainty is the entropy of the input, the remaining uncertainty is the conditional entropy of the input given the output, and  \eqref{eqn:leakage} matches exactly the definition of mutual information. This justifies the notion of leakage adopted in the works mentioned before (\cite{Zhu:05:ICDCS,Clark:05:JLC,Malacaria:07:POPL,Malacaria:08:PLAS,Moskowitz:03:CNIS,Moskowitz:03:WPES,Chatzikokolakis:08:IC}). 

The analogy between information flow in a system and  a (simple) channel works well when:

	\begin{enumerate}[(i)]
	
		\item \label{item:rest-a} there is no nondeterminism, i.e. either the system is deterministic, or purely probabilistic; and
	
		\item \label{item:rest-b} there is a precise temporal relation between secrets and observables in the computations; namely, the value of the  secret is chosen at the beginning of  the computation, and the computation of the system produces  an observable outcome with  a probability that depends solely on the chosen input and on the system. Furthermore, each new run of the system is independent from the previous ones. 

	\end{enumerate}

\revision{Restriction~{(\ref{item:rest-a})} implies that  for each secret there is exactly one conditional probability distribution on the observables, where the condition is the secret value. If a system is deterministic, then under the same input each run produces always the same output, with probability $1$. Therefore the matrix contains only $0$'s and $1$'s. Yet the problem of inferring the secret is interesting, because the same output may correspond to different inputs. If the system is probabilistic, i.e. it uses some randomized mechanisms, then the matrix usually contains probabilities different from $0$ and $1$. 

Restriction~{(\ref{item:rest-b})} ensures that this conditional distribution depends uniquely on the system (not on the input distribution). These conditional probabilities constitute the channel matrix. Note that in a (basic) information-theoretic  channel the matrix must be invariant with respect to the input distribution, which is exactly what condition~{(\ref{item:rest-b})} guarantees.}

Unfortunately, usually conditions~{(\ref{item:rest-a})} and {(\ref{item:rest-b})} are too restrictive for real-life systems:  

	\begin{itemize}

		\item Specifications typically need to use  nondeterminism in order to abstract from implementation details. This is particularly  compelling in the case of concurrent and distributed systems: The order in which the various components get executed and their interactions depend on scheduling policies that may differ from implementation to implementation. Furthermore, even if the scheduling  policy is fixed, there are run time circumstances that may influence the relative speed of the processes. Nondeterminism is, in practice, an unavoidable aspect of concurrency. 
	
		\item Secrets and observables often alternate and  interact during an execution. In particular, the choice of a new secret may depend on previous observables. Furthermore, new \revision{executions} of the systems may depend on previous ones. This may be due to the way the system works, or to the  presence of an active adversary that may use the knowledge derived from previous  observations to try to  tamper with the mechanisms of the system, with the purpose of increasing the leakage. Examples of such systems, that we call here \emph{interactive} systems (where interaction refers to the interplay between secrets and observables), can be found in the areas of game theory, auction protocols, web servers, GUI applications, etc.
		
	\end{itemize}

In this thesis we consider the challenges of extending the information-theoretic approach to cases where these conditions are \revision{relaxed}. More specifically, Chapter~\ref{chapter:interactive-systems} concerns the suppression of condition~{(\ref{item:rest-b})}, and Chapter~\ref{chapter:safe-equivalences} deals with the suppression of  condition~{(\ref{item:rest-a})}. 

\section{Uncertainty and leakage}
\label{section:uncertainty-leakage}

In this section we recall various definitions of uncertainty based on information theory proposed in \revision{the} literature, and we discuss the relation with  security attacks and the way of measuring their success. In general we consider the kind of threats that in the model of K\"opf and Basin \cite{Kopf:07:CCS} are called \emph{brute-force guessing attacks}, which can be summarized as follows: The goal of the adversary is to determine the value of a random variable. He can make a series of queries to an oracle. Each query must have a  \emph{yes}/\emph{no} answer. In general the adversary is \emph{adaptive}, i.e. he can choose the next query depending on the \revision{answers} to the previous ones. We assume that the adversary knows the \revision{a priori} probability distribution. In this section, when we talk about the meaning in security of a particular measure of uncertainty, we refer to the work in \cite{Kopf:07:CCS}.

In the following, $A, B$ denote two discrete random variables with finitely many values $\Aset = \{\Asym{1}{}, \ldots, \Asym{n}{}\}, \ 	\Bset = \{ \Bsym{1}{}, \ldots, \Bsym{m}{} \}$, and  probability distributions $\prob{p_{A}}{\cdot}$,  $\prob{p_{B}}{\cdot}$,  respectively. We will use  $A\wedge B$ to represent the random variable with carrier $\Aset\times\Bset$ and joint probability distribution $\prob{p_{A \wedge B}}{a,b} = \prob{p_{A}}{a} \cdot \prob{p}{b\mid A=a}$, while  $A \cdot B$ will denote the random variable with carrier $\Aset\times\Bset$ and probability distribution defined as product, i.e. $\prob{p_{A\cdot B}}{a,b} = \prob{p_{A}}{a}\cdot \prob{p_{B}}{b}$. Clearly, if $A$ and $B$ are independent, we have $A\, \wedge \,B =  A \,\cdot \,B$. We shall omit the subscripts on the probabilities when they are clear from the context. In reference to a channel, in general $A$ will denote the input (secret), and $B$ the output (observable).
 
\subsection{Shannon entropy}
\label{section:shannon-entropy}
	
The (Shannon) \emph{entropy} of $A$ is defined as
\begin{equation*}
	H(A) = -\sum_{{\mathcal A}} \prob{p}{a} \log{\prob{p}{a}}
\end{equation*}

The entropy measures the uncertainty of $A$. It takes its minimum value $H(A) = 0$ when $p_{A}(\cdot)$ is a point mass (also called delta of Dirac). The maximum value $H(A) = \log{|{\mathcal A}|}$ is obtained when $p_{A}(\cdot)$ is the uniform distribution. Usually the base of the  logarithm is set to  be $2$ and the entropy is measured in \emph{bits}. Roughly speaking, $m$ bits of entropy means that we have $2^m$ values to choose from, assuming a uniform distribution.

The \emph{conditional entropy} of $A$ given $B$ is defined as 
\begin{equation}
	\label{eqn:conditional}
		\begin{array}{lcl}
			H(A\mid B) &=& 
		 		{\displaystyle\sum_{b\,\in\, \Bset}
				\prob{p}{\Bsym{}{}} \  H(A \mid B = \Bsym{}{} )}
		\end{array}
\end{equation}
	
\noindent where 
$$
	\begin{array}{lcl}
		H(A \mid B = \Bsym{}{} ) &= & 	{\displaystyle -\sum_{a\,\in\,\Aset}
		\Cprob{p}{\Asym{}{}}{\Bsym{}{}}\allowbreak
		\log \,{\Cprob{p}{\Asym{ }{}}{\Bsym{ }{}}}  }
	\end{array}
$$
The conditional entropy measures the uncertainty of $A$ when $B$ is known. It is well-known that $0 \leq H(A|B) \leq H(A)$. The minimum value, $0$, is obtained when $A$ is completely determined by $B$. The maximum value $H(A)$ is obtained when $A$ and $B$ are independent. 

The \emph{mutual information} between $A$ and $B$ is defined as 
\begin{equation}
	\label{eqn:ShannonMutualInfo}
	I(A;B) = H(A) - H(A|B)
\end{equation}

The mutual information measures the amount of information about $A$ that we gain by observing $B$. It can be shown that $I(A;B) = I(B;A)$ and $0 \leq I(A;B) \leq H(A)$. If $C$ is a third random variable, the \emph{conditional mutual information} between $A$ and $B$ given  $C$ is defined as 
\begin{equation*}
	I(A;B|C) = H(A|C) - H(A|B,C)
\end{equation*}
	
The (conditional) entropy and mutual information respect the \emph{chain rules}. Namely, given the random variables $A_{1}, A_{2}, \ldots, A_{k}$, $B$ and $C$, we have:
\begin{equation*}
	H(A_{1}, A_{2}, \ldots, A_{k} | C) = \sum_{i=1}^{k}
	H(A_{i} | A_{1}, \ldots, A_{i-1}, C)
\end{equation*}

\begin{equation}
	\label{eq:chain-rule-mutual-information}
	I(A_{1}, A_{2}, \ldots, A_{k}; B| C) = \sum_{i=1}^{k}
	I(A_{i} ; B | A_{1}, \ldots, A_{i-1}, C)
\end{equation}
	
A \emph{discrete memoryless channel} is a tuple $({\mathcal A}, {\mathcal B}, p(\cdot|\cdot))$, where ${\mathcal A}, {\mathcal B}$ are the sets of input and output symbols, respectively, and $p(b|a)$ is the probability of observing the output symbol $b$ when the input symbol is $a$. These conditional probabilities constitute the \emph{channel matrix}. An input distribution $p_A(\cdot)$ over ${\mathcal A}$ together with the channel determine the joint distribution $p(a,b) = p(a|b) \cdot p(a)$ and consequently  $I(A;B)$. The maximum $I(A;B)$ over all possible input distributions is the channel's \emph{capacity} $C$:
\begin{equation*}
	C = \max_{p_A(\cdot)} I(A;B)
\end{equation*}

The famous \emph{Channel Coding Theorem} by Shannon relates the capacity of the channel to its maximum transmission rate. In brief, the channel capacity is a tight upper bound for the maximum rate by which information can be \revision{reliably} transmitted using the channel. \revision{Given an acceptable probability of error $\xi$, there is a natural number $n$ and a coding for which $n$ uses of the channel will result in messages being transmitted with at most the acceptable probability of error $\xi$.}

\paragraph{Meaning in security} To explain what $H(A)$ represents from  the security point of view, consider a partition $\{{\mathcal A}_i\}_{i\in I}$ of $\mathcal A$. The adversary is allowed to ask questions of the form \qm{does $A \in {\mathcal A}_i$?} according to some strategy. Let $n(a)$ be the number of questions that are needed to determine the value of $a$, when $A=a$. Then $H(A)$ represents the lower bound to the expected value of $n(\cdot)$, with respect to all possible partitions and strategies of the adversary \cite{Pliam:00:INDOCRYPT,Kopf:07:CCS}.

\subsection{Min-entropy}
\label{section:minEntropy}
	
In \cite{Renyi:61:Berkeley}, R\'enyi introduced a one-parameter family of entropy measures, intended as a generalization of Shannon entropy. The R\'enyi entropy of order $\alpha$ ($\alpha > 0$,  $\alpha \neq 1$) of  a random variable $A$ is defined as
\begin{equation*}
	H_\alpha(A) \ =\  \frac{1}{1-\alpha}\log\sum_{a \,\in\,\Aset} \prob{p}{\Asym{}{}}^\alpha
\end{equation*}
	
R\'enyi's motivations were of \revision{an} axiomatic nature: Shannon entropy satisfies four axioms, namely symmetry, continuity, value $1$ on the Bernoulli uniform distribution, and the chain rule\footnote{The original axiom, called the grouping axiom, does not mention the conditional entropy. It corresponds, however, to the chain rule if the conditional entropy is defined as in (\ref{eqn:conditional}).}: 
\begin{equation}
	\label{eqn:chain}
		H(A \wedge B) \ =\  H(A) + H(B\,|\,A)
\end{equation}

The entropy of the joint probability, $H(A \wedge B)$,  is more commonly denoted by $H(A,B)$. We will use the latter notation in the following.
	
Shannon entropy is also the \emph{only} function that satisfies those axioms. If we replace, however, \eqref{eqn:chain} with a weaker property representing the additivity of entropy for independent distributions:
	\begin{equation*}	
		H(A \cdot B) \ = \ H(A) + H(B)
	\end{equation*}

\noindent then there are more functions satisfying the axioms, among which \revision{are} all those of R\'enyi's family. 
	
Shannon entropy is obtained by taking the limit of $H_\alpha$ as $\alpha$ approaches $1$. In fact we can easily prove, using l'H{\^o}pital's rule, that 
\begin{equation*}
	H_1(A) \ \defsym \lim_{\alpha\rightarrow 1}H_\alpha(A) \ =\  -\sum_{a \,\in\,\Aset} \prob{p}{\Asym{}{}} \log\,	 \prob{p}{\Asym{}{}}
\end{equation*}
	
We are particularly interested in the limit of $H_\alpha$ as $\alpha$ approaches $\infty$. This is  called \emph{min-entropy}. It can \revision{easily} be proven that
\begin{equation*}
	H_\infty(A) \ \defsym \lim_{\alpha\rightarrow \infty}H_\alpha(A) \ =\  - \log\,\max_{ \Asym{}{}\in\Aset}\,\prob{p}{\Asym{}{}}
\end{equation*}
	
R\'enyi considered also  the $\alpha$-generalization of the Kullback-Liebler divergence, which is defined  as (assuming that $p$ and $q$ are distributions on the same set $\mathcal X$):
\begin{equation*}
	D_{KL}(p\parallel q)\ = \ \sum_{x\,\in\,{\mathcal X}}p(x)\,\log\frac{p(x)}{q(x)}
\end{equation*}
	
R\'enyi's $\alpha$-generalization  is:
\begin{equation*}
	D_\alpha(p\parallel q) \ =\ \frac{1}{1-\alpha}\log\sum_{x\,\in\,{\mathcal X}}p(x)^\alpha \,q(x)^{\alpha-1}
\end{equation*}
	
The standard case, i.e. the Kullback-Liebler divergence, is again obtained by taking the limit of $D_\alpha$ as $\alpha\rightarrow 1$.
		
The interest of the above for our purposes lies on the fact that Shannon mutual information can equivalently be defined in terms of the Kullback-Liebler divergence (see  for instance \cite{Cover:91:BOOK}):
\begin{equation*}
	I(A;B) \ = \ D_{KL}(A\wedge B\parallel A\cdot B)
\end{equation*}

Therefore, it seems natural to define  the $\alpha$-generalization of the  mutual information as:
\begin{equation*}
	I_\alpha^{*}(A;B) \ = \ D_{\alpha}(A\wedge B\parallel A\cdot B)	
\end{equation*}	
	
Other $\alpha$-generalizations of the mutual information, based on the same idea, are explored in \cite{Csiszar:95:TIT}.
        
As $\alpha\rightarrow\infty$, the above definition gives the following min-version of the mutual information: 
\begin{equation}
	\label{eqn:alphaMutualInfo}
   I_\infty^{*}(A ; B) \ \defsym \lim_{\alpha\rightarrow \infty} I_\alpha(A;B) \ = \ \log \, \max_{a,b} \frac{p(a,b)}{p(a)\ p(b)}
\end{equation}

Another natural way to generalize $I(A;B)$  would be to replace $H$ by $H_\alpha$ in Definition~\ref{eqn:ShannonMutualInfo}. R\'enyi did not define, however, the $\alpha$-generalization of the  conditional entropy, and there is no agreement on what it should be. 
	
Various researchers, including Cachin~\cite{Cachin:97:PhD}, have considered  the following definition, based on \eqref{eqn:conditional}:
\begin{equation*}
 	\begin{array}{lcl}
 		H^{\textrm Cachin}_\alpha(A\mid B) &=& 
 			{\displaystyle\sum_{b\,\in\, \Bset}
			\prob{p}{\Bsym{}{}} \  H_\alpha(A \mid B = \Bsym{}{} )}
	\end{array}
\end{equation*}
	
\noindent which, as $\alpha\rightarrow \infty$, becomes
\begin{equation}
	\label{eqn:CachinCondEntropyInfty}
	\begin{array}{lclcl}
 		H^{\textrm Cachin}_\infty(A\mid B)  &=& - {\displaystyle\sum_{b\,\in\, \Bset} \prob{p}{\Bsym{}{}} \ \log\, \max_{a\in{\mathcal A}} p(a\mid b)}
	\end{array}
\end{equation}
	
An alternative proposal for $H_\infty(\cdot\mid\cdot)$ came from Smith~\cite{Smith:09:FOSSACS}\footnote{The same formulation had been already used by Dodis et al. in \cite{Dodis:04:EUROCRYPT}, and Smith proposed it independently. Since it is Smith's work on the subject that motivates the approach used in this thesis, we opt to refer to this formulation as Smith's.}:
\begin{equation}
	\label{eqn:SmithCondEntropyInfty-prem}
 	\begin{array}{lcl}
		H^{\textrm Smith}_\infty(A\mid B) \ = \  - \log \sum_{b\in {\mathcal B}} \max_{a\in {\mathcal A}} \ p(a,b)
	\end{array}
\end{equation}
	
Using \eqref{eqn:CachinCondEntropyInfty} and \ref{eqn:SmithCondEntropyInfty-prem}), and the analogue of \eqref{eqn:ShannonMutualInfo} we can define $I^{\textrm Cachin}_\infty$	and $I^{\textrm Smith}_\infty$~\footnote{The notation $I^{\textrm Smith}_\infty$ is ours. Smith himself opts for not adopting it, since $I^{\textrm Smith}_\infty$ is not symmetric.}.

\paragraph{Meaning in security}

The min-entropy can be related to a model of  adversary who is allowed to ask exactly one question, which must be of the form \qm{is $A = a?$} (one-try attacks). More precisely, the min-entropy $H_\infty(A)$ represents the (logarithm of the inverse of the) probability of success for this kind of \revision{attack} and with the best strategy, which consists, of course, in choosing the $a$ with the maximum probability. 

As for $H_\infty(A\mid B)$ and $I_\infty(A;B)$, the most interesting versions in terms of security seem to be those of Smith. In fact, in this thesis we adopt his approach to information leakage, and we will, from now on, use the following notation:

\begin{itemize}
	\item $H_\infty(A\mid B)$ stands for $H^{\textrm Smith}_\infty(A\mid B)$ and is referred to as \emph{conditional min-entropy};
	
	\item $I_\infty(A;B)$ stands for $I^{\textrm Smith}_\infty(A;B)$ and is referred to as \emph{min-entropy leakage}.
\end{itemize}

In fact, the conditional min-entropy $H_\infty(A\mid B)$ represents \revision{the log of} the inverse of the (expected value of the) probability that the same kind of adversary succeeds in guessing the value of $A$ \emph{a posteriori}, i.e.  after observing the  result of $B$. The complement of this probability is also known as \emph{probability of error} or \emph{Bayes risk}. Since in general $B$ and $A$ are correlated, observing $B$ increases the probability of success. In fact, we can prove formally that $H_\infty(A\mid B) \leq H_\infty(A)$, with equality if $A$ and $B$ are independent. The min-entropy leakage $I_\infty(A;B)$ corresponds to the \emph{ratio} between the probabilities of success a priori and a posteriori, which is a natural notion of leakage. Here $I_\infty(A;B)$ is in the format of \eqref{eqn:leakage}, but the difference becomes a ratio due to the presence of the logarithms. Note that $I_\infty(A;B)\geq 0$, which seems desirable for a good notion of leakage. It has been proven in~\cite{Braun:09:MFPS} that $C_\infty$ is obtained at the uniform distribution, and that it is equal to the sum of the maxima of each column in the channel matrix, i.e. $C_\infty = \sum_{b \,\in\,\mathcal{B}} \max_{a \,\in\,\mathcal{A}} p(b\mid a)$.

The definition of $I_\infty^{*}(A;B)$ in \eqref{eqn:alphaMutualInfo} has also an interpretation in security: it represents the maximum gain in the probability of success, i.e. the maximum ratio between the a posteriori and the a priori probability. Note that also $I_\infty^{*}(A;B)$ is always non-negative and it is $0$ if and only if $A$ and $B$ are independent. \revision{More generally}, $D_{KL}(p\parallel q)$ and its $\alpha$-extension $D_{\alpha}(p\parallel q)$ should represent the \qm{inefficiency} of an adversary who bases its strategy on the distribution $q$, when in fact the real distribution is $p$. Hence $I_{\alpha}^{*}(A;B)$ defined as $D_{\alpha}(A\wedge B\parallel A\cdot B)$ should represent the gain of the adversary in revising his strategy according to the knowledge of the correlation between $A$ and $B$. 

Concerning $H^{\textrm Cachin}_\alpha$ and $I^{\textrm Cachin}_\alpha$, they have some nice properties. For instance they enjoy weak versions of the chain rule \eqref{eqn:chain}. More precisely, the \qm{$=$} in \eqref{eqn:chain} becomes \qm{$\geq$} for $\alpha < 1$, and \qm{$\leq$} for $\alpha > 1$. There is no general relation between $H^{\textrm Cachin}_\infty(A\mid B)$ and $H_\infty(A)$, and in particular $I^{\textrm Cachin}_\infty$ is not guaranteed to be non-negative. 

\subsection{Guessing entropy}
\label{section:guessing}

The notion of guessing entropy was introduced by Massey in \cite{Massey:94:IT}. Let us assume, for simplicity, that the elements of $\mathcal A$ are ordered  by decreasing probabilities, i.e. if $1\leq i<j\leq n$ then $p(a_i)\geq p(a_j)$. Then the guessing entropy is defined as follows:
\begin{equation*}
	H_G(A) \, =\, \sum_{1\leq i \leq |{\mathcal A}|} i \, p(a_i)
\end{equation*}

Massey did not define the notion of conditional guessing entropy. In some works, like \cite{Cachin:97:PhD,Kopf:07:CCS}, it is defined analogously to \eqref{eqn:conditional}:
\begin{equation*}
	H_G(A\mid B)  \ =\ {\displaystyle\sum_{b\,\in\, \Bset} \prob{p}{\Bsym{}{}} \  H_G(A \mid B = \Bsym{}{} )}
\end{equation*}

\paragraph{Meaning in security}

Guessing entropy represents an adversary who is allowed to ask repeatedly questions of the form \qm{is $A = a?$}. More precisely, $H_G(A)$ represents the expected number of questions that the adversary needs to ask to determine the value of $A$, assuming that he follows  the best strategy, which consists, of course, in choosing the $a$'s in order of decreasing probability. 

$H_G(A\mid B)$ represents the expected number of questions \emph{a posteriori}, i.e. after observing the value of $B$ and reordering the queries according to the updated probabilities (i.e. the queries will be chosen in order of decreasing a posteriori probabilities). 

Also in this case, $H_G(A\mid B)$ is not necessarily smaller than or equal to $H_G(A)$, so the corresponding notion of mutual information is not guaranteed to be non-negative\footnote{This problem is inherent to the probabilistic case, and therefore it does not occur in \cite{Kopf:07:CCS}, since that work considers only deterministic systems.}. 

\subsection{Marginal guesswork}

The marginal guesswork is a variant of guessing entropy that was proposed by Pliam  \cite{Pliam:00:INDOCRYPT}. It is parametric \revision{in} a number $\eta>0$, and is defined as follows. Again, we assume that the elements of $\mathcal A$ are ordered by decreasing probabilities.
\begin{equation*}
	H_\eta(A) \, =\,\min\{j\mid  \sum_{1\leq i \leq j}  \, p(a_i) > \eta\}
\end{equation*}	

Pliam did not define the  conditional version of marginal guesswork, but in \cite{Kopf:07:CCS}  it is defined following \eqref{eqn:conditional}:
\begin{equation*}
	H_\eta(A\mid B)  \ =\ {\displaystyle\sum_{b\,\in\, \Bset} \prob{p}{\Bsym{}{}} \  H_\eta(A \mid B = \Bsym{}{} )}
\end{equation*}	
	
\paragraph{Meaning in security}

Consider again an adversary who is allowed to ask repeatedly questions of the form \qm{is $A = a?$}. $H_\eta(A)$ represents the minimum number of questions that the adversary needs to ask to determine the value of $A$ with probability at least $\eta$. 

$H_\eta(A\mid B)$ represents the same notion, but using the a posteriori probabilities. Again, it is not necessarily the case that $H_\eta(A\mid B) \leq H_\eta(A)$. 

\subsection{Comparison and discussion}

The  various notions of entropy discussed in this section have been carefully compared with Shannon entropy, to conclude that in general there is no tight relation. Fano's inequality gives a lower bound to the Bayes risk in terms of (conditional) Shannon entropy, and R\'enyi \cite{Renyi:61:Berkeley}, Hellman-Raviv \cite{Hellman:70:TIT}, and Santhi-Vardi \cite{Santhi:06:ACCCC} give upper bounds as well, but all these are rather weak. Smith has shown in \cite{Smith:09:FOSSACS} that the orderings induced on channels by the Bayes risk and by Shannon entropy are in general unrelated.

Massey has shown that the exponential of the Shannon entropy is a lower bound for the guessing entropy, and that, in case of a geometric distribution, the bound is tight. Massey has also shown that in the general case the Shannon entropy can be arbitrarily close to $0$ while the guessing entropy is constant \cite{Massey:94:IT}.

As for the marginal guesswork. Pliam has shown that it is essentially unrelated with Shannon entropy \cite{Pliam:00:INDOCRYPT}.

In this thesis we focus on the concepts of leakage based on Shannon entropy (Chapter~\ref{chapter:interactive-systems}) and min-entropy (Chapter~\ref{chapter:differential-privacy}).

\chapter{Information flow in interactive systems}
\label{chapter:interactive-systems}
\mscite{True interactivity is not about clicking on icons or downloading files, \\ it's about encouraging communication.}{Edwin Schlossberg}
The key idea behind the information-theoretic approaches to information flow is to interpret the system as an information-theoretic channel, where the secrets are the input and the observables are the output. The channel matrix consists of the conditional probabilities $p(b\,|\, a)$, defined as the  measure of the executions producing the observable $b$, relative to those which contain  the secret $a$. The leakage is represented by the mutual information, and the worst-case leakage by the capacity of the channel \revision{(see Chapter~\ref{chapter:probabilistic-info-flow} for reference)}. 

In information theory, however, there are several different models of channels. So far the works in the literature about information theory applied to information flow have focused on the simplest kind of channels: discrete memoryless channels where the absence of feedback is implicitly assumed. This classical approach has been successfully used in scenarios where the secret value is assumed to be chosen at the beginning of the computation. In this chapter, however, we are interested in the more general scenario in which secrets can be chosen at any point. More precisely, we consider \emph{interactive systems}, i.e. systems in which the generation of secrets and the occurrence of observables can alternate during the computation and influence each other. Examples of interactive systems include \emph{auction protocols} like \cite{Vickrey:61:JoF,Subramanian:98:SRDS,Stajano:99:IH}. Some of these have become very popular thanks to their integration in  Internet-based electronic commerce platforms~\cite{eBay,eBid,Mercadolibre}. Other examples of interactive programs include web servers, GUI applications, and command-line programs \cite{Bohannon:09:CCS}.

Unfortunately, the information-theoretic approach which interprets interactive systems as classical channels is not valid. More specifically, in such systems the channel matrix is not invariant with \revision{respect} to the input distribution, so the channel capacity cannot be calculated in the traditional way. Therefore, the notion of maximum leakage as standard capacity is also compromised.

The goal of this chapter is to extend the classical information-theoretic approach to information flow to the more complicated scenario of interactive systems. 

\textbf{Contribution} The main contributions of this chapter can be summarized as follows.

\begin{itemize}
	
	\item We show that by considering the richer channels that support memory and feedback it is possible to retrieve the correspondence between systems and channels. We prove that there is a complete correspondence between interactive systems and channels with memory and feedback, and we show how to model the latter as the former.
	
	\item We propose the use of directed information, as opposed to mutual information, to represent leakage in interactive systems. Recent results in information theory~\cite{Tatikonda:09:TIT} have shown that, in channels with memory and feedback, the transmission rate does not correspond to the maximum mutual information (the standard notion of capacity), but rather to the maximum normalized \emph{directed information}, a concept introduced by Massey~\cite{Massey:90:SITA}. We argue that in interactive channels the real leakage is due to the directed information from secrets to observables, whereas the directed information from observables to secrets (corresponding to feedback) is a characteristic of the system itself and should not be counted as leakage.
	
	\item We show that our model is a proper extension of the classical one, i.e. in the absence of interactivity the model of channels with memory and feedback collapses into the model of memoryless channels without feedback. Moreover, in that case also the concepts of mutual information and directed information from input to output coincide, the same \revision{holds} for the concepts of capacity and directed capacity. We argue that in the classical approach mutual information is a good measure of leakage exactly because of this property: in the absence of feedback mutual information and directed information from input to output are the same.
	
	\item We show that the capacity of the channels associated to interactive systems is a continuous function with respect to a pseudometric based on the Kantorovich metric. The continuity of the channel capacity was also proved in \cite{Desharnais:02:LICS} for simple channels, but the proof does not adapt to the case of channels with memory and feedback and we had to devise a different technique.

\end{itemize}

\textbf{Plan of the Chapter} This chapter is organized as follows. In Section~\ref{section:int-systems} we introduce the concept of interactive systems and we show why channels without memory and feedback are inadequate in this scenario. In Section~\ref{section:dcmfb} we review the notion of channels with memory and feedback, which is the core of the model we propose. We discuss the concept of directed information and also the concept of capacity in the presence of feedback. Section~\ref{section:our-model} contains the main contribution in this chapter: We explain how  Interactive Information Hiding Systems (\IIHSs) can be modeled using channels with memory and feedback. In particular we show that for any \IIHS there is always a channel that simulates its probabilistic behavior. In Section~\ref{section:anonymity-properties} we discuss our notion of adversary and we define the quantification of information leakage as the channel's directed information from input to output, or as the directed capacity, depending on whether the input distribution is fixed or not. In Section~\ref{section:full-example} we apply our model to an example, the Cocaine Auction protocol. In Section~\ref{section:topological-properties} we propose a pseudometric structure on \IIHSs  based on the Kantorovich metric. We also show that the capacity of the channels associated to interactive systems is a continuous function with respect to this pseudometric. In Section~\ref{section:related-work-interactive-systems} we present some related work, and in Section~\ref{section:conclusion-interactive-systems} we review and discuss the main results of the chapter, and consider future work.

\section{Interactive systems}
\label{section:int-systems}

In this section we exemplify the problems that arise when we try to apply the classical information-theoretic approach to interactive systems. In order to derive an information-theoretic channel, at a first glance it would seem natural to  define the channel matrix by using the definition of $p(b\,|\, a)$ in terms of the joint and marginal probabilities $p(a, b)$ and $p(b)$. Namely,  the entry $p(b\,|\, a)$ would be defined as the measure of the traces with (secret, observable)-projection $(a,b)$, divided by the measure of the traces with secret projection $a$. An approach of this kind was proposed in \cite{Desharnais:02:LICS}. In the interactive case, however, this construction does not really produce an information-theoretic channel. In fact, by definition a channel should be invariant with respect to the input distribution, and this is not the case here, as shown by the following example.

\begin{example} 
	\label{exa:website}
	Figure~\ref{fig:website-example-short} represents a web-based interaction between one seller and two possible buyers, \emph{rich} and \emph{poor}. The seller can offer two   different products, \emph{cheap} and \emph{expensive}, with given probabilities. Once the product is offered, each buyer may try to buy it, with a certain probability. For simplicity we assume that the buyers' offers are mutually exclusive. We assume that the offers are observables, in the sense that they are made public on the website, while the identity of the buyer that actually buys the product should be kept  secret from an external observer.  The symbols $r$, $q_1$, $q_2$, $\overline{r}$, $\overline{q_1}$, $\overline{q_2}$ represent probabilities, with the convention that $\overline{r} = 1 - r$ (and the same for the pairs $q_1$, $\overline{q_1}$ and $q_2$, $\overline{q_2}$).
\end{example}

	\begin{figure}[!htb]
		\centering
		\includegraphics[width=0.30\textwidth]{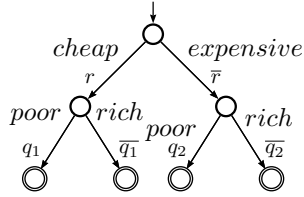}
		\caption{Interactive system of Example~\ref{exa:website}}
		\label{fig:website-example-short}
	\end{figure}

Following~\cite{Desharnais:02:LICS} we can compute the conditional probabilities \revision{using} $\Cprob{p}{b}{a} = \frac{\prob{p}{a,b}}{\prob{p}{a}}$, thus obtaining the matrix in Table~\ref{tab:matrix}.	The matrix  however is not invariant with respect to  the input distribution. For instance for $r = \overline{r} = \frac{1}{2}$, $q_1=\frac{2}{3}$, and $q_2=\frac{1}{3}$ we obtain the  matrix in Table~\ref{tab:matrices-a}. If we change the input distribution, for instance by changing the value of $q_2$ to be   $\frac{1}{6}$, also the matrix changes. We obtain, indeed, the new matrix illustrated in Table~\ref{tab:matrices-b}.
	
	\begin{table}[!htb]
		\centering       
		\begin{tabular}{|c||c|c|}
			 \hline
			 		& ${\;\textit{cheap}\;}$ & $\;\textit{expensive}\;$ \\
			 \hline \hline
			 		& & \\ [-2ex]
			 		$\textit{poor}$ &	$\frac{r q_1}{r q_1+\overline{r} q_2}$ & $\frac{\overline{r} q_2 }{r q_1+\overline{r} q_2}$ \\
			 		& & \\[-2ex]
			 \hline
			 		& & \\[-2ex]
			 		$\textit{rich}$ & $\frac{r \overline{q_1}}{r \overline{q_1}+\overline{r} \overline{q_2}}$ & $\frac{\overline{r} \overline{q_2}}{r \overline{q_1}+\overline{r} \overline{q_2}}$ \\
					[-2ex] 
					& & \\
			 \hline
		\end{tabular}
		\caption{Channel matrix for Example~\ref{exa:website}}
		\label{tab:matrix}
	\end{table}	
	
Consequently, when the secrets occur \emph{after} the observables and \emph{depend on them}, we cannot consider the conditional probabilities (of the observables given the secrets) as representing a classical channel from secrets to observables, and we cannot apply the standard information-theoretic concepts. In particular, we cannot  use \qm{the capacity of the matrix} (defined by considering the matrix as a channel matrix, and taking the maximum mutual information over all possible inputs) because in general the maximum is given by a distribution different from the one that \revision{was used to define} the matrix, hence the result would be unsound.

	\begin{table}[!htb]
			\centering 
			\subbottom[$r = \frac{1}{2}, q_1 = \frac{2}{3}, q_2 = \frac{1}{3}$]{ 
				\centering
						\begin{tabular}{|c||c|c||c|}
						\hline
										 & $\;cheap\;$ & $\;expensive\;$ &\; Input distr. \;\\ \hline \hline
										 & & & \\ [-2ex]
								$poor$ & $\frac{2}{3}$  & $\frac{1}{3}$  & $\;\prob{p}{poor} = \frac{1}{2}\;$ \\
							 			 & & & \\ [-2ex] \hline & & & \\ [-2.ex]							 			
							 	$rich$ & $\frac{1}{3}$ & $\frac{2}{3}$ & $\prob{p}{rich} = \frac{1}{2}$ \\
							 			 [-2ex] & & & \\ \hline
						\end{tabular}
				\label{tab:matrices-a}
			} \subbottom[$r = \frac{1}{2}, q_1 =	\frac{2}{3}, q_2 = \frac{1}{6}$]{
				\centering
						\begin{tabular}{|c||c|c||c|}
						\hline
										 & $\;cheap\;$ & $\;expensive\;$ &\; Input distr. \;\\ \hline \hline
										 & & & \\ [-2ex]
								$poor$ & $\frac{4}{5}$  & $\frac{1}{5}$  & $\;\prob{p}{poor} = \frac{5}{12}\;$ \\
									   & & & \\ [-2ex] \hline & & & \\ [-2.ex]							 			
								$rich$ & $\frac{2}{7}$ & $\frac{5}{7}$ & $\prob{p}{rich} = \frac{7}{12}$ \\
										 [-2ex] & & & \\ \hline
						\end{tabular}
				\label{tab:matrices-b}
			}
			\caption{Two different channel matrices	induced by two different input distributions for Example~\ref{exa:website}}
			\label{tab:matrices}
		\end{table}		

The first contribution of this chapter is to consider an extension of the theory of channels which  makes the information-theoretic approach applicable also in the case of interactive systems. A richer notion of channels, known in information theory as \emph{channels with memory and feedback}, serves our purposes. The dependence of inputs on previous outputs corresponds  to feedback, and the dependence of outputs on previous inputs and outputs corresponds to memory.	Recent results in information theory~\cite{Tatikonda:09:TIT} have shown that, in such channels, the transmission rate does not correspond to the maximum mutual information (the standard notion of capacity), but rather to the maximum normalized \emph{directed information}, a concept introduced by Massey~\cite{Massey:90:SITA}. We propose to adopt this latter notion to represent leakage.

Our model of attacker is the interactive version of the attacker associated to Shannon entropy in the classification of K\"opf and Basin \cite{Kopf:07:CCS}, discussed in Chapter~\ref{chapter:probabilistic-info-flow}. In the case of a standard single-use channel, the invulnerability degree of the secret \emph{before} the attacker observes the output is the entropy of the input, determined by its a priori distribution. The  invulnerability degree \emph{after} the attacker observes the output is the conditional entropy of the input given the output, determined by its a posteriori distribution. The latter is always smaller than or equal to the first. The difference between these invulnerability degrees corresponds to the mutual information, and represents the leakage of the system. In our interactive framework we consider the same scenario, but iterated. At each time step, we consider the input sequence so far; and the increase of its vulnerability caused by the observation of the new output is given by the contribution of the present step to the leakage. The sum of all  these contributions represents the total leakage and, as we will see, corresponds to Massey's directed information. We will come back to the  model of attacker in Section~\ref{section:anonymity-properties}, and discuss also a variant of this interpretation.

A second contribution of our work is the proof that the channel capacity is a continuous function of a pseudometric on interactive systems based on the  Kantorovich metric. The reason why we are interested in the continuity of the capacity is for computability purposes. Given a function $f$ from a (pseudo)metric space ${X}$ to a (pseudo)metric space ${Y}$ the continuity of $f$ means that, given a sequence of objects $x_1, x_2,\ldots \in \mathcal{X}$ converging to $x\in \mathcal{X}$, the \revision{sequence} $f(x_1), f(x_2),\ldots \in \mathcal{Y}$ converges to $f(x)\in \mathcal{Y}$. Hence $f(x)$ can be approximated by the objects $f(x_1), f(x_2), \ldots$. The typical use of this property is in the case of execution trees generated by programs containing loops. Generally the automaton expressing the semantics of the program can be seen as the (metric) limit of the sequence of trees generated by unfolding the loop to an increasingly deeper level. The continuity of the capacity means that we can approximate the real capacity by the capacities of these trees.
 
\section{Discrete channels with memory and feedback}
\label{section:dcmfb}

In this section we present the notion of channel with memory and feedback. We assume a scenario in which the channel is used repeatedly, in a finite temporal sequence of steps $1,\ldots , T$. Intuitively, memory means that the output at time $t (1 \leq t \leq T)$ depends on the input and output histories, i.e. on the inputs up to time $t$, and on the output up to time $t-1$. Feedback means that the input at time $t$ depends on the outputs up to time $t-1$. 
	
We adopt the following notation.
	
\begin{convention}
	\label{convention:notation1}
	Given sets of symbols (alphabets) \revision{$\Aset = \{\Asym{1}{}, \ldots, \Asym{n}{}\}$}, $\Bset = \{\Bsym{1}{}, \ldots, \Bsym{n}{}\}$, we use a Greek letter ($\alpha$, $\beta$, \ldots) to denote a sequence of symbols ordered in time. Given a sequence $\alpha = a_{i_1} a_{i_2} \ldots  a_{i_m}$,  the notation $\Aseq{t}{}$ represents the symbol at time $t$, i.e. $a_{i_t}$, while $\Aseq{}{t}$ represents the sequence $\Aseq{i_1}{} \Aseq{i_2}{} \ldots \Aseq{i_t}{}$. For instance, in the sequence $\alpha = \Asym{3}{} \Asym{7}{} \Asym{5}{}$, we have $\Aseq{2}{} = \Asym{7}{}$ and $\Aseq{}{2} = \Asym{3}{} \Asym{7}{}$. Analogously, if	$X$ is a random variable, then $X^{t}$ denotes the sequence of $t$ consecutive instances $X_{1},\ldots,X_{t}$ of $X$.	
\end{convention}
	
We now define formally the concepts of memory and feedback. Consider a channel from input $A$ to output $B$. The channel behavior after $T$ uses can be fully described by the joint distribution of $A^T\times B^T$, namely by the probabilities $\prob{p}{\Aseq{}{T},\Bseq{}{T}}$. Using the chain rule, we can decompose these probabilities as follows:
\begin{equation}
	\label{eq:joint-prob}
		\prob{p}{\Aseq{}{T},\Bseq{}{T}} = \prod_{t=1}^{T} \Cprob{p}{\Aseq{t}{}}{\Aseq{}{t-1},\Bseq{}{t-1}}\Cprob{p}{\Bseq{t}{}}{\Aseq{}{t},\Bseq{}{t-1}} 
\end{equation}
			
\begin{definition}We say that a channel \emph{has feedback} if, in general, \linebreak $\Cprob{p}{\Aseq{t}{}}{\Aseq{}{t-1},\Bseq{}{t-1}}\neq \Cprob{p}{\Aseq{t}{}}{\Aseq{}{t-1}}$, i.e. the probability of $\Aseq{t}{}$  depends not only on  $\alpha^{t-1}$, but also on $\beta^{t-1}$. Analogously, we say that the channel \emph{has memory} if,  in general, $\Cprob{p}{\Bseq{t}{}}{\Aseq{}{t},\Bseq{}{t-1}}\neq \Cprob{p}{\Bseq{t}{}}{\Aseq{t}{}}$, i.e. the probability of $\Bseq{t}{}$  depends on $\Aseq{}{t}$ and  $\Bseq{}{t-1}$.
\end{definition}
	
Note that in the opposite case, i.e. when	$\Cprob{p}{\Aseq{t}{}}{\Aseq{}{t-1},\Bseq{}{t-1}}$ coincides with $\Cprob{p}{\Aseq{t}{}}{\Aseq{}{t-1}}$ and $\Cprob{p}{\Bseq{t}{}}{\Aseq{}{t},\Bseq{}{t-1}}$ coincides with $\Cprob{p}{\Bseq{t}{}}{\Aseq{t}{}}$, we have a classical channel (memoryless, and without feedback), in which each use is independent from the previous ones. The only possible dependency on the history is the one of  $a_t$ on $a^{t-1}$. This is because $A_1, \ldots, A_T$ are in general correlated, due to the fact that they are produced by an encoding function. Note that in absence of memory and feedback \eqref{eq:joint-prob} reduces to:	
\begin{align}
	\label{eq:classical-joint}
	\nonumber \prob{p}{\Aseq{}{T},{\Bseq{}{T}}} & = \prod_{t=1}^{T} \Cprob{p}{\Aseq{t}{}}{{\Aseq{}{t-1}}}\,\Cprob{p}{\Bseq{t}{}}{\Aseq{t}{}} & \text{} \\
			                                & = \prob{p}{\Aseq{}{T}}  \prod_{t=1}^{T} \Cprob{p}{\Bseq{t}{}}{\Aseq{t}{}} & \text{(by the chain rule)}
\end{align}	
 		
\noindent from which we can derive the standard formula for a classical channel after $T$ uses.
\begin{align*}
	\Cprob{p}{\Bseq{}{T}}{\Aseq{}{T}} & = \frac{\prob{p}{\Aseq{}{T},{\Bseq{}{T}}}}{\prob{p}{\Aseq{}{T}}} & \text{} \\
			                              & = \prod_{t=1}^{T} \Cprob{p}{\Bseq{t}{}}{\Aseq{t}{}} & \text{(by \eqref{eq:classical-joint})} \\
\end{align*}		
	
So far we have given a very abstract description of a channel with memory and feedback. We now discuss a more concrete notion  following the presentation of ~\cite{Tatikonda:09:TIT}. Such a channel, represented in Figure~\ref{fig:dcmfb}, consists of  a sequence of components formally defined as  a family of stochastic kernels  $\{ \Cprob{p}{\cdot\,}{\Aseq{}{t},\Bseq{}{t-1} } \}_{t=1}^{T}$ over $\Bset$. 

The probabilities $\Cprob{p}{\Bseq{t}{}}{\Aseq{}{t},\Bseq{}{t-1}}$ represent \emph{innermost behavior} of the channel at time $t$, $1 \leq t \leq T$: the internal channel takes the input $\Aseq{t}{}$ and, depending on the history of inputs and outputs so far, it produces an output symbol $\Bseq{t}{}$. The output is then fed back to the encoder with delay one. On the input side, at time $t$ the encoder takes the message and the past output symbols $ \Bseq{}{t-1}$ and produces a channel input symbol $\Aseq{t}{}$ according to the code function $\Fseq{t}{}$ (we will explain this concept in the next paragraph). At final time $T$ the decoder takes all the channel outputs $\Bseq{}{T}$ and produces the decoded message $\hat{W}$. The order in time is the following:
	
	$$ \mbox{Message \ } W, \quad \Aseq{1}{}, \Bseq{1}{},  \quad\Aseq{2}{},	\Bseq{2}{},  \quad\ldots, \quad \Aseq{T}{}, \Bseq{T}{}, \quad \mbox{Decoded Message \ } \hat{W} $$
	
	\begin{figure}[!htb]
  	\centering
  	\includegraphics[width=\linewidth]{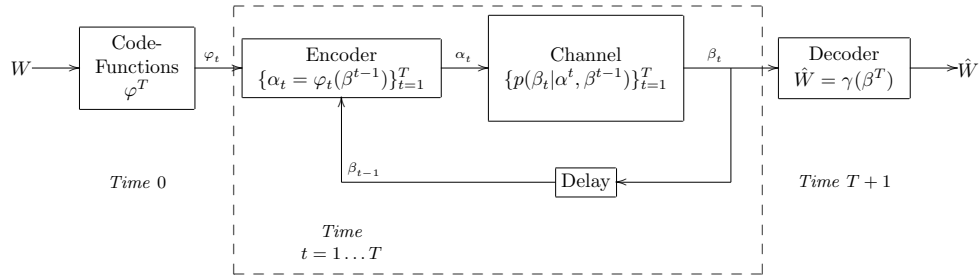}	
		\caption{Model for discrete channel with memory and feedback}
		\label{fig:dcmfb}
	\end{figure}	
	
	Let us now explain the concept of code function. Intuitively, a code function is a strategy to encode the message into a suitable representation to be transmitted through the channel. There is a code function for each possible message, and the functions are fixed at the very beginning of the transmission (time $t=0$). The encoding, however, can use the information provided via feedback, so each component $\Fseq{t}{}$ ($1 \leq t \leq T$) of the code function takes as parameter the history of feedback $\Bseq{}{t-1}$ to generate the next input symbol $\Aseq{t}{}$.
			
Formally, let $\Fset_{t}$ be the set of all measurable maps $\Fseq{t}{}: \Bset^{t-1} \rightarrow \Aset$ endowed with a probability distribution, and let $F_t$ be the corresponding random variable. Let $\Fset^{T}$, $F^T$ denote the Cartesian product on the domain and the random variable, respectively. A \emph{channel code function}  is an element $\Fseq{}{T} = ( \Fseq{1}{}, \ldots, \Fseq{T}{} ) \in \Fset^{T}$.
	
Note that, by the chain rule, $\prob{p}{\Fseq{}{T}} = \prod_{t=1}^{T} \Cprob{p}{\Fseq{t}{}}{\Fseq{}{t-1}}$. Hence the distribution on $\Fset^{T}$ is uniquely determined by a sequence $\{\Cprob{p}{\Fseq{t}{}}{\Fseq{}{t-1}} \}_{t=1}^{T}$. The notation $\Fseq{}{t}(\Bseq{}{t-1})$ will represent the $\mathcal A$-valued $t$-tuple  $( \Fseq{1}{}, \Fseq{2}{}(\Bseq{}{1}), \ldots, \Fseq{t}{}(\Bseq{}{t-1}) )$.

In Information Theory this kind of channel is used to encode and transmit messages. If $\Wset$ is a set of messages of cardinality $M$ with typical element $w$, endowed with a probability distribution, a \emph{channel code} is a set of $M$ channel code functions $\Fseq{}{T}[w]$, interpreted as follows: for message $w$, if at time $t$ the channel feedback is $\Bseq{}{t-1}$, then the channel encoder outputs $\Fseq{t}{}[w](\Bseq{}{t-1})$. A \emph{channel decoder} is a map from $ \Bset^{T}$ to $ \Wset$ which attempts to reconstruct  the input message after observing all the output history $\Bseq{}{T}$ from the channel.
	
\subsection{The power of feedback}
\label{section:channel-feedback-example}

The original purpose of \emph{communication channel} models is to represent data transmission from a source to a receiver. Shannon's Channel Coding Theorem states that for every channel there is an encoding scheme that allows a transmission rate arbitrarily close to the channel capacity with a negligible probability of error (if the number of uses of the channel is large enough). A general way to find an optimal encoding scheme that is also easy to decode has not been found yet. The use of feedback, however, can simplify the design of the encoder and of the decoder. The following example illustrates the idea. 

\begin{table}[!htb]
	\centering
	\begin{tabular}{|c||c|c|c|}
		\hline
	 	 & $0$ & $1$ & \texttt{e} \\ \hline \hline
	 	 $0$ & $0.8$ & 0 & $0.2$ \\ \hline
	 	 $1$ & 0 & $0.8$ & $0.2$ \\ \hline
	\end{tabular}
	\caption{Channel matrix for binary erasure channel}
	\label{tab:example}
\end{table}

\begin{example}
	\label{exa:erasure}
		
		Consider a discrete memoryless binary channel $\{\Aset, \Bset, \Cprob{p}{.}{.} \}$ with $\Aset = \{0,1\}$, $\Bset = \{0,1,\texttt{e}\}$ and the channel matrix of Table~\ref{tab:example}. This kind of channel is called \emph{erasure channel} because it can lose (or \emph{erase}) bits  during the transmission with a certain probability.  Namely, any bit has $0.8$ probability of being correctly transmitted, and $0.2$ probability of being lost. On the output side the encoder is able to detect whether the bit was erased (by receiving an \texttt{e} symbol), but it cannot tell which was the actual value of the original bit. The Channel Coding Theorem guarantees that the maximum information transmission rate in this channel is ($2$ to the power of) the channel capacity, i.e. $0.8$ bits per use of the channel.
	
		Following simple principles described in~\cite{Cover:06:BOOK}, an encoding that achieves the capacity can be easily obtained if the channel can be used with feedback. The idea is an adaptation of the stop-and-wait protocol~\cite{Sallings:06:BOOK,Tanenbaum:89:BOOK}. Suppose that every bit received on the output end of the channel is fed back noiselessly to the source with delay $1$. Define the encoding as follows: for each bit transmitted, the encoder checks via feedback whether the bit was erased. If not, the encoder moves on to transmit the text of the message. If yes, the encoder transmits the same bit again.
	
		It is easy to see that with this encoding scheme the transmission rate is $0.8$ bit per usage of the channel, since in $80\%$ of the cases the bit is transmitted properly, and in $20\%$ it is lost and a retransmission is needed. 
\end{example}

We now proceed to illustrate in more detail the design and the function of the  encoder and decoder. 

\subsubsection{An example illustrating the the encoder/decoder design}
\label{section:encoder-decoder}

We proceed with the erasure channel of Example~\ref{exa:erasure} to show how the enriched model of channels with memory and feedback can be used to transmit the message, and in particular how the feedback can be used to design the encoder. We assume that the set $\Wset$ of possible messages consists of all finite sequences of bits. The role of the code functions is to encode the message $W$ into a suitable representation for the stochastic kernels within the channel. The input and output alphabets for the stochastic kernels are $\Aset = \{0,1\}$ and $\Bset = \{0,1, \texttt{e} \}$, respectively. We assume that at most $T$ uses of the channel are allowed and we use $t$, with $1\leq t \leq T$, to represent the $t^{th}$ time step.

We consider a sort of memory that depends only on the input history and we abstract from its specific form by defining a function $\eta: \pow(\Aset^{t}) \rightarrow [0,1]$ that maps each possible input history to a correction factor to be added to (or subtracted from) a base probability value. We compute the contribution of $\eta$ to the base values using \revision{arithmetic} modulo $2$, in such a way that the resulting values are still a probability distribution. More precisely, the stochastic kernels are defined as follows.
\begin{equation}
	\begin{array}{l}
		\Cprob{p}{\Bseq{t}{}=0}{\Aseq{}{t-1}0, \Bseq{}{t-1}} = 0.8 - \eta(\Aseq{}{t-1}) \\
		\Cprob{p}{\Bseq{t}{}=1}{\Aseq{}{t-1}0, \Bseq{}{t-1}} = 0 \\
		\Cprob{p}{\Bseq{t}{}=\texttt{e}}{\Aseq{}{t-1}0, \Bseq{}{t-1}} = 0.2 + \eta(\Aseq{}{t-1}) \\
		\Cprob{p}{\Bseq{t}{}=0}{\Aseq{}{t-1}1, \Bseq{}{t-1}} = 0 \\
		\Cprob{p}{\Bseq{t}{}=1}{\Aseq{}{t-1}1, \Bseq{}{t-1}} = 0.8 - \eta(\Aseq{}{t-1}) \\
		\Cprob{p}{\Bseq{t}{}=\texttt{e}}{\Aseq{}{t-1}1, \Bseq{}{t-1}} = 0.2 + \eta(\Aseq{}{t-1}) \\
	\end{array}
\end{equation}

Correspondingly, the general form of the channel matrix for each time $1 \leq t \leq T$ is shown in Table~\ref{tab:general-form}.

\begin{table}[!htb]
	\centering
			\setlength{\extrarowheight}{3pt}	
			$$
				\begin{array}{|c||c|c|c|}
					\hline
															 				 & 0 & 1 & \texttt{e} \\ \hline \hline
						\Aseq{t}{}=0, \Bseq{}{t-1} & 0.8 - \eta(\Aseq{}{t-1}) & 0 & 0.2 + \eta(\Aseq{}{t-1}) \\ \hline
						\Aseq{t}{}=1, \Bseq{}{t-1} & 0 & 0.8 - \eta(\Aseq{}{t-1}) & 0.2 + \eta(\Aseq{}{t-1}) \\ \hline  
				\end{array}
			$$
			\setlength{\extrarowheight}{0pt}	
	\caption{General form of channel matrix}
	\label{tab:general-form}
\end{table}

The code functions are chosen at time $t=0$, based on the message to be transmitted. For illustration purposes, let us suppose that the message is the sequence of three bits $W = 011$. The other cases of $W$ are analogous. 

At time $t=1$, the channel is used for its first time and the feedback history so far is empty $\Bseq{}{0} = \epsilon$. The encoder selects the input symbol $\Aseq{0}{} = 0$, as in \eqref{eq:f1}.
\begin{equation}
	\label{eq:f1}
	\Fsym{1}{}[W=011](\Bseq{}{0}=\epsilon) = 0
\end{equation}

At time $t=2$, the feedback history consists of only one symbol, and in principle the possibilities are either $\Bseq{}{1} = 0$, $\Bseq{}{1} = 1$ or $\Bseq{}{1} = \texttt{e}$. In the first case, the first bit was successfully transmitted and the encoder can go on to the second bit of the message. By the way the channel is defined, the second case is not really possible, so it is not important how the reaction function is defined for this case. We will denote this indifference by attributing to the function the symbol $\mathtt{x}$ instead of a $0$ or a $1$. In the last case, $\Bseq{}{1} = \texttt{e}$, the first bit was erased and the encoder tries to retransmit the bit $0$. We can write it formally as below.
\begin{equation}
	\label{eq:f2}
	\begin{array}{l}
		\Fsym{2}{}[W=011](\Bseq{}{1}=0) = 1 \\
		\Fsym{2}{}[W=011](\Bseq{}{1}=1) = \mathtt{x} \\
		\Fsym{2}{}[W=011](\Bseq{}{1}=\texttt{e}) = 0   \\
	\end{array}
\end{equation}

At time $t=3$ the feedback histories allowed by the channel are $\Bseq{}{2} \in \{01, 0\texttt{e}, \texttt{e}0, \texttt{e}\texttt{e}\}$ (the other ones have zero probability). In the first case, $\Bseq{}{2} = 01$ the two first bits of the message have been transmitted correctly and the encoder can send the third bit. If $\Bseq{}{2} = 0\texttt{e}$, the transmission of the first bit was successful, but the second bit was erased and needs to be resent. In the case $\Bseq{}{2} = \texttt{e}0$, the first bit was erased in the first try but was successfully transmitted in the second try, so now the encoder can move to the second bit of the message. In the last case, $\Bseq{}{2} = \texttt{e}\texttt{e}$, the two tries were unsuccessful and the encoder still needs to transmit the first bit of the message. Formally:
\begin{equation}
	\label{eq:f3}
	\begin{array}{l}
		\Fsym{3}{}[W=011](\Bseq{}{2}=00) = \mathtt{x} \\
		\Fsym{3}{}[W=011](\Bseq{}{2}=01) = 1 \\
		\Fsym{3}{}[W=011](\Bseq{}{2}=0\texttt{e}) = 1 \\
		\Fsym{3}{}[W=011](\Bseq{}{2}=10) = \mathtt{x} \\ 
		\Fsym{3}{}[W=011](\Bseq{}{2}=11) = \mathtt{x} \\
		\Fsym{3}{}[W=011](\Bseq{}{2}=1\texttt{e}) = \mathtt{x} \\
		\Fsym{3}{}[W=011](\Bseq{}{2}=\texttt{e}0) = 1 \\
		\Fsym{3}{}[W=011](\Bseq{}{2}=\texttt{e}1) = \mathtt{x} \\
		\Fsym{3}{}[W=011](\Bseq{}{2}=\texttt{e}\texttt{e}) = 0 \\
	\end{array}
\end{equation}

We can easily extend the construction of code functions $\Fsym{t}{}$ for $3 \leq t \leq T$ using this encoding scheme.

The decoder is very simple: once all time steps $1, \ldots, T$ have taken place, it just takes the whole output trace $\Bseq{}{T}$ and removes the occurrences of the erased bit symbol $\texttt{e}$ in order to recover the original message. 

Table~\ref{tab:dcmf-example} shows a possible behavior of a binary erasure channel with memory and feedback in a scenario where the message is $W=011$ and the channel can be used at most $T=3$ times. Note that in this particular example the maximum number of uses of the channel is achieved before the whole message is successfully sent: the decoder can recover only the two first bits of the original message.

\begin{table}[!htb]
	\centering
	\setlength{\extrarowheight}{3pt}
	\tiny
	\begin{tabular}{|c||c|c|c|c|c|}
		\hline
		Time & Code  & Feedback & Encoder & Channel & Decoder \\ 
		 $t$ & functions & history & $\Aseq{t}{} =$ & $\Cprob{p}{\Bseq{t}{}}{\Aseq{}{t},\Bseq{}{t-1}}$ & $\hat{W} = $ \\
		     & $\Fsym{t}{}(\Bseq{}{t-1})$                           & $\Bseq{}{t-1}$ & $\Fsym{t}{}[W](\Bseq{}{t-1})$ &           & $\gamma(\Bseq{}{T})$ \\ \hline \hline
		        & Code       &  &  &  &  \\
		$t = 0$ & functions  & --------- & --------- & --------- & --------- \\ 
		        & for $W=011$ &  &  &  &  \\
		        & are selected.  &  &  &  &  \\ \hline
		        &  &  & $\Aseq{1}{} =$ & According to &  \\
		$t = 1$ & As in (\ref{eq:f1})  & $\epsilon$ & $\Fsym{1}{}[W=011](\epsilon)$ &  $\Cprob{p}{\Bseq{1}{}}{0, \epsilon}$ & --------- \\	
						&      &  &   $= 0$   &  produces & \\
						&      &  &           &   $\Bseq{1}{} = \texttt{e}$ & \\ \hline
		        &  &  & $\Aseq{2}{} = $ & According to &  \\
		$t = 2$	& As in (\ref{eq:f2})  & $\texttt{e}$ & $\Fsym{2}{}[W=011](\texttt{e})$ & $\Cprob{p}{\Bseq{2}{}}{00, \texttt{e}}$ &  --------- \\ 
						& &  & $=0$ & produces  &  \\
						& &  &      & $\Bseq{2}{} = 0$ &  \\ \hline
		        &  &  & $\Aseq{3}{} = $ & According to &  \\
		$t = 3$ & As in (\ref{eq:f3}) & $\texttt{e}0$ & $ \Fsym{3}{}[W=011](\texttt{e}0)$  & $\Cprob{p}{\Bseq{3}{}}{001, \texttt{e}0}$ & --------- \\ 
						&  &  & $= 1$ & produces &  \\
						&  &  &       & $\Bseq{3}{} = 1$ &  \\ \hline
		        &  &  &  &  & Decoded \\		
		$t = 4$ & --------- & --------- & --------- & --------- &  message $\hat{W} = $\\
		        &  &  &  &  & $ \gamma(\Bseq{}{3} = \texttt{e}01)$\\
		        &  &  &  &  & $= 01$\\		\hline		
	\end{tabular}
	\setlength{\extrarowheight}{0pt}
	\caption{A possible evolution of the binary channel with time, for $W=011$ and $T=3$}
	\label{tab:dcmf-example}
\end{table}	

We can observe that the channel capacity in the above example does not increase with the addition of feedback (it is $0.8$ bit per usage of the channel with or without feedback).  This is because the channel is memoryless: \emph{feedback does not increase the capacity of discrete memoryless channels} \cite{Cover:06:BOOK}. In general however, feedback \emph{does} increase the capacity \revision{of channels with memory}.

\subsection{Directed information and capacity of channels with feedback}
\label{section:dirinfo}
	
In classical Information Theory, the channel capacity, which  is related to the channel's transmission rate by Shannon's Channel Coding Theorem, can be obtained as the supremum of the mutual information over all possible input distributions. In the presence of feedback, however, this correspondence no longer holds.	More specifically, mutual information no longer represents the information flow from $A^{T}$ to $B^{T}$. Intuitively, this is due to the fact that mutual information expresses correlation, and therefore it is increased by feedback \revision{(Example~\ref{exa:small} in Section~\ref{section:anonymity-properties} depicts this fact)}. Yet feedback, i.e. the way the output influences the next input, is not part of the information to be transmitted. If we want to maintain the correspondence between the transmission rate and  capacity, we need to replace the mutual information with \emph{directed information}~\cite{Massey:90:SITA}.

\begin{definition}
	\label{def:directed-information}
	In a channel with feedback, the directed information from input $A^{T}$ to output $B^{T}$ is defined as
		$$ I(A^{T} \rightarrow B^{T}) = \sum_{t=1}^{T}	I(A^{t};B_{t} | B^{t-1})$$
	In the other direction, the directed information from   $B^{T}$ to   $A^{T}$ is defined as
		$$ I(B^{T} \rightarrow A^{T}) = \sum_{t=1}^{T} I(A_{t};B^{t-1} | A^{t-1})$$
\end{definition}

In Section~\ref{section:anonymity-properties} we will discuss the relation between directed information and mutual information, as well as the correspondence with information leakage. For the moment, we only present the extension of the concept of capacity.
	
Let $\Dset_{T} = \{ \Cprob{p}{\Aseq{t}{}}{\Aseq{}{t-1},\Bseq{}{t-1}} \}_{t=1}^{T}$ be the set of all  input distributions in presence of feedback. For finite $T$, the capacity of a channel with memory and feedback  is:
\begin{equation}
	\label{def:cmfb-capacity}
	C_{T} = \sup_{\Dset_{T}} \frac{1}{T} I(A^{T} \rightarrow B^{T})
\end{equation}
	
The capacity is also defined when $T$ is infinite, see~\cite{Tatikonda:09:TIT}. In this thesis, however, we only need to consider the finite case.
	
\section{Interactive systems as channels with memory and feedback}
\label{section:our-model}

Interactive Information Hiding Systems (\IIHS) were introduced in \cite{Andres:10:TACAS} to represent systems where secrets (inputs) and observables (outputs) can interleave and influence each other. They are  a variant of probabilistic automata in which  actions are divided \revision{into} secrets and observables.  They can be of two kinds:   \emph{fully probabilistic}, and \emph{secret-nondeterministic} (or \emph{input-nondeterministic}). In the former there is no nondeterminism, while in the latter every secret choice is  fully nondeterministic. In this chapter we consider \emph{normalized} \IIHSs, in which secrets and observables alternate, and the actions at the first level are secrets. We note that this is not really a restriction,  because given an \IIHS which is not normalized, it is always possible to transform it into a normalized \IIHS which is equivalent to the former one up to a given execution level. The reader can find \revision{further below in this Section} the formal definition of the transformation. Furthermore, we require that for each state $s$ and each action $\ell$ there is at most one state that can be reached from $s$ by performing an $\ell$ transition. 

In this section we formalize the notion of \IIHS and we show how to associate to an IIHS a channel with memory and feedback.
	
\begin{definition}
	\label{def:IIHS}
	A (normalized) \IIHS\ is a triple $\Isys=(M, \Aset, \Bset )$, where $\Aset$ and $\Bset$ are disjoint sets of secrets and observables  respectively,  $M$ is a probabilistic automaton $(\Sset, \Lset, \si, \vartheta)$ with $\Lset=\Aset\cup\Bset $, and, for each $s\in\Sset$:

	\begin{enumerate}
		\item \label{item:restr-theta1} either $\vartheta(s)\subseteq \distr(\Aset  \times \Sset)$ or  $\vartheta(s)\subseteq \distr(\Bset  \times \Sset)$. We call $s$ a \emph{secret state} in the first case, and an \emph{observable state} in the second case;
		
		\item \label{item:restr-theta2} if $\smash{ s\stackrel{\ell}\rightarrow r}$ then: if $s$ is a secret state then $r$ is an observable state, and if $s$ is an observable state then $r$ is a secret state;
			
		\item \label {item:restr-initial} $\si$ \ is a secret state;
		
		\item \label{item:restr-theta3} if $s$ is an observable state then $|\vartheta(s)|\leq 1$ ;
		
		\item \label{item:restr-theta4} either:
				
			\begin{enumerate}[(i)]
			 \item \label{item:restr-theta4-a} for every secret state $s$ we have $|\vartheta(s)|\leq 1$ \emph{(fully probabilistic \IIHS)}, 
			
			or
			
				\item \label{item:restr-theta4-b} for every secret state $s$ there exist $a_{i}$ and $s_{i}$ ($i=1,\ldots,n$) such that $\vartheta(s)$ $=$ $\{\delta(a_i,s_i)\}_{i=1}^n$, where $\delta(a_i,s_i)$ is the Dirac measure \emph{(secret-nondeterministic \IIHS);}
				\end{enumerate}
				
		\item \label{item:restr-add} for every state $s$ and action $\ell$ there exists a unique state $r$ such that $\smash{ s\stackrel{\ell}\rightarrow r}$. 	

	\end{enumerate}				
\end{definition}	

In the rest of the \revision{chapter} we will omit the adjective \qm{normalized} for simplicity. In the above definition, Conditions \ref{item:restr-theta1} and \ref{item:restr-theta2} imply that the \IIHS is alternating between secrets and observables. Moreover, all the transitions between nodes at two consecutive depths have either secret actions only, or observable actions only. Condition~\ref{item:restr-initial} means that the first level contains secret actions. Condition~\ref{item:restr-theta3} means that all observable transitions are fully probabilistic. Condition~\ref{item:restr-theta4} means that \revision{either all secret transitions are fully probabilistic, either they are all fully nondeterministic}. The term \qm{nondeterministic} is justified by the fact that the scheme of \revision{Condition~\ref{item:restr-theta4-b}} represented in Figure~\ref{fig:nondeterministic-dirac-inputs}, is equivalent to the one of Figure~\ref{fig:nondeterministic-input}.		
			
			\begin{figure}[!htb]
				\centering
				\subbottom[Nondeterministic input using Dirac measures]{
					\centering
					\includegraphics[width=0.6\linewidth]{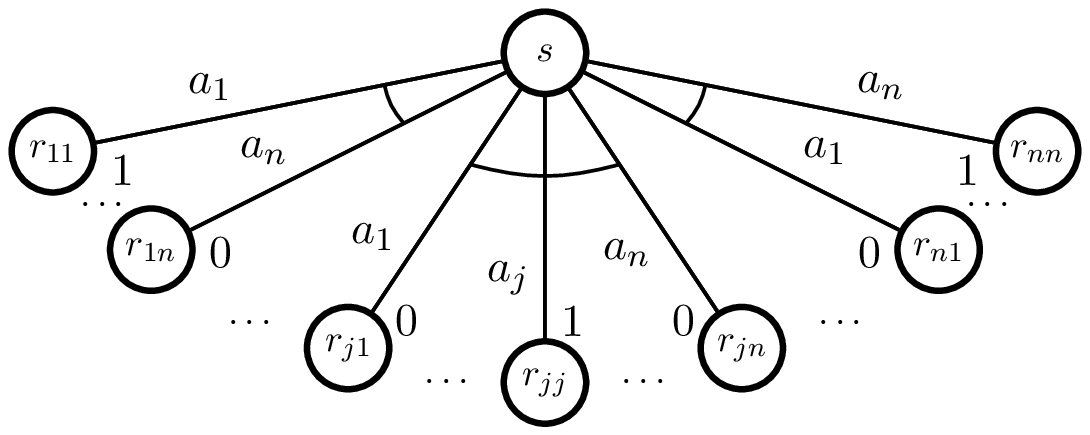}%
					\label{fig:nondeterministic-dirac-inputs}%
				}
				\subbottom[Equivalent scheme]{
					\centering
					\includegraphics[width=0.3\linewidth]{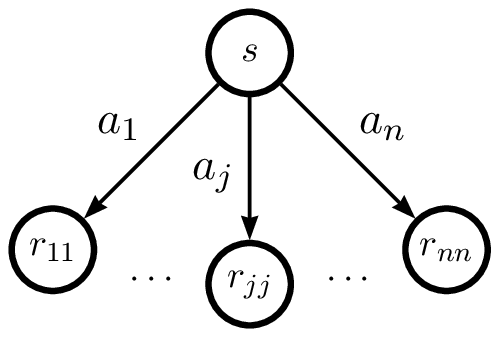}%
					\label{fig:nondeterministic-input}%
				}
				\caption{Scheme of secret transitions for secret-nondeterministic \IIHSs }
				\label{fig:inout-scheme}
			\end{figure}
			
Note that we do not consider here internal nondeterminism \revision{which can arise} from interleaving of concurrent processes. This means that we make a rather restricted use of probabilistic automata, but this is enough for our purposes. The nondeterminism generated by concurrency gives rise to a new set of problems (see for example~\cite{Chatzikokolakis:08:IC}) which are orthogonal to those considered in this \revision{chapter}. 	

Condition~\ref{item:restr-add} means that the secret and observable actions  determine the states. As a consequence, the actions are enough to retrieve the path. This is expressed by the following proposition:
				
\begin{proposition}
	\label{prop:traces-states}
	Given an  \IIHS, consider two paths $\sigma$ and $\sigma'$. If  $\trace_{\Aset}(\sigma)=\trace_{\Aset}(\sigma')$ and $\trace_{\Bset}(\sigma)=\trace_{\Bset}(\sigma')$, then $\sigma =\sigma'$.
\end{proposition}
			
\begin{proof}
	By induction on the length of the traces. The initial state of the automaton is  uniquely determined by the empty  (secret and observable) traces. Assume now we are in a state $s$ uniquely determined by secret and observable traces $\Aseq{}{}$ and $\Bseq{}{}$, respectively. If $s$ makes a secret transition $\smash{ s\stackrel{a}\rightarrow s'}$, then by Condition~\ref{item:restr-add} there is only one state $s'$ reachable from $s$ via an $a$-transition, and therefore $s'$  is uniquely determined by the secret trace $\Aseq{}{}'= \Aseq{}{}a$ and the observable trace $\Bseq{}{}$.  The case in which $s$ makes an observable transition is similar.  		
\end{proof}
			
\subsubsection{The normalization of \IIHS\ trees}
\label{section:tree-trans}

In this section we will address the problem of \emph{normalizing} an \IIHS, namely transforming it into a stratified automaton in which secret and observable actions alternate level by level. The process of normalization described bellow is general enough to be applied to any \IIHS\ without loss of generality or expressive power.
	
Let $\Aset$ and $\Bset$ represent the secret and observable actions, respectively. Consider a general \IIHS\ $\Isys = (M, \Aset, \Bset)$ with $M=(Q, \Lset, \si, \vartheta)$, where $\Lset = \Aset \cup \Bset$. Assume that we are only interested in executions that involve up to $T$ interactions, i.e. $T$ uses of the system, with one secret taking place and one observable produced at each time. 
	
In the normalization process, we unfold the automaton up to level $2T$, since there is one secret symbol and one observable symbol for each step. We also extend the secret alphabet $\Aset$ with a new symbol $\Asym{*}{} \notin \Aset$ and  the observable alphabet $\Bset$ with a new symbol $\Bsym{*}{} \notin \Bset$. These new symbols will be used as placeholders when we need to re-balance the tree. Let $\Aset' = \Aset \cup \{ \Asym{*}{} \}$ and $\Bset' = \Bset \cup \{ \Bsym{*}{} \}$.
	
For a given level $t$ let $\Lfun(\Isys,t)$ be the set of all labels of transitions that can be performed with a non-zero probability from the states at the $t^{th}$ level of the automaton. Formally:

\[\Lfun(\Isys,t) \equiv \{ \ell \in \Lset \ | \ \exists \sigma, s \ . \ |\sigma| = t, \ \last(\sigma) \stackrel{\ell}{\to} s\}\]
	
The normalization of the \IIHS $\Isys$ leads to an equivalent \IIHS\ $\Isys' = (M', \Aset', \Bset')$, where $M'=(Q', \Lset', \si', \vartheta')$ and $\Lset'= \Aset' \cup \Bset'$; and such that, for every $1 \leq t \leq 2T$:
	
	\begin{enumerate}
		\item \label{eq:norm-cond1}
		$\Lfun(\Isys',t) \subseteq \Aset' \mbox{\ \ \ or \ \ \ } \Lfun(\Isys',t) \subseteq \Bset'$;
		\item \label{eq:norm-cond3}
		$ \Lfun(\Isys',t) \subseteq \Aset' \mbox{\ \ \ if and only if \ \ \ } \Lfun(\Isys',t+1) \subseteq \Bset'$, for $1\leq t \leq T-1$; 
		\item \label{eq:norm-cond4}
		$\Lfun(\Isys',1) \subseteq \Aset'$;
	\end{enumerate}

Condition~\ref{eq:norm-cond1} states that each level consists of either the secret actions only, or the observable actions only. Condition~\ref{eq:norm-cond3} states that secret and observable levels alternate. Condition~\ref{eq:norm-cond4} says that  the automaton starts with a secret level. 
	
The proof is straightforward.  First, the new symbols $\Asym{*}{}$ and $\Bsym{*}{}$ are placeholders for the absence of a secret and observable symbol, respectively. If in a given level $t$ we want to have only secret symbols, we can postpone the occurrences of observable symbols at this level as follows: add $\Asym{*}{}$ to the secret level and \qm{move} all the observable symbols to the subtree of $\Asym{*}{}$. Figure~\ref{fig:tree-trans-exa} exemplifies the local transformations we need to make on the tree.

	\begin{figure}[!htb]
		\centering        
		\subbottom[Local nodes of the tree before the transformation]{\label{fig:trans-exa-T1}
			\centering
			\includegraphics[width=0.45\linewidth]{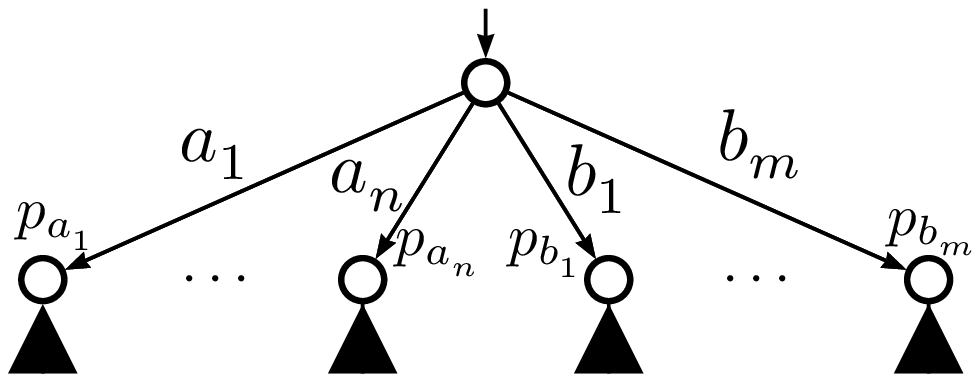}
		}
		\subbottom[Local nodes of the tree after the transformation]{\label{fig:trans-exa-T2}
			\centering
			\includegraphics[width=0.35\linewidth]{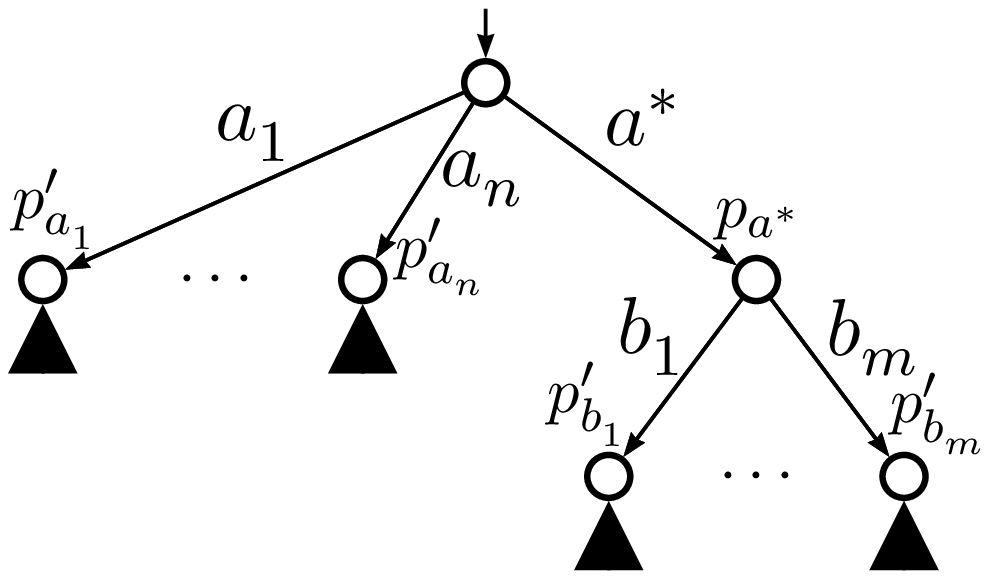}
		}
		\caption{Local transformation in an \IIHS tree}
		\label{fig:tree-trans-exa}
	\end{figure}

	Note that in~\ref{fig:trans-exa-T2} the introduction of new nodes changed the probabilities of the transitions in the tree. In general, whenever we need to introduce ${\Asym{*}{}}$ in order to postpone the observable symbols, the probabilities change as follows:
	
	\begin{enumerate}
		\item For every $\Asym{i}{}$, $1 \leq i \leq n$, the associated probability is maintained as $p_{\Asym{i}{}}' = p_{\Asym{i}{}}$;
		\item The probability of the new symbol $\Asym{*}{}$ is introduced as $p_{\Asym{*}{}} = \sum_{k=0}^{m}p_{\Bsym{k}{}}$;
		\item If $p_{\Asym{*}{}} \neq 0$, then for $1 \leq i \leq m$, the associated probability of $\Bsym{j}{}$ is updated to $p_{\Bsym{j}{}}' = {p_{\Bsym{j}{}}} / {p_{\Asym{*}{}}} = {p_{\Bsym{j}{}}} / {\sum_{k=0}^{m}p_{\Bsym{k}{}}}$. If $p_{\Asym{*}{}} = 0$, then $p_{\Bsym{j}{}}' = 0$, for $1 \leq i \leq m$, and $p_{\Bsym{*}{}} = 1$.
	\end{enumerate}
	
The subtrees of each node of the original tree are preserved as they are, until we apply the same transformation to them. If a node does not have a subtree (i.e. no descendants), we create a subtree by adding all the possible actions in $\Bset$ with probability $0$, and the action $\Bsym{*}{}$ with probability $1$.
		
If we are normalizing an observable level, the same rules apply, guarding the proper symmetry between secrets and observables. We then \revision{proceed in} the same way on the deeper levels of the tree. Figure~\ref{fig:tree-trans} shows an example of a full transformation on a tree (for the sake of readability, we omit the levels where only $\Asym{*}{} = 1$ or $\Bsym{*}{} = 1$).
	
	\begin{figure*}[!htb]
		\centering
		\subbottom[Tree before transformation]{
			\centering
			\label{fig:trans-T1}
			\includegraphics[width=0.3\linewidth]{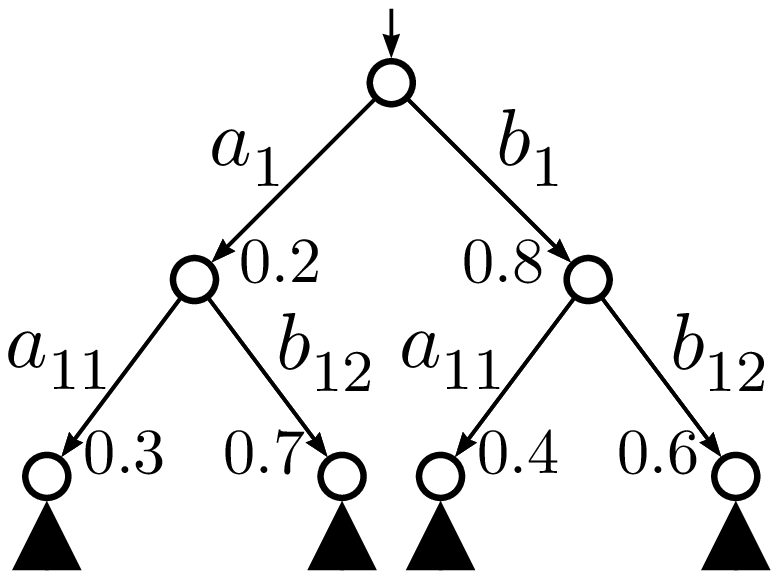}
		}
		\subbottom[Tree after transformation]{
			\centering
			\label{fig:trans-T2}
			\includegraphics[width=0.6\linewidth]{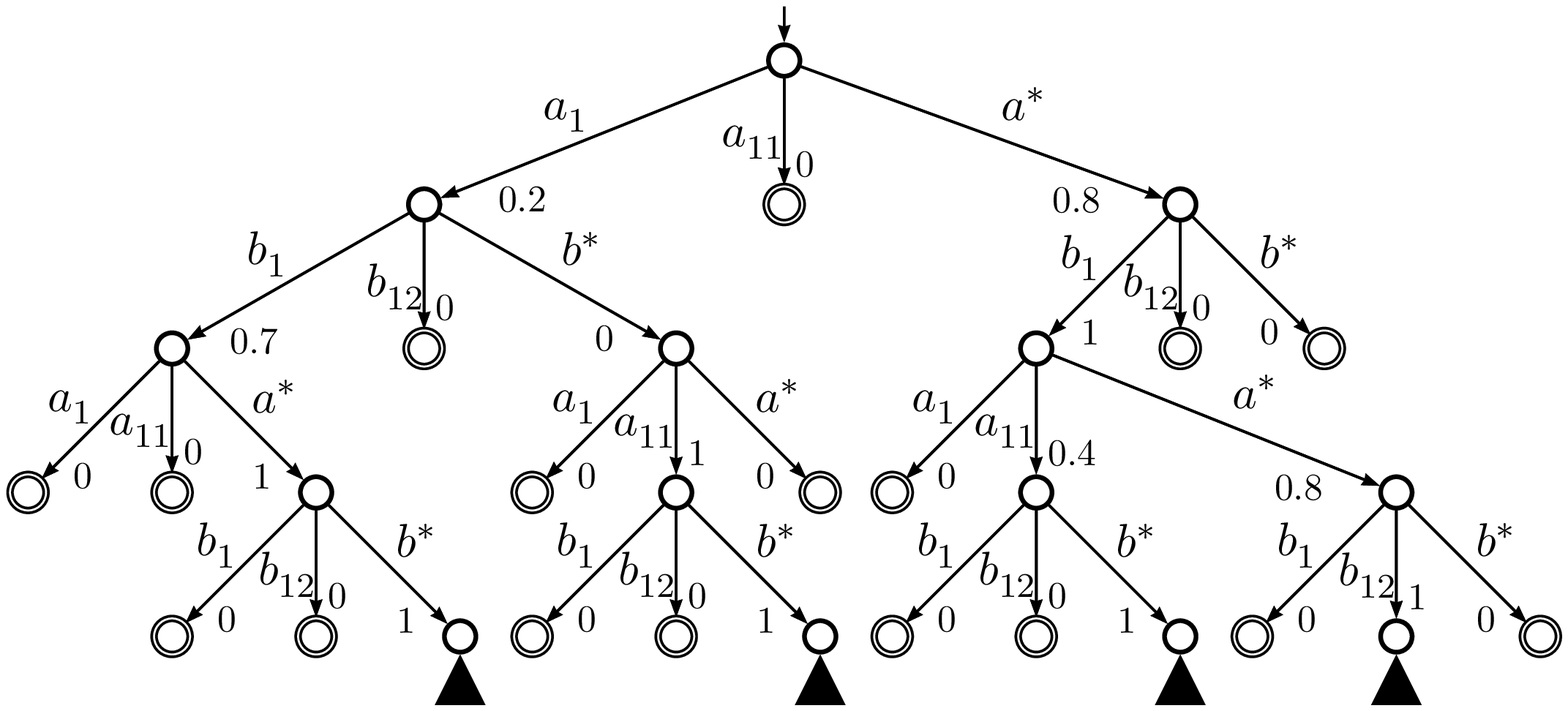}
		}
		\caption{Transformation in an \IIHS tree}
		\label{fig:tree-trans}
	\end{figure*}

	\subsection{Construction of the channel associated to an \IIHS}			
	We now show how to associate a channel to an \IIHS. 
	
In an interactive system secrets and observables may interleave and influence each other. Considering a channel with memory and feedback is a way to capture this rich behavior. Secrets have a causal influence on observables via the channel, and, in the presence of interactivity, observables have a causal influence on secrets via feedback. This alternating mutual influence between secrets and observables can be modeled by repeated uses of the channel. Each time the channel is used it represents a different state of the computation, and the conditional probabilities of observables on secrets can depend on this state. The addition of memory to the model allows expressing the dependency of  the channel matrix on such a state.

We will see that a secret-nondeterministic \IIHS\  determines a channel as specified by its stochastic kernels, while a fully probabilistic \IIHS\ determines, additionally, the input distribution.
	
In Section~\ref{section:full-example} we will give an extensive and detailed example of how to make  such a construction for an actual security protocol. 
	
Given a path $\sigma$ of length $2t-1$, we will denote $\trace_{\Aset}(\sigma)$ by $\Aseq{}{t}$, and $\trace_{\Bset}(\sigma)$ by $\Bseq{}{t-1}$. 				

	\begin{definition} \label{def:def-matrix}
	Let  $\Isys$ be an \IIHS.
	For each $t$, the channel's stochastic kernel corresponding to  $\Isys$
		is defined as
 			$\Cprob{p}{\Bseq{t}{}}{\Aseq{}{t},\Bseq{}{t-1}}
			= \vartheta(s)(\Bseq{t}{},s')$,
		 where $s$ is the state reached from the root via the path $\sigma$ whose
	secret and observable traces are  $\Aseq{}{t}$ and  $\Bseq{}{t-1}$ respectively.	
	\end{definition}
	
	Note that $s$ and $s'$ in the previous definition are well defined: by Proposition~\ref{prop:traces-states}, $s$
	is unique, and since the choice of $\Bseq{t}{}$ is fully probabilistic, $s'$ is also unique. 
	
	The following example illustrates how to apply Definition~\ref{def:def-matrix}, with the help of Proposition~\ref{prop:traces-states}, to
	build the channel matrix of a simple example.
	
	\begin{example}
		\label{exa:channel-construction}
		
		Let us consider an extended version of the website interactive system of Figure~\ref{fig:website-example-short}. We maintain the general definition of
		the system, i.e. there are two possible buyers ($rich$ and $poor$, represented by $rc.$ and $pr.$, respectively) and two possible products ($cheap$ and $expensive$, represented by $chp.$ and $exp.$, respectively). We still assume that offers are observable, since they are visible to everyone on the website, but the identity of buyers should be kept secret. We consider two consecutive rounds of offers and buys, which implies that, after normalization, $T = 3$. Figure~\ref{fig:website-example-full} shows an automaton for this example in normalized form. Transitions with null probability are omitted, and the symbol $\Asym{*}{}$ is used as a place holder to achieve the normalized \IIHS.
		
		To construct the stochastic kernels $\{ \Cprob{p}{\Bseq{t}{}}{\Aseq{}{t},\Bseq{}{t-1}} \}_{t=1}^{T}$, we need to determine the conditional probability of an observable at time $t$ given the history up to time $t$. 
		
		Let us take  the case $t=2$ and compute the conditional probability of the observable $\Bseq{2}{} = cheap$ given that the history of secrets up to time $t=2$ is $\Aseq{}{2} = \Asym{*}{}, poor$ and the history of observables is $\Bseq{}{1} = expensive$. Applying Definition~\ref{def:def-matrix}, we see that $\Cprob{p}{\Bseq{2}{} = cheap}{\Aseq{}{2} = \Asym{*}{}, poor, \Bseq{}{1} = expensive} = \vartheta(s)(cheap,s')$.  By Proposition~\ref{prop:traces-states}, the traces $\Aseq{}{2} = \Asym{*}{}, poor, \Bseq{}{1} = expensive$ determine a unique state $s$ in the automaton, namely, the state $s = 5$. Moreover, from the state $5$ a unique transition labeled with the action $cheap$ is possible, leading to the state $s' = 11$. Therefore, we can conclude that $\Cprob{p}{\Bseq{2}{} = cheap}{\Aseq{}{2} = \Asym{*}{}, poor, \Bseq{}{1} = expensive} = \vartheta(s = 5)(cheap, s' = 11) = p_{23}$.
		
		Similarly, with $t=1$ and history $\Aseq{}{1} = \Asym{*}{}, \Bseq{}{0} = \epsilon$, the observable symbol $\Bseq{1}{} = expensive$ can be observed with probability $\Cprob{p}{\Bseq{1}{} = expensive}{\Aseq{}{1} = \Asym{*}{}, \Bseq{}{0} = \epsilon} = \vartheta(s = 0)(cheap, s' = 2) = \overline{p_{1}}$.
		
	\end{example}		
	
	If $\Isys$ is fully probabilistic, then it determines also the input distribution and the dependency of $\Aseq{t}{}$ on $\Bseq{}{t-1}$ (feedback) and on $\Aseq{}{t-1}$.
			
	\begin{definition}\label{def:input-distribution}
	Let  $\Isys$ be an \IIHS. If   $\Isys$ is fully probabilistic, the associated channel has a conditional input distribution for each $t$ defined as
	   $\Cprob{p}{\Aseq{t}{}}{\Aseq{}{t-1},\Bseq{}{t-1}}
			= \vartheta(s)(\Aseq{t}{},s')$,		
		where $s$ is the state reached from the root via the path $\sigma$ whose
		secret and  observable traces are $\Aseq{}{t-1}$ and $\Bseq{}{t-1}$ respectively.
	\end{definition}
	
	\begin{example}
		\label{exa:input-construction}
		
		Since the system of Example~\ref{exa:channel-construction} is fully probabilistic, we can calculate the values of the conditional probabilities $\{ \Cprob{p}{\Aseq{t}{}}{\Aseq{}{t-1},\Bseq{}{t-1}} \}_{t=1}^{T}$.
		
		Let us take, for instance, the case where $t=2$ and compute the conditional probability of secret $\Aseq{2}{} = poor$ given that the history of secrets up to time $t=2$ is $\Aseq{}{1} = \Asym{*}{}$ and the history of observables is $\Bseq{}{1} = expensive$. Applying Definition~\ref{def:input-distribution}, we see that $\Cprob{p}{\Aseq{2}{} = poor}{\Aseq{1}{} = \Asym{*}{}, \Bseq{}{1} = expensive} = \vartheta(s)(poor,s')$. By Proposition~\ref{prop:traces-states}, the traces $\Aseq{}{1} = \Asym{*}{}, \Bseq{}{1} = expensive$ determine a unique state $s$  in the automaton, namely, the state $s = 2$. Moreover, from the state $2$ a unique transition labeled with the action $poor$ is possible, leading to the state $s' = 5$. Therefore, we can conclude that $\Cprob{p}{\Aseq{2}{} = poor}{\Aseq{1}{} = \Asym{*}{}, \Bseq{}{1} = expensive} = \vartheta(s = 2)(poor,s' = 5) = q_{12}$.
		
		Similarly, with $t=3$ and history $\Aseq{}{2} = \Asym{*}{}, rich, \Bseq{}{2} = cheap, expensive$, the secret symbol $\Aseq{3}{} = rich$ can be observed with probability $\Cprob{p}{\Aseq{3}{} = rich}{\Aseq{}{2} = \Aseq{*}{},rich, \Bseq{}{0} = cheap,expensive} = \vartheta(s = 10)(cheap, s' = 22) = \overline{q_{24}}$.
		
	\end{example}		

	\begin{figure}[!htb]%
		\includegraphics[width=\linewidth]{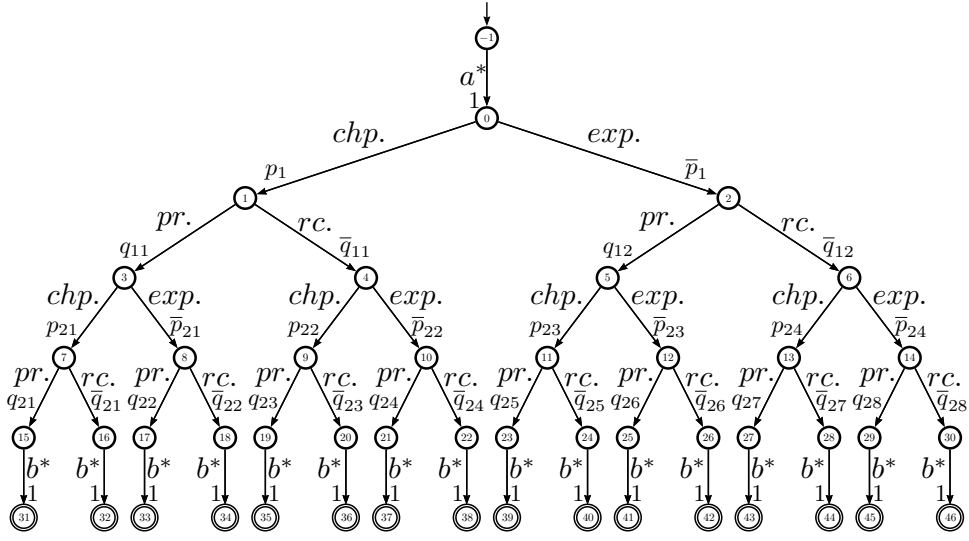}%
		\caption{The normalized \IIHS for the extended website example}
		\label{fig:website-example-full}%
	\end{figure}

	\subsection{Lifting the channel inputs to reaction functions}
	\label{section:reaction-functions}

	Taken together, Definitions~\ref{def:def-matrix} and~\ref{def:input-distribution} show how to obtain the the joint probabilities 	$\prob{p}{\Aseq{}{t},\Bseq{}{t}}$ for a fully probabilistic  \IIHS. We still need to show, however, in what sense this joint probability distribution defines an information-theoretic channel.
	
	The  $\{ \Cprob{p}{\Bseq{t}{}}{\Aseq{}{t},\Bseq{}{t-1}} \}_{t=1}^{T}$ determined by the \IIHS\ trivially correspond to a channel's stochastic kernel. The problem 	resides in the conditional probabilities $\{ \Cprob{p}{\Aseq{t}{}}{\Aseq{}{t-1},\Bseq{}{t-1}} \}_{t=1}^{T}$. In an information-theoretic channel, the value of $\Aseq{t}{}$ is determined in the encoder by a deterministic function $\Fseq{t}{}(\Bseq{}{t-1})$. Therefore, inside the encoder there is no possibility for a probabilistic description of $\Aseq{t}{}$. The solution is to externalize this probabilistic behavior to the code functions.
	
	As shown in~\cite{Tatikonda:09:TIT}, the original channel with feedback from input symbols $\Aset^{T}$ to output symbols $\Bset^{T}$ can be lifted to an equivalent channel without feedback from code functions $\Fset^{T}$ to output symbols $\Bset^{T}$. This transformation also allows us to calculate the channel capacity. Let $\{\Cprob{p}{\Fseq{t}{}}{\Fseq{}{t-1}}\}_{t=1}^{T}$ be a sequence of code function stochastic kernels and let $\{ \Cprob{p}{\Bseq{t}{}}{\Aseq{}{t},\Bseq{}{t-1}} 	\}_{t=1}^{T}$ be a channel with memory and feedback. The channel from  $F^{T}$ to $B^{T}$ is constructed using a joint  measure $Q(\Fseq{}{T},\Aseq{}{T},\Bseq{}{T})$ that respects the following constraints:
	
	\begin{definition}
	\label{def:Q}
	A measure $\prob{Q}{\Fseq{}{T},\Aseq{}{T},\Bseq{}{T}}$ is said to be \emph{consistent} with respect to the code function stochastic kernels $\{\Cprob{p}{\Fseq{t}{}}{\Fseq{}{t-1}}\}_{t=1}^{T}$ and the channel $\{ \Cprob{p}{\Bseq{t}{}}{\Aseq{}{t},\Bseq{}{t-1}} \}_{t=1}^{T}$ if, for each $t$:	

		\begin{enumerate}

			\item There is no feedback to the code functions:
				
				$$\Cprob{Q}{\Fseq{t}{}}{ \Fseq{}{t-1},  \Aseq{}{t-1}, \Bseq{}{t-1}} =	\Cprob{p}{\Fseq{t}{}}{\Fseq{}{t-1}}$$

			\item  The input is a function of	the past outputs:
				
				$$\Cprob{Q}{\Aseq{t}{}}{ \Fseq{}{t},  \Aseq{}{t-1}, \Bseq{}{t-1}} = \delta_{\{\Fseq{t}{}(\Bseq{}{t-1})\}}(\Aseq{t}{})$$
			
			\noindent where $\delta$ is the Dirac measure;

			\item  The properties of the underlying channel are preserved:
				
				$$\Cprob{Q}{\Bseq{t}{}}{F^{t}=\Fseq{}{t}, A^{t}=\Aseq{}{t}, B^{t-1} = \Bseq{}{t-1}} = \Cprob{p}{\Bseq{t}{}}{\Aseq{}{t},\Bseq{}{t-1}}$$

		\end{enumerate}

	\end{definition}
	
	The following result states that there is only one consistent measure \linebreak $\prob{Q}{\Fseq{}{T},\Aseq{}{T},\Bseq{}{T}}$.
	
	\begin{theorem}[\cite{Tatikonda:09:TIT}]
			Given the probability distributions $\{\Cprob{p}{\Fseq{t}{}}{\Fseq{}{t-1}}\}_{t=1}^{T}$ and a channel defined by $\{\Cprob{p}{\Bseq{t}{}}{\Aseq{}{t},\Bseq{}{t-1}}\}_{t=1}^{T}$, there exists only one   consistent measure $Q(\Fseq{}{T},\Aseq{}{T},\Bseq{}{T})$. Furthermore the channel from $\Fset^{T}$ to $\Bset^{T}$ is given by:
		
		$$ \Cprob{Q}{\Bseq{t}{}}{\Fseq{}{t}, \Bseq{}{t-1}} = \Cprob{p}{\Bseq{t}{}}{\Fseq{}{t}(\Bseq{}{t-1}), \Bseq{}{t-1}} $$
		
	\end{theorem} 	
	
	Since in our setting the concept of encoder makes little sense as there is no information to encode, we externalize the probabilistic behavior of $\Aseq{t}{}$ as follows. Code functions become \emph{a single set of reaction functions $\{\Fseq{t}{}\}_{t=1}^{T}$ with $\Bseq{}{t-1}$ as parameter} (the message $w$ does not play a role any more). Reaction functions can be seen as a model of how the environment reacts to given system outputs, producing new system inputs (they do not play a role of encoding a message). These reaction functions are endowed with a probability distribution  that generates the probabilistic behavior of the values of $\Aseq{t}{}$.
	
	\begin{definition}
		\label{def:reactor}
		A \emph{reactor} is a distribution on reaction functions, i.e. a sequence of stochastic kernels $\{ \Cprob{p}{\Fseq{t}{}}{\Fseq{}{t-1}} \}_{t=1}^{T}$. A reactor $R$ is \emph{consistent with  a fully probabilistic  \IIHS\ $\ihs$} if it induces the compatible distribution $Q(\Fseq{}{T},\Aseq{}{T},\Bseq{}{T})$ such that, for every $1 \leq t \leq T$, $\Cprob{Q}{\Aseq{t}{}}{\Aseq{}{t-1}, \Bseq{}{t-1}} = \Cprob{p}{\Aseq{t}{}}{\Aseq{}{t-1}, \Bseq{}{t-1}}$, where the latter is the probability distribution induced by $\Isys$.
	\end{definition}

	The main result of this section states that for any fully probabilistic \IIHS\ there is a reactor that generates the probabilistic behavior of the \IIHS. Before moving to this result, we need to introduce a lemma.
	
	\begin{lemma}
		\label{lemma:prob1}
		Let $\Xset, \Yset$ be non-empty finite sets, and let $\tilde{x} \in \Xset, \tilde{y} \in \Yset$. Let $p: \Xset \times \Yset \rightarrow [0,1]$ 
		be a function such that, for every $x \in \Xset$, we have: $\sum_{y \in \Yset}p(x,y) = 1$. Then:
		$$ 
			\sum_{\begin{array}{c}
							\scriptstyle f \in \Xset \rightarrow \Yset \\[-1 mm]
							\scriptstyle f(\tilde{x}) = \tilde{y}
			\end{array}}
			 \prod_{x \in \Xset} p(x, f(x)) = p(\tilde{x},\tilde{y}) 
		$$
	\end{lemma}

\begin{proof}
	
	By induction on the number of elements of $\Xset$.
	
	\begin{description}
	
		\item[Base case:] $\Xset = \{\tilde{x}\}$. In this case: 
\begin{equation*}
	\sum_{\begin{array}{c}\scriptstyle  f \in \Xset \rightarrow \Yset \\[-1mm] \scriptstyle  f(\tilde{x}) = \tilde{y} \end{array}} \prod_{x \in \Xset} p(x, f(x)) = p(\tilde{x},f(\tilde{x})) = p(\tilde{x},\tilde{y})
\end{equation*}
	
		\item[Inductive case:] Let $\Xset = \Xset' \cup \{\mathring{x}\}$, with $\tilde{x} \in \Xset'$ and $\mathring{x} \notin \Xset'$. Then:
	\begin{align*}
			\sum_{\begin{array}{c} \scriptstyle f \in \Xset' \cup \{\mathring{x}\} \rightarrow \Yset \\[-1mm] \scriptstyle f(\tilde{x}) = \tilde{y} \end{array}} \prod_{x \in \Xset' \cup \{ \mathring{x} \}} p(x, f(x)) & = & \text{(by distributivity)} \\
			\left ( \sum_{\begin{array}{c} \scriptstyle f \in \Xset' \rightarrow \Yset \\[-1mm]
				\scriptstyle  f(\tilde{x}) = \tilde{y} \end{array}} \prod_{x \in \Xset'} p(x, f(x)) \right ) \sum_{g \in \{\mathring{x}\} \rightarrow \Yset }p(\mathring{x},g(\mathring{x})) & = &  \text{(by the assumption)} \\
			\sum_{\begin{array}{c} \scriptstyle  f \in \Xset' \rightarrow \Yset \\[-1mm] \scriptstyle f(\tilde{x}) = \tilde{y} \end{array}} \prod_{x \in \Xset'} p(x, f(x)) & = &  \text{(by the ind. hyp.)} \\
			p(\tilde{x},\tilde{y}) & & 
	\end{align*}
	
	\end{description}

\end{proof}

\begin{theorem}
\label{thm:channelFromTree}
	Let $\Isys$ be a fully probabilistic \IIHS inducing the joint probability distribution $p(\Aseq{}{t},\Bseq{}{t})$, $1 \leq t \leq T$, on secret and observable traces. 	It is always possible to construct a channel with memory and feedback, and an associated probability distribution $Q({\Fseq{}{T},\Aseq{}{T},\Bseq{}{T}})$, which corresponds to 	$\Isys$ in the sense that, for every $1 \leq t \leq T$, $\Aseq{}{t}$, $\Bseq{}{t}$, the equality $ Q(\Aseq{}{t},\Bseq{}{t}) = p(\Aseq{}{t},\Bseq{}{t})$ holds.
\end{theorem}

\begin{proof}

	First note that, by laws of probability, $ \prob{Q}{\Aseq{}{t},\Bseq{}{t}} = \sum_{\Fseq{}{t}} \prob{Q}{\Fseq{}{t},\Aseq{}{t},\Bseq{}{t}} $. So we need to show that $\sum_{\Fseq{}{t}} \prob{Q}{\Fseq{}{t},\Aseq{}{t},\Bseq{}{t}} = \prob{p}{\Aseq{}{t},\Bseq{}{t}} $ by induction on $t$.
	\begin{description}
		\item[Base case:] $t=1$. Let us define $\Cprob{Q}{\Fseq{1}{}}{\epsilon} = \prob{p}{\Fseq{1}{}(\epsilon)}$ and $\Cprob{Q}{\Bseq{1}{}}{\Aseq{}{1},\epsilon} = \Cprob{p}{\Bseq{1}{}}{\Aseq{1}{}}$. Then:			
		\begin{align*}
				\sum_{\Fseq{}{1}}\prob{Q}{\Fseq{}{1},\Aseq{}{1},\Bseq{}{1}} & = & \\
				\sum_{\Fseq{1}{}}\prob{Q}{\Fseq{1}{},\Aseq{1}{},\Bseq{1}{}} & = & \mbox{(by the chain rule)} \\
				\sum_{\Fseq{1}{}} \left( \Cprob{Q}{\Fseq{1}{}}{\epsilon,\epsilon,\epsilon} \cdot \Cprob{Q}{\Aseq{1}{}}{\Fseq{1}{},\epsilon,\epsilon} \cdot\right. & & \\
				\left. \quad \quad  \Cprob{Q}{\Bseq{1}{}}{\Fseq{1}{},\Aseq{1}{},\epsilon} \right) & = & \mbox{(by Definition~\ref{def:Q})} \\
				\sum_{\Fseq{1}{}} \Cprob{Q}{\Fseq{1}{}}{\epsilon} \delta_{\{\Fseq{1}{}(\epsilon)\}}(\Aseq{1}{}) \Cprob{Q}{\Bseq{1}{}}{\Aseq{}{1},\epsilon}  & = & \mbox{(by construction of $Q$)} \\
				\sum_{\Fseq{1}{}} \prob{p}{\Fseq{1}{}(\epsilon)} \delta_{\{\Fseq{1}{}(\epsilon)\}}(\Aseq{1}{}) \Cprob{p}{\Bseq{1}{}}{\Aseq{1}{}}  & = & \mbox{(by definition of $\delta$)} \\
				\prob{p}{\Aseq{1}{}} \Cprob{p}{\Bseq{1}{}}{\Aseq{1}{}} & = & \\
				\prob{p}{\Aseq{1}{},\Bseq{1}{}} & = & \\
				\prob{p}{\Aseq{}{1},\Bseq{}{1}} & &
		\end{align*}
	
		\item[Inductive case:]
		
			Let us define $\Cprob{Q}{\Bseq{t}{}}{\Aseq{}{t},\Bseq{}{t-1}} = \Cprob{p}{\Bseq{t}{}}{\Aseq{}{t},\Bseq{}{t-1}}$, and			
			\begin{equation*}
				\Cprob{Q}{\Fseq{t}{}}{\Fseq{}{t-1}} = \prod_{\Bseq{}{t-1}} \Cprob{p}{\Fseq{t}{}(\Bseq{}{t-1})}{\Fseq{}{t-1}(\Bseq{}{t-2}), \Bseq{}{t-1}}
			\end{equation*}
			
			Note that, if we consider $\Xset = \{\Bseq{}{t-1} \mid  \Bseq{i}{} \in \Bset, 1\leq i \leq t-1 \}$, $\Yset = \Aset $, and $\prob{p}{\Bseq{}{t-1},\Aseq{t}{}} = \Cprob{p}{\Aseq{t}{}}{\Fseq{}{t-1}(\Bseq{}{t-2}),\Bseq{}{t-1}}$, then $\Xset$, $\Yset$ and $p$ satisfy the hypothesis of Lemma \ref{lemma:prob1}.
			
			Then:	
			\begin{align*}				
					\sum_{\Fseq{}{t}} \prob{Q}{\Fseq{}{t},\Aseq{}{t},\Bseq{}{t}} & = & \mbox{(by the chain rule)} \\																													
					\sum_{\Fseq{}{t}} \left( \prob{Q}{\Fseq{}{t-1},\Aseq{}{t-1},\Bseq{}{t-1}} \cdot \right. & & \\
					\left. \Cprob{Q}{\Fseq{t}{}}{\Fseq{}{t-1},\Aseq{}{t-1},\Bseq{}{t-1}} \cdot \right. & & \\
					\left. \Cprob{Q}{\Aseq{t}{}}{\Fseq{}{t},\Aseq{}{t-1},\Bseq{}{t-1}} \cdot \Cprob{Q}{\Bseq{t}{}}{\Fseq{}{t},\Aseq{}{t},\Bseq{}{t-1}} \right) & = & \mbox{(by Definition~\ref{def:Q})} \\				
					\sum_{\Fseq{}{t}} \left( \prob{Q}{\Fseq{}{t-1},\Aseq{}{t-1},\Bseq{}{t-1}} \cdot \Cprob{Q\textsc{}}{\Fseq{t}{}}{\Fseq{}{t-1}} \right. & & \\ 			
					\left. \delta_{\{\Fseq{t}{}(\Bseq{}{t-1})\}}(\Aseq{t}{}) \cdot \Cprob{Q}{\Bseq{t}{}}{\Aseq{}{t},\Bseq{}{t-1}} \right) & = & \mbox{(by constr. of $Q$)} \\				
					\sum_{\Fseq{}{t}} \left( \prob{Q}{\Fseq{}{t-1},\Aseq{}{t-1},\Bseq{}{t-1}} \cdot \right. & & \\
					\left ( \prod_{\Bseq{}{'t-1}} \Cprob{p}{\Fseq{t}{}(\Bseq{}{'t-1})}{\Fseq{}{t-1}(\Bseq{}{'t-2}),\Bseq{}{'t-1}}  \right ) \cdot & & \\
					\left. \delta_{\{ \Fseq{t}{}(\Bseq{}{t-1}) \}}(\Aseq{t}{}) \cdot \Cprob{p}{\Bseq{t}{}}{\Aseq{}{t},\Bseq{}{t-1}} \right) & = & \mbox{(by definition of $\delta$)} \\			
					\sum_{{\begin{array}{c}\scriptstyle\Fseq{}{t} \\[-1mm] \scriptstyle\Fseq{t}{}(\Bseq{}{t-1}) = \Aseq{t}{} \end{array}}} \left( \prob{Q}{\Fseq{}{t-1},\Aseq{}{t-1},\Bseq{}{t-1}} \right. \cdot & & \\
					\left ( \prod_{\Bseq{}{'t-1}} \Cprob{p}{\Fseq{t}{}(\Bseq{}{'t-1})}{\Fseq{}{t-1}(\Bseq{}{'t-2}),\Bseq{}{'t-1}}  \right ) \cdot & & \\ 
					\left. \Cprob{p}{\Bseq{t}{}}{\Aseq{}{t},\Bseq{}{t-1}} \right) & = &\mbox{} \\				
					\sum_{\Fseq{}{t-1}} \left( \right. \prob{Q}{\Fseq{}{t-1},\Aseq{}{t-1},\Bseq{}{t-1}} \Cprob{p}{\Bseq{t}{}}{\Aseq{}{t},\Bseq{}{t-1}} & & \\
					\sum_{\begin{array}{c} \scriptstyle\Fseq{t}{} \\[-1mm] \scriptstyle \Fseq{t}{} (\Bseq{}{t-1}) = \Aseq{t}{} \end{array}} \prod_{\Bseq{}{'t-1}} \Cprob{p}{\Fseq{t}{}(\Bseq{}{'t-1})}{\Fseq{}{t-1}(\Bseq{}{'t-2}),\Bseq{}{'t-1}} \left. \right) & = & \mbox{(by Lemma~\ref{lemma:prob1})} \\				
					\sum_{\Fseq{}{t-1}} \left( \right. \prob{Q}{\Fseq{}{t-1},\Aseq{}{t-1},\Bseq{}{t-1}} \cdot \Cprob{p}{\Bseq{t}{}}{\Aseq{}{t},\Bseq{}{t-1}} \cdot & & \\
					\Cprob{p}{\Aseq{t}{}}{\Aseq{}{t-1},\Bseq{}{t-1}} \left. \right) & = & \mbox{} \\				
					\Cprob{p}{\Bseq{t}{}}{\Aseq{}{t},\Bseq{}{t-1}} \cdot \Cprob{p}{\Aseq{t}{}}{\Aseq{}{t-1},\Bseq{}{t-1}} \cdot & & \\
					\sum_{\Fseq{}{t-1}} \prob{Q}{\Fseq{}{t-1},\Aseq{}{t-1},\Bseq{}{t-1}} & = & \mbox{(by ind. hyp.)} \\				
					\Cprob{p}{\Bseq{t}{}}{\Aseq{}{t},\Bseq{}{t-1}} \cdot \Cprob{p}{\Aseq{t}{}}{\Aseq{}{t-1},\Bseq{}{t-1}} \cdot \prob{p}{\Aseq{}{t-1},\Bseq{}{t-1}} & = & \mbox{(by the chain rule)} \\				
					\prob{p}{\Aseq{}{t},\Bseq{}{t}} & & \\				
			\end{align*}

	\end{description}

\end{proof}

\begin{corollary}
		\label{corollary:distribution-F}
		Let $\Isys$ be a fully probabilistic \IIHS. Let $\{ \Cprob{p}{\Bseq{t}{}}{\Aseq{}{t},\Bseq{}{t-1}} \}_{t=1}^{T}$  be a sequence of stochastic kernels and $\{ \Cprob{p}{\Aseq{t}{}}{\Aseq{}{t-1},\Bseq{}{t-1}} \}_{t=1}^{T}$ a sequence of input distributions defined by $\Isys$ according to Definitions~\ref{def:def-matrix} and~\ref{def:input-distribution}. Then the reactor $R = \{ \Cprob{p}{\Fseq{t}{}}{\Fseq{}{t-1}} \}_{t=1}^{T}$ compatible with respect to the $\Isys$ is given by:
	\begin{align}
		\prob{p}{\Fseq{1}{}} & = &\Cprob{p}{\Aseq{1}{}}{\Aseq{}{0},\Bseq{}{0}} = \prob{p}{\Aseq{1}{}}\\
		\Cprob{p}{\Fseq{t}{}}{\Fseq{}{t-1}} & = &\prod_{\Bseq{}{t-1}}\Cprob{p}{\Fseq{t}{}(\Bseq{}{t-1})}{\Fseq{}{t-1}(\Bseq{}{t-2}),\Bseq{}{t-1}},  \ \ \ 2 \leq t \leq T
	\end{align}
		
\end{corollary}

  Figure~\ref{fig:dcmfb-iihs} depicts the model for \IIHS. Note that, in relation to Figure~\ref{fig:dcmfb}, there are some simplifications: (1) no message $W$ is needed; 2) the encoder becomes an \qm{interactor}; (3) the decoder is not used. At the beginning, a reaction function sequence $\Fseq{}{T}$ is chosen and then the channel is used $T$ times. At each usage $t$, the interactor produces the next input symbol $\Aseq{t}{}$ by applying the reaction function $\Fseq{t}{}$ to the fed back output $\Bseq{}{t-1}$. Then the channel produces an output $\Bseq{t}{}$ based on the stochastic kernel $\Cprob{p}{\Bseq{t}{}}{\Aseq{}{t},\Bseq{}{t-1}}$. The output is then fed back to the encoder, which uses it for producing the next input.

	\begin{figure}[!htb]
    \centering
		\includegraphics[width=0.8\linewidth]{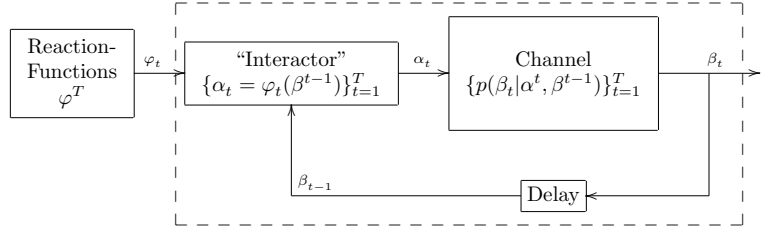}
		\caption{Channel with memory and feedback model for \IIHS}
		\label{fig:dcmfb-iihs}
	\end{figure}
		
    We conclude this section by remarking on an intriguing coincidence: The notion of reaction function sequence $\Fseq{}{T}$, on the \IIHSs, corresponds to the notion of deterministic scheduler~\cite{Segala:95:PhD}. In fact, each reaction function $\Fseq{t}{}$ selects the next step, $\Aseq{t}{}$, on the basis of the $\Bseq{}{t-1}$ and $\Aseq{}{t-1}$ (generated by $\Fseq{}{t-1}$), and $\Bseq{}{t-1}$, $\Aseq{}{t-1}$ represent the path \revision{up to that state}.

\section{Leakage in interactive systems}
\label{section:anonymity-properties}
	
In this section we propose a definition for the notion of leakage in interactive systems. We first argue that mutual information is not the correct notion, and we propose to replace it with the directed information instead.
	 
In the case of channels with memory and feedback, mutual information   is defined as $I(A^T;B^T) = H(A^T) - H(A^T|B^T)$, and it is still symmetric (i.e.  $I(A^T;B^T) = I(B^T;A^T)$). Since the roles of $A^T$ and $B^T$ in $I(A^T;B^T)$ are interchangeable,  this concept cannot capture \emph{causality}, in the sense that it does not imply that $A^T$ causes $B^T$, nor conversely.	Mutual information expresses	\emph{correlation} between the sequences of random variables $A^T$ and $B^T$. 
		
Mathematically the mutual information $I(A^{T};B^{T})$ for $T$ uses of the channel can be expressed with the help of the chain rule of \eqref{eq:chain-rule-mutual-information} in the following way.
\begin{equation*}
		I(A^T;B^T)= \sum_{t=1}^{T}I(A^{T};B_{t}|B^{t-1}) \nonumber
\end{equation*}
	
	In the equation above, each term of the sum is the mutual information between the random variable $B_{t}$ and the whole sequence of random
	variables $A^T = A_{1}, \ldots, A_{T}$, given the history $B^{t-1}$. The equation emphasizes that at time $1 \leq t \leq T$, even though only the inputs $\Aseq{}{t} = \Aseq{1}{}, \Aseq{2}{}, \ldots, \Aseq{t}{}$ have been fed to the channel, the whole sequence $A^{T}$, including $A_{t+1}, A_{t+2}, \ldots, A_{T}$, has a statistical correlation with $B_{t}$. Indeed, in the presence of feedback, $B_{t}$ may influence $A_{t+1}, A_{t+2}, \ldots, A_{T}$.
	
In order to show how the concept of directed information contrasts with the above, let us recall its definition:
	
		$$
			I(A^{T} \rightarrow B^{T}) = \sum_{t=1}^{T}	I(A^{t};B_{t} | B^{t-1}).
		$$
		$$
			I(B^{T} \rightarrow A^{T}) = \sum_{t=1}^{T} I(A_{t};B^{t-1} | A^{t-1}).
		$$
		
These notions capture the concept of \emph{causality}, to which the definition of mutual information is indifferent. The correlation between inputs and outputs $I(A^{T}; B^{T})$ is split into the information $I(A^{T} \rightarrow B^{T})$ that flows from input to output through the channel and the information $I(B^{T} \rightarrow A^{T})$ that flows from output to the input via feedback. Note that the directed information is not symmetric: the flow from $A^T$ to $B^T$ takes into account the correlation between $A^{t}$ and $B_{t}$, while the flow from $B^{T}$ to $A^{T}$ takes into account the correlation between  $B^{t-1}$ and $A_{t}$.
	
It was proved in~\cite{Tatikonda:09:TIT} that 
\begin{equation}
	\label{eq:directed-info-sum}
	I(A^{T};B^{T}) = I(A^{T} \rightarrow B^{T}) + I(B^{T}\rightarrow A^{T} )	
\end{equation}

\noindent i.e. the mutual information is the sum of the directed information flow in both senses. Note that this formulation highlights the symmetry of mutual information from yet another perspective.

Once we split mutual information into directed information in the two opposite directions, it is important to understand the different roles that the information flow in each direction plays. $I(A^{T} \rightarrow B^{T})$ represents the system behavior: via the channel the information flows from inputs to outputs according to the specification of the system, modeled by the channel stochastic kernels. This flow represents the amount of information an attacker can gain from the inputs by observing the outputs, and we argue that this is the real information leakage.

On the other hand, $I(B^{T} \rightarrow A^{T})$ represents how the environment reacts to the system: given the system outputs, the environment produces new inputs. We argue that the information flow from outputs to inputs is independent of any particular system: it is a characteristic of the environment itself. Hence, if an attacker knows how the environment reacts to outputs (the probabilistic behavior of the reactions of the environment given the system outputs), this knowledge is part of the \emph{a priori} knowledge of the adversary. As a further justification, observe that this is a natural extension of the classical approach, where the choice of secrets is seen as external to the system, i.e. determined by the environment. The probability distribution on the secrets constitutes the a priori knowledge and does not count as leakage. In order to encompass the classical approach, in our extended model we should preserve this principle, and a natural way to do so is to consider the secret choices, at every stage of the computation, as external. Their probability distributions, which are now in general conditional probability distributions depending on the history of secrets and observables, should therefore be considered as part of the external knowledge, and   not    counted as leakage.

The following example supports our claim that, in the presence of feedback, mutual information is not a correct notion of leakage.

\begin{example} 
	\label{exa:small}		
	
	Consider the discrete memoryless channel with secret alphabet $\Aset = \{\Asym{1}{},\Asym{2}{}\}$ and observable alphabet $\Bset=\{\Bsym{1}{},\Bsym{2}{}\}$ whose matrix is represented in Table~\ref{tab:small-example}. 

	\begin{table}[!htb]
		\centering
		\begin{tabular}{|c||c|c|}
		\hline
									 & $\Bsym{1}{}$ & $\Bsym{2}{}$ \\ \hline \hline
			$\Asym{1}{}$ & $0.5$    		& $0.5$			   \\ \hline
			$\Asym{2}{}$ & $0.5$    		& $0.5$			   \\ \hline			
		\end{tabular}
		\caption{Channel matrix for Example~\ref{exa:small}}
		\label{tab:small-example}
	\end{table}
			
Suppose that the channel is used with feedback, in such a way that, for all $1 \leq t \leq T$, we have $\Aseq{t+1}{}  = \Asym{1}{}$ if $\Bseq{t}{} = \Bsym{1}{}$, and $\Aseq{t+1}{}  = 	 \Asym{2}{}$ if $\Bseq{t}{} = \Bsym{2}{}$. It is easy to show that if $T\geq 2$ then  $I(A^{T};B^{T}) \neq 0$. Yet there is no leakage from  $A^T$ to $B^T$, since the rows of the matrix are all equal. We have indeed that $I(A^{T} \rightarrow B^{T}) = 0$, and the mutual information $I(A^{T};B^{T})$ is only due to the feedback information flow $I(B^{T} \rightarrow A^{T})$.	

\end{example}
		
Having in mind the above discussion, we now propose a notion of information flow based on our model.  We  follow the idea of defining leakage and maximum leakage using the concepts of mutual information and capacity, making the necessary adaptations.
	
As discussed in Chapter~\ref{chapter:probabilistic-info-flow}, in the \revision{non-interactive} case the definition of leakage as mutual information, for a single use of the channel, is \[ I(A;B) = H(A) - H(A|B)\] (cfr. for instance \cite{Chatzikokolakis:08:IC,Kopf:07:CCS}). This amounts to viewing the leakage as the difference between the a priori invulnerability and the a posteriori one. As explained in Chapter~\ref{chapter:probabilistic-info-flow}, these correspond to $H(A)$ and $H(A|B)$, respectively. This corresponds to the model of an attacker based on Shannon entropy discussed by K\"opf and Basin in~\cite{Kopf:07:CCS}. 
		
In the interactive case, we can extend this notion by considering the leakage at every step $t$ as given by \[I(A^t;B_t | B^{t-1}) = H(A^t|B^{t-1}) - H(A^t|B_t,B^{t-1})	\] The notion of attack is the same modulo the fact that we consider all the input from the beginning up to step $t$, and the difference in its vulnerability induced by the observation of $B_t$ (the output at step $t$), taking into account the observation history $B^{t-1}$. It is then natural to consider as total leakage the summation of the contributions $I(A^t;B_t | B^{t-1})$ for all the steps $t$. This is exactly the notion of directed information (cfr. Definition~\ref{def:directed-information}):
\[
\revision{I(A^{T} \rightarrow B^{T})} = \sum_{t=1}^{T} I(A^t;B_t | B^{t-1})
\]
			
	\begin{definition}
		\label{def:leakage}
		The information leakage of a fully probabilistic \IIHS\ is defined as the directed information
		$I(A^{T} \rightarrow B^{T})$ of the associated channel with memory and feedback.
	\end{definition}
	
	We now show an equivalent formulation of directed information that leads to a new interpretation in terms of an attack model. First we need the following lemma.
	
	\begin{lemma}
		\label{lemma:feedback-directed-info}
		$I(B^{T} \rightarrow A^{T}) = H(A^{T}) -\sum_{t=1}^{T} H(A_{t}|A^{t-1}, B^{t-1})$
	\end{lemma}
	
	\begin{proof}
		\begin{align*}
				I(B^{T} \rightarrow A^{T}) & = \sum_{t=1}^{T} I(A_{t};B^{t-1}|A^{t-1}) & \mbox{(by Definition~\ref{def:directed-information})} \\
				 & = \sum_{t=1}^{T} \left( H(A_{t}|A^{t-1}) \right. & \\
				 & \left. - H(A_{t}|A^{t-1}, B^{t-1}) \right)& \mbox{(by def. of mutual info.)} \\
				 & = H(A^{T}) -\sum_{t=1}^{T} H(A_{t}|A^{t-1}, B^{t-1}) & \mbox{(by the chain rule)} \\
		\end{align*}			
	\end{proof}
		
		The next proposition points out the announced alternative formulation of directed information from input to output:
		
		\begin{proposition}
			\label{prop:directedinfo}
				$I(A^{T} \rightarrow B^{T}) = \sum_{t=1}^{T} H(A_{t}|A^{t-1},B^{t-1}) - H(A^{T}|B^{T})$
		\end{proposition}
		
		\begin{proof}
			\begin{align*}	
				 I(A^{T} \rightarrow B^{T}) & = I(A^{T};B^{T}) - I(B^{T} \rightarrow A^{T}) & \mbox{(by (\ref{eq:directed-info-sum}))} \\
				  & = I(A^{T};B^{T}) - H(A^{T}) & \\
				  &   + \sum_{t=1}^{T} H(A_{t}|A^{t-1}, B^{t-1}) & \mbox{(by Lemma~\ref{lemma:feedback-directed-info})} \\
				  & = H(A^{T}) - H(A^{T}|B^{T}) - H(A^{T})& \\
				  &   + \sum_{t=1}^{T} H(A_{t}|A^{t-1}, B^{t-1}) & \mbox{(by def. of mutual info.)} \\
				  & = \sum_{t=1}^{T} H(A_{t}|A^{t-1}, B^{t-1}) - H(A^{T}|B^{T}) & \\
			\end{align*}
		\end{proof}		
	
 		We note that the term $\sum_{t=1}^{T} H(A_{t}|A^{t-1},B^{t-1})$ can be seen as the entropy $H_{R}$ of the reactor $R$, i.e. the entropy of the inputs, taking into account their dependency on the previous outputs. This brings us to an intriguing alternative interpretation of leakage. 
		
		\begin{remark}
		The leakage can be seen as the difference between the a priori invulnerability degree of the whole secret $A^T$, assuming that the attacker knows the distribution of the reactor, and the a posteriori invulnerability degree, after the adversary has observed the whole output $B^T$.
		\end{remark}

In Section~\ref{section:full-example} we give an extensive and detailed example of how to calculate the leakage for an actual security protocol.
	
In the case of secret-nondeterministic \IIHS, we have a stochastic kernel but no distribution on the reaction functions. In this case it seems natural to consider the worst leakage over all possible distributions on reaction functions. This is exactly the concept of capacity.
	
	\begin{definition}
		\label{def:maximum-leakage}
		The \emph{maximum leakage} of a secret-nondeterministic \IIHS is defined as the capacity $C_T$ of the associated channel with memory and feedback (cfr. \mbox{\textrm{{(\ref{def:cmfb-capacity})}}}).
	\end{definition}
	
A comparison with the definition of Gray (cfr. \cite{Gray:91:SSP}, Definition 5.3) is in order. As explained in the introduction, Gray's model is more complicated than ours, because it assumes that low and high variables are present at both ends of the channel. If we restrict the definition of Gray's capacity $C^G$ to our case, by eliminating the low input and the high output, we obtain the following formula:
\begin{equation}
	\label{def:Gray-capacity}
	C^{G }_{T} = \sup_{\Dset_{T}} \frac{1}{T} \sum_{t=1}^{T} I(A^{t-1};B_t|B^{t-1})
\end{equation}
		
By comparing (\ref{def:cmfb-capacity}), which is based on Definition~\ref{def:directed-information}, to (\ref{def:Gray-capacity}), we can see that the only difference is that (\ref{def:Gray-capacity}) considers the correlation between $B_t$ and $A^{t-1}$ instead of $A^t$. This seems to be intentional (cfr. \cite{Gray:91:SSP}, discussion after Definition 4.1). We are not sure why $C^G$ is defined in this way, our best guess is that the high values must be those of the previous time step in order to encompass the theory of McLean~\cite{McLean:90:SSP}. In any case, Gray's conjecture that $C^{G }_{T}$ corresponds to the channel transmission rate does not hold. For instance, it is easy to see that for $T=1$ we always have $C^{G }_{T}=0$, but there obviously are channels which can transmit a non-zero amount of information even with one single use.
	
We conclude this section by showing that our approach to the notion of leakage generalizes the classical approach (based on mutual information) to the case of feedback. The idea is that, if a channel does not have feedback, then $I(B^{T} \rightarrow A^{T}) = 0$ and therefore $I(A^{T};B^{T}) = I(A^{T} \rightarrow B^{T})$. In our opinion, the fact that mutual information turns out to be a particular case of directed information helps to justify the former as a good measure of information flow, despite its symmetry: in channels without feedback it is a good measure \emph{because it coincides with directed information} from input to output.
	
	\begin{lemma}
		\label{lemma:zero-feedback}
		In absence of feedback, $I(B^T \rightarrow A^T) = 0$
	\end{lemma}
	
	\begin{proof}
		When feedback is not allowed, $B^{t-1}$ and $A_t$ are independent for \revision{every} $1 \leq t \leq T$. Then:
		\begin{align*}
			 I(B^T \rightarrow A^T) = & \sum_{t=1}^{T} I(A_t;B^{t-1}|A^{t-1}) & \mbox{(by Definition~\ref{def:directed-information})} \\
			 & = \sum_{t=1}^{T} ( H(A_t|A^{t-1}) & \\
			 & \quad - H(A_t|A^{t-1}, B^{t-1}) ) & \mbox{(by def. of mutual info.)} \\
			 & = \sum_{t=1}^{T} ( H(A_t|A^{t-1}) & \\
			 & \quad - H(A_t|A^{t-1}) ) & \mbox{($B^{t-1}$ and $A^{t}$ are independent)} \\
			 & = 0 & 
		\end{align*}
	\end{proof}
	
	\begin{proposition}\label{prop:generalization}
		In absence of feedback, leakage can be equivalently defined as directed information or as mutual information. 
		Similarly, in absence of feedback, the maximum leakage can be equivalently defined as directed capacity or as capacity. 
		\label{prop:collapse}
	\end{proposition}
	
	\begin{proof}
		It follows directly from Lemma~\ref{lemma:zero-feedback} and \eqref{eq:directed-info-sum}.
	\end{proof}
	
\section{An example: the Cocaine Auction protocol}
\label{section:full-example}
In this section we show the application of our approach to the \emph{Cocaine Auction Protocol}~\cite{Stajano:99:IH}. The formalization of this protocol in terms of IIHSs using our framework makes it possible to prove the claim in \cite{Stajano:99:IH} suggesting that if the seller knows the identity of the bidders then the (strong) anonymity guaranties are no longer assured.
	
Let us consider a scenario in which several mobsters are gathered around a table. An auction is about to be held in which one of them  offers his next shipment of cocaine to 	the highest bidder. The seller describes the merchandise and proposes a starting price. The others then bid increasing amounts until there are no bids for, say, $30$ consecutive seconds. At that point the seller declares the auction closed and arranges a secret appointment with the winner to deliver the goods.
		
The basic protocol is fairly simple and is organized as a succession of rounds of bidding. Round $i$ starts with the seller announcing the bid price $b_i$ for that round. Buyers have $t$ seconds to make an offer (i.e. to say yes, meaning \qm{I'm willing to buy at the current	bid price $b_i$}). As soon as one buyer anonymously says yes, he becomes the winner $w_i$ of that round and a new round begins. If nobody says anything for $t$ seconds, round $i$ is concluded by timeout and the auction is won by the winner $w_{i-1}$ of the previous round, if one exists. If the timeout occurs during round $0$, this means that nobody made any offers at the initial price $b_0$, so there is no sale.

Although our framework allows the formalization of this protocol for an arbitrary number of bidders and bidding rounds, for illustration purposes we will consider the case of two bidders (\emph{Candlemaker} and \emph{Scarface}) and two rounds of bids. Furthermore, we assume that the initial bid is always $100$ euros, so the first bid does not need to be announced by the seller. In each turn the seller can choose how much he wants to increase the current bid value. This is done by adding an increment to the last bid. There are two options of increments, namely $inc_{1}$ ($100$ euros) and $inc_{2}$ ($200$ euros). In that way, $b_{i+1}$ is either $b_{i} + inc_{1}$ or $b_{i} + inc_{2}$. We can describe this protocol as a \emph{normalized} \IIHS $\ihs = (M, \Aset, \Bset)$, where $\Aset = \{ \mbox{Candlemaker}, \mbox{Scarface}, \Asym{}{*} \}$ is 	the set of secret actions, $\Bset = \{ inc_{1}, inc_{2}, \Bsym{*}{} \}$ is the set of observable actions, and the probabilistic automaton $M$ is represented in Figure~\ref{fig:cocaine-example}.	For clarity reasons, transitions with probability $0$ are not represented in the automaton. Note that the special secret action $\Asym{*}{}$ represents the situation where neither \emph{Candlemaker} nor \emph{Scarface} bid. The special observable action $\Bsym{*}{}$ represents the end of the auction and it can only occur if no one has bid in the round. 

	\begin{figure}
		\center
		\includegraphics[width=\textwidth]{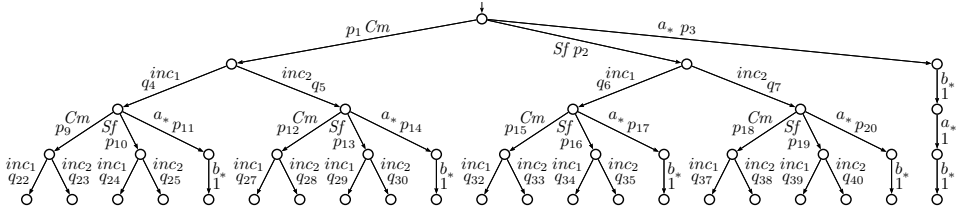}
		\caption{Cocaine auction example}
		\label{fig:cocaine-example}
	\end{figure}

Table~\ref{tab:sk-cocaine-example} shows all the stochastic kernels for this example. 

	\begin{table}[!htb]		
		\centering		
		\subbottom[$t\!=\!1,\Cprob{p}{\Bseq{1}{}}{\Aseq{}{1},\Bseq{}{0}}$]{
			\centering
			\small
			\begin{tabular}{|l||c|c|c|}
				\hline
				$\ \ \ \Aseq{1}{} \rightarrow \Bseq{1}{}$ & \ \ \ $inc_{1}$ \ \ \  &\ \ \ $inc_{2}$ \ \ \ & $\Bsym{*}{}$ \\ \hline \hline
				\textit{Candlemaker} & $q_{4}$ & $q_{5}$ & 0 \\ \hline
				\textit{Scarface} & $q_{6}$ & $q_{7}$ & 0 \\ \hline
				$\Asym{}{*}$  & $0$ & $0$ & 1 \\ \hline
			\end{tabular}
			\label{tab:sk-t1}	
		}					
		\subbottom[$t=2, \Cprob{p}{\Bseq{2}{}}{\Aseq{}{2},\Bseq{}{1}}$]{
			\centering
			\small
			\begin{tabular}{|l||c|c|c|}
				\hline
				$\Aseq{1}{},\Bseq{1}{},\Aseq{2}{} \rightarrow \Bseq{2}{}$ & $inc_{1}$ & $inc_{2}$ & $\Bsym{*}{}$ \\ \hline \hline
				\textit{Candlemaker},$inc_{1}$,\textit{Candlemaker}						& $q_{22}$ & ${q}_{23}$ & 0 \\ \hline				
				\textit{Candlemaker},$inc_{1}$,\textit{Scarface} 							& $q_{24}$ & ${q}_{25}$ & 0 \\ \hline				
				\textit{Candlemaker},$inc_{1}$,$\Asym{*}{}$ 					& 0 & 0 & 1 \\ \hline
				\textit{Candlemaker},$inc_{2}$,\textit{Candlemaker} 					& $q_{27}$ & ${q}_{28}$ & 0 \\ \hline				
				\textit{Candlemaker},$inc_{2}$,\textit{Scarface} 							& $q_{29}$ & ${q}_{30}$ & 0 \\ \hline			
				\textit{Candlemaker},$inc_{2}$,$\Asym{*}{}$					& 0 & 0 & 1 \\ \hline	
				\textit{Scarface},$inc_{1}$,\textit{Candlemaker}								& $q_{32}$ & ${q}_{33}$ & 0 \\ \hline				
				\textit{Scarface},$inc_{1}$,\textit{Scarface} 									& $q_{34}$ & ${q}_{35}$ & 0 \\ \hline			
				\textit{Scarface},$inc_{1}$,$\Asym{*}{}$ 							& 0 & 0 & 1 \\ \hline
				\textit{Scarface},$inc_{2}$,\textit{Candlemaker} 							& $q_{37}$ & ${q}_{38}$ & 0 \\ \hline			
				\textit{Scarface},$inc_{2}$,\textit{Scarface} 								& $q_{39}$ & ${q}_{40}$ & 0 \\ \hline				
				\textit{Scarface},$inc_{2}$,$\Asym{*}{}$							& 0 & 0 & 1 \\ \hline	
				$\Asym{*}{}$,$\Bsym{*}{}$,$\Asym{*}{}$	& 0 & 0 & 1 \\ \hline	
				All other lines													& 0 & 0 & 1 \\ \hline
			\end{tabular}
			\label{tab:sk-t2}	
		}			
		\caption{Stochastic kernels for the Cocaine Auction example}
		\label{tab:sk-cocaine-example}	
	\end{table}
	
The next step is to construct all possible reaction functions $\{\Fseq{t}{}(\Bseq{}{t-1})\}_{t=1}^{T}$. As seen in Section~\ref{section:reaction-functions}, the reaction functions \revision{correspond to} the encoder in the channel. They take the feedback story and decide how the world will react to this situation. Table~\ref{tab:rf-cocaine-example} contains the reaction functions for each time $t \leq 2$.

\setlength{\extrarowheight}{3pt}
\begin{table}[!htb]
	\centering
	\small
	\subbottom[All $3$ reaction functions $\Fseq{1}{}$]{
		\centering
		\begin{tabular}{|c||c|c|c|}
			\hline
			$\Bseq{}{0}$ & $\Fsym{1(1)}{}$ & $\Fsym{1(2)}{}$ & $\Fsym{1(3)}{}$  \\ \hline \hline
			$\emptyset$  & \textit{Candlemaker} & \textit{Scarface} & $\Asym{*}{}$ \\ \hline
		\end{tabular}
	}
	\subbottom[All $27$ reaction functions $\Fseq{2}{}(\Bseq{}{1})$]{
		\centering
		\begin{tabular}{|c||c|c|c|c|}
			\hline 
			$\Bseq{}{1}$	& $\Fsym{2(1)}{}(\Bseq{}{1})$ & $\Fsym{2(2)}{}(\Bseq{}{1})$ & $\Fsym{2(3)}{}(\Bseq{}{1})$ & $\Fsym{2(4)}{}(\Bseq{}{1})$ \\ \hline \hline
			$inc_{1}$	    & \textit{Candlemaker}        & \textit{Candlemaker}        & \textit{Candlemaker}        & \textit{Candlemaker}        \\ \hline	 	
			$inc_{2}$     & \textit{Candlemaker}        & \textit{Candlemaker}        & \textit{Candlemaker}        & \textit{Scarface}           \\ \hline		
			$\Bsym{*}{}$	& \textit{Candlemaker}        & \textit{Scarface}           & $\Asym{*}{}$                & \textit{Candlemaker}        \\ \hline 
			\hline
			$\Bseq{}{1}$	& $\Fsym{2(5)}{}(\Bseq{}{1})$ & $\Fsym{2(6)}{}(\Bseq{}{1})$ & $\Fsym{2(7)}{}(\Bseq{}{1})$ & $\Fsym{2(8)}{}(\Bseq{}{1})$ \\ \hline \hline
			$inc_{1}$	    & \textit{Candlemaker}        & \textit{Candlemaker}        & \textit{Candlemaker}        & \textit{Candlemaker}        \\ \hline	 	
			$inc_{2}$	    & \textit{Scarface}           & \textit{Scarface}           & $\Asym{*}{}$                & $\Asym{*}{}$                \\ \hline		
			$\Bsym{*}{}$	& \textit{Scarface}           & $\Asym{*}{}$                & \textit{Candlemaker}        & \textit{Scarface}           \\ \hline 
			\hline		
			$\Bseq{}{1}$  & $\Fsym{2(9)}{}(\Bseq{}{1})$ & $\Fsym{2(10)}{}(\Bseq{}{1})$ & $\Fsym{2(11)}{}(\Bseq{}{1})$ & $\Fsym{2(12)}{}(\Bseq{}{1})$ \\ \hline \hline
			$inc_{1}$	    & \textit{Candlemaker}        & \textit{Scarface}            & \textit{Scarface}            & \textit{Scarface}            \\ \hline	
			$inc_{2}$	    & $\Asym{*}{}$                & \textit{Candlemaker}         & \textit{Candlemaker}         & \textit{Candlemaker}         \\ \hline			
			$\Bsym{*}{}$	& $\Asym{*}{}$                & \textit{Candlemaker}         & \textit{Scarface}            & $\Asym{*}{}$                 \\ \hline 
			\hline 		
			$\Bseq{}{1}$  & $\Fsym{2(13)}{}(\Bseq{}{1})$ & $\Fsym{2(14)}{}(\Bseq{}{1})$ & $\Fsym{2(15)}{}(\Bseq{}{1})$ & $\Fsym{2(16)}{}(\Bseq{}{1})$ \\ \hline \hline
			$inc_{1}$	    & \textit{Scarface}            & \textit{Scarface}            & \textit{Scarface}            & \textit{Scarface}            \\ \hline	
			$inc_{2}$	    & \textit{Scarface}            & \textit{Scarface}            & \textit{Scarface}            & $\Asym{*}{}$                 \\ \hline			
			$\Bsym{*}{}$	& \textit{Candlemaker}         & \textit{Scarface}            & $\Asym{*}{}$                 & \textit{Candlemaker}         \\ \hline 
			\hline 		
			$\Bseq{}{1}$  & $\Fsym{2(17)}{}(\Bseq{}{1})$ & $\Fsym{2(18)}{}(\Bseq{}{1})$ & $\Fsym{2(19)}{}(\Bseq{}{1})$ & $\Fsym{2(20)}{}(\Bseq{}{1})$ \\ \hline \hline
			$inc_{1}$	    & \textit{Scarface}            & \textit{Scarface}            & $\Asym{*}{}$                 & $\Asym{*}{}$                 \\ \hline	
			$inc_{2}$	    & $\Asym{*}{}$                 & $\Asym{*}{}$                 & \textit{Candlemaker}         & \textit{Candlemaker}         \\ \hline			
			$\Bsym{*}{}$	& \textit{Scarface}            & $\Asym{*}{}$                 & \textit{Candlemaker}         & \textit{Scarface}            \\ \hline 
			\hline 		
			$\Bseq{}{1}$  & $\Fsym{2(21)}{}(\Bseq{}{1})$ & $\Fsym{2(22)}{}(\Bseq{}{1})$ & $\Fsym{2(23)}{}(\Bseq{}{1})$ & $\Fsym{2(24)}{}(\Bseq{}{1})$ \\ \hline \hline
			$inc_{1}$	    & $\Asym{*}{}$                 & $\Asym{*}{}$                 & $\Asym{*}{}$                 & $\Asym{*}{}$                 \\ \hline	
			$inc_{2}$	    & \textit{Candlemaker}         & \textit{Scarface}            & \textit{Scarface}            & \textit{Scarface}            \\ \hline
			$\Bsym{*}{}$	& $\Asym{*}{}$                 & \textit{Candlemaker}         & \textit{Scarface}            & $\Asym{*}{}$                 \\ \hline
			\hline 		
			$\Bseq{}{1}$  & $\Fsym{2(25)}{}(\Bseq{}{1})$ & $\Fsym{2(26)}{}(\Bseq{}{1})$ & $\Fsym{2(27)}{}(\Bseq{}{1})$ & --- \\ \hline \hline
			$inc_{1}$	    & $\Asym{*}{}$                 & $\Asym{*}{}$                 & $\Asym{*}{}$                 & --- \\ \hline	
			$inc_{2}$	    & $\Asym{*}{}$                 & $\Asym{*}{}$                 & $\Asym{*}{}$                 & --- \\ \hline			
			$\Bsym{*}{}$	& \textit{Candlemaker}         & \textit{Scarface}            & $\Asym{*}{}$                 & --- \\ \hline \hline 	
		\end{tabular}
	}		
	\caption{Reaction functions for the cocaine auction example}
	\label{tab:rf-cocaine-example}
\end{table}
\setlength{\extrarowheight}{0pt}	
	
Now we need to define the reactor, i.e. the probability distribution on reaction functions. Corollary~\ref{corollary:distribution-F} shows that we can do so by using the following equations:
\begin{align*}
		\prob{p}{\Fseq{1}{}} & = \Cprob{p}{\Aseq{1}{}}{\Aseq{}{0},\Bseq{}{0}} = \prob{p}{\Aseq{1}{}} \\
		\Cprob{p}{\Fseq{t}{}}{\Fseq{}{t-1}} & = \prod_{\Bseq{}{t-1}}\Cprob{p}{\Fseq{t}{}(\Bseq{}{t-1})}{\Fseq{}{t-1}(\Bseq{}{t-2}),\Bseq{}{t-1}},  \ \ \ 2 \leq t \leq T 
\end{align*}
		
For instance, $\prob{p}{\Fsym{1(1)}{}} = \prob{p}{\textit{Candlemaker}} = p_{1}$. In the same way, $\prob{p}{\Fsym{1(2)}{}} = \prob{p}{\mbox{\textit{Scarface}}} = p_{2}$ and $\prob{p}{\Fsym{1(3)}{}} = \prob{p}{\Asym{*}{}} = p_{3}$.

Let us take as an example the calculation of \revision{$\Cprob{p}{\Fsym{2(6)}{}}{{\Fsym{1(1)}{}}}$}:
	\begin{align*}
		\Cprob{p}{\Fsym{2(6)}{}}{{\Fsym{1(1)}{}}} & = \prod_{\Bseq{}{1}} \Cprob{p}{\Fsym{2(6)}{}(\Bseq{}{1}){}}{\Fseq{1(1)}{},\Bseq{}{1}} \\
			& = \Cprob{p}{\Fsym{2(6)}{}(inc_{1})}{\mbox{\textit{Candlemaker}},inc_{1}} \cdot \\
			& \quad \quad  \Cprob{p}{\Fsym{2(6)}{}(inc_{2})}{\mbox{\textit{Candlemaker}},inc_{2}} \cdot \\
			& \quad \quad  \Cprob{p}{\Fsym{2(6)}{}(\Bsym{*}{})}{\mbox{\textit{Candlemaker}},\Bsym{*}{}} \\
			& = \Cprob{p}{\mbox{\textit{Candlemaker}}}{\mbox{\textit{Candlemaker}},inc_{1}}	\cdot \\
			& \quad \quad \Cprob{p}{\mbox{\textit{Scarface}}}{\mbox{\textit{Candlemaker}},inc_{2}} \\
			& \quad \quad  \Cprob{p}{\Asym{*}{}}{\mbox{\textit{Candlemaker}},\Bsym{*}{}} \\
			& = p_{9} \cdot p_{13} \cdot 1 \\
			& = p_{9} p_{13} \\
	\end{align*}
	
Note that some reaction functions can have probability $0$, which is consistent with the probabilistic automaton. For instance:
	\begin{align*}
		\Cprob{p}{\Fsym{2(25)}{}}{{\Fsym{1(3)}{}}} & = \prod_{\Bseq{}{1}} \Cprob{p}{\Fsym{2(25)}{}(\Bseq{}{1}){}}{\Fseq{1(3)}{},\Bseq{}{1}} \\
			& = \Cprob{p}{\Fsym{2(25)}{}(inc_{1})}{\Asym{*}{},inc_{1}} \cdot \Cprob{p}{\Fsym{2(25)}{}(inc_{2})}{\Asym{*}{},inc_{2}} \cdot \\
			& \quad \quad  \Cprob{p}{\Fsym{2(25)}{}(\Bsym{*}{})}{\Asym{*}{},\Bsym{*}{}} \\
			& = \Cprob{p}{\Asym{*}{}}{\Asym{*}{},inc_{1}}	\cdot \Cprob{p}{\Asym{*}{}}{\Asym{*}{},inc_{2}} \cdot \Cprob{p}{\mbox{\textit{Candlemaker}}}{\Asym{*}{},\Bsym{*}{}} \\
			& = 1 \cdot 1 \cdot 0 \\
			& = 0 \\
	\end{align*}
	 		
\subsection{Calculating the information leakage}
	
Let us now calculate the information leakage for this example using the concepts from Section~\ref{section:anonymity-properties}. We will analyze three different scenarios:
	
\begin{description}
	
	\item [Example \texttt{a}:] There is feedback, but the probability of an observable does not depend on the history of secrets. In the auction protocol, this corresponds to a scenario where the probability of one of the mobsters to bid can depend on the increment imposed by the seller, but the history of who has previously bid in the past has no influence on how the seller chooses the bid increment in the coming turns. In other words, the seller cannot use the information of who has been bidding to change his strategy of defining the new increments. This situation corresponds to the original description of the protocol in~\cite{Stajano:99:IH}, where the seller does not have access to the identity of the bidder, for the sake of anonymity preservation. In general, we have $\Cprob{p}{\Bseq{t}{}}{\Aseq{}{t},\Bseq{}{t-1}} = \Cprob{p}{\Bseq{t}{}}{\Bseq{}{t-1}}$ for every $1 \leq t \leq T$. There is an exception, however: if there is no bidder, the case modeled by
the secret being $\Asym{*}{}$, then the auction terminates, which is signaled by the observable $\Bsym{*}{}$.
				
	\item [Example \texttt{b}:] This is the most general case, without any restrictions. The presence of feedback allows the probability of the bidder to depend of the \revision{increment in} the price. For instance, if \textit{Candlemaker} is richer than \textit{Scarface}, it is more likely that the former bids if the increment in the price is $inc_{2}$ instead of $inc_{1}$. Also, the probability of an observable can depend on the history of secrets, i.e. in general $\Cprob{p}{\Bseq{t}{}}{\Aseq{}{t},\Bseq{}{t-1}} \neq \Cprob{p}{\Bseq{t}{}}{\Bseq{}{t-1}}$ for $1 \leq t \leq T$. This scenario can represent a situation where the seller is corrupted and can use his information to affect the outcome of the auction. As an example, suppose that the seller is a friend of \textit{Scarface} and he wants to help him in the auction. One way of doing so is to check who was the winner of the last bidding round. Whenever the winner is \textit{Candlemaker}, the seller chooses as increment the small value $inc_{1}$, hoping that it will give  \textit{Scarface} a good chance to bid in the next round. On the other hand, whenever the seller detects that the winner is \textit{Scarface}, he chooses as the next increment the greater value $inc_{2}$, hoping that it will minimize the chances of \textit{Candlemaker} to bid in the next round (and therefore maximizing the chances of the auction to end up having \textit{Scarface} as the final winner).
		
	\item [Example \texttt{c}:] There is no feedback. In the cocaine auction, we can have the (perhaps unrealistic) situation in which the increment added to the bid has no influence on the probability of \textit{Candlemaker} or \textit{Scarface} being the bidder. Mathematically, we have $\Cprob{p}{\Aseq{t}{}}{\Aseq{}{t-1},\Bseq{}{t-1}} = \Cprob{p}{\Aseq{t}{}}{\Aseq{}{t-1}}$ for every $1 \leq t \leq T$. As in Example \texttt{b}, however, we do not impose any restriction on $\Cprob{p}{\Bseq{t}{}}{\Aseq{}{t},\Bseq{}{t-1}}$.
		
\end{description}
	
For each scenario we need to fill in the values of the probabilities in the protocol tree in Figure~\ref{fig:cocaine-example}. The probabilities for each example are listed in Table~\ref{tab:example-probabilities}.
\begin{table}[!htb]
	\centering
	\small
	\begin{tabular}{|c||c|c|c|}
		\hline
		Probability & Example \texttt{a} & Example \texttt{b} & Example \texttt{c} \\
		variable & value        & value        & value        \\ \hline \hline
		$p_{1}$  & 0.75 & 0.70 & 0.70 \\ \hline
		$p_{2}$  & 0.24 & 0.24 & 0.24 \\ \hline
		$p_{3}$  & 0.01 & 0.01 & 0.01 \\ \hline
			
		$q_{4}$  & 0.50 & 0.55 & 0.30 \\ \hline
		$q_{5}$  & 0.50 & 0.45 & 0.70 \\ \hline
		$q_{6}$  & 0.50 & 0.45 & 0.70 \\ \hline
		$q_{7}$  & 0.50 & 0.55 & 0.30 \\ \hline
			
		$p_{9}$  & 0.04 & 0.80 & 0.75 \\ \hline
		$p_{10}$ & 0.95 & 0.19 & 0.20 \\ \hline
		$p_{11}$ & 0.01 & 0.01 & 0.05 \\ \hline
		$p_{12}$ & 0.95 & 0.19 & 0.75 \\ \hline
		$p_{13}$ & 0.04 & 0.80 & 0.20 \\ \hline
		$p_{14}$ & 0.01 & 0.01 & 0.05 \\ \hline
		$p_{15}$ & 0.04 & 0.90 & 0.65 \\ \hline
		$p_{16}$ & 0.95 & 0.09 & 0.35 \\ \hline
		$p_{17}$ & 0.01 & 0.01 & 0.05 \\ \hline
		$p_{18}$ & 0.95 & 0.09 & 0.65 \\ \hline
		$p_{19}$ & 0.04 & 0.90 & 0.35 \\ \hline
		$p_{20}$ & 0.01 & 0.01 & 0.05 \\ \hline

		$q_{22}$ & 0.50 & 0.80 & 0.45 \\ \hline
		$q_{23}$ & 0.50 & 0.20 & 0.55 \\ \hline
		$q_{24}$ & 0.50 & 0.20 & 0.55 \\ \hline
		$q_{25}$ & 0.50 & 0.80 & 0.45 \\ \hline
		$q_{27}$ & 0.45 & 0.75 & 0.45 \\ \hline
		$q_{28}$ & 0.55 & 0.25 & 0.55 \\ \hline
		$q_{29}$ & 0.45 & 0.35 & 0.55 \\ \hline
		$q_{30}$ & 0.55 & 0.65 & 0.45 \\ \hline
		$q_{32}$ & 0.50 & 0.55 & 0.45 \\ \hline
		$q_{33}$ & 0.50 & 0.45 & 0.55 \\ \hline
		$q_{34}$ & 0.50 & 0.40 & 0.55 \\ \hline
		$q_{35}$ & 0.50 & 0.60 & 0.45 \\ \hline
		$q_{37}$ & 0.45 & 0.60 & 0.45 \\ \hline
		$q_{38}$ & 0.55 & 0.40 & 0.55 \\ \hline
		$q_{39}$ & 0.45 & 0.35 & 0.55 \\ \hline
		$q_{40}$ & 0.55 & 0.55 & 0.45 \\ \hline
			
	\end{tabular}

	\caption{Values of the probabilities in Figure~\ref{fig:cocaine-example} for  Examples \texttt{a}, \texttt{b}, and \texttt{c}}
	\label{tab:example-probabilities}		
\end{table}
Table~\ref{tab:example-values} shows a comparison between some relevant values \revision{for} the three cases.
\setlength{\extrarowheight}{3pt}	
\begin{table}[!htb]		
	\centering
	\small
		\begin{tabular}{|l|l||c|c|c|}
		\hline
		Interpretation           & Symbol                       & Example a & Example b & Example c \\ \hline \hline
		Input uncertainty        & $H(A^{T})$                   & 1.9319    & 1.9054    & 1.9158    \\ \hline
		Reactor uncertainty      & $H_{R}$                      & 1.1911    & 1.5804    & 1.9158    \\ \hline
		A posteriori uncertainty & $H(A^{T}|B^{T})$             & 1.0303    & 1.2371    & 1.4183    \\ \hline
		Mutual information       & $I(A^{T};B^{T})$             & 0.9016    & 0.6684    & 0.4975    \\ \hline
		Leakage                  & $I(A^{T} \rightarrow B^{T})$ & 0.1608    & 0.3433    & 0.4975    \\ \hline
		Feedback information     & $I(B^{T} \rightarrow A^{T})$ & 0.7408    & 0.3250    & 0.0000    \\ \hline						
	\end{tabular}		
	\caption{Values of the entropy and directed information for Examples \texttt{a}, \texttt{b}, and \texttt{c}, where $I(A^{T};B^{T}) = H(A^{T}) - H(A^{T}|B^{T})$ and $I(A^{T} \rightarrow B^{T}) = H_{R} - H(A^{T}|B^{T})$}
	\label{tab:example-values}
\end{table}
\setlength{\extrarowheight}{0pt}
		
In Example \texttt{a}, since the probability of observables does  not depend on the history of secrets, there is (almost) no information flowing from the input to the output, and the directed information $I(A^{T} \rightarrow B^{T})$ is close to zero, i.e. the leakage is low. The only reason why the leakage is not zero is because the end of an auction needs to be signaled. Due to presence of feedback, however, the directed information in the other sense $I(B^{T} \rightarrow A^{T})$ is non-zero, and so is the mutual information $I(A^{T};B^{T})$. This is an example where the mutual information does not correspond to the real information leakage, since  some (in this case, most) of the correlation between input and output can be attributed to the feedback.
	
In Example \texttt{b} the information flow from input to output $I(A^{T} \rightarrow B^{T})$ is significantly higher than zero, but still, due to feedback, the information flow from outputs to inputs $I(B^{T} \rightarrow A^{T})$ is not zero and  the mutual information $I(A^{T};B^{T})$ is higher than the directed information $I(A^{T} \rightarrow B^{T})$.
		
In Example \texttt{c}, the absence of feedback implies that $I(B^{T} \rightarrow A^{T})$ is zero. In that case the values of $I(A^{T};B^{T})$ and $I(A^{T} \rightarrow B^{T})$ coincide, and represent the real leakage.

Finally, Figure~\ref{fig:examples-values} shows a comparison between the values of the entropy and of the directed information in the examples. The totality of the mutual information $I(A^{T};B^{T})$ is represented by the height of the correspondent bar, and we emphasize the contribution of the directed information in each direction by splitting the bar into two parts. This figure highlights the fact that mutual information can be misleading as a measure of leakage. The greatest mutual information is obtained in Example \texttt{a}, followed by Example \texttt{b} and then by Example \texttt{c}. The \emph{real leakage}, however, given by $I(A^{T} \rightarrow B^{T})$, respects exactly the inverse order, namely Example \texttt{a} presents the lowest value while Example \texttt{c} presents the highest one. Indeed, in Example \texttt{a} the value of $I(A^{T} \rightarrow B^{T})$ represents only $18\%$ of the mutual information, while in Example \texttt{b} it represents $51\%$ and in Example \texttt{c} it amounts to $100\%$.
	\begin{figure}[!htb]
		\centering
		\includegraphics[width=0.8\textwidth]{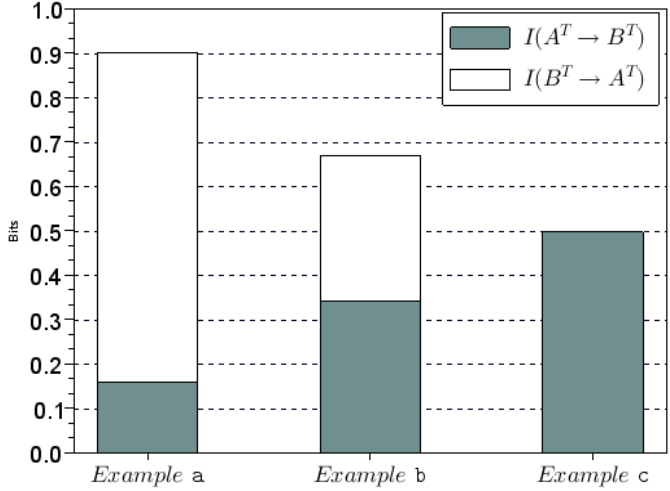}	
		\vspace{-.5cm}			
		\caption{Comparison between the leakage in Examples \texttt{a}, \texttt{b}, and \texttt{c}}		
		\label{fig:examples-values}			
	\end{figure}

\section{Topological properties of \IIHSs and their capacity}	
\label{section:topological-properties}
In this section we show how to extend to \IIHSs the notion of pseudometric defined in \cite{Desharnais:02:LICS} for Concurrent Labeled Markov Chains, and we prove that the capacity of the corresponding channels  is a continuous function with respect to this pseudometric. The pseudometric construction is sound for general \IIHSs, but the result on capacity is only valid for secret-nondeterministic \IIHSs.

Given a set of states $S$, a pseudometric is a function $d$ that yields a non-negative real number for each pair of states and satisfies the following: 

\begin{itemize}[(i)]
	\item $d(s,s)=0$;
	\item $d(s,t)=d(t,s)$; and
	\item $d(s,t)\leq d(s,u)+d(u,t)$. 
\end{itemize}
  
We say that a pseudometric $d$ is $c$-bounded if $\forall s,t: d(s,t)\leq c$, where $c$ is a positive real number.

Note that, in contrast to metrics, in pseudometrics two elements  can have distance $0$ without being identical. We consider pseudometrics instead of metrics because our purpose is to extend the notion of (probabilistic) bisimulation:  having distance $0$ will correspond to being bisimilar.

We now define a complete lattice \revision{structure} on pseudometrics, in order  to define the distance between \IIHSs as the greatest fixpoint of a particular transformation, in line with the coinductive theory of bisimilarity. Since larger bisimulations identify more, the natural extension of the ordering to pseudometrics must shorten the distances as we go up in the lattice:

\begin{definition}
$\CM$ is the class of $1$-bounded pseudometrics on states with the
  ordering
\[d\preceq d'\ \textrm{if}\ \forall s,s'\in S: d(s,s')\geq
  d'(s,s').\]
\end{definition}

It is easy to see that $(\CM,\preceq)$ is a complete lattice. In order to define pseudometrics on \IIHSs, we now need to lift the pseudometrics on states to pseudometrics on distributions in $\distr({\mathcal L} \times S)$. Following standard lines \cite{vanBreugel:01:ICALP,Desharnais:02:LICS,Deng:05:QAPL}, we apply the construction based on the Kantorovich metric \cite{Kantorovich:42:DAN}.

\begin{definition}
	For $d\in\CM$, and $\mu,\mu'\in\distr({\mathcal L}\times S)$, we define $d(\mu,\mu')$ (overloading the notation $d$) as 
	\begin{equation*}
		d(\mu,\mu') \ = \  \max\sum_{(\ell_i,s_i)\in{\mathcal L} \times S}(\mu(\ell_i,s_i)-\mu'(\ell_i,s_i))x_i
	\end{equation*} 
	
	\noindent where the maximum is taken over all possible values of the  $x_i$'s, subject to the constraints $0\leq x_i\leq 1 $ and $x_i-x_j\leq  \hat{d}((\ell_i,s_i),(\ell_j,s_j))$, where
	\begin{equation*}
		\hat{d}((\ell_i,s_i),(\ell_j,s_j)) \ = \ \left\{\begin{array}{ll}
    	                                     1 & \textrm{if}\ \ell_i\not=\ell_j\\
      	                          d(s_i,s_j)\quad & \textrm{otherwise}
        		                          \end{array}\right.
	\end{equation*}
\end{definition}

It can be shown that with this definition $m$ is a pseudometric on $\distr({\mathcal L} \times S)$.

\begin{definition}
	\label{def:sm}
	A pseudometric $d\in\CM$ is a {\em bisimulation  pseudometric} \footnote{In literature a pseudometric with this property is also known as bisimulation metric, although it is still a pseudometric.} if, for all $\epsilon\in [0,1)$, $d(s,s')\leq\epsilon$ implies that if $s\to\mu$, then there exists some $\mu'$ such that $s'\to\mu'$ and $d(\mu,\mu')\leq\epsilon$.
\end{definition}

Note that it is not necessary to require the converse of the condition in Definition \ref{def:sm} to get a complete analogy with bisimulation: the converse is indeed implied by the symmetry of $d$ as a pseudometric.  Note also that we prohibit $\epsilon$ to be $1$ because, throughout this chapter, $1$ represents the maximum distance,  which includes the case where one state may perform a transition and the other may not. 

The greatest bisimulation pseudometric is 
\begin{equation}
	\label{eqn:max}
	d_{\textit{max}}=\bigsqcup\{d\in\CM\mid d \mbox{ is a bisimulation pseudometric}\}
\end{equation}

We now characterize $d_\textit{max}$ as a fixed point of a monotonic function $\Phi$ on $\CM$. Eventually we are interested in the distance between \IIHSs, and for the sake of simplicity, from now on we consider only the distance between states belonging to different  \IIHSs. The extension to the general case is trivial. For  clarity purposes, we assume that  different \IIHSs have disjoint sets of states.

\begin{definition}
	Given two \IIHSs   with transition relations $\theta$ and $\theta'$ respectively, and a pseudometric $d$ on states, define  $\Phi:\CM\rightarrow\CM$ as:
\begin{equation*}
	\Phi(d)(s,s')=\left\{\begin{array}{ll}
	\max_i{d(s_i,s'_i)}  & \mbox{if  }\  \ \ \ \ \ \  \vartheta(s)  =  \{\delta_{(a_1,s_1)},\ldots,\delta_{(a_m,s_m)}\} \\
  	                   & \mbox{and } \ \vartheta'(s') \ \! =  \{\delta_{(a_1,s'_1)},\ldots,\delta_{(a_m,s'_m)}\}
	\\[2mm]
	d(\mu,\mu') & \mbox{if  }  \vartheta(s) = \{\mu\} \mbox{ and } \vartheta'(s')= \{\mu'\}
	\\[2mm]
	0 &\mbox{if } \vartheta(s)= \vartheta'(s')=\emptyset
	\\[2mm]
	1 & \mbox{otherwise }
	\end{array}
	\right.
\end{equation*}
\end{definition}

It is easy to see that the definition of $\Phi$ is a particular case of the function $F$ defined in \cite{Desharnais:02:LICS,Deng:05:QAPL}, which is characterized as follows (cf. Lemma 3.8 in the full version of  \cite{Desharnais:02:LICS}, and Definition 2.7 in \cite{Deng:05:QAPL}):
\begin{equation*}
	F(d)(s,s')=\max\{\sup_{s\to \mu}\inf_{s'\to\mu'} d(\mu,\mu')\ , \  \sup_{s'\to \mu'}\inf_{s\to\mu} d(\mu,\mu')\}
\end{equation*}

Hence it can be proved, as an instance of the analogous result for $F$ (cf. Lemma 2.8  in \cite{Deng:05:QAPL}), that $\Phi(d)$ is a pseudometric, and that the following property holds.

\begin{lemma}
	For  $\epsilon\in [0,1)$, $\Phi(d)(s,s')\leq\epsilon$ holds if and only if whenever  $s\to\mu$, there exists some $\mu'$ such that $s'\to\mu'$ and $d(\mu,\mu')\leq\epsilon$.
\end{lemma}

From the above lemma and Definition~\ref{def:sm} we derive (see also Lemma 2.9 in \cite{Deng:05:QAPL}):

\begin{corollary}
\label{l:smfix}
	A pseudometric $d$ is a bisimulation pseudometric if and only if $d\preceq \Phi(d)$.
\end{corollary}

By applying Corollary~\ref{l:smfix} to \eqref{eqn:max} we obtain 
	\[d_{\textit{max}}=\bigsqcup\{d\in\CM\mid d\preceq \Phi(d)\}\]
Furthermore,  by adapting the proof of the \revision{monotonicity of} $F$ (cf. Lemma 3.9 in the full version of \cite{Desharnais:02:LICS}) we can prove the following: 
\begin{lemma}
	\label{l:monotonicity}
	$\Phi$ is monotonic on $(\CM\preceq)$.
\end{lemma}

Thanks to Lemma~\ref{l:monotonicity}, and using Tarski's fixed point theorem as formulated in \cite{Tarski:55:PJM}, we have that $d_{\textit{max}}$ is the greatest fixed point of $\Phi$. Furthermore, by Corollary~\ref{l:smfix} we know that $d_{\textit{max}}$ is indeed a bisimulation pseudometric, and that it is the greatest bisimulation pseudometric.

In addition, the finite branching property of \IIHSs ensures that the closure ordinal of $\Phi$ is $\omega$ (cf. Lemma 3.10 in the full version of \cite{Desharnais:02:LICS}). Therefore we can proceed in a standard way  to show that 
\[d_{\textit{max}}=\bigsqcap\ \{\Phi^i(\top)\mid i\in{\mathbb N}\},\]
\noindent where $\top$ is the greatest pseudometric  (i.e. $\top(s,s')=0$ for every $s,s'$), and $\Phi^0(\top)=\top$.

Given two  \IIHSs $\Isys$ and ${\Isys}'$, with initial states $s$ and $s'$ respectively, we define the distance between $\Isys$ and ${\Isys}'$ as $d({\Isys}, {\Isys}') = d_{\textit{max}}(s,s').$ The following properties are auxiliary to the theorem which states the continuity of the capacity.

\begin{lemma}
	\label{lem:decrdist}
	Consider two \IIHSs $\Isys$ and ${\Isys}'$ with transition functions $\vartheta$ and $\vartheta'$ respectively. Given $t\geq 2$ and two sequences $\alpha^t$ and $\beta^{t}$, assume that both ${\Isys}(\alpha^{t-1}, \beta^{t-1})$ and  ${\Isys}'(\alpha^{t-1}, \beta^{t-1})$ are defined. Assume also it is the case that $d_\textit{max}({\Isys}(\alpha^{t-1}, \beta^{t-1}), {\Isys}'(\alpha^{t-1}, \beta^{t-1}))<p(\beta_t\mid \alpha^t,\beta^{t-1})$, and $\vartheta({\Isys}(\alpha^{t}, \beta^{t-1}))\neq \emptyset$. Then:

	\begin{enumerate}
		\item $\vartheta'({\Isys}'(\alpha^{t}, \beta^{t-1}))\neq \emptyset$ holds as well,
		\item ${\Isys}(\alpha^{t}, \beta^{t})$ and ${\Isys}'(\alpha^{t}, \beta^{t})$ are both defined, $p(\beta_t\mid \alpha^t,\beta^{t-1})>0$, and 
		\[
d_\textit{max}({\Isys}(\alpha^t, \beta^{t}),{\Isys}'(\alpha^t, \beta^{t}))\leq \frac{d_\textit{max}({\Isys}(\alpha^{t-1}, \beta^{t-1}),{\Isys}'(\alpha^{t-1}, \beta^{t-1}))}{p(\beta_t\mid \alpha^t,\beta^{t-1}).}
		\]
\end{enumerate}
\end{lemma}

\begin{proof}\ \ \ \ \
	\begin{enumerate}
	
		\item Assume $\vartheta({\Isys}(\alpha^{t}, \beta^{t-1}))\neq \emptyset$ and, by contradiction, $\vartheta'({\Isys}'(\alpha^{t}, \beta^{t-1}))= \emptyset$. Since $d_{\emph{max}}$ is a fixed point of $\Phi$, we have $d_{\emph{max}} = \Phi(d_{\emph{max}})$, and therefore
			\begin{equation*}
				\begin{array}{lll}
					d_{\emph{max}}({\Isys}(\alpha^{t}, \beta^{t-1}), {\Isys}'(\alpha^{t}, \beta^{t-1}))
					&=& \Phi(d_{\emph{max}})({\Isys}(\alpha^{t}, \beta^{t-1}), {\Isys}'(\alpha^{t}, \beta^{t-1})) \\[2mm]
					&=&1\\[2mm]
					&\geq& p(\beta_t\mid \alpha^t,\beta^{t-1}),
				\end{array}
			\end{equation*}
 
 \noindent which contradicts the hypothesis.\\
  
	\item If $\vartheta({\Isys}(\alpha^{t}, \beta^{t-1}))\neq\emptyset$, then, by the first point of this lemma, we have that $\vartheta'({\Isys}'(\alpha^{t}, \beta^{t-1}))\neq\emptyset$ holds as well, and therefore both ${\Isys}(\alpha^{t}, \beta^{t})$ and ${\Isys}'(\alpha^{t}, \beta^{t})$ are defined. The hypothesis $d_\textit{max}({\Isys}(\alpha^{t-1}, \beta^{t-1}), {\Isys}'(\alpha^{t-1}, \beta^{t-1}))<p(\beta_t\mid \alpha^t,\beta^{t-1})$ ensures that $p(\beta_t\mid \alpha^t,\beta^{t-1}) \geq 0$. 

Let us now prove the bound on $d_\textit{max}({\Isys}(\alpha^{t}, \beta^{t}), {\Isys}'(\alpha^{t}, \beta^{t}))$. By definition of $\Phi$, we have \[\Phi(d_\textit{max})({\Isys}(\alpha^{t-1}, \beta^{t-1}),{\Isys}'(\alpha^{t-1}, \beta^{t-1}))\geq d_\textit{max}({\Isys}(\alpha^{t}, \beta^{t-1}),{\Isys}'(\alpha^{t}, \beta^{t-1})).\] Since $d_\textit{max} = \Phi(d_\textit{max})$, we have
\begin{equation}
	\label{eqn:bounddist}
	d_\textit{max}({\Isys}(\alpha^{t-1}, \beta^{t-1}),{\Isys}'(\alpha^{t-1}, \beta^{t-1}))\geq d_\textit{max}({\Isys}(\alpha^{t}, \beta^{t-1}),{\Isys}'(\alpha^{t}, \beta^{t-1})).
\end{equation}

By definition of $\Phi$ and of the Kantorovich metric, we have
\begin{equation*} 
	\begin{array}{lcl}
		\Phi(d_\textit{max})({\Isys}(\alpha^{t}, \beta^{t-1}),{\Isys}'(\alpha^{t}, \beta^{t-1}))
		&\geq & p(\beta_t\mid \alpha^t,\beta^{t-1}) \cdot \\
		&&   d_\textit{max}({\Isys}(\alpha^{t}, \beta^{t}),{\Isys}'(\alpha^{t}, \beta^{t})).
		\end{array}
\end{equation*}

Using again $d_\textit{max} = \Phi(d_\textit{max})$, we get
\begin{equation*}
	\begin{array}{lcl}
		d_\textit{max}({\Isys}(\alpha^{t}, \beta^{t-1}),{\Isys}'(\alpha^{t}, \beta^{t-1}))
		&\geq& p(\beta_t\mid \alpha^t,\beta^{t-1}) \cdot \\
		&    &d_\textit{max}({\Isys}(\alpha^{t}, \beta^{t}),{\Isys}'(\alpha^{t}, \beta^{t})),
	\end{array}
\end{equation*}

\noindent which, together with \eqref{eqn:bounddist}, allows us to conclude.
\end{enumerate}
\end{proof}

\begin{lemma}
	\label{lem:channeldiff}
Consider two \IIHSs $\Isys$ and ${\Isys}'$, and let $p(\cdot\mid \cdot,\cdot)$ and $p'(\cdot\mid\cdot,\cdot)$ be their distributions on the output nodes. Given $T>0$, and two sequences $\alpha^T$ and $\beta^T$, assume that  $p(\beta_t\mid \alpha^t,\beta^{t-1})>0$ for every $t< T$. Let $m=\min_{1\leq t < T} p(\beta_t\mid \alpha^t, \beta^{t-1})$ and let $\epsilon\in(0,m^{T-1})$. Assume $d({\Isys},{\Isys}')<\epsilon$. Then, for every $t\leq T$, we have
\begin{equation*}
p(\beta_t\mid \alpha^t, \beta^{t-1})-p'(\beta_t\mid \alpha^t, \beta^{t-1}) < \frac{\epsilon}{m^{T-1}}.
\end{equation*}
\end{lemma}

\begin{proof}
	Observe that, for every $t<T$, ${\Isys}(\alpha^{t}, \beta^{t})$ must be defined, and, by repeatedly applying Lemma~\ref{lem:decrdist}(1), we get that also ${\Isys}'(\alpha^{t}, \beta^{t})$ is  defined. By definition of $\Phi$, and of the Kantorovich metric, we have
\begin{equation*}
	p(\beta_t\mid \alpha^t, \beta^{t-1})-p'(\beta_t\mid \alpha^t, \beta^{t-1}) \leq \Phi(d_\textit{max})({\Isys}(\alpha^{t-1}, \beta^{t-1}),{\Isys}'(\alpha^{t-1}, \beta^{t-1})),
\end{equation*}

\noindent and since $d_\textit{max}$ is a fixed point of $\Phi$, we get
\begin{equation}
	\label{eqn:bound}
	p(\beta_t\mid \alpha^t, \beta^{t-1})-p'(\beta_t\mid \alpha^t, \beta^{t-1}) \leq d_\textit{max}({\Isys}(\alpha^{t-1}, \beta^{t-1}),{\Isys}'(\alpha^{t-1}, \beta^{t-1})).
\end{equation}

By applying Lemma~\ref{lem:decrdist}(2) $t-1$ times, from \eqref{eqn:bound} we get
\[
\begin{array}{lll}
p(\beta_t\mid \alpha^t, \beta^{t-1})-p'(\beta_t\mid \alpha^t, \beta^{t-1})
&\leq &
\frac{d_\textit{max}({\Isys}(\alpha^{0}, \beta^{0}),{\Isys}'(\alpha^{0}, \beta^{0}))}{m^{t-1}}\\[2mm]
&=&\frac{d({\Isys},{\Isys}')}{m^{t-1}}\\[2mm]
&\leq&\frac{d({\Isys},{\Isys}')}{m^{T-1}}\\[2mm]
&<&\frac{\epsilon}{m^{T-1}}
\end{array}
\]
\end{proof}

Note that previous lemma states a sort of continuity property of the matrices obtained from \IIHSs, but not uniform continuity, because of the dependence on one of the two \IIHSs. It is easy to see (from the proof of the Lemma) that uniform continuity does not hold. 

The main contribution of this section, stated in \revision{the} next theorem,  is the continuity of the capacity with respect to the pseudometric  on \IIHSs. For this theorem, we assume that the  \IIHSs are normalized. Furthermore, it is crucial that they are secret-nondeterministic (while the definition of the pseudometric holds in general).

\begin{theorem}
	\label{theo:cont}
	Consider two normalized \IIHSs $\Isys$ and ${\Isys}'$, and fix a  $T>0$. For every $\epsilon > 0$ there exists $\nu>0$ such that $ \mbox{if } \ d({\Isys},{\Isys}')<\nu \  \  \mbox{ then }\  \  |C_T({\Isys})-C_T({\Isys}') | < \epsilon.$
\end{theorem}

\begin{proof}
Consider two normalized \IIHSs $\Isys$ and ${\Isys}'$ and choose $T,\epsilon >0$. Let $\Dset_{T}$ be the set of all  input distributions in presence of feedback. Observe that
\[
\begin{array}{rcl} |C_T({\Isys})-C_T({\Isys}') | &=& | \displaystyle\max_{\Dset_{T}}\frac{1}{T} I(A^T\rightarrow B^T) -  \displaystyle\max_{\Dset_{T}}\frac{1}{T}I(A'^T\rightarrow B'^T) | \\[4mm]
                                                &\leq& \frac{1}{T}\displaystyle\max_{\Dset_{T}} | I(A^T\rightarrow B^T) -I(A'^T\rightarrow B'^T) |
\end{array}
 \]
 
Since the directed information $I(A^T\rightarrow B^T)$ is defined by means of arithmetic operations and logarithms on the joint probabilities $p(\alpha^t,\beta^t)$ and on the conditional probabilities $p(\alpha^t,\beta^t)$,  $p(\alpha^t,\beta^{t-1})$, which in turn can be obtained by means of arithmetic operations from the probabilities $p(\beta_t \mid \alpha^t,\beta^{t-1})$ and ${p}_F(\varphi^t)$, we have that $I(A^T\rightarrow B^T)$ is a continuous \revision{function}  of the distributions   $p(\beta_t \mid \alpha^t,\beta^{t-1})$ and ${p}_F(\varphi^t)$, for every $t\leq T$. Let $p(\beta_t \mid \alpha^t,\beta^{t-1})$, $p'(\beta_t \mid \alpha^t,\beta^{t-1})$ be the distributions on the output nodes of $\Isys$ and ${\Isys}'$, modified in the following way: starting from level $T$, whenever $p(\beta_t \mid \alpha^t,\beta^{t-1}) = 0$, then we redefine the distributions at all the output nodes of the subtree rooted in  ${\Isys}(\alpha^t,\beta^{t})$ so that they coincide with the distribution of the corresponding nodes of  in ${\Isys}'$, and analogously for $p'(\beta_t \mid \alpha^t,\beta^{t-1})$. Note that this transformation does not change the directed information, because the subtree rooted in ${\Isys}(\alpha^t,\beta^{t})$  does not contribute to it, due to the fact that the probability of reaching any of its nodes is $0$. The continuity of   $I(A^T\rightarrow B^T)$ implies that there exists $\epsilon' >0$ such that, if $|p(\beta_t \mid \alpha^t,\beta^{t-1})-p'(\beta_t \mid \alpha^t,\beta^{t-1})| < \epsilon'$ for all $t\leq T$ and all sequences $\alpha^t$, $\beta^t$, then, for any ${p}_F(\varphi^t)$, we have $| I(A^T\rightarrow B^T) -I(A'^T\rightarrow B'^T) | < \epsilon$. The result then follows from Lemma~\ref{lem:channeldiff}, by choosing
\begin{align*}
	\nu & = \epsilon' \cdot \min \left( \min_{{\begin{array}{c}1\leq t < T\\ p(\beta_t \mid \alpha^t,\beta^{t-1}) > 0\end{array}}}p(\beta_t \mid \alpha^t,\beta^{t-1}), \right. \\
	& \quad \quad \quad \quad \quad \quad \quad \quad \quad \quad \quad \quad \quad \quad  \left. \min_{{\begin{array}{c}1\leq t < T\\ p'(\beta_t \mid \alpha^t,\beta^{t-1}) > 0\end{array}}}p'(\beta_t \mid \alpha^t,\beta^{t-1}) \right).
\end{align*}
	
\end{proof}

We conclude this section with an example showing that the continuity result for the capacity does not hold if the construction of the channel is done starting from a system in which the secrets are endowed with a probability distribution. This is also the reason  why we could not simply adopt the proof technique of the continuity result in \cite{Desharnais:02:LICS} and we had to come up \revision{with different reasoning}. 

\begin{example}\label{exa:non-cont}
Consider the two following programs, where $a_1, a_2$ are secrets, $b_1$, $b_2$ are observable, $\parallel$ is the parallel operator, and $+_p$ is a binary probabilistic choice that assigns probability $p$ to the left branch, and probability $1-p$ to the right one.
\begin{description}
\item[s)]  $(\textit{send}(a_1) +_p\  \textit{send}(a_2)) \parallel \textit{receive}(x). \textit{output}(b_2)$
\item[t)]  $(\textit{send}(a_1) +_q \ \textit{send}(a_2)) \parallel \textit{receive}(x). \textit{if} \ x = a_1 \ \textit{then} \ \textit{output}(b_1) \ \textit{else} \ \textit{output}(b_2)$.
\end{description}

Table~\ref{tab:non-cont} shows the fully probabilistic \IIHSs corresponding to these programs, and their associated channels, which in this case (since the secret actions are all at the top-level) are classical channels, i.e. memoryless and without feedback. As usual for \revision{classical} channels, they do not depend on $p$ and $q$. It is easy to see that the capacity of the first channel is $0$ and the capacity of the second one is $1$. Hence their difference is $1$, \revision{independently of} $p$ and $q$.

Let now $p=0$ and $q=\epsilon$. It is easy to see that the distance between $s$ and $t$ is $\epsilon$. Therefore (when the automata have probabilities on the secrets), the capacity is not a continuous function of the distance.
\end{example}

	\begin{figure}[!htb]
		\centering
    \includegraphics[width=0.6\linewidth]{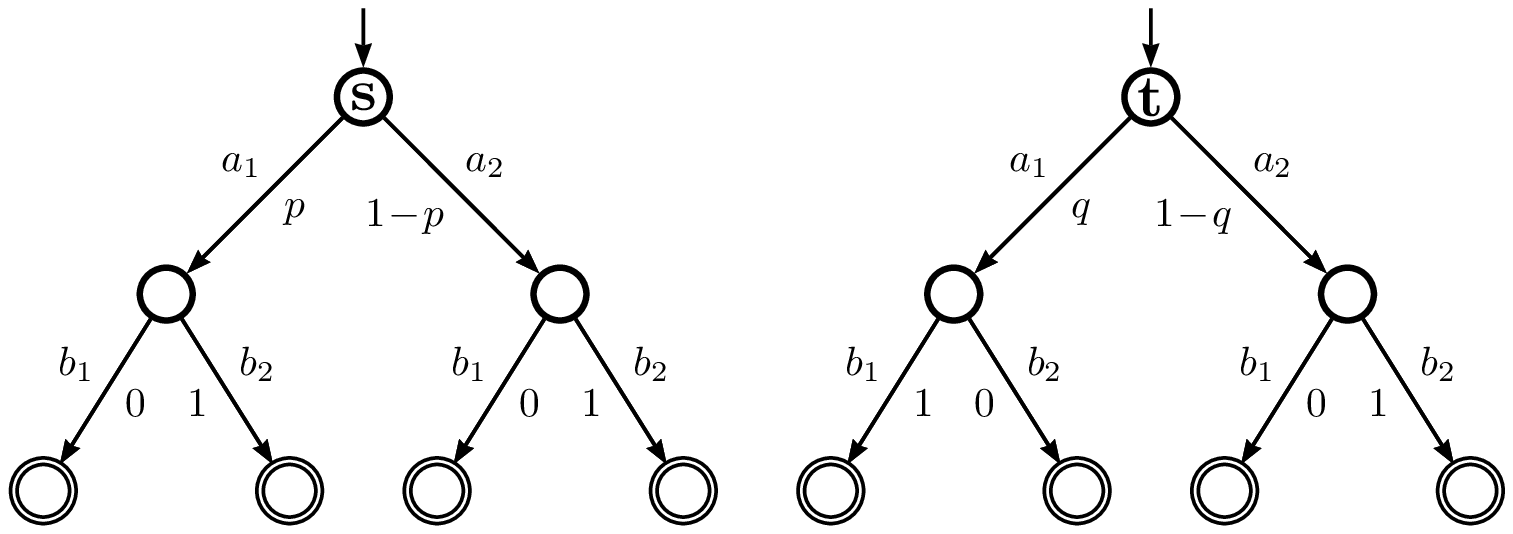}	
		\label{fig:continuityCEx}
	\end{figure}
     
	\begin{table}[ht]
		\centering	
		\subbottom[(Channel for $s$]{
			\centering
			\begin{tabular}{|c|c|c|}
			\hline
               $\ \textbf{s}\ $ & $\ b_1\ $ & $\ b_2\ $ \\ \hline
							 $\ a_1\ $ & $\ 0\ $ & $\ 1\ $ \\ \hline
							 $\ a_2\ $ & $\ 0\ $ & $\ 1\ $ \\ \hline						
			\end{tabular}
		}
		\hspace{0.5cm}
   \subbottom[Channel for $t$]{
  		\centering
			\begin{tabular}{|c|c|c|}
				\hline
               $\ \textbf{t}\ $ & $\ b_1\ $ & $\ b_2\ $ \\ \hline
							 $\ a_1\ $ & $\ 1\ $ & $\ 0\ $ \\ \hline
							 $\ a_2\ $ & $\ 0\ $ & $\ 1\ $ \\ \hline
			\end{tabular}
   	}	
   	\caption{The \IIHSs of Example~\ref{exa:non-cont} and their corresponding channels }
   	\label{tab:non-cont}
	\end{table}

\section{Related work}
\label{section:related-work-interactive-systems}

Gray investigated a concept similar to directed information in~\cite{Gray:91:SSP}. In contrast to our model, which is based on an eavesdropper scenario, he considered leakage in a sender-receiver model. More precisely, he considered a system based on Millen's synchronous state machine \cite{Millen:90:CSFW}, and connected to \qm{low} and \qm{high} environments via communication channels. His purpose was to measure the flow of information from the high environment to the low one, assuming that the only way for the low environment to learn about the high one (and vice versa) is through the system. To this end, he defined a notion of \qm{quasi-directed information} by extending Gallager's formula for discrete finite state channels~\cite{Gallager:68:BOOK}. He also conjectured a correspondence between the quasi-directed information and the transmission rate of the channel. His formulation of quasi-directed information, however, is not completely the same as directed information, and as a result the conjecture does not hold. 

The continuity of the channel capacity was also proved in \cite{Desharnais:02:LICS} for simple channels, but the proof does not adapt to the case of channels with 
memory and feedback and we had to devise a different technique.

\section{Chapter summary and discussion}
\label{section:conclusion-interactive-systems}

In this chapter we have investigated the problem of information leakage in interactive systems, and proved that these systems can be modeled as channels with memory and feedback. We have also proved that the channel capacity is a continuous function of a pseudometric based on the Kantorovich metric.

We have considered various kinds of automata corresponding to different combinations of nondeterministic and probabilistic choice, as summarized in Table \ref{tab:comp-model-iihs-a}. Note that in this the third row corresponds to the limit case in which the reactor is a Dirac measure, i.e. the probability is all concentrated on exactly one  $\Fseq{}{T}\in \Fset$. It is easy to see that in this case $I(A^{T} \rightarrow B^{T}) = 0$ (all the entropies that constitute $I(A^{T} \rightarrow B^{T})$ are $0$), although $I(B^{T} \rightarrow A^{T}) \neq 0$. Therefore there is no leakage. In the classic case this corresponds to the situation in which the input distribution is a Dirac measure. 

	\begin{table}[!htb]
		\small
		\centering
		\subbottom[The various models considered in this chapter]{
			\centering
				\begin{tabular}{|l|l|l|}
					\hline
					\textbf{\IIHSs as automata} & \textbf{\IIHSs as channels} & \textbf{Notion of leakage} \\ \hline \hline
					Normalized \IIHSs with & Sequence of & \\
					nondeterministic secrets & stochastic kernels & Leakage as capacity \\
					and probabilistic observables & $\{ \Cprob{p}{\Bseq{t}{}}{\Aseq{}{t},\Bseq{}{t-1}} \}_{t=1}^{T}$ & \\ \hline
					& Sequence of & \\
					Fully probabilistic & stochastic kernels & Leakage as directed \\
					normalized \IIHSs & $\{ \Cprob{p}{\Bseq{t}{}}{\Aseq{}{t},\Bseq{}{t-1}} \}_{t=1}^{T}$ & information \\
					& + reactor & $I(A^{T} \rightarrow B^{T})$ \\
					& $\{ \Cprob{p}{\Fseq{t}{}}{\Fseq{}{t-1}} \}_{t=1}^{T}$ & \\ \hline
					& Sequence of & \\
					Normalized \IIHSs with a & stochastic kernels	& \\
					deterministic scheduler & $\{ \Cprob{p}{\Bseq{t}{}}{\Aseq{}{t},\Bseq{}{t-1}} \}_{t=1}^{T}$ & No leakage \\
					solving the nondeterminism & + reaction function & \\			
					& sequence $\Fseq{}{T}$ & \\ \hline
				\end{tabular}
			\label{tab:comp-model-iihs-a}			
		}
		\subbottom[Classical channels vs. channels with memory and feedback]{
			\centering
			\begin{tabular}{|l||l|}
				\hline
				\textbf{Classical channels}   & \textbf{Channels with memory and feedback} \\ \hline \hline
				The system is modeled in     & The system is modeled in several \\
				independent uses of the channel, & consecutive uses of the channel. \\
				often a unique use.       & \\
			 \hline
			 & \\ [-2ex]
				The channel is defined on                & \\
				$\Aset^{T} \to \Bset^{T}$, i.e. its input is      & The channel is defined on $\Fset \to \Bset$, i.e. \\
				a single string $\Aseq{}{T}=\Aseq{1}{}\ldots\Aseq{T}{}$ & its input is a reaction function $\Fseq{t}{}$ \\
				of secret symbols and its output            & and its output is an observable $\Bseq{t}{}$. \\
				is a single string $\Bseq{}{T}=\Bseq{1}{}\ldots\Bseq{T}{}$ & \\
				of observable symbols.                 & \\
			 \hline
				The channel is memoryless and     & The channel has memory. Despite the\\
				in general it is implicitly assumed  & fact that the channel defined on $\Fset \to \Bset$ \\
				the absence of feedback.        & does not have feedback, the internal \\
				                & stochastic kernels do.\\
			 \hline
				The capacity is calculated using   & The capacity is calculated using mutual \\
				mutual information $I(A^{T};B^{T})$. & directed information $I(A^{T} \rightarrow B^{T})$. \\
			 \hline
			\end{tabular}
			\label{tab:comp-model-iihs-b}			
		}
		\caption{Summary of results}
		\label{tab:comp-model-iihs}
	\end{table}	

Table~\ref{tab:comp-model-iihs-b} summarizes the comparison between the channels with memory and feedback investigated in this \revision{chapter}, and the  classic channels.

Throughout this chapter we have assumed that the dependence of the secret choices on the observables is part of the external knowledge and, therefore, not considered leakage. The reader may wonder what would happen if this assumption were dropped. We argue that in this case $I(B^{T} \rightarrow A^{T})$ \emph{could be considered as part of the leakage}. In the cases \texttt{a} and \texttt{b} of the cocaine auction example in Section \ref{section:full-example}, for instance, one may want to consider the information that we can deduce about the secrets (the identities of the bidder) from the observables (the increments of the seller) as a leak due to the protocol. 

In some other cases the flow of information from the observables to the secrets may even be considered as a consequence of the active attacks of an adversary, which uses the observables to modify the probability of the secrets. In this case $I(B^{T} \rightarrow A^{T})$ could represent a measure of the effectiveness of the adversary.

As future work, we would like to provide algorithms to compute the leakage and maximum leakage of interactive systems. These are rather challenging problems given the exponential growth of reaction functions (needed to compute the leakage) and the quantification over infinitely many reactors (given by the definition of maximum leakage in terms of capacity). One possible solution is to study the relation between deterministic schedulers and sequence of reaction functions. In particular, we believe that for each sequence of reaction functions and distribution over it there exists a probabilistic scheduler for the automata representation of the secret-nondeterministic IIHS. In this way, the problem of computing the leakage and maximum leakage would reduce to a standard probabilistic model checking problem (where the challenge is to compute probabilities ranging over infinitely many schedulers).

In addition, we plan to investigate measures of leakage for interactive systems other than mutual information and capacity.

We intend to study the applicability of our framework to the area of game theory. In particular, the interactive nature of games such as \emph{Prisoner Dilemma}~\cite{Poundstone:92:Doubleday} and \emph{Stag and Hunt}~\cite{Skyrms:03:BOOK} (in their iterative versions) can be modeled as channels with memory and feedback following the techniques proposed in this work. Furthermore, (probabilistic) strategies can be encoded as reaction functions. In this way, optimal strategies are attained by reaction functions maximizing the leakage of the channel.

\chapter{Differential privacy: the trade-off between leakage and utility}
\label{chapter:differential-privacy}
\mscite{If you have nothing to hide, then you don't have a life.}{cited by Daniel J. Solove}

In this chapter we consider the differential privacy approach to the problem of statistical disclosure control. In general a statistical database contains data of a group of individuals, and users can pose queries to obtain statistical information about the sample in the dataset. To preserve the privacy of the the participants in the database, it is desirable to restrict the amount of information that the system leaks about their individual values. One way of dealing with the problem is by using randomization mechanisms: to avoid leakage, the real answer is modified with some carefully added noise before being reported to the users. A very popular and studied way of doing so is based on the concept of differential privacy.

In our work we consider the relation between differential privacy and quantitative information flow. We address the problem of characterizing the protection that differential privacy provides to individuals with respect to information leakage, and the problem of the utility, i.e. the measure of how close the reported answer is to the true answer.

\paragraph{Contribution} The main contributions of this chapter can be summarized as follows.

\begin{itemize}
		\item We propose an information-theoretic framework to reason about both information leakage and utility.
		
		\item We explore the graph-theoretic foundations of the adjacency relation on databases\footnote{The adjacency relation on databases will be defined precisely in Section~\ref{section:model}.}, and we point out two types of symmetries which allow us to establish a strict link between differential privacy and information leakage.

		\item We prove that $\epsilon$-differential privacy implies a tight bound on the min-entropy leakage.

		\item We prove that $\epsilon$-differential privacy implies a bound on the utility, measured in terms of binary gain functions. We prove that, under certain conditions, the bound is tight.

		\item We identify a method that, under certain conditions, constructs \revision{randomization mechanisms that maximize} utility while providing $\epsilon$-differential privacy. 
\end{itemize}

\textbf{Plan of the Chapter} This chapter is organized as follows. In Section~\ref{section:differential-privacy} we formalize the notion of differential privacy and present an alternative interpretation for it in the special case where the adjacency relation on databases is complete (i.e. every two distinct databases are adjacent). In Section~\ref{section:model} we introduce our model to reason about leakage and utility for randomized functions in the case where the query and the randomization mechanism can be split into two distinct channels. In Section~\ref{section:graph-symmetries} we review some concepts from graph theory and present two special classes of graphs having symmetries that we will explore to make the connection between differential privacy and quantitative information flow. We also show that the graph structure on databases, induced by the adjacency relation and the query, presents these symmetries. In Section~\ref{section:graph-bounds} we use the results of the previous section to prove a bound on the a posteriori min-entropy of the channel matrix. Then we apply this bound to derive our results for leakage in Section~\ref{section:leakage} and for utility in Section~\ref{section:utility}. Finally, in Section~\ref{section:related-work-diff-priv} we review some of the related work in the literature, and in Section~\ref{section:conclusion-diff-priv} we make our final remarks and conclude this chapter.

\section{Differential privacy}
\label{section:differential-privacy}

Databases are commonly used for obtaining statistical information about their participants. Simple examples of statistical queries are, for instance, the predominant disease in a certain population, or the average salary of a group of people. The fact that the answer is publicly available may, however, constitute a threat for the privacy of the individuals. 

In order to illustrate the problem, consider a database that stores the values of the salaries of a set of individuals, and assume that a user can pose the query \qm{what is the average salary of the participants in the database?}. In principle we would like to consider the \emph{global information} relative to the database as \emph{public}, and the \emph{individual information} about a participant as \emph{private}. In this example, we would like to obtain the average salary without being able to infer the salary of any specific participant. Unfortunately this is not always possible. In particular, if the number of participants in the database is known, and an individual is removed from (or included in) the database, it is possible to infer his salary by querying again the database and calculating the influence of the removal (or inclusion) on the reported answer to the query.

Another kind of private information we may want to protect is whether a specific individual is \emph{participating or not} in a database. If we know that a particular individual earns, say, $5.000$\euro\ a month, and all the other individuals earn less than $4.000$\euro\ a month, then learning that the average salary is greater that $4.000$\euro\ will reveal immediately the presence of our individual of interest in the database.

A common approach to this problem is to introduce some output perturbation mechanism based on randomization: instead of the exact answer, the querying mechanism reports a \qm{noisy} answer. Namely, a randomized function is used to produce answers according to some probability distribution that depends on the database. The goal is to report this randomized answer, \revision{which} ideally should be \qm{close enough} to the real one, yet should make it harder for the user to guess the values of individual participants. For certain distributions, however, it may still be possible to guess the value of an individual with a high probability of success. The notion of \emph{differential privacy}, due to Dwork \cite{Dwork:06:ICALP,Dwork:09:STOC,Dwork:10:SODA,Dwork:11:CACM}, is a proposal to control the risk of violating privacy for both kinds of threats described above (value and participation). The idea is to say that a randomized function $\mathcal{K}$ satisfies $\epsilon$-differential privacy (for some $\epsilon>0$) if the ratio between the probabilities that two adjacent databases give \revision{a certain} answer is bound by $e^\epsilon$, where by \qm{adjacent} we mean that the databases differ in only one individual (either for the value of an individual or for the presence/absence of an individual). The notion of differential privacy was developed to be independent of the \emph{side (or auxiliary) information} the user can have about the database, and how it can affect his knowledge about the database before posing the query. This information can come from external sources (e.g. newspapers, common knowledge, etc), but does not affect the  guarantees assured by differential privacy.

In this chapter we explore the similarities between differential privacy and quantitative information flow. We base our approach on the following observations: at the motivational level, the concern about privacy is akin the concern about information leakage. At the conceptual level, the randomized function $\mathcal{K}$ can be seen as an information-theoretic channel, and the limit case of $\epsilon = 0$, for which the privacy protection is total, corresponds to a $0$-capacity channel, which does not allow any leakage. More specifically, we investigate the notion of differential privacy and its implications in the light of the min-entropy framework for information flow discussed in Chapter~\ref{chapter:probabilistic-info-flow}. 

\subsection{Formal definition}
\label{section:formal-definition}
	
Let $\calx$ be the set of all possible databases. Two databases $x,x' \in \calx$ are \emph{adjacent} (or \emph{neighbors}), written $x \sim x'$, if they differ in the value of exactly one individual. Note that the structure $(\calx, \sim)$ forms an undirected graph.

Intuitively, differential privacy is based on the idea that a randomized query function provides sufficient protection if the ratio between the probabilities of two adjacent databases to give a certain answer is bound by $e^\epsilon$, for some $\epsilon > 0$. Formally:

\begin{definition}[\cite{Dwork:11:CACM}]
	\label{def:diff-privacy-1}
	A randomized function $\mathcal{K}$ from $\calx$ to $\calz$ satisfies {$\epsilon$-differential privacy} if for all pairs $x,x'\in \calx$, with $x\sim x'$, and all $S \subseteq \calz$, we have:
	\begin{equation*}
		\mathit{Pr}[\mathcal{K}(x) \in S] \leq e^{\epsilon} \times \mathit{Pr}[\mathcal{K}(x') \in S]		
	\end{equation*}	
	
\end{definition}

In this thesis we consider $\calz$ to be finite, \revision{therefore each of its probability distributions is finite} and we can rewrite the property of $\epsilon$-differential privacy more simply. Using the notation of conditional probabilities, and considering both quotients, we can say that $\epsilon$-differential-privacy holds in the discrete case if, for all $x,x'\in \calx$ with $x\sim x'$, and all $z\in \calz$:
\begin{equation}
	\label{eq:dp-classical}
	\frac{1}{e^\epsilon} \leq \frac{\mathit{Pr}[Z=z|X=x] }{\mathit{Pr}[Z=z|X=x'] } \leq e^\epsilon 
\end{equation}

\noindent where $X$ and $Z$ represent the random variables associated to $\calx$ and $\calz$, respectively. 

Intuitively, \eqref{eq:dp-classical} implies that, if a value of one single individual changes in a dataset (either by inclusion, removal or modification), the probability of the querying mechanism to report a specific answer will not \qm{vary much}. In other words, the influence of a single individual in a database is \qm{negligible} with respect to the whole set of individuals. Of course the notion of what is meant by \qm{much} and \qm{negligible} depends on the value of $\epsilon$.


\subsection{Alternative interpretation in the case of cliques}
\label{section:dp-interpretation-cliques}

A special interpretation of differential privacy is possible in the case where every two distinct databases in $\calx$ are neighbors. More precisely, if $(\calx, \sim)$ is a clique (i.e. a complete graph), it is possible to ensure that he ratio between any a priori knowledge $\mathit{Pr}[X=x]$ of the user (before the query is posed) and his a posteriori knowledge $\mathit{Pr}[X=x|Z=z]$ (after the answer to the query is reported) is bound by $e^\epsilon$. Formally, if for every $x,x'\in \calx$ with $x \neq x'$ we have $x\sim x'$ then:
\begin{align}
	\label{eq:dp-alternative}
	\frac{1}{e^\epsilon} \leq \frac{\mathit{Pr}[X=x|Z=z]}{\mathit{Pr}[X=x]} \leq e^\epsilon & \quad \quad \text{for all priors $\mathit{Pr}[X=x]$, } \\[-4mm]
	\nonumber   																																						& \quad \quad \text{all $x \in \calx$, and all $z \in \calz$}
\end{align}

\noindent where $X$ and $Z$ represent the random variables associated to $\calx$ and $\calz$, respectively. 

Intuitively, \eqref{eq:dp-alternative} states that the observation of the reported answer should not \qm{change much} the user's knowledge about the database. The next proposition shows that in the special case of every pair of distinct databases are neighbors, the above formulation of differential privacy is equivalent to the classic one.

\begin{proposition}
	If for all $x,x' \in \calx$ with $x \neq x'$ we have $x \sim x'$, then \eqref{eq:dp-classical} and \eqref{eq:dp-alternative} are equivalent.
\end{proposition}

\begin{proof}
		Let us represent by $X$ and $Z$ the random variables associated to $\calx$ and $\calz$, respectively. For better readability, we will denote $\mathit{Pr}[X=x]$, $\mathit{Pr}[Z=z]$, $\mathit{Pr}[Z=z|X=x]$ and $\mathit{Pr}[X=x|Z=z]$ by $\mathit{Pr}(x)$, $\mathit{Pr}(z)$, $\mathit{Pr}(x|z)$ and $\mathit{Pr}(z|x)$, respectively.

	\begin{itemize}
		\item \eqref{eq:dp-classical} $\implies$ \eqref{eq:dp-alternative} \\	\\
			\begin{align*}
				\mathit{Pr}(x|z) & =    \frac{\mathit{Pr}(z|x) \mathit{Pr}(x)}{\mathit{Pr} (z)} & \text{(by the Bayes law)} \\[2mm]
												 & =    \frac{\mathit{Pr}(z|x) \mathit{Pr}(x)}{\sum_{x'\in \calx} \left( \mathit{Pr}(x') \mathit{Pr}(z|x') \right) } &  \\[2mm]
												 & \geq \frac{\mathit{Pr}(z|x) \mathit{Pr}(x)}{\sum_{x'\in \calx} \left( \mathit{Pr}(x') \cdot e^\epsilon \mathit{Pr}(z|x) \right) } & \text{by \eqref{eq:dp-classical}} \\[2mm]
												 & =    \frac{\mathit{Pr}(z|x) \mathit{Pr} (x)}{e^\epsilon \mathit{Pr}(z|x)} &  \\[2mm]
												 & =    \frac{\mathit{Pr}(x)}{e^\epsilon} &  \\
			\end{align*}
			
			\noindent from which it follows that $\frac{\mathit{Pr}(x)}{\mathit{Pr}(x|z)} \leq e^\epsilon$. The case of $\frac{1}{e^\epsilon} \leq \frac{\mathit{Pr}(x)}{\mathit{Pr}(x|z)}$ is a analogous: just take the symmetrical step when applying \eqref{eq:dp-classical} in the derivation above.
			
		\item \eqref{eq:dp-alternative} $\implies$ \eqref{eq:dp-classical} \\ \\
		For every prior $\mathit{Pr}(x)$ we have
		\begin{align*}
			\frac{\mathit{Pr}(x|z)}{\mathit{Pr}(x)} & = \frac{\mathit{Pr}(z|x)}{p(z)} & \text{(by the Bayes law)} \\[2mm]
																							& = \frac{\mathit{Pr}(z|x)}{\sum_{x''} \left( \mathit{Pr}(x'') \mathit{Pr}(z|x'') \right) }
		\end{align*}
		
		In particular, the above is valid for every prior of the form $\mathit{Pr}(x) = \delta_{x'}(x)$, where $x' \in \calx$. Therefore, for all $x' \in \calx$ 
		\begin{align*}
		\frac{\mathit{Pr}(x|z)}{\mathit{Pr}(x)} & = \frac{\mathit{Pr}(z|x)}{\sum_{x''} \left( \delta_{x'}(x'') \mathit{Pr}(z|x'') \right) } &  \\[2mm]
																						& = \frac{\mathit{Pr}(z|x)}{\mathit{Pr}(z|x')}
		\end{align*}
		
		Since by \eqref{eq:dp-alternative} we have $\frac{1}{e^\epsilon} \leq \frac{\mathit{Pr}(z|x)}{\mathit{Pr}(x)} \leq e^{\epsilon}$ for every prior $\mathit{Pr}(x)$, it follows from the derivation above that also $\frac{1}{e^\epsilon} \leq \frac{\mathit{Pr}(z|x)}{\mathit{Pr}(z|x')} \leq e^{\epsilon}$ for all $x' \in \calx$.
	
	\end{itemize}
	
\end{proof}

\section{A model of utility and privacy for statistical databases}
\label{section:model}

In this section we present a model of statistical queries on databases, where noise is carefully added to protect the privacy of the participants in the sample, and the reported answer to a query does not need to be the real one. In this model, the notion of information leakage is to measure the amount \revision{of} information that an adversary can learn about the database by posing queries and then analyzing the reported answers. Note that in principle the adversary can be a user of the database, and therefore the privacy guarantees should not depend on distinctions of who is posing the queries. Our model will also allow us to quantify the utility of the query, i.e. how much information about the real answer can be obtained from the reported one. In our work we focus on the case in which all the values of interest are discrete.

We fix a finite set \revision{$\ind=\lbrace 0,1,\ldots, u-1\rbrace$} of $u$ individuals participating in the database. In addition, we fix a finite set \revision{$\val = \lbrace  {\textsl{v}}_{0}, {\textsl{v}}_{1}, \ldots, {\textsl{v}}_{v-1} \rbrace$}, representing the set of ($v$ different) possible values for the \emph{sensitive attribute} of each individual (e.g. disease-name in a medical database). In the more general case where there are several sensitive attributes in the database (e.g. salary and security number in a census sample), we can think of the elements of $\val$ as tuples. The absence of an individual in the database, if allowed, can be modeled with one special value in $\val$ (see the discussion in Section~\ref{section:note-values}). A database $D = d_0 \ldots d_{u-1}$ is a $u$-tuple where each $d_i \in \val$ is the value of the corresponding individual. The set of all databases is $\calx = \val^{u}$. Two databases $x,x'$ are \emph{adjacent}, written $x \sim x'$, if and only if they differ in the value of exactly one individual. \revision{As we already pointed out, the} structure $(\calx, \sim)$ forms an undirected graph, and we \revision{call} $\sim$ its \emph{adjacency relation}.

Let $\calk$ be a randomized function from $\calx$ to $\calz$, where $\calz=Range(\calk)$ (see Figure~\ref{fig:mechanism-k}). This function can be modeled by a channel $(\calx, \revision{\calz}, p_{Z|X}(\cdot|\cdot))$, where $\calx$ and $\calz$ are the input and output alphabets, respectively, and $p_{Z|X}(\cdot|\cdot)$ is the channel matrix. The random variables modeling the input and output of the channel are denoted by $X$ and $Z$, respectively. The definition of differential privacy can be directly expressed as a property of the channel: it satisfies $\epsilon$-differential privacy if
\begin{equation*}
	p(z|x)\le e^\epsilon p(z|x') \quad \text{for all $x,x'\in \calx$ with $x\sim x'$, and all $z \in \calz$}
\end{equation*}

\begin{figure}[!htb]
	\centering
	\includegraphics[width=0.35\textwidth]{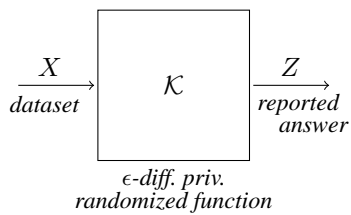}%
	\caption{Randomized function $\mathcal{K}$}%
	\label{fig:mechanism-k}%
\end{figure}

Intuitively, the correlation between $X$ and $Z$ measures how much information about the complete database the attacker can obtain by observing the reported answer. We will refer to this correlation as the \emph{leakage} of the channel, denoted by $\call(X,Z)$. In Section~\ref{section:leakage} we will discuss how this leakage can be quantified using notions from information theory, and we will study the behavior of the leakage for differentially private queries.

In our model the true answer to the query $f$ is modeled by the random variable $Y$ ranging over $\caly = Range(f)$. The correlation between $Y$ and $Z$ measures how much we can learn about the real answer from the reported one. We will refer to this correlation as the \emph{utility} of the channel, denoted by $\calu(Y,Z)$. In Section~\ref{section:utility} we will discuss in detail how the utility can be quantified, and we will investigate how to construct a randomization mechanism, i.e. a way of adding noise to the query outputs, so that utility is maximized while preserving differential privacy.

In practice, the randomization mechanism is often \emph{oblivious}, meaning that the reported answer $Z$ only depends on the real answer $Y$ and not on the database $X$. In this case, the randomized function $\mathcal{K}$, seen as \revision{a} channel, can be decomposed into two parts: a channel modeling the query $f$, and a channel modeling the oblivious randomization mechanism $\calh$. \revision{These two channels are said to be \emph{in cascade}, as the output of the first one is the input for the second one.} The definition of utility can be then simplified as it only depends on properties of the sub-channel \revision{corresponding} to $\mathcal{H}$. The leakage relating $X$ and $Y$ and the utility relating $Y$ and $Z$ for a decomposed randomized function are shown in Figure~\ref{fig:utility-privacy}. 

\begin{figure}[!bt]
	\centering
	\includegraphics[width=0.85\textwidth]{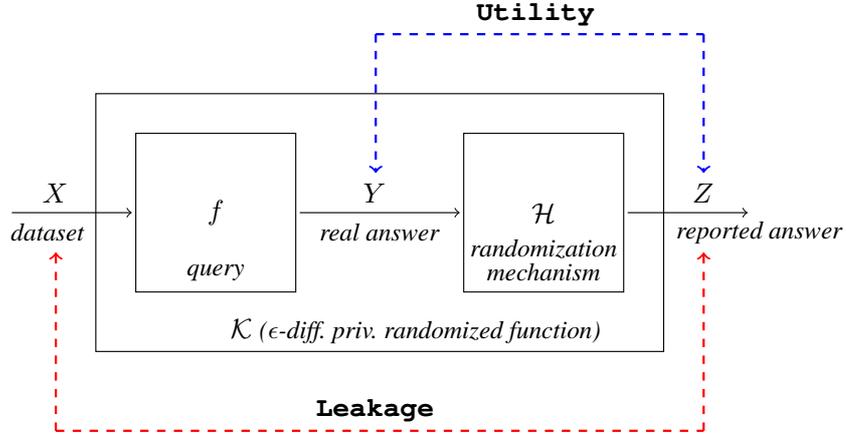}%
	\caption{Leakage and utility for oblivious mechanisms}%
	\label{fig:utility-privacy}%
\end{figure}

We capture the notion of the attacker's side information as the prior distribution on $X$, which is standard in information flow and also in papers on differential privacy \cite{Ghosh:09:STC,Kasiviswanathan:08:CORR}.

\subsection{Leakage about an individual}

As already discussed, $\call(X,Z)$ can be used to quantify the information that the attacker can learn about the whole database. Protecting the entire database at once, however, is not the main goal of differential privacy. In fact, some information will necessarily be revealed, otherwise the query would not be useful. Instead, differential privacy aims at protecting the value of \emph{any single individual}, even in the worst case where the values of all other individuals are known. To quantify this information leakage we can define smaller channels, where only the information of a specific individual varies. Let $x^- \in \val^{u-1}$ be a $(u-1)$-tuple with the values of all individuals but one (the individual whose degree of protection we want to quantify). We create a channel $\mathcal{K}_{x^-}$ whose input alphabet is the set of all databases in which the $u-1$ other individuals have the same values as in $x^-$. Note that, since $x^-$ is fixed, to define the input of the channel it is enough to specify the value of the individual of interest. In this way the input for the channel can be seen as a random variable $V$ ranging over the set $\val$. Intuitively, the information leakage of this channel measures how much information about one particular individual the attacker can learn if the values of all others are known to be $x^-$. This leakage will be studied in Section~\ref{section:individual-leakage}. 

\subsection{A note on the choice of values} 
\label{section:note-values}

The choice of the set $\val$ depends on the assumptions about the attacker's knowledge. In particular, if the attacker does not know which individuals participate in the database, a distinguished value in $\val$ could be interpreted as absence (e.g. the value $0$ or the special value $null$). As discussed in \cite{Dwork:11:CACM}, a database $x'$ adjacent to $x$ can be though of either \revision{as} being a superset (or subset) of $x$ with one extra (or missing) row, or \revision{as} being exactly the same database as $x$ in all rows \revision{except for} one which has a different (non-\emph{null}) value. Our definition of $\sim$ with the possibility of $null$ values covers all these cases.

At this point an important observation should be made about the choice of $\val$. Most often we are interested in protecting the \emph{actual value} of an individual, not only his participation in the database. In this case, the definition of differential privacy (as well as the channels we are constructing) should include databases with all possible values for each individual, not just the \qm{real} ones. In other words, to prevent the attacker from finding out the individual's value, the probability $p(z|x)$, where $x$ contains the individual's true value, should be close to $p(z|x')$ where $x'$ contains a hypothetical value for this individual. This might seem unnecessary at first sight, since differential privacy is often \revision{thought of} as protecting \revision{the} participation of an individual in a database. Hiding the participation of an individual, however, does not imply hiding his value. Consider the following example: we aim at learning the average salary of employees in a small company, and it happens that all of them have exactly the same salary $s$. We allow anyone to participate or not, while offering $\epsilon$-differential privacy. If we only consider $s$ as the value in all possible databases, then the query is always constant, so answering it any number of times without any noise should satisfy differential privacy for any \revision{$\epsilon \geq 0$}. Since all reported answers are $s$, the attacker can deduce that the salary of all employees, including those not participating in the query, is $s$. Indeed, the attacker cannot find out who participated, despite the value of all individuals is revealed.

In other cases, we are only interested in hiding the \revision{identity of the participants} (e.g. in a database with information about anonymous donations). Thus, $\val$ should be properly selected according to the application. If \revision{who has participated} is known and we only wish to hide the values, then $\val$ should contain all possible values, e.g. all possible salaries in the example above. If the values are known and participation is to be hidden, then $\val$ can contain just the values $0$ and $1$ denoting absence and presence respectively. Finally, if both the value and the \revision{the identities of the participants} are to be protected, then $\val$ should contain all values plus $\mathit{null}$.

\subsection{The questions we explore with the help of our model}
\label{section:goals}

We will use the model we just introduced to explore the following questions: 

\begin{enumerate}
	\item \label{item:q1} Does $\epsilon$-differential privacy induce a bound on the information leakage of the randomized function $\mathcal{K}$? 
	\item \label{item:q2} Does $\epsilon$-differential privacy induce a bound on the information leakage \emph{relative to an individual}? 
	\item \label{item:q3} Does $\epsilon$-differential privacy induce a bound on the utility? 
	\item \label{item:q4} Given a query $f$ and a value $\epsilon > 0$, can we construct a randomized function $\mathcal{K}$ which satisfies $\epsilon$-differential privacy and also presents maximum utility?
\end{enumerate}

We will see that the answers to \ref{item:q1} and \ref{item:q2} are positive in case we take the measure of leakage to be the min-entropy leakage, and we provide bounds that are tight (i.e. for every $\epsilon$ there is a $\mathcal{K}$ whose leakage reaches the bound). For \ref{item:q3} we are able to give a tight bound in some cases which depend on the structure of the query, and for the same cases, we are able to construct an oblivious $\mathcal{K}$ with maximum utility (defined in terms of a binary gain function), as requested by \ref{item:q4}.

\section{Graph symmetries}
\label{section:graph-symmetries}

In this section we explore some classes of graphs that will allow us to derive a strict correspondence between $\epsilon$-differential privacy and the a posteriori entropy of the input. As we already mentioned, the input domain of databases and the adjacency relation forms an undirected graph, and this fact will be used to derive bounds on information leakage and utility. We will present two classes of graphs, distance-regular and $\vtt$, that will be used in the next section to transform a generic channel matrix into a matrix with a symmetric structure, while preserving the a posteriori min-entropy and the $\epsilon$-differential privacy.

Let us first recall some basic notions. Given a graph $G=(\mathcal{V}, \sim)$, the \emph{distance} $d(\textsl{v},w)$ between two vertices $\textsl{v},w\in \mathcal{V}$ is the number of edges in a shortest path connecting them. The \emph{diameter} $\diameter$ of $G$ is the maximum distance between any two vertices in $\mathcal{V}$. The \emph{degree} of a vertex is the number of edges incident to it. $G$ is called \emph{regular} if every vertex has the same degree. A regular graph with vertices of degree $k$ is called a $k$\emph{-regular graph}. An \emph{automorphism} of $G$ is a permutation $\sigma$ on the vertex set $\mathcal{V}$, such that for any pair of vertices $\textsl{v},w$, if $\textsl{v}\sim w$, then $\sigma(\textsl{v})\sim\sigma(w)$. If $\sigma$ is an automorphism, and $\textsl{v}$ is a vertex, the \emph{orbit} of $\textsl{v}$ under $\sigma$ is the set $\{\textsl{v}, \sigma(\textsl{v}), \ldots, \sigma^{k-1}(\textsl{v})\}$ where $k$ is the smallest positive integer such that $\sigma^k(\textsl{v}) = \textsl{v}$. Clearly, the orbits of the vertices under $\sigma$ define a partition of $\mathcal{V}$. If $\mathcal{V}$ is the set of vertices of $G$, we denote by $\border{\mathcal{V}}{d}{\textsl{v}}$ the subset of vertices in $\mathcal{V}$ that are at distance $d$ from the vertex $\textsl{v}$.

The following two definitions introduce the classes of graphs that we are interested in. The first class is well known in literature. 

\begin{definition}[Distance-regular graph]
	\label{def:distance-regular}
	A graph $G=(\mathcal{V},\sim)$ is called \emph{distance-regular} if there exist integers $b_d$ and $c_d$ ($d \in \{ 0,\ldots,\diameter \}$) (called \emph{intersection numbers}) such that, for all vertices $\textsl{v},w$ at distance $d(\textsl{v},w)=d$, there are exactly
	\begin{itemize}
		\item $b_d$ neighbors of $w$ in $\border{\mathcal{V}}{d+1}{\textsl{v}}$
		\item $c_d$ neighbors of $w$ in $\border{\mathcal{V}}{d-1}{\textsl{v}}$
	\end{itemize}
\end{definition}

Some examples of distance-regular graphs are illustrated in Figure~\ref{fig:dist-reg}.

	\begin{figure}[!htb]%
		\centering
		\subbottom[Tetrahedral graph]{
			\includegraphics[width=0.23\columnwidth]{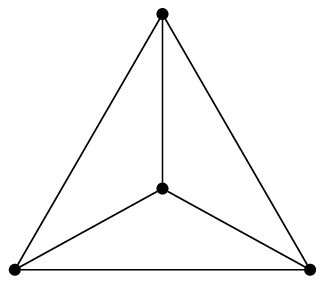}%
			\label{fig:dist-reg3}%
		}
		\qquad
		\subbottom[Cubical graph]{
			\includegraphics[width=0.22\columnwidth]{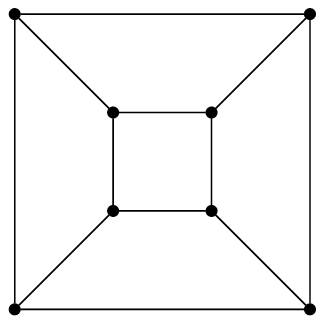}%
			\label{fig:dist-reg4}%
		}
		\qquad
		\subbottom[Petersen graph]{
			\includegraphics[width=0.22\columnwidth]{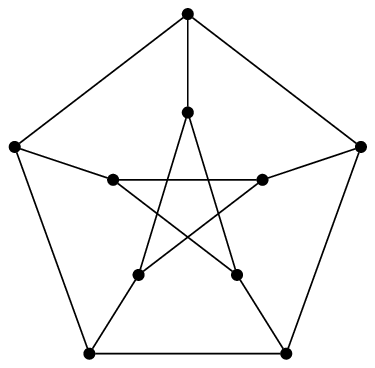}%
			\label{fig:dist-reg5}%
		}	
		\caption{Some distance-regular graphs with degree $3$}
		\label{fig:dist-reg}	
	\end{figure}
	
The second class we are interested in is a variant of the VT (vertex-transitive\footnote{A graph $G = (\mathcal{V},\sim)$ is said to be \emph{vertex-transitive} if for any pair $\textsl{v}, w \in \mathcal{V}$ there exists an automorphism $\sigma$ such that $\sigma(\textsl{v}) = w$.}) class: 

\begin{definition}[$\vtt$ graph]
A graph $G = (\mathcal{V},\sim)$ is $\vtt$ (\emph{vertex-transitive +}) if there are $n$ automorphisms $\sigma_0$, $\sigma_1$, \ldots $\sigma_{n-1}$, where $n = |\mathcal{V}|$, such that, for every vertex $\textsl{v} \in \mathcal{V}$, we have that $\{\sigma_i(\textsl{v})\mid 0\leq i \leq n-1\} = \mathcal{V}$. 
\end{definition}

In particular, the graphs for which there exists an automorphism $\sigma$ which induces only one orbit are $\vtt$: it is sufficient to define $\sigma_i=\sigma^i$ for all $i$ from $0$ to $n-1$. Figure \ref{fig:hexagons} illustrates some $\vtt$ graphs with a single-orbit automorphism.

\begin{figure}[!htb]%
	\centering
	\subbottom[Cycle: degree $2$]{
		\includegraphics[width=0.22\columnwidth]{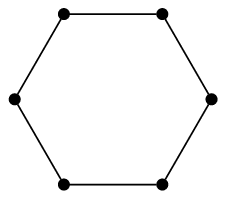}%
		\label{fig:hex1}%
	}
	\qquad
	\subbottom[Degree 4]{
		\includegraphics[width=0.22\columnwidth]{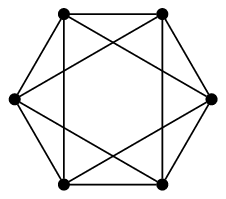}%
		\label{fig:hex2}%
	} 
	\qquad
	\subbottom[Clique: degree 5]{
		\includegraphics[width=0.22\columnwidth]{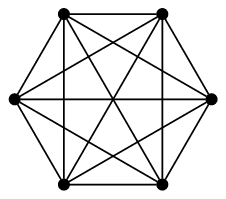}%
		\label{fig:hex3}%
	}	
	\caption{Some $\vtt$ graphs}
	\label{fig:hexagons}	
\end{figure}
	
From graph theory we know that neither of the two classes subsumes the other. They have however a non-empty intersection, which contains in particular all the structures of the form $(\val^{u},\sim)$, i.e. the database domains. 

The two next propositions show that the structure $(\mathcal{X},\sim) = (\val^{u},\sim)$ is both a distance-regular graph and a $\vtt$ graph.

\begin{proposition}
	\label{prob:database-is-dr}
	\revision{If $v \geq 2$, the} graph $(\val^{u}, \sim)$ is a connected distance-regular graph with diameter $\delta = u$, and intersection numbers $b_d = (u - d)(v-1)$ and $c_d = d$, for all $0 \leq d \leq \delta$.
\end{proposition}

\begin{proof}
	The vertices of ($\val^{u}, \sim$) are $u$-tuples $(\textsl{v}_1,\ldots,\textsl{v}_u), \textsl{v}_i\in \val$ and two vertices are adjacent if and only if the differ in exactly one element $\textsl{v}_i$. It is easy to see that the distance between two vertices is the number of elements in which they differ. Let $x_1,x_2 \in \val^{u}$ with $d(x_1,x_2)=d$, so they differ in exactly $d$ elements. To go at distance $d+1$ from $x_1$ we can select any of the remaining $u-d$ elements and change it in $v-1$ possible ways, so the total number is $(u-d)(v-1)$ and depends only on $d$, not on $x_1,x_2$. Similarly, by changing one of the differing elements of $x_2$ to match the value of $x_1$ we get a vertex at distance $d-1$, and there are $d$ such elements.
\end{proof}

\begin{proposition}
	\label{prob:database-is-vtt}
	The graph $(\val^{u}, \sim)$ is a $\vtt$ graph.
\end{proposition}

\begin{proof}

	Recall that we assume the values in the set $\val$ to be indexed, i.e. $ \val = \{\textsl{v}_{0}, \ldots, \textsl{v}_j, \ldots, \textsl{v}_{v-1} \}$, where $v = |\val|$. Note that, for convenience, we opt to use here the indexing from $0$ to $v-1$. Let us define an \revision{bijective} function $\rho: \val \rightarrow \val$ as
	\begin{equation*}
		\rho(\textsl{v}_j) = \textsl{v}_{j \oplus 1}
	\end{equation*}
	
	\noindent for every $\textsl{v}_j \in \val$, and where $\oplus$ represents the sum modulo $v$. We define the composition of $\rho$ with itself $i$ times as
	\begin{equation*}
		\rho^{i}(\textsl{v}_j) = \underbrace{\rho \circ \rho \circ \ldots \circ \rho}_{i \ \text{times}}(\textsl{v}_j)
	\end{equation*}
	
	\noindent Note that since $\rho$ is injective, $\rho^i$ is injective as well. 
	
	We represent a database in $\val^u$ as $x = \textsl{v}_{k_0} \dots \textsl{v}_{k_\ell} \ldots \textsl{v}_{k_{u-1}}$, with $0 \leq \ell \leq u-1$ and \revision{$0 \leq k_\ell \leq v-1$}. We now define a family $\{ \sigma_{\iota} \}_{\iota = 0}^{v^{u}-1}$ of automorphisms as follows. Given a $0 \leq \iota \leq v^u-1$, consider the representation in base $v$ of $\iota$: 
	\begin{equation}
		\label{eq:iota}
		\iota = i_0 \cdot v^{0} + \ldots + i_\ell \cdot v^\ell + \ldots + i_{u-1} \cdot v^{u-1} 
	\end{equation}
	
	\noindent where $0 \leq i_\ell \leq v-1$. Then define
	\begin{equation}
		\label{eq:automorphisms}
		\sigma_{\iota}(x) =   \rho^{i_0}(\textsl{v}_{k_0}) \ldots \rho^{i_{\ell}}(\textsl{v}_{k_{\ell}}) \ldots \rho^{i_{u-1}}(\textsl{v}_{k_{u-1}})  \\
	\end{equation}
	
	\noindent where $x = \textsl{v}_{k_0} \ldots \textsl{v}_{k_\ell} \ldots \textsl{v}_{k_{u-1}}$.
		
	We have to show that:
	
	\begin{itemize}
	
		\item \emph{$\sigma_\iota$ is an automorphism for all $0 \leq \iota \leq v^u - 1$.}
		
		First we show that $\sigma_\iota$ is injective. Let us consider two arbitrary databases $x = \textsl{v}_{k_{0}} \ldots \textsl{v}_{k_\ell} \ldots \textsl{v}_{k_{u-1}}$ and $x' =   \textsl{v}_{k_{0}'} \ldots \textsl{v}_{k_\ell'} \ldots \textsl{v}_{k_{u-1}'}$, and assume $\sigma_{\iota} = \rho^{i_0}(\cdot) \ldots \rho^{i_\ell}(\cdot) \ldots \rho^{i_{u-1}}(\cdot)$. If $x \neq x'$ then $\textsl{v}_{k_\ell} \neq \textsl{v}_{k_\ell'}$ for some $\ell$, and since \revision{an arbitrary} $\rho^{i_\ell}$ is injective we have $\rho^{i_\ell}(\textsl{v}_{k_\ell}) \neq \rho^{i_\ell}(\textsl{v}_{k_{\ell}'})$. Therefore $\sigma_\iota(x) \neq \sigma_\iota(x')$.
		
		Now we show that if $x \sim x'$ then $\sigma_\iota(x) \sim \sigma_\iota(x')$. Consider \revision{an} arbitrary pair of adjacent databases $x =   \textsl{v}_{k_{0}} \ldots \textsl{v}_{k_{\ell}} \ldots \textsl{v}_{k_{u-1}}  $ and $x' =   \textsl{v}_{k_{0}} \ldots \textsl{v}_{k_{\ell}'} \ldots \textsl{v}_{k_{u-1}}$, where $x$ and $x'$ differ exactly for  $\textsl{v}_{k_{\ell}} \neq \textsl{v}_{k_{\ell}'}$. We know that $\sigma_{\iota}(x) =   \rho^{i_{0}}(\textsl{v}_{k_{0}}) \ldots \rho^{i_{\ell}}(\textsl{v}_{k_{\ell}}) \ldots \rho^{i_{u-1}}(\textsl{v}_{k_{u-1}})$ and we also know that $\sigma_{\iota}(x') =   \rho^{i_{0}}(\textsl{v}_{k_{0}}) \ldots \rho^{i_{\ell}}(\textsl{v}_{k_{\ell}'}) \ldots \rho^{i_{u-1}}(\textsl{v}_{k_{u-1}})$. Therefore $\sigma_{\iota}(x)$ and $\sigma_{\iota}(x')$ can differ at most in $\rho^{i_{\ell}}(\textsl{v}_{k_{\ell}})$ and $\rho^{i_{\ell}}(\textsl{v}_{k_{\ell}'})$. Since $\rho^{i_\ell}$ is injective, we have $\rho^{i_{\ell}}(\textsl{v}_{k_{\ell}}) \neq \rho^{i_{\ell}}(\textsl{v}_{k_{\ell}'})$, and it follows that $\sigma_{\iota}(x) \sim \sigma_{\iota}(x')$.
	
		\item \emph{For every $x = \textsl{v}_{k_{0}} \ldots \textsl{v}_{k_{\ell}} \ldots \textsl{v}_{k_{u-1}}  $ in $\val^u$ we have $\bigcup_{\iota = 0}^{v^u-1}\{\sigma_{\iota}(x)\} = \val^u$.}
			
		Take an arbitrary element $x' =   \textsl{v}_{k_{0}'} \ldots \textsl{v}_{k_{\ell}'} \ldots \textsl{v}_{k_{u-1}'}  $ in $\val^u$. Note that $\rho^{k_m}(\textsl{v}_{k_n}) = \textsl{v}_{k_{m \oplus n}}$ for all $0 \leq m,n \leq v-1$. Therefore the automorphism $\sigma_{} =   \rho^{k_{0}' \ominus k_{0}}(\cdot) \ldots \rho^{k_\ell' \ominus k_\ell}(\cdot) \ldots \rho^{k_{u-1}' \ominus k_{u-1}}(\cdot)$, where $\ominus$ represents the subtraction modulo $v$, satisfies $\sigma_{}(x) = x'$. Since $0 \leq k_{\ell}' \ominus k_\ell \leq v-1$ we have that $\sigma_{} = \sigma_{\iota}$ for $\iota = (k_{0}' \ominus k_{0}) \cdot v^0 + \ldots + (k_{\ell}' \ominus k_{\ell}) \cdot v^\ell + \ldots + (k_{u-1}' \ominus k_{u-1}) \cdot v^{u-1}$, and therefore $\sigma_{}$ belongs to the family $\{\sigma_{\iota}\}_{\iota = 0}^{v^u-1}$.
	
	\end{itemize}

\end{proof}

Figure \ref{fig:hypercubes} illustrates some examples of structures $(\val^{u},\sim)$. Note that when $|\val| = 2$, $(\val^{u},\sim)$ is the $u$-dimensional hypercube. 

\begin{figure}[!htb]%
	\centering
	\subbottom[$u=4, \val=\{a,b\}$ ($4$-dimensional hypercube)]{
		\includegraphics[width=0.40\columnwidth]{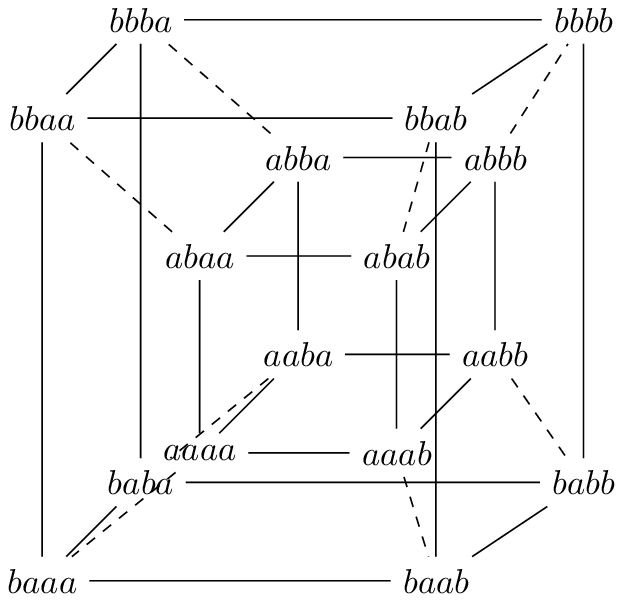}%
		\label{fig:hyp1}%
	}
	\qquad
	\subbottom[$u=3, \val = \{a,b,c\}$ (for readability sake we show only part of the graph)]{
		\includegraphics[width=0.40\columnwidth]{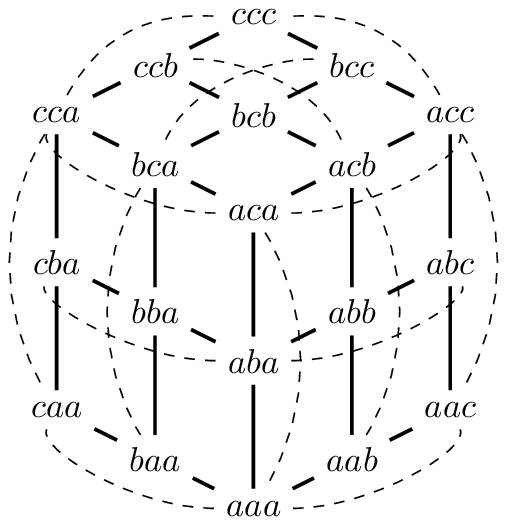}%
		\label{fig:hyp2}%
	}
	\caption{Some $(\val^{u},\sim)$ graphs}
	\label{fig:hypercubes}	
\end{figure}

The relation between graph structures we consider in this chapter is summarized in Figure \ref{fig:venn}. We remark that in general the graphs $(\val^{u},\sim)$ do not have a single-orbit automorphism.

\begin{figure}[!htb]%
	\centering
	\includegraphics[width=0.4\columnwidth]{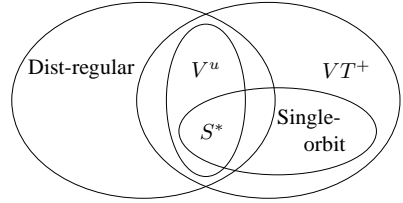}%
	\caption{\revision{Venn diagram for the classes of graphs considered in this section. Here $S^* = \{\val^u \ \ | \ \ |\val| = 2, u \leq 2 \}$}}
	\label{fig:venn}%
\end{figure}

\section{Deriving the relation between differential privacy and quantitative information flow on the basis of the graph structure}
\label{section:graph-bounds}

In this section we present the main technical contribution of the chapter: a general technique that explores the graph structure induced by the adjacency relation $\sim$ on $\calx$ and the query $f$ to determine  relations between $\epsilon$-differential privacy and min-entropy leakage, and between $\epsilon$-differential privacy and utility. We use the symmetries of the graph structure $(\calx,\sim)$ to transform the channel matrix into an equivalent matrix with certain regularities. These regularities are the key that allow us to establish the link between $\epsilon$-differential privacy and the a posteriori min-entropy (i.e. the conditional min-entropy associated to the channel). The establishment of bounds on the a posteriori \revision{entropy} will allow us to derive bounds on leakage and utility: in Section~\ref{section:leakage} we will cope with leakage and in Section~\ref{section:utility} we will cope with utility.

But first, in Section~\ref{section:matrix-transformation} we will present how to perform the transformation on the channel matrix, and in Section~\ref{section:bound-a-post-entropy} we will show how to derive a bound on the a posteriori min-entropy for the matrix obtained. It is important to note that we consider the case where \emph{the channel input has the uniform distribution}. This is not a restriction for our bounds on the leakage: as seen in Chapter~\ref{chapter:probabilistic-info-flow}, the maximum min-entropy leakage is achieved in the uniform input distribution and, therefore, any bound for the uniform input distribution is also a bound for all other input distributions. In the case of utility the assumption of uniform input distribution is more restrictive, but we will see that it still provides interesting results for several practical cases.

Before we present formally our technique, let us fix some notation.

\subsection{Assumptions and notation}
\label{section:notation}

In the rest of this section we consider channels (usually referred to by $M$, $M'$, $M''$ or $N$) with input $A$ and output $B$, with finite carriers \revision{$\mathcal{A} = \{ a_0,\ldots,a_{n-1}\}$ and $\mathcal{B} = \{b_0,\ldots,b_{m-1}\}$}, respectively, and we assume that the probability distribution of $A$ is uniform. Furthermore, we assume that $|\mathcal{A}| = n \leq |\mathcal{B}| = m$. If it is the case that $n > m$, we just add to the matrix enough zero-ed columns, i.e. columns containing only $0$'s, so \revision{as} to match the number of rows. Note that adding zero-ed columns does not change the min-entropy leakage nor the conditional min-entropy of the channel. We assume as well an adjacency relation $\sim$ on $\mathcal{A}$, i.e. that $(\mathcal{A}, \sim)$ is an undirected graph structure. With a slight abuse of notation, we will also write $i\sim h$ when $i$ and $h$ are associated to adjacent elements of $\mathcal{A}$, and we will write $d(i,h)$ to denote the distance between the elements of $\mathcal{A}$ associated to $i$ and $h$. \revision{More generally, we may use the number $i$ to denote the element $a_i$ of $\mathcal{A}$ (or, equivalently, the element $b_i$ of $\mathcal{B}$) whenever it is clear from the context}.

We note that a channel matrix $M$ satisfies $\epsilon$-differential privacy if for each column $j$ and for each pair of rows $i$ and $h$ such that $i\sim h$ we have that:
\begin{equation*}
	\frac{1}{e^\epsilon}\leq \frac{M_{i,j}}{M_{h,j}}\leq e^\epsilon.
\end{equation*}

The a posteriori entropy of a channel with matrix $M$ will be denoted by $H^M_\infty(A|B)$, and its min-entropy leakage by $I^M_\infty(A;B)$.

We denote by $M[l \to k]$  the matrix obtained by \qm{collapsing} the column $l$ into $k$, i.e.
\begin{equation*}
	M[l\to k]_{i,j} =
	\begin{cases}
		M_{i,k} + M_{i,l}	& \text{if } j = k, \\
		0				  				& \text{if } j = l, \\
		M_{i,j}						& \text{otherwise}
	\end{cases}
\end{equation*}

Given a partial function $\rho: \mathcal{A} \rightarrow \mathcal{B}$, the image of $\mathcal{A}$ under $\rho$ is $\rho(\mathcal{A}) = \{ \rho(a) | a \in \mathcal{A}, \rho(a) \neq \bot \}$, where $\bot$ stands for \qm{undefined}. 

In the proofs we will need to use several indices, and we will typically use the letters $i,j,h,k,l$  to range over rows and columns (usually $i,h,l$ will range over rows and $j, k$ will range over columns). Given a matrix $M$, we denote by $\maxj{M}{j}$ the maximum value of column $j$ over all rows $i$, i.e. $\maxj{M}{j} = \max_{i}M_{i,j}$, and by $\max^M = \max_{i,j} M_{i,j}$ the maximum element of the matrix.

Finally, given a graph $G = (\mathcal{V},\sim)$ with diameter $\diameter$, we denote by $\Delta_{G}$ the set $\{0, 1, \ldots, \diameter\}$. We may omit the subscript and denote the set only by $\Delta$ if the context does not allow any confusion. The notation $\border{\mathcal{V}}{d}{\textsl{v}}$ represents the subset of $\mathcal{V}$ of all elements $w$ at distance $d$ from $\textsl{v}$. For a fixed $d$, we define $n_d = |\border{\mathcal{V}}{d}{\textsl{v}}|$ as the number of vertices in $\mathcal{V}$ at distance $d$ from $\textsl{v}$, and we \revision{intend} that it will be always clear by the context to which set of vertices $\mathcal{V}$ and element $\textsl{v}$ the value $n_d$ is associated to.

\subsection{The matrix transformation}
\label{section:matrix-transformation}

The transformation on the channel matrices is divided into two steps, and we start this section by giving an overview of the process. Consider a channel whose matrix $M$ has at least as many columns as rows and assume that the input distribution is uniform. First, we transform $M$ into a matrix $M'$ in which each of the first $n$ columns has a maximum in the diagonal, and the remaining columns are all $0$'s. Second, under the assumption that the input domain is distance-regular or $\vtt$, we transform $M'$ into a matrix $M''$ whose diagonal elements are all the same, and coincide with the maximum element $\max^{M''}$ of $M''$. The transformation ensures that both $M'$ and $M''$ are valid channel matrices (i.e. each row is a probability distribution), also respect $\epsilon$-differential privacy, and preserve the value of the a posteriori entropy for the uniform input distribution. A scheme of the transformation is shown in Figure \ref{fig:mat-transf}, where Lemma~\ref{lemma:transform-diagonal} (\emph{Step $1$}) is applied on the first step of the transformation, and on the second step either Lemma~\ref{lemma:transform-dist-reg} (\emph{Step $2a$}) or Lemma~\ref{lemma:transform-vtt} (\emph{Step $2b$}) is applied, depending on whether the graph structure is distance-regular or $\vtt$, respectively.

\begin{figure}[!htb]
	\centering
	\includegraphics[width=0.65\columnwidth]{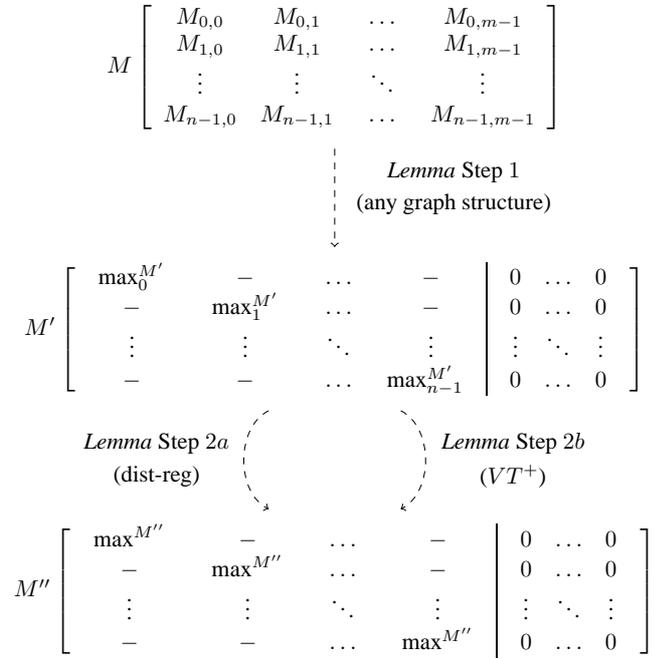}%
	\caption{\revision{Steps of the matrix transformation for distance-regular and $\vtt$ graphs}}%
	\label{fig:mat-transf}%
\end{figure}

We now present formally the transformation. The next Lemma is relative to the first step. 

\begin{lemma}[Step $1$]
	\label{lemma:transform-diagonal}
	Let $M$ be a channel matrix of dimensions $n \times m$ with at least as many columns as rows, and assume that $M$ satisfies $\epsilon$-differential privacy. Then it is possible to transform $M$ into a matrix $M'$ satisfying the following conditions: 
	
	\begin{enumerate}[(i)]
		\item \label{item:td-a} $M'$ is a valid channel matrix: \revision{$\sum_{j=0}^{m-1} M'_{i,j} = 1$ for all $0 \leq i \leq n-1$};
		\item \label{item:td-b} Each of the first $n$ columns has a maximum in the diagonal: $M'_{i,i}= \max^{M'}_{i}$ for all $0 \leq i \leq n-1$;
		\item \label{item:td-c} The $m-n$ last columns contain only $0$'s: $M'_{i,j}= 0$ for all $0 \leq i \leq n-1$ and all $n \leq j \leq m-1$;
		\item \label{item:td-d} $M'$ satisfies $\epsilon$-differential privacy: $\frac{M'_{i,j}}{M'_{h,j}} \leq e^\epsilon$ for all $0 \leq i,h \leq n-1$ \revision{s.t. $i \sim h$} and all $0 \leq j \leq m-1$;
		\item \label{item:td-e} $H^{M'}_\infty(A|B) = H^{M}_\infty(A|B)$, if $A$ has the uniform distribution.
	\end{enumerate}
	
\end{lemma}

\begin{proof}
	We first show that there exists a matrix $N$ of dimensions $n\times m$, and an injective total function $\rho: \mathcal{A} \rightarrow \mathcal{B}$ such that\revision{~\footnote{\revision{To avoid a heavy notation, here we will use the convention established in Section~\ref{section:notation} and denote $N_{a_i,b_j}$, where $a_i \in \mathcal{A}$ and $b_j \in \mathcal{B}$, simply by $N_{i,j}$.}}}:
	
	\begin{itemize}
		\item $N_{i,\rho(i)} = 	\maxj{N}{\rho(i)}$ for all $i \in \mathcal{A}$, and
		\item $N_{i,j} = 0$ for all $j \in \mathcal{B} \backslash \rho(\mathcal{A})$ and all $i \in \mathcal{A}$.
	\end{itemize}
	
	We iteratively construct $\rho$ and $N$ \qm{column by column} via a sequence of approximating partial functions $\rho_s$ and matrices $N_s$ ($0\leq s \leq m$).
	
	\begin{itemize}
	
		\item \emph{Initial step} ($s = 0$)\\
			
			Define $\rho_0(i) = \bot$ for all $i \in \mathcal{A}$ and $N_0 = M$. \\	
		
		\item \emph{$s^{th}$ step} ($1 \leq s \leq m$)\\
			
			Let $j$ be the $s$-th column and let $i \in \mathcal{A}$ be one of the rows containing the maximum value of column $j$ in $M$, i.e. $M_{i,j} = \maxj{M}{j}$. There are two cases:
			\begin{enumerate}
				\item $\rho_{s-1}(i) = \bot$. We define:
					\begin{align*}
						\rho_{s} & = \rho_{s-1} \cup \{ i \mapsto j \} & \text{and} \\
						N_s      & = N_{s-1}                           &
					\end{align*}
					
				\item \label{item:b} $\rho_{s-1}(i) = k \in \mathcal{B}$. We \qm{collapse} column $j$ into column $k$ \revision{(recall the notation introduced in Section~\ref{section:notation})}:
					\begin{align*}
						\rho_s & = \rho_{s-1}      & \text{and} \\
						N_s    &= N_{s-1}[j \to k] & 
					\end{align*}
			\end{enumerate} 
			
	\end{itemize}

	Since the operation of \qm{collapsing} assigns $j$ in $\rho_s$ and then zeroes the column $j$ in $N_s$, all unassigned columns $\mathcal{B} \setminus \rho_m(\mathcal{A})$ must be zero in $N_m$. We finish the construction by taking $\rho$ to be the same as $\rho_m$ after assigning to each unassigned row one of the columns in $\mathcal{B} \setminus \rho_m(\mathcal{A})$ (there are enough such columns since $n \leq m$). We also take $N = N_m$. Note that by construction $N$ is a channel matrix.

	Thus we get a matrix $N$ and a function $\rho: \mathcal{A} \rightarrow \mathcal{B}$ which, by construction, is injective and satisfies $N_{i,\rho(i)} = \maxj{N}{\rho(i)}$ for all $i \in \mathcal{A}$, and $N_{i,j} = 0$ for all $j \in \mathcal{B} \backslash \rho(\mathcal{A})$ and all $i \in \mathcal{A}$. Furthermore, $N$ provides $\epsilon$-differential privacy (condition  (\ref{item:td-d})) because each column is a linear combination of columns of $M$. It is also easy to see that $\sum_{j} \maxj{N}{j} = \sum_{j} \maxj{M}{j}$, and from that it immediately follows that $H_\infty^{N}(A|B) = H_\infty^{M}(A|B)$ (recall that $A$ has the uniform distribution and therefore the a posteriori entropy is a function of the sum of the maximum of each column), so condition (\ref{item:td-e}) is satisfied.

	Finally, we create our claimed matrix $M'$ from $N$ just by rearranging the columns according to $\rho$. Note that the order of the columns is irrelevant, since any permutation represents the same conditional probabilities and therefore the same channel\revision{~\footnote{\revision{Note that by rearranging the columns of the channel matrix we may change the marginal probability of the outputs. This, however, does not pose a problem for our purposes, since the maximum a posteriori entropy of the channel will be maintained. If we want the marginal probability of the outputs to remain unchanged, we can just \qm{relabel} the columns after the rearrangement so they will match the correct outputs.}}}. The resulting matrix $M'$ has all maxima in the diagonal $M'_{i,i}$ for $0 \leq i \leq n-1$, and every element in the columns $n \leq j \leq m-1$ are $0$, which satisfies conditions (\ref{item:td-b}) and (\ref{item:td-c}). Also, since $N$ is a valid channel matrix, so is $M'$ and condition (\ref{item:td-a}) is also satisfied.

\end{proof}

The second step of the transformation depends on the graph structure of $(\mathcal{A}, \sim)$. But before we discuss this step, let us introduce a notion of distance between elements in $\mathcal{B}$, derived from the notion of distance between elements in $\mathcal{A}$. Let $M$ be a channel matrix in which the maximum of each column is in the diagonal, as in Figure~\ref{fig:matrix-distance}. Then we define the distance between two elements $j_1, j_2 \in \mathcal{B}$ as follows:
\begin{equation}
	\label{eq:dist-b}
	d(j_1,j_2) =
		\begin{cases}
			d(i_1,i_2) & \text{if there are $i_1, i_2 \in \mathcal{A}$ such that $i_1 = j_1$ and $i_2 = j_2$,} \\
			\bot       & \text{otherwise}.
		\end{cases}
\end{equation} 

Note that the range of the notion of distance defined above is the set $\Delta = \{0, 1, \ldots, \diameter\}$, where $\diameter$ is the diameter of $(\mathcal{A},\sim)$. Based on \eqref{eq:dist-b}, we define the set $\border{\mathcal{B}}{d}{j}$ as the subset of $\mathcal{B}$ of elements at distance $d$ from an element $j \in \mathcal{B}$. It is clear that for any $j \in \mathcal{B}$, we have $\bigcup_{d \in \Delta} \border{\mathcal{B}}{d}{j} = \mathcal{B}$.

\begin{figure}[!htb]%
	\centering
	\includegraphics[width=0.75\columnwidth]{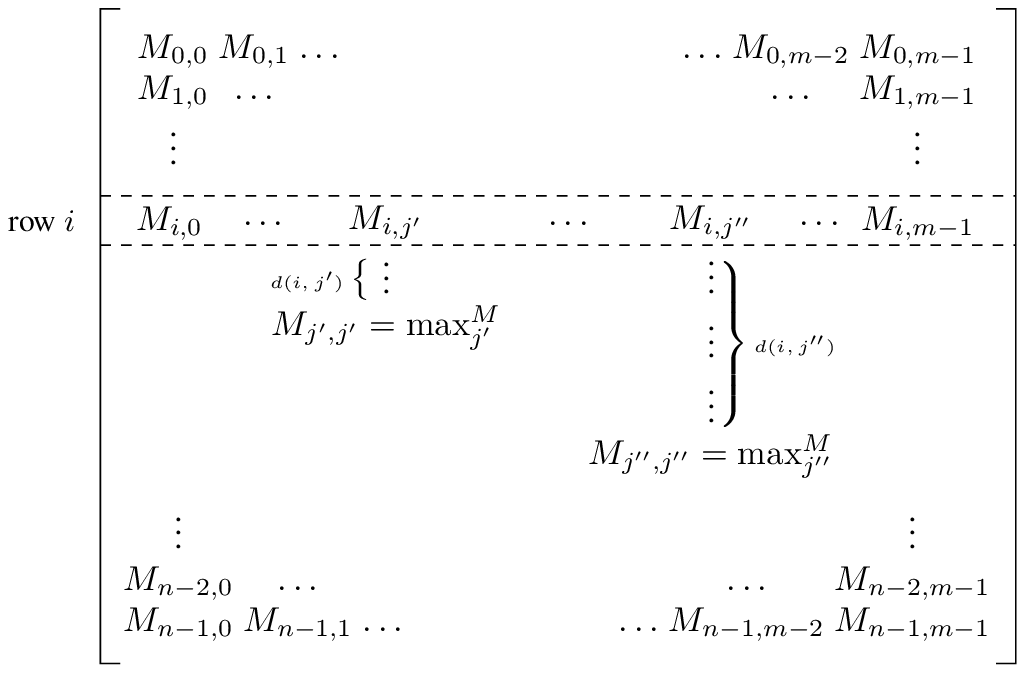}%
	\caption{The relation between elements of a row $i$ and the elements in the diagonal}%
	\label{fig:matrix-distance}%
\end{figure}

We can extend the adjacency relation $\sim$ on $\mathcal{A}$ to an adjacency relation $\sim'$ on $\mathcal{B}$ by using the notion of distance of \eqref{eq:dist-b}. For any $j_1,j_2 \in \mathcal{B}$, we have $j_1 \sim' j_2$ if and only if $d(j_1,j_2) = 1$. Therefore, if $(\mathcal{A},\sim)$ is distance-regular, so it is $(\mathcal{B},\sim')$.

Now we are ready to present the lemma for the second step of the transformation, in the case of distance-regular graphs.

\begin{lemma}[Step $2a$]
	\label{lemma:transform-dist-reg}
	Let $M'$ be a channel matrix of dimensions $n \times m$ with at least as many columns as rows, and assume that $M'$ satisfies $\epsilon$-differential privacy. Let $\sim$ be an adjacency relation on $\mathcal{A}$ such that the graph $(\mathcal{A},\sim)$ is connected and distance-regular. Assume that the maximum value of each column is on the diagonal, that is $M_{i,i} = \maxj{M}{i}$ for all $i\in \mathcal{A}$, and that all the last $m-n$ columns have only zero elements, i.e. $M'_{i,j} = 0$ for all $0\leq i \leq n-1$ and $n \leq j \leq m-1$. Then it is possible to transform $M'$ into a matrix $M''$ satisfying the following conditions: 
	
	\begin{enumerate}[(i)]
		\item \label{item:dr-a} $M''$ is a valid channel matrix: \revision{$\sum_{j=0}^{m-1} M''_{i,j} = 1$ for all $0 \leq i \leq n-1$};			
		\item \label{item:dr-b} The elements of the diagonal are all the same, and are equal to the maximum of the matrix: $M''_{i,i}= \max^{M''}$ for all $0 \leq i \leq n-1$;
		\item \label{item:dr-c} The $m-n$ last columns contain only $0$'s: $M''_{i,j}= 0$ for all $0 \leq i \leq n-1$ and all $n \leq j \leq m-1$;
		\item \label{item:dr-d} $M''$ satisfies $\epsilon$-differential privacy: $\frac{M'_{i,j}}{M'_{h,j}} \leq e^\epsilon$ for all $0 \leq i,h \leq n-1$ \revision{s.t. $i \sim h$} and all $0 \leq j \leq m-1$;
		\item \label{item:dr-e} $H^{M''}_\infty(A|B) = H^{M'}_\infty(A|B)$, if $A$ has the uniform distribution.
	\end{enumerate}

\end{lemma}

\begin{proof}	
	
	Let us define $\mathcal{B}^{*} = \{0, 1, \ldots, n-1\}$, i.e. the subset of $\mathcal{B}$ that excludes the zero-ed columns of $M'$ from $n$ to $m-1$. Note that we can safely use the set $\mathcal{B}^*$ instead of $\mathcal{B}$ in this proof because the zero-ed columns do not contribute to the a posteriori entropy, and trivially respect $\epsilon$-differential privacy.
	
	We then define the matrix $M''$ as follows.
	\begin{equation*}
		M''_{i,j} = 
			\begin{cases}
				\frac{1}{n|\border{\mathcal{A}}{d(i,j)}{i}|} \sum_{k \in \mathcal{B}^*} \sum_{h \in \border{\mathcal{\mathcal{A}}}{d(i,j)}{k}} M'_{h,k} & \text{if $j \in \mathcal{B}^*$,} \\
				0 & \text{otherwise.}
			\end{cases}
	\end{equation*}
	
	By the definition above, condition (\ref{item:dr-c}) is immediately satisfied. We then show that this definition also induces a channel matrix. We have 
	\begin{align*}
		\sum_{j \in \mathcal{B}^*} M''_{i,j} &=
			\sum_{j \in \mathcal{B}^*} 
				\frac{1}{n|\border{\mathcal{A}}{d(i,j)}{i}|} \sum_{k \in \mathcal{B}^*} \sum_{h \in \border{\mathcal{A}}{d(i,j)}{k}} M'_{h,k} \\ 
				& = \frac{1}{n} \sum_{k \in \mathcal{B}^*} \sum_{j \in \mathcal{B}^*} \frac{1}{|\border{\mathcal{A}}{d(i,j)}{i}|}\sum_{h \in \border{\mathcal{A}}{d(i,j)}{k}} M'_{h,k}
	\intertext{%
		Recall that $\Delta=\{0,\ldots,\delta\}$, where $\delta$ is the diameter of the graph. Note that for every $i$, $\mathcal{B}^* = \bigcup_{d\in \Delta}\border{\mathcal{B}^*}{d}{i}$, \revision{and for} different values of $d$ the sets $\border{\mathcal{B}^*}{d}{i}$ are disjoint. Therefore the summation over $j \in \mathcal{B}^*$ can be split as follows
	}
		& = \frac{1}{n} \sum_{k \in \mathcal{B}^*} \sum_{d\in \Delta}\sum_{j\in\border{\mathcal{B}^*}{d}{i}} \frac{1}{|\border{\mathcal{A}}{d}{i}|}\sum_{h \in \border{\mathcal{A}}{d}{k}} M'_{h,k} \\
		& = \frac{1}{n} \sum_{k \in \mathcal{B}^*} \sum_{d\in \Delta} \sum_{h \in \border{\mathcal{A}}{d}{k}} M'_{h,k} \sum_{j\in\border{\mathcal{B}^*}{d}{i}}\frac{1}{|\border{\mathcal{A}}{d}{i}|}
	\intertext{%
		as  $\displaystyle \sum_{j\in \border{\mathcal{B}^*}{d}{i}}\frac{1}{|\border{\mathcal{A}}{d}{i}|} =1$, we obtain 
	}
		& = \frac{1}{n} \sum_{k \in \mathcal{B}^*} \sum_{d\in \Delta} \sum_{h \in \border{\mathcal{A}}{d}{k}} M'_{h,k} 
	\intertext{%
		and now the summations over $h$ can be joined together
	}
		& = \frac{1}{n} \sum_{k \in \mathcal{B}^*} \sum_{h \in \mathcal{A}} M'_{h,k} \\
		& = 1
	\end{align*}
	\noindent which implies that condition (\ref{item:dr-a}) is satisfied. 
	
	We now turn our attention to the elements of the diagonal. We have
	\begin{equation*}
		M''_{i,i} = \frac{1}{n}\sum_{h\in \mathcal{A}}M'_{h,h}
	\end{equation*}
	
	\noindent and so they are all identical. To fulfill condition (\ref{item:dr-b}) we still need to show that $M''_{i,i} = \maxj{M''}{i}$ for all $i\in \mathcal{A}$.
	\begin{align*}
		M''_{i,j} & =    \frac{1}{n |\border{\mathcal{A}}{d(i,j)}{i}|} \sum_{k \in \mathcal{B}^*} \sum_{h \in \border{\mathcal{A}}{d(i,j)}{k}} M'_{h,k} & \\
						  & \leq \frac{1}{n |\border{\mathcal{A}}{d(i,j)}{i}|} \sum_{k \in \mathcal{B}^*} \sum_{h \in \border{\mathcal{A}}{d(i,j)}{k}} M'_{h,h} & \text{(since the biggest element} \\[-7mm]
						  &                                                                                                                          & \text{is in the diagonal)} \\
						  & =    \frac{1}{n} \sum_{k \in \mathcal{B}^*} M'_{h,h} \frac{1}{|\border{\mathcal{A}}{d(i,j)}{i}|} \sum_{h \in \border{\mathcal{A}}{d(i,j)}{k}} 1 & \\
						  & =    \frac{1}{n} \sum_{k \in \mathcal{B}^*} M'_{h,h} \frac{|\border{\mathcal{A}}{d(i,j)}{k}|}{|\border{\mathcal{A}}{d(i,j)}{i}|} & \\
						  & =    \frac{1}{n} \sum_{k \in \mathcal{B}^*} M'_{h,h} \cdot 1 & \text{(since the graph} \\[-6mm]
						  &                                                            & \text{is distance-regular)} \\
						  & =    M''_{i,i}
	\end{align*}
	
	Since $A$ has the uniform distribution, $H_\infty^{M'}(A|B) = H_\infty^{M''}(A|B)$ (condition (\ref{item:dr-e})) follows immediately. 
		
	It remains to show that $M''$ satisfies $\epsilon$-differential privacy (condition (\ref{item:dr-d})). We need to show that
	\begin{equation*}
		M''_{i,j} \le e^\epsilon M''_{i',j} \quad \quad  \forall j \in \mathcal{B}, i,i' \in \mathcal{A} : i \sim i'
	\end{equation*}
	
	From the triangular inequality we have (since $d(i,i')=1$)
	\begin{equation*}
		d(i',j) - 1 \le d(i,j) \le d(i',j) + 1
	\end{equation*}
	
	Thus, there are $3$ possible cases: 
	
	\begin{enumerate}
	
		\item $d(i,j)=d(i',j)$ \\ \\
			The result is immediate since $M''_{i,j}=M''_{i',j}$.
		
		\item $d(i,j)=d(i',j)-1$\\ \\
			We define the set of neighbors of $h$ \qm{one step further away} from $k$:
			\begin{equation*}
				\mathcal{F}_{h,k} = \{ h' \sim h \ | \ h' \in \border{\mathcal{A}}{d(h,k)+1}{k} \}
			\end{equation*}	
	
			Note that $|\mathcal{F}_{h,k}| = b_{d(h,k)}$ since the graph is distance-regular. The following inequalities hold for any $h,h' \in\mathcal{A}$:
			\begin{align}
				M'_{h,k} &\le e^\epsilon M'_{h',k}	\quad \quad \quad \quad \forall h' \in \mathcal{F}_{h,k} & \text{(diff. privacy)}\Rightarrow \notag\\
				b_{d(h,k)} M'_{h,k} &\le e^\epsilon \sum_{h'\in \mathcal{F}_{h,k}} M'_{h',k} &\text{(sum of the above)} \notag
				\intertext{%
				we now fix a distance $d$ and sum the above inequalities for all vertices at distance $d$ from $h$:
				}
				\sum_{h \in \border{\mathcal{A}}{d}{k}} b_d M'_{h,k}
				&\le e^\epsilon \sum_{h \in \border{\mathcal{A}}{d}{k}} \sum_{h'\in \mathcal{F}_{h,k}} M'_{h',k}
				&& \notag
				\intertext{%
				Note that each $h' \in \border{\mathcal{A}}{d+1}{k}$ is contained in $\mathcal{F}_{h,k}$ for exactly $c_{d+1}$ different $h \in \border{\mathcal{A}}{d}{k}$. So the right-hand side above sums all vertices of  $\border{\mathcal{A}}{d+1}{k}$ exactly $c_{d+1}$ times each. Thus we get that for all $k \in \mathcal{B}^*, d\in \Delta$:
				}
				\label{eq1-distreg}
				b_d\sum_{h \in \border{\mathcal{A}}{d}{k}} M'_{h,k}
				&\le e^\epsilon\ c_{d+1} \sum_{h \in \border{\mathcal{A}}{d+1}{k}} M'_{h,k}
			\end{align}
	
			Finally, note that $c_{d+1}|\border{\mathcal{A}}{d+1}{i}| = b_d|\border{\mathcal{A}}{d}{i}|$ (both sides count the number of edges between a vertex at distance $d$ and a vertex at distance $d+1$). So we have
			\begin{align*}
				M''_{i,j} & = \frac{1}{n|\border{\mathcal{A}}{d}{i}|} \sum_{k \in \mathcal{B}^*} \sum_{h \in \border{\mathcal{A}}{d}{k}} M'_{h,k}	\\
									&   \leq e^\epsilon\ \frac{1}{n|\border{\mathcal{A}}{d}{i}|} \frac{c_{d+1}}{b_d} \sum_{k \in \mathcal{B}^*} \sum_{h \in \border{\mathcal{A}}{d+1}{k}} M'_{h,k} & \text{(from \eqref{eq1-distreg})} \\
								  & = e^\epsilon\ \frac{1}{n|\border{\mathcal{A}}{d+1}{i}|} \sum_{k \in \mathcal{B}^*} \sum_{h \in \border{\mathcal{A}}{d+1}{k}} M'_{h,k}	\\
									& = e^\epsilon M''_{i',j}
			\end{align*}
		
		\item $d(i,j)=d(i',j)+1$ \\ \\
			This case is analogous to the case case where $d(i,j)=d(i',j)-1$.
			
	\end{enumerate}
	
\end{proof}

The next lemma is relative to the second step of the transformation, for the case of $\vtt$ graphs. 

\begin{lemma}[Step $2b$]
	\label{lemma:transform-vtt}
	Consider a channel matrix $M'$ satisfying the assumptions of Lemma \ref{lemma:transform-dist-reg}, except for the assumption about distance-regularity, which we replace by the assumption that $(\mathcal{A}, \sim)$ is $\vtt$. Then it is possible to transform $M'$ into a matrix $M''$ with the same properties as in Lemma~\ref{lemma:transform-dist-reg}.
\end{lemma}

\begin{proof}

	Let us define $\mathcal{B}^{*} = \{0, 1, \ldots, n-1\}$, i.e. the subset of $\mathcal{B}$ that excludes the zero-ed columns of $M'$ from $n$ to $m-1$. Note that we can safely use the set $\mathcal{B}^*$ instead of $\mathcal{B}$ in this proof because the zero-ed columns do not contribute to the a posteriori entropy, and trivially respect $\epsilon$-differential privacy.
	
	We then define the matrix $M''$ as follows.
	\begin{equation*}
		M''_{i,j} = 
			\begin{cases}
				\frac{1}{n} \sum_{h=0}^{n-1}M'_{\sigma_{h}(i),\sigma_{h}(j)} & \text{if $j \in \mathcal{B}^*$,} \\
				0 & \text{otherwise.}
			\end{cases}
	\end{equation*}
	
	By the definition above, condition (\ref{item:dr-c}) is immediately satisfied. We then show that this definition also induces a channel matrix. Recall that \revision{$\{\sigma_{h}(j)|0\leq h \leq n-1\}=\mathcal{A}$} since the graph is $\vtt$.
	\begin{align*}
		\sum_{j=0}^{n-1} M''_{i,j} & = \sum_{j=0}^{n-1} \frac{1}{n} \sum_{h=0}^{n-1}M'_{\sigma_{h}(i),\sigma_{h}(j)} & \\
															 & = \sum_{h=0}^{n-1} \frac{1}{n} \sum_{j=0}^{n-1}M'_{\sigma_{h}(i) ,\sigma_{h}(j)} & \\
															 & = \sum_{h=0}^{n-1} \frac{1}{n} \cdot 1 & \text{\revision{(since $\sigma_{h}$ is a permutation)}} \\
															 & = 1 & \\
	\end{align*}
	
	\noindent which implies that condition (\ref{item:dr-a}) is satisfied. 
	
	Now we prove that the diagonal contains the maximum values of the matrix (condition (\ref{item:dr-b})), i.e. for every $i$, \revision{$M''_{i,i} = \max^{M''}$}. It is easy to see that, by definition, the elements of the diagonal are all the same \revision{(they are the average of the diagonal elements of $M'$)}. Then we need to show that they are the maximum of each column, from which it follows that they are the maximum of the matrix.
	\begin{align*}
		M''_{i,i} & =    \frac{1}{n} \sum_{h=0}^{n-1} M'_{\sigma_{h}(i),\sigma_{h}(i)} & \\[3ex]
					    & \geq \frac{1}{n} \sum_{h=0}^{n-1} M'_{\sigma_{h}(i),\sigma_{h}(j)} & \text{{(since $M'_{\sigma_{h}(j),\sigma_{h}(j)}=\textstyle \maxj{M'}{\sigma_i(j)}$)}}& \\[3ex]
						  & =    M''_{i,j} & \\
	\end{align*}
	
	We now prove that $M''$ provides $\epsilon$-differential privacy (condition (\ref{item:dr-d})). For every pair $i \sim i'$ and every $j$:
	\begin{align*}
		M''_{i,j} & =    \frac{1}{n} \sum_{h=0}^{n-1} M'_{\sigma_{h}(i),\sigma_{h}(j)} & \\[3ex]
						  & \leq \frac{1}{n} \sum_{h=0}^{n-1} e^\epsilon M'_{\sigma_{h}(i'),\sigma_{h}(j)}& \text{(by $\epsilon$-diff. privacy, for some $i'$} \\[-5mm]
						  &   																														  						 & \text{s.t. $\sigma_{h}(i') = \sigma_{h}(j)$)} \\[3ex]
						  & =    e^\epsilon M''_{i',j} & \\
	\end{align*}
	
	Finally, we prove condition (\ref{item:dr-e}):
	\begin{align*}
		H_\infty^{M''}(A|B) & = \revision{\frac{1}{n}} \sum_{i=0}^{n-1} M'_{h,h} & \\
										    & = \frac{1}{n} \sum_{i=0}^{n-1} \revision{\frac{1}{n}} \sum_{h=0}^{n-1}   M'_{\sigma_{h}(i),\sigma_{h}(i)} & \text{} \\
				  					    & = \frac{1}{n} \sum_{i=0}^{n-1}  H_\infty^{M'}(A|B) &  \text{(since $M'_{\sigma_{h}(i),\sigma_{h}(i)}= \textstyle{\maxj{M'}{\sigma_i(i)}}$)} \\
									   	  & = H_\infty^{M'}(A|B) & \\
	\end{align*}
	
\end{proof}

\subsection{The bound on the a posteriori entropy of the channel}
\label{section:bound-a-post-entropy}

Once the transformation presented in the previous section has been applied, and the channel matrix respects the properties of $M''$, we can use again the graph structure of $(\mathcal{A}, \sim)$ to determine a bound on the a posteriori entropy $H^{M''}_\infty(A|B)$ of $M''$. Recall that our matrix transformation preserves the value of the a posteriori conditional entropy, so the bound we find is also valid for the original channel matrix we started with.

It is a known result in literature (cfr. ~\cite{Braun:09:MFPS}) that, if the distribution on $A$ is uniform, then the a posteriori entropy of the channel $M$ is given by
\begin{equation*}
	H_{\infty}^{M}(A|B) = -\log_2 \revision{\frac{1}{n}} \sum_{j \in \mathcal{B}} \textstyle{\maxj{M}{j}}
\end{equation*}
	
Hence, under our assumption that the input distribution $A$ is uniform, and knowing that matrix the $M''$ the diagonal elements are all equal to the maximum $\max^{M''}$, we have
\begin{equation}
	\label{eq:cond-entropy-uniform}
		H^{M''}_\infty(A|B) = -\log_2 \textstyle{\max^{M''}}
\end{equation}

Therefore to find a bound on the a posteriori entropy of the channel $M''$ it is enough to find a bound on $\max^{M''}$. This is exactly what we do in this section. 

We proceed by noting that the property of $\epsilon$-differential privacy induces a relation between the ratio of elements at any distance: 

\begin{remark}
	Let $M$ be a matrix satisfying $\epsilon$-differential privacy. Then, for any column $j$, and any pair of rows $i$ and $h$ we have that:
	\begin{equation*}
		\frac{1}{e^{\epsilon \,d(i,h)}}\leq \frac{M_{i,j}}{M_{h,j}}\leq e^{\epsilon \,d(i,h)}
	\end{equation*}
\end{remark}

In particular, as we know that the diagonal elements of $M$ are equal to the maximum element $\max^{M}$, then for each element $M_{i,j}$ we have that:
\begin{equation}
	\label{eq:other-elements}
	M_{i,j}\geq \frac{\max^{M}}{\displaystyle e^{\epsilon \,d(i,j)}}
\end{equation}

\noindent which motivates the next proposition.

\begin{proposition}
	\label{prop:bound-max-M}
	Let $M$ be a channel matrix \revision{satisfying $\epsilon$-differential privacy} where the diagonal elements are the maximum element $max^{M}$ of the matrix. Then:
	\begin{equation*}
		\textstyle{\max^{M}} \leq \frac{1}{\sum_{d \in \Delta} \frac{n_d}{e^{\epsilon d}}}
	\end{equation*}
	\noindent where $\Delta = \{0, 1, \ldots, \diameter\}$, $\delta$ is the diameter of the graph $(\mathcal{A},\sim)$, and $n_d = \border{\mathcal{A}}{d}{j}$ is the number of elements $M_{i,j}$ that are at distance $d$ from the corresponding diagonal element $M_{j,j}$, i.e. such that $d(i,j) = d$.
\end{proposition}

\begin{proof}
	The elements of any given row $i$ of $M$ represent a probability distribution, therefore they \revision{sum} to $1$.
	\begin{equation*}
		\sum_{j} M_{i,j} = 1
	\end{equation*}
	
	 By substituting \eqref{eq:other-elements} in the equation above we obtain:
	 \begin{align*}
	 	\nonumber \sum_{j} \left( \frac{\textstyle \max^{M}}{e^{\epsilon d(i,j)}} \right) \leq 1 \\
	 	\nonumber \sum_{d} \left( \frac{n_{d}}{e^{\epsilon d}} {\textstyle \max^{M}}  \right) \leq 1 &
	 \end{align*}
	 
	 \noindent and therefore
	 \begin{equation*}
	 	{\textstyle \max^{M}} \leq \frac{1}{\sum_{d} \frac{n_d}{e^{\epsilon d}}}
	 \end{equation*}
	 
\end{proof}

Putting together all the steps of this section, we obtain our main result. 

\begin{theorem}
	\label{theo:bound-cond-min-entropy}
	Consider a channel matrix $M$ satisfying $\epsilon$-differential privacy for some $\epsilon > 0$, and assume that $(\mathcal{A}, \sim)$ is either distance-regular or $\vtt$. Then we have:
	\begin{equation}
		\label{eq:bound-impro}
		H^{M}_\infty(A|B) \geq - \log_2 \frac{1}{\sum_{d} \frac{n_d}{e^{\epsilon \,d }}}
	\end{equation}
	
	\noindent where $n_d = |\border{\mathcal{A}}{d}{i}|$ is the number of nodes $j \in \mathcal{A}$ at distance $d$ from $i \in \mathcal{A}$.
	
	Moreover, this bound it tight, in the sense that we can build a matrix for which \eqref{eq:bound-impro} holds with equality. 
	
\end{theorem}

\begin{proof}
		The inequality follows directly from \eqref{eq:cond-entropy-uniform} and Proposition~\ref{prop:bound-max-M}. To prove that the bound is tight, it is sufficient to define each element $M_{i,j}$ according to \eqref{eq:other-elements} with equality instead of inequality. 
\end{proof}

In the next sections we will see how to use this theorem for establishing a bound on the leakage and on the utility. 

\section{Application to leakage}
\label{section:leakage}

As discussed in the Section~\ref{section:model}, the correlation $\call(X,Z)$ between $X$ and $Z$ measures the information that the attacker can learn about the database by observing the reported answers. In this section we consider the min-entropy leakage as a measure of this information, that is $\call(X,Z) = I_\infty(X;Z)$. We then investigate bounds on information leakage imposed by differential privacy.

Before we continue, let us make a very important observation about the results we obtain in this section.

\begin{remark}
The bounds on the min-entropy leakage we present in this section (Theorem~\ref{theo:bound-leakage}, Proposition~\ref{prop:new-bound}, and Proposition~\ref{prop:ind}) are derived under the assumption that the input distribution $X$ for the channel is uniform. As seen in Chapter~\ref{chapter:probabilistic-info-flow}, we know from the literature \cite{Braun:09:MFPS,Smith:09:FOSSACS} that the min-entropy leakage $I_\infty^{M}(X;Z)$ of a given matrix $M$ is maximum when input distribution is uniform (even though it may not be the only case). Therefore the bounds we present in this section, although based on the assumption that $X$ has the uniform distribution, are valid for every possible input distribution. As we model side information as input distributions, and as we provide bounds on the leakage for any possible input distribution, it follows that our bounds on the min-entropy leakage are valid for any possible side information the attacker may have.
\end{remark}

Our first result shows that the min-entropy leakage of a randomized function $\mathcal K$ is bounded by a quantity depending on $\epsilon$, and on the numbers $u = |\ind|$ and $v = |\val|$ of individuals and values respectively. We assume that $v\geq 2$.

As seen in Section~\ref{section:model}, $\mathcal{K}$ can be modeled as a channel with input $X$ and output $Z$. From Propositions \ref{prob:database-is-dr} and \ref{prob:database-is-vtt} we know that $(\calx, \sim)$ is both distance-regular and $\vtt$, and therefore we can apply Theorem~\ref{theo:bound-cond-min-entropy}. Then, by \eqref{eq:other-elements} we know that \revision{for} $j \in \border{\mathcal{X}}{d}{x}$ (i.e. every $j$ in $\calx$ at distance $d$ from a given $x$) it is the case that $M_{x,j} \geq \frac{\max^{M}}{e^{ \epsilon d}}$. Furthermore we note that each element $j$ at distance $d$ from $x$ can be obtained by changing the value of $d$ individuals in the $u$-tuple representing $i$. We can choose those $d$ individuals in $\binom{u}{d}$ possible ways, and for each of these individuals we can change the value (with respect to the one in $x$) in $v-1$ possible ways. Therefore $|\border{\calx}{d}{x}| = \binom{u}{d} (v-1)^{d}$, and we obtain that the number of databases at distance $d$ from $x$ is 
\begin{equation}
	\label{eq:border} 
	n_d = |\border{\mathcal{X}}{d}{x}| = \left(\begin{array}{c}u\\d\end{array}\right)\,(v-1)^{d}
\end{equation}
	
In fact, recall that $x$ can be represented as a $u$-tuple with values in $V$. We need to select $d$ individuals in the $u$-tuple and then change their values, and each of them can be changed in $v-1$ different ways. 

Using the value of $n_d$ from \eqref{eq:border} in Theorem~\ref{theo:bound-cond-min-entropy} we obtain the following result.

\begin{theorem}
	\label{theo:bound-leakage}
	If $\mathcal{K}$ satisfies $\epsilon$-differential privacy, then the information leakage is bound from above as follows: 
	\begin{equation*}
		I_\infty(X;Z)\leq u\, \log_2\frac{v\,e^{\epsilon}}{v-1+e^\epsilon} = \bound(u,v,\epsilon)
	\end{equation*}
\end{theorem}

\begin{proof}
	For this proof we need a matrix with all column maxima on the diagonal, and all equal. We obtain such a matrix by transforming the matrix associated to $\mathcal{K}$ as follows: first we apply Lemma~\ref{lemma:transform-diagonal} to it (with $A=X$ and $B=Z$), and then we apply either Lemma~\ref{lemma:transform-dist-reg} or Lemma~\ref{lemma:transform-vtt} (we can choose \revision{either} of them, since $(\mathcal{X},\sim)$ is both distance-regular and $\vtt$). The final matrix $M$ has all non-zero elements on its $n\times n$ submatrix, with $n = |\mathcal{X}| = \val^{u}$, provides $\epsilon$-differential privacy, and for every row $i$ we have that $M_{i,i} = \max^{M}$. Furthermore, \revision{$I_\infty^M(X;Z)$} is equal to the min-entropy leakage of $\mathcal{K}$\revision{, assuming a uniform distribution on $X$}. 
	
	Then we can derive:
	\begin{align*}
		\sum_{j=1}^{n} M_{i,j} & \geq \sum_{d=0}^{u} n_d \frac{\textstyle{\max^M}}{(e^{ \epsilon})^d} & \\
													 & =    \sum_{d=0}^{u} \binom{u}{d} (v-1)^{d} \frac{\textstyle{\max^M}}{(e^{ \epsilon})^d} & \text{(by \eqref{eq:border})}											 
	\end{align*}

	Since each row represents a probability distribution, the elements of row $i$ must sum up to $1$:
	\begin{equation*}
		\sum_{d=0}^{u} \binom{u}{d} (v-1)^{d} \frac{\textstyle{\max^M}}{(e^{ \epsilon})^d} \leq 1
	\end{equation*}
	
	\noindent and by multiplying both sides of the inequality by $e^{\epsilon u}$ we get
	\begin{equation*}
			\textstyle{\max^M} \sum_{d=0}^{u} \binom{u}{d} (v-1)^{d} e^{\epsilon(u-d)} \leq e^{\epsilon u}
	\end{equation*}
	
	Since by the binomial expansion $\displaystyle \sum_{d=0}^{u} \binom{u}{d} (v-1)^{d} ({e^\epsilon})^{u-d}\, = \, (v-1+e^\epsilon)^u$, we obtain:
	\begin{equation}		
		\label{eq:theorem:bl}
			\textstyle{\max^M} \leq \left(\frac{{e^\epsilon}}{v-1+e^\epsilon}\right)^u
	\end{equation}

	Therefore:	
	\begin{align*}		
	 	 I_{\infty}^{M}(X;Y) & =    H_{\infty}(X) - H_{\infty}^{M}(X|Y) & \text{(by definition)} \\
												 & =    \log_{2} \val^{u} + \log_2 \textstyle{\max^{M}} & \text{(by \eqref{eq:cond-entropy-uniform})} \\
												 & \leq \log_{2} \val^{u} + \log_2 \left(\frac{{e^\epsilon}}{v-1+e^\epsilon}\right)^u & \text{(by \eqref{eq:theorem:bl})} \\[2ex]
					 							 & =    u \log_{2} \frac{v\,{e^\epsilon}}{v-1+e^\epsilon} & \text{} \\
	\end{align*}
	
	To conclude our proof we recall that, since the above bound on $I_{\infty}^{M}(X;Y)$ is valid for the case where $X$ has the uniform distribution, it is also valid for any distribution on $X$.
	
\end{proof}

Note that the bound $\bound(u,v,\epsilon) = u\, \log_2\frac{v\, e^\epsilon}{(v-1 +  e^\epsilon)}$  is a continuous function in $\epsilon$, has value $0$ when $\epsilon=0$, and converges to $u\, \log_2 v$ as $\epsilon$ approaches infinity. Figure~\ref{fig:plots} shows the growth of  $\bound(u,v,\epsilon)$ along with $\epsilon$, for various fixed values of $u$ and $v$.

\begin{figure}[!htb]
	\centering
	\includegraphics[width=0.36\textwidth]{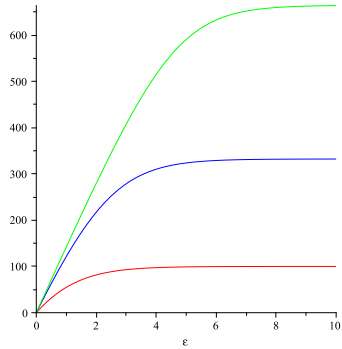}%
	\caption{Graphs of $\bound(u,v,\epsilon)$ for $u\!\!=\!\!100$ and $v\!\!=\!\!2$ (lowest line), $v\!\!=\!\!10$ (intermediate line), and $v\!\!=\!\!100$ (highest line), respectively.}%
	\label{fig:plots}%
\end{figure}

	The next proposition shows that the bound obtained in previous theorem is tight.

\begin{proposition} 
	\label{prop:tight}
	For every $u$, $v$, and $\epsilon$ there exists a randomized function $\mathcal K$ which provides $\epsilon$-differential privacy and whose min-entropy leakage, for the uniform input distribution, is $I_\infty(X;Z)=\bound(u,v,\epsilon)$.
\end{proposition}

\begin{proof}
	The adjacency relation in $\mathcal X$ determines a graph structure $G_\mathcal{X}$. Set $\mathcal{Z}=\mathcal{X}$ and define the matrix of $\mathcal K$ as follows: 
		\begin{equation}
			\label{eq:tight-bound}
			p_{\mathcal K}(z|x) = \frac{\bound(u,v,\epsilon)}{(e^\epsilon)^d}
		\end{equation}
		
	\noindent where $d$ is the distance between $x$ and $z$ in $G_\mathcal{X}$.
	
	We need to show that $p_{\mathcal{K}}(\cdot |x)$ is a probability distribution for every $x$:
	\begin{align*}
		\sum_{z \in \calz} \frac{\bound(u,v,\epsilon)}{(e^\epsilon)^d} & = \bound(u,v,\epsilon) \sum_{z \in \calz} \frac{1}{(e^\epsilon)^d} & \\
																																	 & = \bound(u,v,\epsilon) \sum_{d} \frac{n_d}{(e^\epsilon)^d} & \\
																																	 & = \bound(u,v,\epsilon) \frac{1}{\textstyle{\max^M}} & \text{by Proposition~\ref{prop:bound-max-M}}\\
																																	 & = \bound(u,v,\epsilon) \frac{1}{\bound(u,v,\epsilon)} & \text{take $d=0$ in \eqref{eq:tight-bound}}\\
		                                                    					 & = 1 
	\end{align*} 
	
	To see that $\mathcal K$ provides $\epsilon$-differential privacy, just take $d=1$ in \eqref{eq:tight-bound}, and to see that $I_\infty(X;Z)=\bound(u,v,\epsilon)$ take $d=0$ in the same equation.
	
\end{proof}

We now \revision{give} an example of the use of $\bound(u,v,\epsilon)$ as a bound for the min-entropy leakage.

\begin{example}
	\label{exa:eyes}
	Assume that we are interested in the \revision{eye color} of a certain population $\ind=\{\textit{Alice}, \textit{Bob}\}$. Let $\val=\{\mathtt{a},\mathtt{b},\mathtt{c}\}$ where $\mathtt{a}$ stands for $\mathit{absent}$ (i.e. the \emph{null} value), $\mathtt{b}$ stands for $\mathit{blue}$, and $\mathtt{c}$ stands for $\mathit{coal black}$. We can represent each dataset as a tuple $d_0 d_1$, where $d_0\in \val$ represents the \revision{eye color} of  $\textit{Alice}$ (cases $d_0=b$ and $d_0=c$), or that $\textit{Alice}$ is not in the dataset (case $d_0=a$). $d_1$ provides the same kind of information for $\mathit{Bob}$. Note that $v=3$. Fig~\ref{fig:eyes.a} represents the set $\mathcal X$ of all possible datasets and its adjacency relation. Fig~\ref{fig:eyes.b} represents the matrix with input $\mathcal X$ which provides $\epsilon$-differential privacy and has the highest min-entropy leakage. In the representation of the matrix, the generic entry $\alpha$ stands for $\frac{\max^M}{e^{\epsilon \, \alpha}}$, where $\max^M$ is the highest value in the matrix, i.e. $\max^M= \frac{e^\epsilon}{(v-1 +  e^\epsilon)} =  \frac{e^\epsilon}{(2 +  e^\epsilon)}$.

	\begin{figure}[!htb]%
		\centering
		\subbottom[The datasets and their adjacency relation]{
			\centering
			\includegraphics[width=0.3\linewidth]{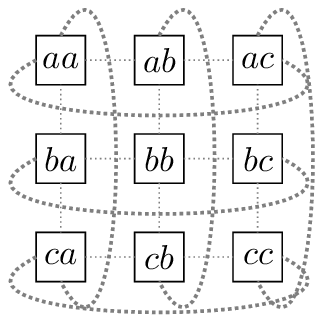}%
			\label{fig:eyes.a}%
		} \hspace{2cm}
		\subbottom[The representation of the matrix]{
			\centering
			\includegraphics[width=0.3\linewidth]{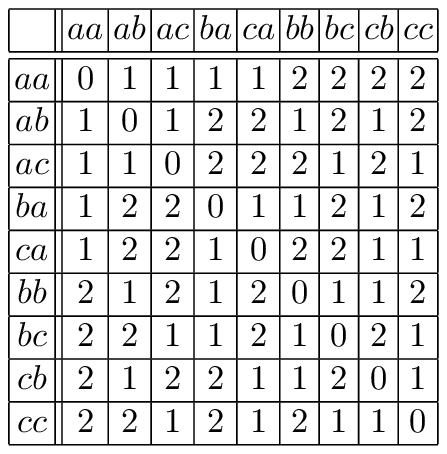}%
			\label{fig:eyes.b}%
		}
		\caption{Universe and highest min-entropy leakage matrix giving $\epsilon$-differential privacy for Example~\ref{exa:eyes}.}%
	\end{figure}

\end{example}
	
Note that the bound $\bound(u,v,\epsilon)$ is guaranteed to be reached with the uniform input distribution. The construction of the matrix for Proposition~\ref{prop:tight} gives a square matrix of dimension $\val^{u} \times \val^{u}$. Often, however, the range of $\mathcal K$ is fixed, as it is usually related to the possible answers to the query $f$. Hence it is natural to consider the scenario in which we are given a number $r<\val^{u}$, and want to consider only those  $\mathcal K$'s whose range has  cardinality  at most $r$. Proposition~\ref{prop:new-bound} shows that in n this restricted setting we can find a better bound than the one given by Theorem~\ref{theo:bound-leakage}. But first we need the following lemma.

\begin{lemma}
	\label{lemma:new-bound}
	Let $\mathcal K$ be a randomized function with input $X$, where $\mathcal{X}=\val^{u}$, providing $\epsilon$-differential privacy. Assume that $r = |\mathit{Range}({\mathcal{K})}|=v^\ell$, for some $\ell<u$. Let $M$ be the matrix associated to $\mathcal K$. Then it is possible to build a square matrix $M'$ of size $v^\ell\times v^\ell$, with row and column indices in $\mathcal{A}\subseteq\mathcal{X}$, and a binary relation $\sim' \subseteq\mathcal{A}\times \mathcal{A}$ such that $(\mathcal{A}, \sim')$ is isomorphic to $(\val^\ell, \sim_\ell)$, and such that:

	\begin{enumerate}[(i)]
		\item \label{item:nb-a} $M'$ is a valid channel matrix: \revision{$\sum_{j=0}^{m-1} M'_{i,j} = 1$ for all $0 \leq i \leq n-1$};		
		\item \label{item:nb-b} $M'_{i,j}\leq (e^\epsilon)^{u-l + d}\, M'_{h,j}$ for all $i,h \in \mathcal{X}$ and $j \in \caly $, where $d$ is the $\sim'$-distance between $i$ and $h$;
	  \item \label{item:nb-c} The elements of the diagonal are all equal to the maximum element of the matrix: $M'_{i,i} = \max^{M'}$ for all $i \in \mathcal{X}$; 
	  \item \label{item:nb-d} $H^{M'}_\infty(X|Y) = H^{M}_\infty(X|Y)$, if $X$ has the uniform distribution.
	\end{enumerate}
			
\end{lemma}

\begin{proof}
We first apply a procedure similar to that of Lemma~\ref{lemma:transform-diagonal} to construct a square matrix of size $v^\ell\times v^\ell$ which has the maximum values of each column in the diagonal. (In this case we construct an injection from the columns to rows containing their maximum value, and we eliminate the rows that at the end are not associated to any column.) Then  define $\sim'$ as the projection of $\sim_u$ on $\val^\ell$. It is easy to see that condition (\ref{item:nb-b}) in is satisfied by this definition of $\sim'$. Finally, apply the procedure in Lemma~\ref{lemma:transform-dist-reg}, or equivalently the procedure in Lemma~\ref{lemma:transform-vtt}, on the structure $(\mathcal{X}, \sim')$ to make all elements in the diagonal equal to the maximum element of the matrix (condition (\ref{item:nb-c})). Note that this procedure preserves the property of condition (\ref{item:nb-b}), and conditional min-entropy ((\ref{item:nb-d})). Also the matrix obtained is a valid channel matrix (condition (\ref{item:nb-a})).
\end{proof}

Now we are ready to prove the proposition.

\begin{proposition}
	\label{prop:new-bound}
	Let $\mathcal K$ be a randomized function with associated channel matrix $M$, and let $r = |\mathit{Range}({\mathcal{K})}|$. If $\mathcal{K}$ provides \emph{$\epsilon$-differential privacy} then the min-entropy leakage associated to $\mathcal{K}$ is bounded from above as follows: 
	
	\begin{equation*}
		I_\infty^{M}(X;Z) \, \leq \,\log_2\frac{r\,(e^{\epsilon})^u}{(v-1+e^\epsilon)^\ell-(e^{\epsilon})^\ell+(e^{\epsilon})^u}
	\end{equation*} 
	
	\noindent where $\ell=\lfloor \log_v r\rfloor$.
\end{proposition}

\begin{proof}
	Assume first that $r$ is of the form $v^\ell$. We transform the matrix $M$ associated to ${\mathcal{K}}$ by applying Lemma~\ref{lemma:new-bound}, and let $M'$ be the resulting matrix. Let us denote by $\textstyle{\max^{M'}}$ the value of every element in the diagonal of $M'$, i.e. $\textstyle{\max^{M'}} = M'_{i,i}$ for every row $i$, and let us denote by  $\border{\mathcal{A}'}{d}{i}$ the set of elements whose $\sim'$-distance from $i$ is $d$. Note that for every  $j \in \border{\mathcal{A}'}{d}{i}$ we have that $ M'_{j,j}\leq {M'_{i,j}}{(e^{ \epsilon})^{u-\ell+d}}$, hence 
	\begin{equation*}
		M'_{i,j} \geq \frac{\textstyle{\max^M}}{(e^{ \epsilon})^{u-\ell+d}}
	\end{equation*}
	
	Furthermore each element $j$ at $\sim'$-distance $d$ from $i$ can be obtained by changing the value of $d$ individuals in the $\ell$-tuple representing $i$ (remember that $(\mathcal{A},\sim')$ is isomorphic to $(\val^\ell, \sim_\ell)$). We can choose those $d$ individuals in $\binom{\ell}{d}$ possible ways, and for each of these individuals we can change the value (with respect to the one in $i$) in $v-1$ possible ways. Therefore 
	\begin{equation*}
		|\border{\mathcal{A}'}{d}{i}| = \binom{\ell}{d} (v-1)^{d}
	\end{equation*}

	Taking into account that for $M'_{i,i}$ we do not need to divide by $(e^{ \epsilon})^{u-\ell+d}$, we obtain:
	\begin{equation*}
		\textstyle{\max^M} +\sum_{d=1}^{\ell} \binom{\ell}{d} (v-1)^{d} \frac{\textstyle{\max^M}}{(e^{ \epsilon})^{u-\ell+d}} \, \leq \, \sum_{j} M'_{i,j} 
	\end{equation*}

	Since each row represents a probability distribution, the elements of row $i$ must sum up to $1$. Hence:
	\begin{equation}
		\label{eq:bl1}
		\textstyle{\max^M} + \sum_{d=1}^{u} \binom{u}{d} (v-1)^{d} \frac{\textstyle{\max^M}}{(e^{ \epsilon})^{u-\ell+d}} \, \leq \, 1
	\end{equation}
	
	By performing some simple calculations, similar to those of the proof of Theorem~\ref{theo:bound-leakage}, we obtain:
	\begin{equation*}
		\textstyle{\max^M} \, \leq \,\frac{(e^{\epsilon})^u}{(v-1+e^\epsilon)^\ell-(e^{\epsilon})^\ell+(e^{\epsilon})^u}
	\end{equation*}
	
	Therefore:	
	\begin{align}
		\label{eq:bl2}
		I_{\infty}^{M'}(X;Z) & =     H_{\infty}(X) - H_{\infty}^{M'}(X|Z) & \quad& \text{(by definition)} \\
											   & =     \log_{2} v^{u} + \log_2 \sum_{j=1}^{v^\ell} \textstyle{\max^M} \frac{1}{v^{u}} & \\
											   & =     \log_{2} v^{u} + \log_2\frac{1}{v^{u}} + \log_2 (v^\ell\,\textstyle{\max^M}) & \\
											   & \leq  \log_2 \frac{v^\ell\,(e^{\epsilon})^u}{(v-1+e^\epsilon)^\ell-(e^{\epsilon})^\ell+(e^{\epsilon})^u} && \text{(by \eqref{eq:bl1} )}
	\end{align}

Consider now the case in which $r$ is not of the form $v^\ell$. Let $\ell$ be the maximum integer such that $v^\ell < r$, and let $m = r-v^\ell$. We transform the matrix $M$ associated to ${\mathcal{K}}$ by collapsing the $m$ columns with the smallest maxima into the $m$ columns with highest maxima. Namely, let $j_1,j_2,\ldots,j_m$ the indices of the columns which have smallest maxima values, i.e. $\maxj{M}{j_t}\leq\maxj{M}{j}$ for every column $j\neq j_1,j_2,\ldots,j_m$. Similarly, let $k_1,k_2,\ldots,k_m$ be the indexes of the columns which have maxima values. Then, define
	\begin{equation*}
		N = M[j_1\rightarrow k_1] 	[j_2\rightarrow k_2] 	\ldots [j_m\rightarrow k_m]
	\end{equation*}
	
	Finally, eliminate the $m$ zero-ed columns to obtain a matrix with exactly $v^\ell$ columns. It is easy to show that 
	\begin{equation*}
		I_\infty^M(X;Z)\, \leq \, I_\infty^N(X;Z) \frac{r}{v^\ell}
	\end{equation*}

	After transforming $N$ into a matrix $M'$ with the same min-entropy leakage as described in the first part of this proof, from \eqref{eq:bl2} we conclude
	\begin{equation*}
		I_\infty^M(X;Z)\, \leq \, I_\infty^{M'}(X;Z) \frac{r}{v^\ell} \, \leq \, \displaystyle \log_2 \frac{r\,(e^{\epsilon})^u}{(v-1+e^\epsilon)^\ell-(e^{\epsilon})^\ell+(e^{\epsilon})^u} 
	\end{equation*}

\end{proof}

Note that this bound can be much smaller than the one provided by Theorem~\ref{theo:bound-leakage}. For instance, if $r=v$  this bound becomes:

\begin{equation*} 
	\log_2 \frac{v\,(e^{\epsilon})^u}{v-1+(e^{\epsilon})^u} 
\end{equation*}

\noindent which  for large  values of $u$ is much smaller than $\bound(u,v,\epsilon)$. 

Let us clarify that there is no contradiction with the fact that the bound $\bound(u,v,\epsilon)$ is strict: in fact it is strict when we are free to choose the range, but here we fix the dimension of the range.

\subsection{Measuring the leakage about an individual}
\label{section:individual-leakage}

As discussed in Section~\ref{section:model}, the main goal of differential privacy is not to protect information about the complete database, but about each of its individual participants. To capture the leakage about a particular individual, we start from a tuple $x^- \in \val^{u-1}$ containing the given (and known) values of all other $u-1$ individuals. Then we create a channel whose input $V$ ranges over the values in $\val$ and represents the value of our individual of interest. Note that this means that we take into consideration all possible input databases where the values of the other individuals are exactly those of $x^-$ and only the value of the selected individual varies. Intuitively, $I_{\infty}^{x^{-}}(V;Z)$ measures the leakage about the individual's value where all other values are known to be as in $x^-$. (Similarly, $H_{\infty}^{x^{-}}(V|Z)$ represents the conditional entropy of $V$ given $Z$ for a fixed database where all other values are $x^-$.) As all these databases are adjacent, differential privacy provides a stronger bound for this leakage.

Therefore, the \emph{leakage for a single individual} can be characterized as follows. 

\begin{proposition}
	\label{prop:ind}
	Assume that $\mathcal{K}$ satisfies $\epsilon$-differential privacy. Then the information leakage for an individual is bound from above by: 
	\begin{equation*}
		I^{x^{-}}_\infty(V;B)\leq \log_2\frac{v\,e^{\epsilon}}{v-1+e^\epsilon}
	\end{equation*}
\end{proposition}

\begin{proof}
	Let us fix a database $x$, and a particular individual $i$ in $\ind$. The possible ways in which we can change the value of $i$ in $x$ are $v-1$. All the new databases obtained in this way are adjacent to each other, i.e. the graph structure associated to the input is a clique of $v$ nodes. Recall that $n_d$ is the number of elements of the input at distance $d$ from a given element $x$. In this case we have
	\begin{equation*}
		n_d = 
			\begin{cases}
	 			1 & \text{for $d=0$,} \\
				v-1 & \text{for $d=1$,} \\
				0 & \text{otherwise.} \\
			\end{cases}
	\end{equation*}

	By substituting this value of $n_d$ in Theorem~\ref{theo:bound-cond-min-entropy}, we get
	\begin{align*}
		H^{x^{-}}_\infty(V|Z) & \geq - \log_2 \frac{1}{\displaystyle 1 + \frac{v-1}{e^{\epsilon}}} & \\
															& = - \log_2\frac{e^\epsilon}{v-1 + e^\epsilon} & \\
	\end{align*}

	The particular individual can present $v$ different values, and thus in the case the input distribution is uniform its min-entropy is $H_{\infty}^{x^{-}}(V) = \log_2 v$.
	\begin{align*}
		I_{\infty}^{x^{-}}(V;Z) & = H_{\infty}^{x^{-}}(V) - H_{\infty}^{x^{-}}(V|Y) & \text{(by definition)} \\
														       & = \log_2 v + \log_2\frac{e^\epsilon}{v-1 + e^\epsilon}  & \text{(by the derivations above)} \\
															     & = \log_2\frac{v \, e^\epsilon}{v-1 + e^\epsilon}  & \text{} \\
	\end{align*}
	
	Since the min-entropy leakage is maximum in the case of the uniform input distribution, the result follows.

\end{proof}

Note that the bound on the leakage for an individual does not depend on the size $u$ of $\ind$, nor on the database $x^{-}$ that we fix.
			
\section{Application to utility}
\label{section:utility}

As discussed in Section~\ref{section:model}, the utility of a randomized function $\mathcal{K}$ is the correlation between the real answers $Y$ for a query and the reported answers $Z$. 
	
For our analysis we assume an oblivious randomization mechanism. As discussed in Section~\ref{section:model}, in this case the system can be decomposed into \revision{the cascade of} two channels, and the utility becomes a property of the channel associated to the randomization mechanism $\calh$ which maps the real answer $y \in \mathcal{Y}$ into a reported answer $z \in \mathcal{Z}$ according to given probability distributions $p_{Z|Y}(\cdot|\cdot)$. The user, however, does not necessarily take $z$ as her guess for the real answer, since she can use some Bayesian post-processing to maximize the probability of success, i.e. a right guess. Thus for each reported answer $z$ the user can remap her guess to a value $y' \in \mathcal{Y}$ according to some strategy that maximizes her expected gain. 

The standard way to define utility is by means of $\gain$ functions (see for instance \cite{Bernardo:94:BOOK}). We define $\mathit{gain}:\caly\times\caly\rightarrow \mathbb{R}$ and the value $\mathit{gain}(y,y')$ represents the reward for guessing the answer $y'$ when the correct answer is $y$. 

It is natural to define the global utility of the mechanism $\mathcal H$ as the expected gain:

\begin{equation}
	\mathcal{U}(Y,Z) = \displaystyle \sum_{y} p(y) \sum_{y'} p(y'|y) \gain(y,y') 
	\label{eq:gain}
\end{equation}

\noindent where $p(y)$ is the prior probability of real answer $y$, and $p(y'|y)$ is the probability of the user guessing $y'$ when the real answer is $y$.
	
Assuming that the user uses a remapping function $\guess : \mathcal{Z} \rightarrow \mathcal{Y}$, we can derive the following characterization of the utility. Recall that  $\delta_x(\cdot)$ represents the probability distribution which has value $1$ on $x$ and $0$ elsewhere.
\begin{align}		
	\label{eq:utility-gain}
	\nonumber \mathcal{U}(Y,Z) & = \sum_{y} p(y) \sum_{y'} p(y'|y) \gain(y,y') & \text{(by \eqref{eq:gain})} \\
	\nonumber	     					   & = \sum_{y} p(y) \sum_{y'} \left( \sum_{z} p(z|y)p(y'|z) \right) \gain(y,y') & \text{} \\
	\nonumber	     				  	 & = \sum_{y} p(y) \sum_{y'} \left( \sum_{z} p(z|y) \delta_{y'}(\guess(z)) \right) \gain(y,y') & \text{($y' = \guess(z)$)} \\
	\nonumber	     			  		 & = \sum_{y} p(y) \sum_{z} p(z|y) \sum_{y'} \delta_{y'}(\guess(z))\gain(y,y') & \text{} \\
	\nonumber	     		  			 & = \sum_{y,z} p(y,z) \sum_{y'} \delta_{y'}(\guess(z))\gain(y,y') & \text{} \\
		     					  				 & = \sum_{y,z} p(y,z) \gain(y,\guess(z)) & \text{}
\end{align}
 
We focus here on the so-called \emph{binary} gain function, which is defined as
\begin{equation*}
	\bingain(y,y') = 
		\begin{cases}
				1 & \text{if } y = y',\\[1mm]
				0 & \text{otherwise.}
		\end{cases}
\end{equation*}

Note that in the above equation the value $y'$ represents the user's guess after the observed answer $z$. Therefore we have
\begin{equation*}
	\bingain = \delta_{y}(\guess(z))
\end{equation*}

This kind of function represents the case in which there is no reason to prefer \revision{one answer over another}, except if it is the \emph{correct} answer. More precisely, we obtain some gain if and only if we guess the right answer. Note that if the answer domain is equipped with a notion of distance (i.e. even if two answers are wrong, one of them may be \qm{closer} to the correct one than the other) then the gain function could take into account the proximity of the reported answer to the real one. In this case a \qm{close} answer, even if wrong, is considered better than a distant one. We do not assume here a notion of distance, and therefore we will focus on the binary case. The use of binary gain functions in the context of differential privacy was also investigated in \cite{Ghosh:09:STC}\footnote{The authors of \cite{Ghosh:09:STC} used the dual notion of \emph{loss functions} instead of gain functions, but the final result is equivalent.}.

By substituting $\gain$ with $\bingain$ in \eqref{eq:utility-gain} we obtain:
\begin{equation}
	\mathcal{U}(Y,Z) = \displaystyle \sum_{y,z} p(y,z) \delta_{y}(\guess(z))
	\label{eq:utility-function}
\end{equation}
	
\noindent which tells us that the expected utility is the greatest when $\guess(z) = y$ is chosen to maximize $p(y,z)$. Assuming that the user chooses such a maximizing remapping, we have:
	\begin{align}
		\label{eq:utility-renyi}
		\nonumber \mathcal{U}(Y,Z) & = \sum_{z} \max_{y} p(y,z) & \text{} \\
															 & = \sum_{z} \max_{y} (p(y) \, p(z|y)) & \text{(by the Bayes law)}	
	\end{align}
	
If the gain function is binary, and the function $\mathit{guess}$ is chosen to optimize utility (i.e. it represents the user's best strategy), then there is a well-known correspondence between $\calu$ and the Bayes risk / the a posteriori min-entropy. \revision{This} correspondence is expressed by the following proposition:

\begin{proposition}
	\label{prop:utility-entropy}
	Assume that function $\mathit{gain}$ is binary and the function $\mathit {guess}$ is optimal. Then:
	\begin{equation*}
		\calu(Y,Z) = \sum_{z} \max_y (p(y) \, p(z|y)) = 2^{-H_\infty(Y|Z)}
	\end{equation*}
\end{proposition}

\begin{proof}
	Just substitute \eqref{eq:utility-renyi} in the definition of conditional min-entropy: $H_\infty(Z \mid Y) = -\log_2 \sum_{z} \max_y ((p(y) \, p(z|y))$.
\end{proof}

\subsection{The bound on the utility}
\label{section:utility-bound}

In this section we show that\revision{,in some special cases,} the fact that $\mathcal K$ provides $\epsilon$-differential privacy induces a bound on the utility as defined in terms of a binary gain function. We start by extending the adjacency relation $\sim$ from the datasets $\calx$ to the real answers $\mathcal Y$\revision{, in such a way that two values in $\caly$ are adjacent if they have pre-images that are adjacent}. Intuitively, the function $f$ associated to the query determines a partition on the set of all databases ($\mathcal X$, i.e. $\val^{u}$), and we say that two classes are adjacent if they contain an adjacent pair. More formally: 

\begin{definition}
	\label{def:neighbor-y}
	Given $y,y'\in\mathcal{Y}$, with $y\neq y'$, we say that $y$ and $y'$ are adjacent (notation $y\sim y'$), if and only if there exist $x,x'\in \val^{u}$ with $x \sim x'$ such that  $ y = f(x)$ and  $y'=f(x')$.
\end{definition}

Since $\sim$ is symmetric on databases, it is also symmetric on $\caly$, therefore also $(\mathcal{Y},\sim)$ forms an undirected graph. 

\revision{Using the above concept of neighborhood for the inputs of the randomization mechanism $\mathcal{H}$, we can show that in an oblivious mechanisms (see Figure~\ref{fig:utility-privacy}) if the query $f$ is deterministic, then the randomized function $\mathcal{K}$ provides $\epsilon$-differential privacy with respect to neighbor databases if and only if $\mathcal{H}$ respects $\epsilon$-differential privacy with respect to neighbor answers. Intuitively, this result follows from the fact that a deterministic query $f$ remaps every database $x \in \mathcal{X}$ to a sole answer $y \in \mathcal{Y}$, working as a sort of \qm{relabeling} that substitutes databases for answers in the adjacency graph structure, and therefore preserving $\epsilon$-differential privacy. Note also that if $\mathcal{K}$ is oblivious, the probability of any reported answer $z \in \mathcal{Z}$ does not depend on the database, but solely on the real answer $y$. Therefore under a deterministic $f$, two databases $x$ and $x'$ can be mapped to same value of $y$ only if, for all $z$, $\mathcal{K}(z|x) = \mathcal{K}(z|x')$.

\begin{proposition}
	\label{prop:k-h}
	If the query function $f$ is deterministic, then the randomized function $\mathcal{K}$ satisfies $\epsilon$-differential privacy with respect to every pair of neighbor databases $x,x' \in \mathcal{X}$ if and only if the randomization mechanism $\mathcal{H}$ satisfies $\epsilon$-differential privacy with respect to every pair of neighbor answers $y,y' \in \mathcal{Y}$. 
\end{proposition}	

\begin{proof}

	Since the matrix $\mathcal{K}$ can be obtained by the product of the two matrices corresponding to $f$ and $\mathcal{H}$, we can derive that, for every pair of neighbor databases $x$ and $x'$ and for all reported answer $z$:
	
	\begin{align*}
		\frac{\mathcal{K}(z|x)}{\mathcal{K}(z|x')} & = \frac{Pr[Z=z|X=x]}{Pr[Z=z|X=x']} & \text{} \\[2mm]
																							 & = \frac{\sum_{y}Pr[Y=y|X=x]Pr[Z=z|Y=y]}{\sum_{y}Pr[Y=y|X=x']Pr[Z=z|Y=y]} & \text{(matrix multiplication)} \\[2mm]
																							 & = \frac{\sum_{y} \delta_{f(x)}(y) Pr[Z=z|Y=y]}{\sum_{y} \delta_{f(x')}(y) Pr[Z=z|Y=y]} & \text{(since $f$ is deterministic)}	\\[2mm]
																							 & = \frac{Pr[Z=z|Y=f(x)]}{Pr[Z=z|Y=f(x')]} & \text{(applying the Dirac $\delta$)} \\[2mm]
																							 & = \frac{\mathcal{H}(z|f(x)]}{\mathcal{H}(z|f(x')]}	& \text{} \\
	\end{align*}
	
	Therefore it follows immediately that $\frac{\mathcal{K}(z|x)}{\mathcal{K}(z|x')} \leq e^\epsilon$ if and only if $\frac{\mathcal{H}(z|f(x))}{\mathcal{H}(z|f(x'))} \leq e^\epsilon$.
	
\end{proof}

The link the above proposition establishes between the randomized function $\mathcal{K}$ and the randomization mechanism $\mathcal{H}$ will help us find determine a bound on the utility of $\mathcal{H}$, since, in the case the query $f$ is deterministic, requiring $\mathcal{K}$ to respect $\epsilon$-differential privacy is equivalent to requiring that $\mathcal{H}$ does.
}

\begin{theorem}
	\label{theo:util}
	Consider a randomized mechanism $\mathcal{H}$, and let $y$ be an element of $\caly$. \revision{Assume that the distribution of $Y$ is uniform and that} $({\caly}, \sim)$ is either distance-regular or $\vtt$ and that $\mathcal{H}$ satisfies $\epsilon$-differential privacy. For each distance $d \in \{0, 1, \ldots, \diameter \}$, where $\diameter$ is the diameter of $({\caly}, \sim)$, we have that:
	\begin{equation}\label{eq:util}
		\mathcal{U}(Y,Z) \leq \frac{1}{\displaystyle \sum_{d} \frac{n_d}{e^{\epsilon \,d}}} 
	\end{equation}
	
	\noindent where $n_d$ is the number of nodes $y' \in \caly$ at distance $d$ from $y$. 

\end{theorem}
	
\begin{proof}
	Since $({\caly}, \sim)$ is distance-regular or $\vtt$, we can apply Theorem~\ref{theo:bound-cond-min-entropy} to derive that $H_{\infty}^{M}(Z|Y) \geq -\log_2 \frac{1}{\sum_{d} \frac{n_d}{e^{\epsilon \,d }}}$. Then we just substitute this result in Proposition~\ref{prop:utility-entropy}.
\end{proof}

The above bound is tight, in the sense that (provided $(\caly,\sim)$ is distance-regular or $\vtt$) we can construct a mechanism $\mathcal{H}$ which satisfies \eqref{eq:util} with equality. More precisely, for $0 \leq i \leq n-1$ and $0 \leq j \leq n-1$, we define $\mathcal{H}$ (here identified with its channel matrix for simplicity) as follows:
\begin{equation}
	\label{eq:construction-h}
	\mathcal{H}_{i,j}= \frac{\gamma}{\displaystyle e^{\epsilon \,d(i,j)}}
\end{equation}

\noindent where 
\begin{equation}
	\label{eq:gamma}
	\gamma = \frac{1}{\displaystyle \sum_{d} \frac{n_d}{e^{\epsilon \,d}}}
\end{equation}

Note that $\mathcal{H}$ is a square matrix of dimension $n \times n$, where $n = |\calx|$. This is not a problem because since we assume $(\caly,\sim)$ to be either distance-regular or $\vtt$, via Theorem~\ref{theo:bound-cond-min-entropy} we can transform the channel matrix into an equivalent one such that all non zero elements are in the submatrix of dimensions $n \times n$. Let us introduce now $\mathcal{Z}^{*} = \{0, 1, \ldots, n-1\}$, i.e. the subset of $\mathcal{Z}$ that excludes the zero-ed columns of the channel matrix from $n$ to $m-1$. Note that for the following result we can safely use the set $\mathcal{Z}^*$ instead of $\mathcal{Z}$ because the zero-ed columns do not contribute to the a posteriori entropy, and trivially respect $\epsilon$-differential privacy.

\begin{theorem}
	\label{theorem:construction-h}
	Assume $({\caly}, \sim)$ is distance-regular or $\vtt$ \revision{and that the distribution of $Y$ is uniform}. Then the matrix $\mathcal{H}$ defined in \eqref{eq:construction-h} satisfies $\epsilon$-differential privacy and has maximal utility: 
	\begin{equation*}
		\label{eq:maxutil}
		\mathcal{U}(Y,Z) = \frac{1}{\displaystyle \sum_{d} \frac{n_d}{e^{\epsilon \,d}}} 
	\end{equation*}
	
\end{theorem}

\begin{proof}
		
	First we prove that the matrix as defined in \eqref{eq:construction-h} is a channel matrix, i.e. that each row is a probability distribution.
	\begin{align*}
		\sum_{j \in \calz^*} \mathcal{H}_{i,j} & = \sum_{j \in \calz^*} \frac{\gamma}{e^{\epsilon d(i,j)}} & \\
															 					   & = \gamma \sum_{j \in \calz^*} \frac{1}{e^{\epsilon d(i,j)}} & \\			
															 					   & = \gamma \sum_{d} \frac{n_d}{e^{\epsilon d}} & \text{by \eqref{eq:gamma}} \\			
																				   & = \gamma \frac{1}{\gamma} & \\			
																				   & = 1 & 
	\end{align*}
	
	Now we show that the utility is maximum.
	\begin{align*}
		\mathcal{U}(Y,Z) & = \sum_{z \in \calz^*} \max_{y} (p(y) \, \mathcal{H}(z|y) ) & \text{by \eqref{eq:utility-renyi}} \\
										 & = \sum_{z \in \calz^*} \max_{y} \frac{1}{|\caly|} \mathcal{H}(z|y) & \text{since $Y$ is uniform} \\
										 & = \frac{1}{|\caly|} \sum_{z \in \calz^*} \max_{y} \frac{\gamma}{\max_{d} e^{\epsilon d(i,j)} }  & \text{by \eqref{eq:construction-h}} \\
										 & = \frac{1}{|\caly|} \sum_{z \in \calz^*} \gamma  & \text{maximum is $d=0$} \\
										 & = \frac{1}{|\caly|} \cdot \revision{|\calz^*|} \gamma  &  \\
										 & = \gamma & \text{since $|\caly| = |\calz^*| = n$}
	\end{align*}

\end{proof}

Therefore we can always define $\mathcal{H}$ as in \eqref{eq:construction-h}: the matrix so defined will be a legal channel matrix, and it will satisfy $\epsilon$-differential privacy. If $({\caly}, \sim)$ is neither distance-regular nor $\vtt$, then the utility of such $\mathcal{H}$ is not necessarily optimal.

The conditions for the construction of the optimal matrix are strong, but there are some interesting scenarios in which they are satisfied. Depending on the degree of connectivity $c$ of the graph $(\caly,\sim)$, we can have $\lfloor \frac{|\mathcal{Y}|}{2}\rfloor - 1$ different cases (note that the case of $c=1$ is not possible because the datasets are fully connected via their adjacency relation), whose extremes are:

\begin{itemize}
	
	\item $(\mathcal{Y},\sim)$ is a \emph{clique}, i.e. every element has exactly $|\mathcal{Y}|-1$ adjacent elements. 
	
	\item $(\mathcal{Y},\sim)$ is a \emph{ring}, i.e. every element has exactly two adjacent elements. This is similar to the case of the counting queries considered in  \cite{Ghosh:09:STC}, with the difference that our \qm{counting} is in arithmetic modulo $|\mathcal{Y}|$. 
	
\end{itemize}

\begin{remark}
	\label{rem:relax2}
	Note that our method can be applied also when the conditions of  Theorem~\ref{theorem:construction-h} are not met: We can always  add \qm{artificial} adjacencies to the graph structure \revision{so as to meet} those conditions. Namely, for computing the distance in \eqref{eq:construction-h} we use,  instead of $(\caly,\sim)$, a structure $(\caly,\sim')$ which satisfies the conditions of  Theorem~\ref{theorem:construction-h}, and such that $\sim\, \subseteq \, \sim'$.  Naturally, the matrix constructed in this way provides $\epsilon$-differential privacy, but in general is not optimal. It is clear that, in general, the smaller $\sim'$ is, \revision{the higher} is the utility. 
\end{remark}

The matrices generated by \eqref{eq:construction-h} can be very different, depending on the value of $c$. The next two examples illustrate queries that give rise to the clique and to the ring structures, and show the corresponding matrices. 

\begin{example} Consider a database with electoral information where each entry corresponds to a voter and contains the following three fields:

	\begin{itemize}
		\item \emph{Id}: a unique (anonymized) identifier assigned to each voter;
		\item \emph{City}: the name of the city where the user voted;
		\item \emph{Candidate}: the name of the candidate the user voted for.
	\end{itemize}

Consider the query \emph{\qm{What is the city with the greatest number of votes for a given candidate $\mathit{cand}$?}}. For such a query the binary utility function \revision{could be taken as the natural choice: from the user's point of view, only the right city could give some gain, and all wrong answers would be equally bad}. It is easy to see that every two answers are neighbors, i.e. the graph structure of the answers is a clique. 

Let us consider the scenario where \emph{City} $=\{A,B,C,D,E,F\}$ and assume for simplicity that there is a unique answer for the query, i.e. there are no two cities with exactly the same number of individuals voting for candidate $\mathit{cand}$. Table~\ref{tab:city-cand-geo} shows two alternative mechanisms providing $\epsilon$-differential privacy (with \revision{$\epsilon = \ln 2$}). The first one, $M_{1}$, is based on the truncated geometric mechanism method used in \cite{Ghosh:09:STC} for counting queries (here extended to the case where every two distinct answers are neighbors). The second mechanism, $M_{2}$, is obtained by applying the definition of \eqref{eq:construction-h}. From Theorem~\ref{theorem:construction-h} we know that for the uniform input distribution $M_{2}$ gives optimal utility.

For the uniform input distribution, it is easy to see that $\mathcal{U}(M_1) = 0.2242 < 0.2857 = \mathcal{U}(M_2)$. Even for non-uniform distributions, our mechanism still provides better utility. For instance, for $p(A) = p(F) = 1/10$ and $p(B) = p(C) = p(D) = P(E) = 1/5$, we have $\mathcal{U}(M_1) = 0.2412 < 0.2857 = \mathcal{U}(M_2)$. This is not too surprising: the geometric mechanism, as well as the Laplacian mechanism proposed by Dwork, perform very well when the domain of answers is provided with a metric and the utility function is not binary\footnote{As we mentioned before, in the metric case the gain function can take into account the proximity of the reported answer to the real one, the idea being that a close answer, even if wrong, is better than a distant one.}. It also works well when $(\mathcal{Y}, \sim)$ has low connectivity, in particular in the cases of a ring and of a line. But in this example, we are not in these cases, because we are considering \emph{binary gain functions} and \emph{high connectivity}. 

\begin{table}[tb]
	\centering	
	\small
	\subbottom[$M_{1}$: truncated geometric mechanism]{	
		$
			\begin{array}{|c||c|c|c|c|c|c|}
				\hline
				\text{In/Out} & A & B & C & D & E & F \\ \hline \hline
				A & 0.535 & 0.060 & 0.052 & 0.046 & 0.040 & 0.267 \\ \hline
				B & 0.465 & 0.069 & 0.060 & 0.053 & 0.046 & 0.307 \\ \hline
				C & 0.405 & 0.060 & 0.069 & 0.060 & 0.053 & 0.353 \\ \hline
				D & 0.353 & 0.053 & 0.060 & 0.069 & 0.060 & 0.405 \\ \hline
				E & 0.307 & 0.046 & 0.053 & 0.060 & 0.069 & 0.465 \\ \hline
				F & 0.267 & 0.040 & 0.046 & 0.052 & 0.060 & 0.535 \\ \hline
			\end{array}
		$
		\label{tab:city-cand-geo-a}
	}	\ \ \ \ \ \ 
	\subbottom[$M_{2}$: our mechanism]{		
		$
			\begin{array}{|c||c|c|c|c|c|c|}
			 \hline
			 \text{In/Out} & A & B & C & D & E & F \\ \hline \hline
			 A & 2/7 & 1/7 & 1/7 & 1/7 & 1/7 & 1/7 \\ \hline
			 B & 1/7 & 2/7 & 1/7 & 1/7 & 1/7 & 1/7 \\ \hline
			 C & 1/7 & 1/7 & 2/7 & 1/7 & 1/7 & 1/7 \\ \hline
			 D & 1/7 & 1/7 & 1/7 & 2/7 & 1/7 & 1/7 \\ \hline
			 E & 1/7 & 1/7 & 1/7 & 1/7 & 2/7 & 1/7 \\ \hline
			 F & 1/7 & 1/7 & 1/7 & 1/7 & 1/7 & 2/7 \\ \hline
			\end{array}
		$
		\label{tab:city-cand-geo-b}
	}
	\caption{Mechanisms for the city with higher number of votes for candidate $\mathit{cand}$}
	\label{tab:city-cand-geo}
\end{table}

\end{example}

\begin{example}

Let us consider the same database as the previous example, but now assume a counting query of the form \emph{\qm{What is the number of votes for candidate $\mathit{cand}$?}}. 
It is easy to see that each answer has at most two neighbors. More precisely, the graph structure on  the answers is a line. For illustration purposes, let us assume that only $5$ individuals have participated in the election. Table~\ref{tab:count-geo} shows two alternative mechanisms providing $\epsilon$-differential privacy ($\epsilon = \log 2$): the truncated geometric mechanism  $M_{1}$ proposed in \cite{Ghosh:09:STC} and the mechanism we propose $M_{2}$. Note that in order to apply our method we have first to apply Remark~\ref{rem:relax2} to transform the graph structure from a line into a ring.


\revision{Let} us consider the uniform prior distribution. We see that the utility of $M_1$ is higher than the utility of $M_2$, in fact the first is $4/9$ and the second is $8/21$. This does not contradict our theorem, because our matrix is guaranteed to be optimal only in the case of a ring structure, not a line as we have in this example. If the structure were a ring, i.e. if the last row were adjacent to the first one, then $M_1$ would not provide $\epsilon$-differential privacy. In case of a line as in this example, the truncated geometric mechanism has been proved optimal \cite{Ghosh:09:STC}. 

\begin{table}[!htb]
	\centering
	\small
	\subbottom[$M_{1}$: truncated $\frac{1}{2}$-geom. mechanism]{
		$
			\begin{array}{|c||c|c|c|c|c|c|}
				\hline
				\text{In/Out} & 0    & 1    & 2    & 3    & 4    & 5    \\ \hline \hline
				0                   & 2/3  & 1/6  & 1/12 & 1/24 & 1/48 & 1/48 \\ \hline
				1                   & 1/3  & 1/3  & 1/6  & 1/12 & 1/24 & 1/24 \\ \hline
				2                   & 1/6  & 1/6  & 1/3  & 1/6  & 1/12 & 1/12 \\ \hline
				3                   & 1/12 & 1/12 & 1/6  & 1/3  & 1/6  & 1/6  \\ \hline
				4                   & 1/24 & 1/24 & 1/12 & 1/6  & 1/3  & 1/3  \\ \hline
				5                   & 1/48 & 1/48 & 1/24 & 1/12 & 1/6  & 2/3  \\ \hline
			\end{array}
		$
		\label{tab:count-geo-a}
	}\ \ \ 
	\subbottom[$M_{2}$: our mechanism]{
		$
			\begin{array}{|c||c|c|c|c|c|c|}
			 	\hline
				 \text{In/Out} & 0     & 1    & 2    & 3    & 4    & 5    \\ \hline \hline
				 0             & 8/21  & 4/21 & 2/21 & 1/21 & 2/21 & 4/21 \\ \hline
				 1             & 4/21  & 8/21 & 4/21 & 2/21 & 1/21 & 2/21 \\ \hline
				 2             & 2/21  & 4/21 & 8/21 & 4/21 & 2/21 & 1/21 \\ \hline
				 3             & 1/21  & 2/21 & 4/21 & 8/21 & 4/21 & 2/21  \\ \hline
				 8             & 2/21  & 1/21 & 2/21 & 4/21 & 8/21 & 4/21  \\ \hline
				 5             & 4/21  & 2/21 & 1/21 & 2/21 & 4/21 & 8/21  \\ \hline
			\end{array}
		$
		\label{tab:count-geo-b}
	}
	\caption{Mechanisms for the counting query ($5$ voters)}
	\label{tab:count-geo}	
\end{table}
\end{example}

\section{Related work}
\label{section:related-work-diff-priv}

To the best of our knowledge, the first work to investigate the relation between differential privacy and information-theoretic leakage \emph{for an individual} was \cite{Alvim:10:TechRep}. In this work, the definition of channel was relative to a given database $x$, and the channel inputs were all possible databases adjacent to $x$. Two bounds on leakage were presented, one for the min-entropy, and one for Shannon entropy. Our bound in Proposition~\ref{prop:ind} is an improvement with respect to the (min-entropy) bound in \cite{Alvim:10:TechRep}.

Barthe and K\"opf \cite{Barthe:11:CSF} were the first to investigate the (more challenging) connection between differential privacy and the min-entropy leakage \emph{for the entire universe of possible databases}. They considered the \qm{end-to-end differentially private mechanisms}, which correspond to what we call the randomized function $\mathcal{K}$ in this chapter, and proposed, like we do, to interpret them as information-theoretic channels. They provided a bound for the leakage, but pointed out that it was not tight in general. They also showed that there cannot be a domain-independent bound, by proving that for any number of individuals $u$ the optimal bound must be at least a certain expression $f(u,\epsilon)$. Finally, they showed that the question of providing optimal upper bounds for the leakage of $\epsilon$-differentially private randomized functions in terms of rational functions of $\epsilon$ is decidable, and left the actual function as an open question. In our work we used rather different techniques and found (independently) the same function $f(u,\epsilon)$ (the bound in Theorem~\ref{theo:bound-cond-min-entropy}), but we actually proved that $f(u,\epsilon)$ is the optimal bound\footnote{When discussing our result with Barthe and K\"opf, they said that they also conjectured that $f(u,\epsilon)$ is the optimal bound.}. Another difference between their work and ours is that \cite{Barthe:11:CSF} captures the case in which the focus of differential privacy is on hiding \emph{participation} of individuals in a database, whereas we consider both the participation and the \emph{values} of the participants.

Clarkson and Schneider also considered differential privacy as a case study of their proposal for quantification of integrity \cite{Clarkson:11:TECHREP}. There, the authors analyzed database privacy conditions from the literature (such as differential privacy, $k$-anonymity, and $l$-diversity) using their framework for utility quantification. In particular, they studied the relationship between differential privacy and a notion of leakage (which is different from ours - in particular their definition is based on Shannon entropy) and they provided a tight bound on leakage. 

Heusser and Malacaria \cite{Heusser:09:FAST} were among the first to explore the application of information-theoretic concepts to databases queries. They proposed to model database queries as programs, which allows for \revision{statistical} analysis of the information leaked by the query. \cite{Heusser:09:FAST}, however, did not attempt to relate information leakage to differential privacy.

In \cite{Ghosh:09:STC} the authors aimed at obtaining optimal-utility randomization mechanisms while preserving differential privacy. The authors proposed adding noise to the output of the query according to the geometric mechanism. Their framework is very interesting in the sense it provides a general definition of utility for a mechanism $M$ that captures any possible side information and preference (defined as a loss function) the users of $M$ may have. They proved that the geometric mechanism is optimal in the particular case of counting queries. Our results in Section \ref{section:utility} do not restrict to counting queries, but on the other hand we only consider the case of binary loss function.

\section{Chapter summary and discussion}
\label{section:conclusion-diff-priv}

In this chapter we have investigated the relation between $\epsilon$-differential privacy and leakage, and between $\epsilon$-differential privacy and utility. Our main contribution was the development of a general technique for determining these relations depending on the graph structure of the input domain, induced by the adjacency relation and by the query. We have considered two particular structures, the distance-regular graphs, and the $\vtt$ graphs, which allowed us to obtain tight bounds on the leakage and on the utility. We also constructed an optimal randomization mechanism satisfying $\epsilon$-differential privacy for some special cases.
 
As future work, we plan to extend our result to other kinds of utility functions. In particular, we are interested in the case in which the the answer domain is provided with a metric, and we are interested in taking into account the degree of accuracy of the inferred answer.

\chapter{Safe equivalences for security properties}
\label{chapter:safe-equivalences}
\mscite{Too much may be the equivalent of none at all.}{Lee Loevinger}
In the field of Security, process equivalences have been used to characterize various information-hiding properties (for instance secrecy, anonymity and noninterference) based on the principle that a protocol $P$ with a variable $x$ satisfies such a property if and only if, for every pair of secrets $s_1$ and $s_2$, $P[^{s_1}/ _x]$ is equivalent to $P[^{s_2}/ _x]$. We argue that, in the presence of nondeterminism, the above principle \revision{may rely} on the assumption that the scheduler \qm{works for the benefit of the protocol}, and this usually is not a safe assumption. Non-safe equivalences, in this sense, include complete-trace equivalence and bisimulation. 

The goal of this chapter is to present a formalism in which we can specify admissible schedulers and, correspondingly, safe versions of these equivalences. Then we are able to show that safe equivalences can be used to establish information-hiding properties.

\paragraph{Contribution} The main contributions of this chapter can be summarized as follows.

\begin{itemize}
  
	\item We propose a formalism for concurrent distributed systems  which accounts for both probabilistic and nondeterministic behavior, and in which the latter is of two kinds: \emph{global} and \emph{local}. The global nondeterminism represents the possible interleavings produced by the parallel components, which may be influenced by the attacker. The local nondeterminism is associated to the possible internal choices of each component, which may depend on the secrets or other unknown parameters, not controlled by the attacker. Correspondingly, we split the scheduler into two constituents: a global one and a local one. The latter is actually a tuple of local schedulers, one for each component of the system.
  
	\item We propose a notion of \emph{admissible scheduler} for the above systems, in which the global constituent is not allowed to see the secrets, and each local constituent  is not allowed to see any information about the other components. We then generalize the standard definition of strong (probabilistic) information hiding (such as noninterference and strong anonymity) to the case in which also nondeterminism is present, under the assumption  that the schedulers are admissible.
	
  \item We use admissible schedulers to define safe versions of complete-trace\footnote{In this chapter we may refer to \qm{complete traces} simply as \qm{traces}.} equivalence and bisimilarity which are specially tuned for security. This means that we account for the possibility that the global constituent of the scheduler is in collusion with the attacker,   and therefore does not necessarily help the system to obfuscate the secret. We show that the bisimilarity is still a congruence, \revision{as} in the classical case.

  \item We finally show that our notions of safe complete-trace equivalence and bisimilarity imply strong information hiding in the sense discussed above.
  
\end{itemize}

\paragraph{Plan of the Chapter} This chapter is organized as follows. In Section~\ref{section:equivalences-in-security} we review the role equivalences traditionally play in formalizing security properties. In Section~\ref{section:systems} we formalize the notions of distributed systems and components used in this chapter. In Section~\ref{section:admissible-schedulers} we focus on restricting the discerning power of global and local schedulers, and in Section~\ref{section:equivalences} we present our proposal for safe equivalences, namely safe complete-traces and safe bisimilarity. In Section~\ref{section:nd-ih} we define the notion of information hiding under the novel assumption that nondeterminism is handled partly in a demonic way and partly in an angelic way. Finally, in Section~\ref{section:related-work-equivalences} we review the related bibliography, and in Section~\ref{section:conclusion-safe-equivalences} we summarize the chapter and outline some future work. 

\section{The use of equivalences in security}
\label{section:equivalences-in-security}

As we have seen in Chapter~\ref{chapter:introduction}, one technique used to prevent an attacker of inferring the secret from the observables is to create \emph{noise}, namely to make sure that for every execution in which a given secret produces a certain observable, there is at least another execution in which a different secret produces the same observable. In practice this is often done by using randomization.

In the literature about the foundations  of computer security, however, the quantitative aspects are often abstracted away, and probabilistic  behavior is replaced by nondeterministic behavior. Correspondingly, there have been various approaches in which information-hiding properties are expressed in terms of equivalences based on nondeterminism, especially in a concurrent setting. For instance, \cite{Schneider:96:ESORICS} defines \emph{anonymity} as follows\footnote{The actual definition of \cite{Schneider:96:ESORICS} is more complicated, but the spirit is the same.}: A protocol $S$ is anonymous if, for every pair of culprits $a$ and $b$, $S[^a/ _x]$ and $S[^b/ _x]$ produce the same observable traces. A  similar definition is given in \cite{Abadi:99:IC} for \emph{secrecy}, with the difference that
$S[^a / _x]$ and $S[^b / _x]$ are required to be bisimilar. In \cite{Delaune:09:JCS}, an electoral system $S$ preserves the \emph{confidentiality of the vote} if
for any voters $v$ and $w$, the observable behavior of $S$ is the same if we swap the votes of $v$ and $w$, i.e. if $S[^a / _v\mid ^b / _w]$ is bisimilar to $S[^b / _v \mid ^a / _w]$.

These proposals are based on the implicit assumption that \emph{all the nondeterministic executions present in the  specification of $S$ will always be possible under every implementation of $S$}. Or at least, that the adversary will believe so. In concurrency, however, as argued in \cite{Chatzikokolakis:09:FOSSACS}, nondeterminism has a rather different meaning: if a specification $S$ contains some nondeterministic alternatives, typically it is because we want to abstract from specific implementations, such as the scheduling policy. A specification is considered correct, with respect to some property,  if every alternative satisfies the property. Correspondingly, an implementation is considered correct if all executions are among those possible in the specification, i.e. if the implementation is a refinement of the specification.  There is no expectation that the  implementation will actually make possible all the alternatives indicated by the specification.

We argue that the use of nondeterminism in concurrency corresponds to a \emph{demonic} view: the scheduler, i.e. the entity that will decide which alternative to select,
may try to choose the \qm{worst} alternative. Hence we need to make sure that all alternatives are \qm{good}, in the sense that they satisfy the intended property. In the approaches to formalize security properties mentioned above, on the contrary, the interpretation of  nondeterminism is \emph{angelic}: the scheduler is expected to actually help the protocol to confuse the adversary and thus protect the secret information.

There is another issue, orthogonal to the angelic/demonic dichotomy, but relevant for the achievement of security properties: the scheduler \emph{should not be able to make its choices dependent on the secret}, or else nearly every protocol would be insecure, i.e. the scheduler would always be able to leak the secret to an external observer (for instance by producing different interleavings of the observables, depending on the secret). This remark has been made several times already, and several approaches have  been proposed to cope with the problem of full-information schedulers (aka almighty, omniscient, clairvoyant, etc.), see for example  \cite{Canetti:06:WODES,Canetti:06:DISC,Chatzikokolakis:07:CONCUR,Chatzikokolakis:09:FOSSACS,Andres:10:TechRep}.

The risk of a naive use of nondeterminism to specify a security property is not only that it may rely on an implicit assumption that the scheduler behaves angelically, but also that it is clairvoyant (fully-informed), i.e. that it peeks at the secrets (that it is not supposed to be able to see) to achieve its angelic strategy.

\begin{example}\label{exa:CCS}
Consider the following system, \revision{presented} in a CCS-like syntax: $S \defsym (c) (A \ \parallel \ H_1 \ \parallel \ H_2 \ \parallel\ \mathit{Corr} )$, with $A \defsym \overline{c}\langle \mathit{sec}\rangle$, $H_1 \defsym c(s).\overline{out}\langle a \rangle$, $H_2 \defsym c(s).\overline{out}\langle b \rangle$, $\mathit{Corr} \defsym c(s).\overline{out}\langle s \rangle$. 
The name $\mathit{sec}$ represents a secret.

It is easy to see that we have $S\left[ ^{a}/_{sec}\right]\sim S\left[ ^{b}/_{sec}\right]$, as shown in the execution tress in Figure~\ref{fig:CCSExample}. Note that, in order to simulate the rightmost branch in $S\left[ ^{a}/_{sec}\right]$, the process $S\left[ ^{b}/_{sec}\right]$ needs to follow its leftmost branch. Vice-versa, in order to simulate the rightmost branch in $S\left[ ^{b}/_{sec}\right]$, the process $S\left[ ^{a}/_{sec}\right]$ needs to follow its middle branch. This means that, in order to achieve bisimulation, the scheduler needs to know the secret, and change its choice accordingly.
\end{example}
\begin{figure}[htb]%
	\centering
	\subbottom[$S \mathopen{[} ^{a}/_{sec} \mathclose{]}$]{
		\centering
		\includegraphics[width=0.80\columnwidth]{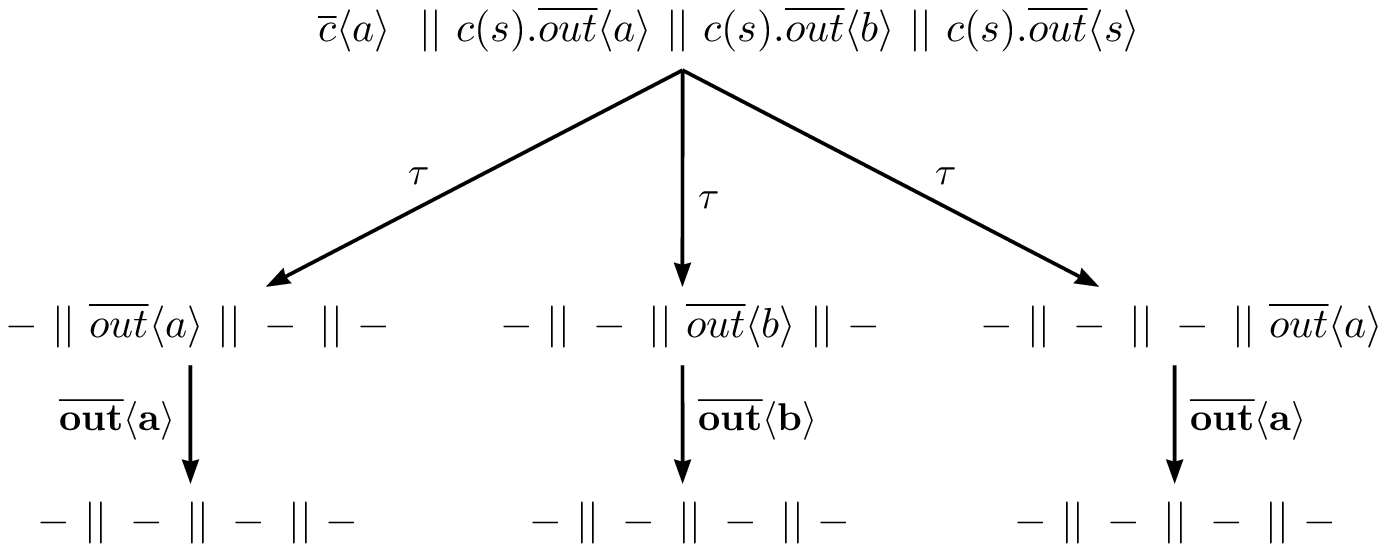}
		\label{fig:CCS-a}
		} \\[5mm]
  	\subbottom[$S \mathopen{[} ^{b}/_{sec} \mathclose{]}$]{  		
  		\centering
  		\includegraphics[width=0.80\columnwidth]{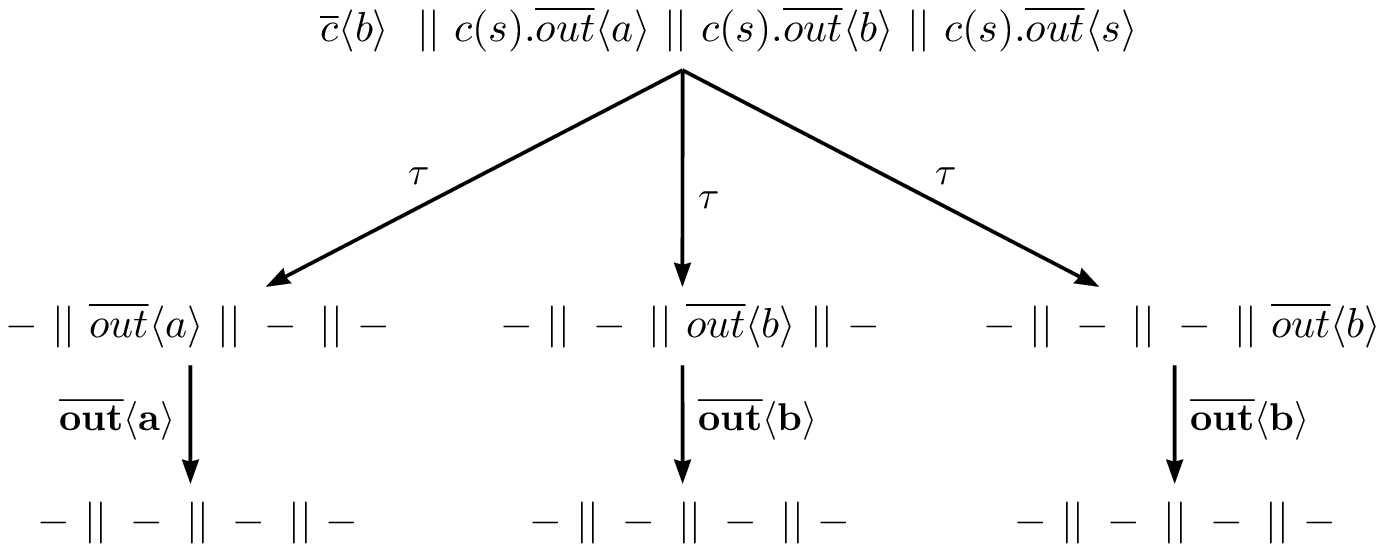}
  		\label{fig:CCS-b}
  	}
	\caption{Execution trees for Example~\ref{exa:CCS}}
	\label{fig:CCSExample}
\end{figure}
This example shows a distributed system that intuitively is not secure, because one of its components, $\mathit{Corr}$, reveals whatever secret it receives. According to the equivalence-based notions of security discussed above, however, \emph{it is secure}. But it is considered secure thanks to a scheduler that:

\begin{enumerate}[(i)]
	\item angelically helps the system to protect the secret; and
	\item does so by making its choices dependent on the secret.
\end{enumerate} 

We consider these assumptions on the scheduler to be excessively strong. 

Here we do not claim, however, that  we should  rule out the use of angelic nondeterminism in security: on the contrary, angelic nondeterminism can be a powerful specification concept. We only advocate a cautious use of this notion. In particular, it should not be used in a context in which the scheduler may be in collusion with the attacker. The goal of this chapter is to define a framework in which we can combine both angelic and demonic nondeterminism in a setting in which also probabilistic behavior may be present, and in a context in which the scheduler is restricted (i.e. not fully-informed). We define  \qm{safe} variant of  typical equivalence relations (complete traces and bisimulation), and we show how to use them to characterize information-hiding properties.

\section{Distributed systems and components}
\label{section:systems}

In this section we describe the kind of distributed systems we are dealing with. We start by introducing a variant of probabilistic automata, that we call \emph{Tagged Probabilistic Automata} (TPA). These systems are parallel compositions of probabilistic processes, called \emph{components}. Each component is equipped with a unique identifier, called \emph{tag}. Whenever a component (or a pair of components in case of synchronization) makes a step, the corresponding transition will be decorated with the associated tag (or pair of tags).

Similar systems have been already introduced in \cite{Andres:10:TechRep}. The main differences are that here the components  may contain nondeterminism

\subsection{Tagged Probabilistic Automata}

We now formalize the notion of TPA.

\begin{definition}
	\label{dfn:tpa}
  	A \emph{Tagged Probabilistic Automaton} (or \emph{TPA}) is a tuple \linebreak $(\Qset,\Tset,\Lset,\qi,\vartheta)$, where $\Qset$ is a set of \emph{states}, $\Tset$ is a set of \emph{tags}, $\Lset$ is a set of \emph{actions}, $\qi \in \Qset$ is the \emph{initial state}, and $\vartheta \colon \Qset \to \mathcal{P}(\Tset \times \Lset \times \distr(\Qset))$ is a \emph{transition function}.
\end{definition}

In the following we write $q\stackrel{t_{g}:a}{\longrightarrow}\mu$ for  $(t_{g},a, \mu) \in \vartheta(q)$, and we use $\avail(q)$ to denote the tags of the components that are enabled to make a transition. More formally:
\begin{equation*}
	\avail(q) \defsym \{t_{g}\in \Tset\mid \mathit{there \ exists}\ a\in \Lset, \mu \in \distr(\Qset)\mathit{ \  such\  that} \ q \stackrel{t_{g}:a}{\longrightarrow}\mu\}
\end{equation*}
In these systems, we can decompose the scheduler into two: a \emph{global scheduler}, which, via tags, decides which component or pair of components makes the next move, and a \emph{local scheduler}, which, also via tags, solves the internal nondeterminism of the selected component.

We  assume that the local scheduler can only select enabled transitions, and that the global scheduler can only select enabled components. This means that the execution does not stop unless all components are blocked. This is in line with the tradition  of process algebra and  of Markov Decision Processes, but contrasts with that of Probabilistic Automata \cite{Segala:95:NJC}. The results in this chapter, however, do not depend on this assumption.

\begin{definition}
	Let $M = (\Qset,\Tset,\Lset,\qi,\vartheta)$ be a TPA. Then:
	\begin{itemize}
		\item  A global scheduler for $M$ is a function $\zeta \colon \fpaths(M) \to (\Tset\cup \{\bot\})$ such that for all finite paths $\sigma$, if $\avail(\last(\sigma))\neq \emptyset$ then $\zeta(\sigma)\in\avail(\last(\sigma))$, and $\zeta(\sigma)= \bot$ otherwise.
		\item  A local scheduler for $M$ is a function $\xi \colon \fpaths(M) \to (\Tset \times \Lset \times \distr(\Qset)\cup \{\bot\})$ such that, for all finite paths $\sigma$, if $\vartheta(\last(\sigma))\neq \emptyset$ then  $\xi(\sigma)\in \vartheta(\last(\sigma))$, and $\xi(\sigma)=\bot$ otherwise.
		\item A global scheduler $\zeta$ and a local scheduler $\xi$ for $M$ are \emph{compatible} if,  for all finite paths $\sigma$, $\xi(\sigma) = (t_{g},a,\mu)$ implies $\zeta(\sigma) = t_{g}$, and $\xi(\sigma) = \bot$ implies $\zeta(\sigma) =\bot$.
		\item A scheduler  is a pair $(\zeta,\xi)$ of compatible global and local schedulers.
	\end{itemize}
\end{definition}

\subsection{Components}

We will use a simple probabilistic process calculus, very close to the \ccsp{} we introduced in Chapter~\ref{chapter:preliminaries}, to specify the components.

We assume a set of \emph{actions} or \emph{channel names}  $\Lset$ with elements $a, a_1,a_2, \cdots$, including the special symbol $\tau$ denoting a \emph{silent step}. Except \revision{for} $\tau$, each action $a$ has a co-action $\bar{a}\in\Lset$ and we assume $\bar{\bar{a}}=a$.
Components are specified by the following grammar:
\begin{equation*}
	q \quad \mbox{::=} \quad 0 \quad  \mid \quad a.q  \quad \mid \quad  q_1 + q_2  \quad \mid \quad  \sum_{i} p_i: q_i   \quad\mid \quad  q_1 | q_2   \quad\mid \quad  (a) q   \quad\mid \quad Q  
\end{equation*}

The constructs $0$, $ a.q$, $q_1 + q_2$, $q_1 | q_2$ and  $(a) q$ represent termination,   prefixing,   nondeterministic choice, parallel composition, and the restriction operator, respectively.  $\sum_{i} p_i: q_i$ is a  probabilistic choice, where $p_i$ represents the probability of the $i$-th branch and must satisfy $0\leq p_i\leq 1$ and
$\sum_{i} p_i=1$. The process call $Q$ is a simple process identifier. For each identifier, we assume a corresponding unique process declaration of the form $\smash{Q \stackrel{\textrm def}{=}q}$. The idea is that, whenever $Q$ is  executed, it triggers the execution of $q$. Note that  $q$ can contain $Q$ or another process identifier,
which means that our language allows (mutual) recursion. We will denote by ${\mathit fn}(q)$ the \emph{free channel names} occurring in $q$, i.e. the channel names not bound by a restriction operator.

\paragraph{Components' semantics:} The operational semantics consists of probabilistic transitions of the form $q{\stackrel{a}\rightarrow}\mu$ where $q\in \Qset$ is a process, $a\in \Lset$ is an action and $\mu\in\distr(\Qset)$ is a distribution on processes. They are specified by the following rules:
$$
\begin{array}{lcl}

	\text{\small PRF } \ \ \begin{tabular}{ c }
	  \\ \vspace{-0.35cm}\\
	  \hline \vspace{-0.35cm}
	  \\
	  $a.q\stackrel{a}\rightarrow \delta_{q}$
	\end{tabular}
	& &
	\text{\small NDT } \ \ \begin{tabular}{ c }
	$q_1\stackrel{a}\rightarrow \mu$
	  \\ \vspace{-0.35cm}\\
	  \hline \vspace{-0.35cm}
	  \\
	  $q_1 + q_2\stackrel{a}\rightarrow \mu$
	\end{tabular}\\[5mm]

	\text{\small PRB}\ \ \begin{tabular}{ c }
	  \\ \vspace{-0.35cm}\\
	  \hline \vspace{-0.35cm}
	  \\
	  $\sum_{i}p_i: q_i\stackrel{\tau}{\rightarrow }\probsum{i} p_i \cdot \delta_{q_i}$
	\end{tabular}
	& \quad &
	\text{\small PAR}\ \ \begin{tabular}{ c }
	$q_1\stackrel{a}{\rightarrow}\mu$
	  \\ \vspace{-0.35cm}\\
	  \hline \vspace{-0.35cm}
	  \\
	  $q_1 \mid q_2 \stackrel{a}{\rightarrow}\mu \mid q_2$
	\end{tabular} \\[5mm]
	
	\text{\small CALL }\ \	\begin{tabular}{ c }
	  $q \stackrel{a}{\rightarrow}\mu$
	  \\ \vspace{-0.35cm}\\
	  \hline \vspace{-0.35cm}
	  \\
	  $A \stackrel{a}{\rightarrow} \mu$
	\end{tabular} \ \  \text{\small if } {\scriptstyle A\defsym q}
	& \quad &
	\text{\small COM}\ \ \begin{tabular}{ c }
		$q_1\stackrel{a}{\rightarrow}\delta_{r_1}\quad q_2\stackrel{\bar{a}}{\rightarrow}\delta_{r_2}$
  	\\ \vspace{-0.35cm}\\
  	\hline \vspace{-0.35cm}
  	\\
  	$q_1 \mid q_2  \stackrel{\tau}{\rightarrow}\delta_{r_1 \mid r_2}$
	\end{tabular} \\
	
\end{array}
$$
$$
\begin{array}{c}
	\text{\small RST }\ \
	\begin{tabular}{ c }
	  $q  \stackrel{a}{\rightarrow} \mu$
	  \\ \vspace{-0.35cm}\\
	  \hline \vspace{-0.35cm}
	  \\
	 $(b) q  \stackrel{a}{\rightarrow} (b)\mu$
	\end{tabular} \ \ {\scriptstyle a,\bar{a}\neq b}\\[10mm]
\end{array}
$$
We assume also the symmetric versions of the rules NDT, PAR and COM. Recall that the symbol $\delta_q$ is the delta of Dirac, which assigns probability $1$ to $q$ and $0$ to all other processes. The symbol $\probsum{i}$ is the summation on distributions. Namely, $\probsum{i}p_i \cdot \mu_i$ is the distribution $\mu$ such that $\mu(x)=\sum_{i}p_i\cdot \mu_i{(x)}$. The notation $\mu\mid q$ represents the distribution $\mu'$ such that $\mu'(r) = \mu(q')$ if $r = q'\mid q$, and $\mu'(r) = 0$ otherwise. Similarly, $(b) \mu$ represents the distribution $\mu'$ such that $\mu'(q) = \mu(q')$ if $q = (b)q'$, and $\mu'(q) = 0$ otherwise.

\begin{remark} In some of the examples in this chapter we use an extension of our process calculus that allows message passing (cfr. Chapter~\ref{chapter:preliminaries}). Since the expressive power of our calculus with message passing or without it is the same, we consider explicit message passing simply as an alias for the correspondent encoding into the presentation of the calculus given above.
\end{remark}

\subsection{Distributed systems}

A distributed system has  the form $(A) \ q_1 \parallel q_2 \parallel \cdots \parallel q_n$, where the $q_i$'s are  components and $A\subseteq \Lset$. The restriction on $A$ enforces synchronization  on the channel names belonging to $A$, in accordance with the CCS spirit.

\paragraph{Systems' semantics} The semantics of a system gives rise to a TPA, where the states are terms representing systems during their evolution. A transition now is of the form $\smash{q\stackrel{t_{g}:a}{\longrightarrow} \mu}$ where $a\in\Lset$, $\mu\in\distr(\Qset)$, and $t_{g}\in \Tset$ is either the tag of the component which makes the move, or a (unordered) pair of tags representing the two partners of a synchronization. We can simply define $\Tset$ as $\Tset = I \cup I^2$ where $I = \{1,2,\ldots,n\}$ is the set of components' identifiers.

\[  \text{\small Interleaving} \ \ 
 \begin{tabular}{ c }
 $q_i \stackrel{a}{\rightarrow} \probsum{k} p_k \cdot \delta_{q_{ik}}$
  \\ \vspace{-0.35cm}\\
\hline \vspace{-0.35cm}
  \\
  $(A)\ q_1\parallel \cdots \parallel q_i\parallel \cdots \parallel q_n\stackrel{  i:a }{ \longrightarrow} \probsum{k}p_k \cdot \delta_{(A)q_1\parallel \cdots \parallel q_{ik}\parallel \cdots \parallel q_n}$
\end{tabular} \ {\scriptstyle a\not\in A}\]\\[-2mm]

\noindent where $i $ is the tag indicating that the component $i$ is making the step. Note that we assume that probabilistic choices are finite. This implies that every transition $\smash{q\stackrel{t_{g}:a}\longrightarrow \mu}$ can be written $\smash{q\stackrel{t_{g}:a}\longrightarrow  \probsum{k} p_k\cdot\delta_{q_k}}$, and justifies the notation used in the interleaving rule.

\[\text{\small Synch.}\ \ \begin{tabular}{ c }
  $q_i \stackrel{a}\rightarrow \delta_{ q_i^\prime}\qquad q_j \stackrel{\bar{a}}\rightarrow \delta_{q_j^\prime}$
  \\ \vspace{-0.35cm}\\
  \hline \vspace{-0.35cm}
  \\
  $(A)\ q_1\parallel \cdots \parallel q_i\parallel\cdots\parallel q_j \parallel \cdots \parallel q_n \stackrel{\lbrace i,j \rbrace:\tau}{\longrightarrow}\delta_{(A)  q_1\parallel \cdots \parallel q_{i}^\prime \parallel\cdots\parallel q_j^\prime \parallel \cdots \parallel q_n}$
\end{tabular} \]\\[-2mm]

\noindent here $\lbrace i,j \rbrace$ is the tag indicating that the components making the step are $i$ and $j$. Note that it is an unordered pair. Sometimes we will write $i,j$ instead of $\lbrace i,j\rbrace$, for simplicity.

	\begin{example}
		\label{exa:TPAs} 
		Consider again the systems of Example~\ref{exa:CCS}. Figures \ref{fig:semantics-ex-a} and  \ref{fig:semantics-ex-b} show the TPAs for $S\left[ ^{a}/_{sec}\right]$ and for $S\left[ ^{b}/_{sec}\right]$ respectively. For simplicity we do not write  the restriction on channels $c$ and $out$, nor the termination symbol $0$. We use '$-$' to denote a component that is stuck. The corresponding tags are indicated in the figure with numbers above the components.

The set of enabled transitions should be clear from the figures. For instance, we have $\avail(S\left[ ^{b}/_{sec}\right])=\lbrace \lbrace 1,2\rbrace , \lbrace 1,3\rbrace, \lbrace 1,4\rbrace\rbrace$ and $\avail(\  - \ ||\ \overline{out}\langle a \rangle\ ||\ -\ ||\ -\ )= \lbrace 2\rbrace$. The scheduler $\zeta$ defined as
\[\zeta(\sigma) \defsym
\begin{cases}
           \lbrace 1,4\rbrace & {~if~} \sigma=S\left[ ^{a}/_{sec}\right],\\
           2 & {~if~} \sigma=S\left[ ^{a}/_{sec}\right] \stackrel{1,2:\tau}\longrightarrow (\ -\  ||\ \overline{out}\langle a \rangle\ ||\ -\ ||\ - \ ),\\
           3 & {~if~} \sigma=S\left[ ^{a}/_{sec}\right] \stackrel{1,3:\tau}\longrightarrow (\ -\ ||\ -\ ||\ \overline{out}\langle b \rangle\ ||\ -\ ),\\
           4 & {~if~} \sigma=S\left[ ^{a}/_{sec}\right] \stackrel{1,4:\tau}\longrightarrow (\ -\ ||\ -\ ||\ -\ ||\ \overline{out}\langle a \rangle\ ),\\
          \bot & {~otherwise,}\\
\end{cases}
\]
\noindent is a global scheduler for $S\left[ ^{a}/_{sec}\right]$.
\end{example}

\begin{figure}[htb]%
	\centering
	\subbottom[$S \mathopen{[} ^{a}/_{sec} \mathclose{]}$]{
		\centering
		\includegraphics[width=0.80\columnwidth]{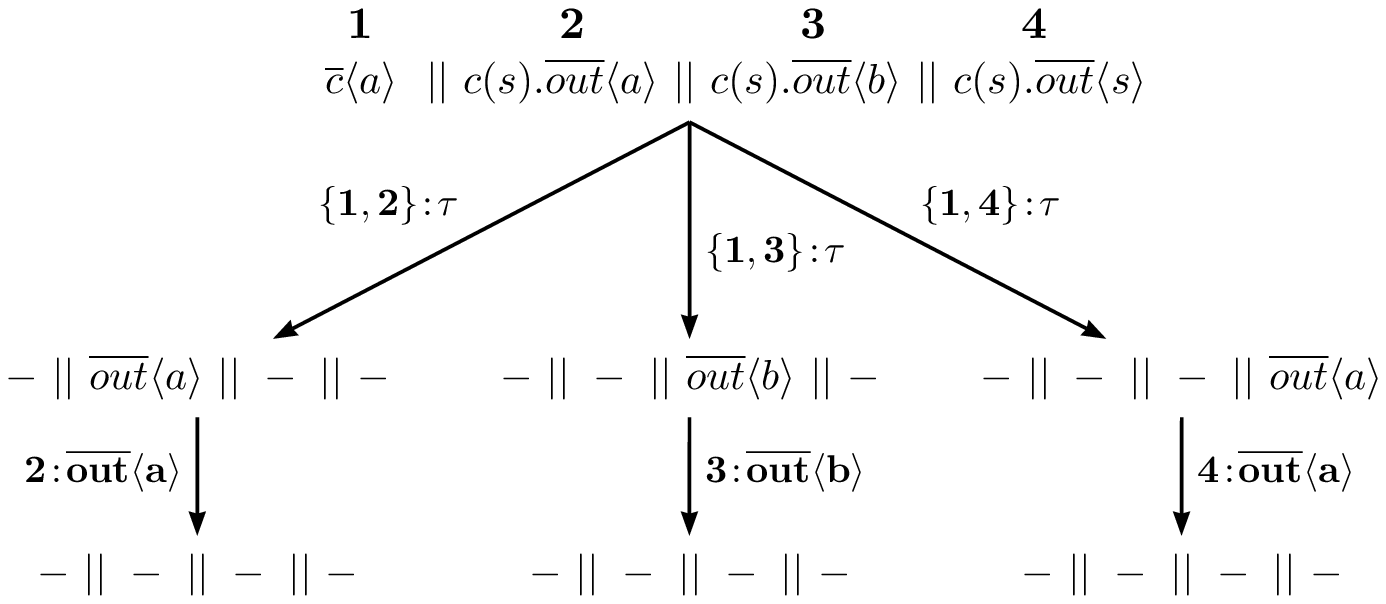}
		\label{fig:semantics-ex-a}
	} \\[5mm]
  \subbottom[$S \mathopen{[} ^{b}/_{sec} \mathclose{]}$]{  		
  	\centering
		\includegraphics[width=0.80\columnwidth]{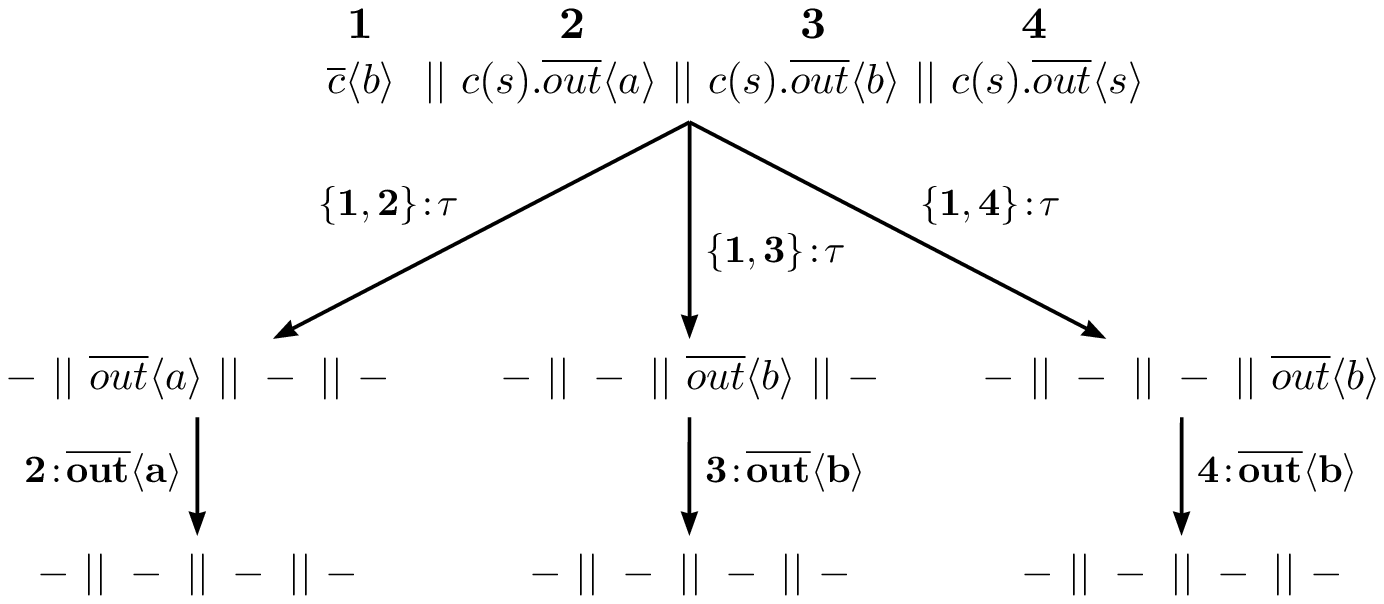}
		\label{fig:semantics-ex-b}
  }
	\caption{TPAs for Example~\ref{exa:TPAs}}
	\label{fig:semantics-ex}
\end{figure}

\section{Admissible schedulers}
\label{section:admissible-schedulers}

In this section we restrict the discerning power of the global and local schedulers in order to avoid the problem of the  information leakage induced by clairvoyant schedulers. We impose two kinds of restrictions: For the global scheduler, following  \cite{Andres:10:TechRep}, we assume that it can only see, and keep memory of, the observable actions and the components that are enabled, but not the secret actions. As for the local scheduler, we assume that the local nondeterminism of each component is solved on the basis of the view of the history local to that component, i.e. the projection of the  history of the system on that component. In other words, each component has to make decisions based only on the history of its own execution; it cannot see anything of the other components.

\subsection{Restricting global schedulers}

We assume that the set of actions $\Lset$ is divided in two disjoint sets, the \emph{secret actions} $\Sset$ and the \emph{observable actions} $\Oset$, such that $\Sset \cup \Oset = \Lset$. The secret actions are supposed to be invisible to the global scheduler. Formally, this can be achieved using a function $\sift$ with
\begin{equation*}
	\sift(a) = \left\{
		\begin{aligned}
			\tau  & \quad \text{if $a \in \Sset$,} \\ 
			a     & \quad \text{otherwise.}
		\end{aligned}
		\right.
\end{equation*}
Then, we restrict the power of the global scheduler by forcing it to make the same decisions on paths he cannot tell apart.

\begin{definition}
	Given a TPA $M$, a global scheduler $\zeta$ for $M$ is \emph{admissible} if for all paths $\sigma_1$ and $\sigma_2$ we have $\view(\sigma_1)=\view(\sigma_2)\mbox{~implies~}\zeta(\sigma_1)=\zeta(\sigma_2)$, where 	
\begin{equation*}
	\begin{split}
		\view \left( \qi \stackrel{t_{g1}:a_1}{\longrightarrow} q_1 \stackrel{t_{g2}:a_2}{\longrightarrow} \cdots \stackrel{t_{gn}:a_{n}}{\longrightarrow} q_{n+1} \right) \defsym \left( \avail(\qi), \sift(a_1), t_{g1} \right) \\ \left( \avail(q_1), \sift(a_2), t_{g2} \right) \cdots \left( \avail(q_n), \sift(a_n), t_{gn} \right)
	\end{split}
\end{equation*}	
\end{definition}

The idea is that $\view$ sifts the information of the path that the scheduler can see. Since $\sift$ \qm{hides} the secrets, the scheduler cannot take different decisions based on them.

\subsection{Restricting local schedulers}

The restriction on local schedulers is based on the idea that a step of the component $i$ of a system can only be based on the view that $i$ has of the history, i.e. its own history. In order to formalize this restriction, it is convenient to introduce the concept of $i$-view of a path $\sigma$, or \emph{projection} of $\sigma$ on $i$, which we will
denote by $\sigma_{\upharpoonright  i}$. We define it inductively:
\begin{equation*}
	(\sigma \stackrel{t_{g}:a}{\longrightarrow}\mu )_{\upharpoonright  i} = 
	\left\{\begin{array}{ll}
					\sigma_{\upharpoonright  i} \stackrel{i:b}{\longrightarrow} \delta_{q_i} & \text{if } t_{g} = \{i,j\} \text{ and } \mu = \delta_{(A)\ q_1\parallel  \ldots \parallel q_i\parallel   \ldots \parallel q_j\parallel  \ldots \parallel q_n}\\
					\sigma_{\upharpoonright  i} \stackrel{i:a}{\longrightarrow} \mu& \text{if } t_{g} =i\\	
  				\sigma_{\upharpoonright  i} &\text{otherwise}						
					\end{array}
	\right.
\end{equation*}
\commentMS{This is a bit of cheating. In the journal version we need to change the operational semantics so that we can retrieve the form of the premise from the conclusion.}

In the above definition, the first line represents the case of a synchronization step involving the component $i$, where we assume that the premise for $i$ is of the form
$\smash{q'_i\stackrel{b}{\longrightarrow} \delta_{q_i} }$.
The second line represents an interleaving step in which $i$ is the active component.
The third line represents  step in which the component $i$ is idle.

The restriction to the local scheduler can now be expressed as follows:

\begin{definition}
	Given a TPA $M$ and a local scheduler $\xi$  for $M$, we say that $\xi$ is \emph{admissible} if for all paths $\sigma$ and $\sigma'$, if whenever $\xi(\sigma)= (t_{g}, a,\mu)$, and $\xi(\sigma')= (t_{g}', a',\mu')$ we have:
	
	\begin{itemize}
		
		\item if $t_{g} =t_{g}'=i$ and $\sigma_{\upharpoonright  i}=\sigma'_{\upharpoonright  i}$, then  $\xi(\sigma)=\xi(\sigma')$,
	
		\item if $t_{g} =t_{g}'=\{i,j\}$, $\sigma_{\upharpoonright  i}=\sigma'_{\upharpoonright  i}$, and  $\sigma_{\upharpoonright  j}=\sigma'_{\upharpoonright  j}$ then $\xi(\sigma)=\xi(\sigma')$.
		
	\end{itemize}
	
\end{definition}
A pair of compatible schedulers $(\zeta,\xi)$ is called \emph{admissible} if $\zeta$ and $\xi$ are admissible.

\section{Safe equivalences}
\label{section:equivalences}
In this section we  revise  process equivalence notions  to make them safe for security.

\subsection{Safe complete traces}

We define here a safe version of complete-trace semantics. The idea is that we  compare two processes based not only on their traces, but also on the choices that the global scheduler makes at every step. We do this by recording explicitly the tags in the traces.

\begin{definition} Here we define the notion of safe complete traces.

	\begin{itemize}
		\item Given a TPA  $M = (\Qset,\Tset,\Lset,\qi,\vartheta)$,  the (complete) safe traces of $M$, denoted here by $\Traces_s$, are defined as the probabilities of sequences of tags and actions corresponding to all possible complete executions, i.e.
\begin{align*}
	\Traces_s(M) = & \{ \ f: (\Tset \times \Lset )^\infty \rightarrow [0,1] \mid  \\
	               & \  \ \text{there exists an admissible scheduler} (\zeta,\xi) \text{ s.t. } \\
	               & \  \ \forall t\in (\Tset \times \Lset )^\infty \\
	               & \  \ f(t) = \psp_{M,\zeta,\xi}(\{ \sigma\in\cpaths(M) \mid \trace_{ta}(\sigma) = t\}) \  \}
\end{align*}
\noindent where  $\psp_{M,\zeta,\xi}$ is the probability measure in $M$ under $(\zeta,\xi)$, and $\trace_{ta}$ extracts from a path  the sequence of tags and actions, i.e.
\begin{align*}
	\trace_{ta}(\epsilon) & = \epsilon \\
	\trace_{ta}(q\stackrel{t_{g}:a}{\longrightarrow}\sigma) & = t_{g}:a\cdot  \trace_{ta}(\sigma)
\end{align*}
		\item We denote by $\Traces_s(q)$ the safe traces of the automaton associated to a system $q$.
		
		\item Two systems $q_1$ and $q_2$ are safe-trace equivalent, denoted by $q_1\simeq_s q_2$,  if and only if $\Traces_s(q_1)=\Traces_s(q_2)$.	
	\end{itemize}
	
\end{definition}

The following example points out the difference between $\simeq_s$ and the standard (complete) trace equivalence.


\begin{example}
	\label{exa:CCSlab} 
Consider the TPAs of Example~\ref{exa:TPAs}. The two TPAs have the same complete traces. In fact we have
\begin{equation*}
	\Traces(S\left[ ^{a}/_{sec}\right]) \quad = \quad \lbrace \tau \cdot \overline{out}\langle a \rangle \ , \ \tau\cdot \overline{out}\langle b \rangle \rbrace \quad \quad = \quad \Traces(S\left[ ^{b}/_{sec}\right])
\end{equation*}
But on the other hand, we have 
\begin{equation*}
	\Traces_s(S\left[ ^{a}/_{sec}\right]) = \{f_1,f_2,f_3\} \quad \neq \quad  \{f_1,f_2,f_4\} = \Traces_s(S\left[ ^{a}/_{sec}\right])
\end{equation*}
\noindent where 
\begin{align*}
	f_1(t) = 
		\left\{ 
			\begin{aligned} 
				1 & \quad \text{if $t = \{1,2\}:\tau \cdot 2:\overline{out}\langle a \rangle$,} \\
				0 & \quad \text{for all other values of $t \in (\Tset \times \Lset )^\infty$.}
			\end{aligned} 
		\right. \\
	f_2(t) = 
		\left\{ 
			\begin{aligned} 
				1 & \quad \text{if $t = \{1,3\}:\tau \cdot 3:\overline{out}\langle b \rangle$,} \\
				0 & \quad \text{for all other values of $t \in (\Tset \times \Lset )^\infty$.}
			\end{aligned} 
		\right. \\
	f_3(t) = 
		\left\{ 
			\begin{aligned} 
				1 & \quad \text{if $t = \{1,4\}:\tau \cdot 4:\overline{out}\langle a \rangle$,} \\
				0 & \quad \text{for all other values of $t \in (\Tset \times \Lset )^\infty$.}
			\end{aligned} 
		\right. \\
	f_4(t) = 
		\left\{ 
			\begin{aligned} 
				1 & \quad \text{if $t = \{1,4\}:\tau \cdot 4:\overline{out}\langle b \rangle$,} \\
				0 & \quad \text{for all other values of $t \in (\Tset \times \Lset )^\infty$.}
			\end{aligned} 
		\right. \\
\end{align*}
\end{example}

\subsection{Safe bisimilarity}

In this section we propose a security-safe version of strong  bisimulation, that we call \emph{safe bisimulation}. This is an equivalence relation stricter than safe-trace equivalence, with  the advantage of being a congruence. Since in this chapter we assume that schedulers can always observe which component is making a step (even a silent step), it does not seem natural to consider weak bisimulation.

We start with some notation. Given a TPA  $M = (\Qset,\Tset,\Lset,\qi,\vartheta)$, and a global scheduler $\zeta$, we write
$ q\stackrel{a}{\longrightarrow}_\zeta \mu$  if there exists $\sigma \in \fpaths(M)$ such that $\zeta(\sigma)\not=\bot$,
$(\zeta(\sigma), a,\mu) \in \vartheta(q)$, and $q =\last(\sigma)$.
Note that the  restriction to  $\zeta$ 
still allows nondeterminism, i.e. there may be $\mu_1,  \mu_2$, such that  $ q\stackrel{a_1}{\longrightarrow}_\zeta \mu_1$ and $ q\stackrel{a_2}{\longrightarrow}_\zeta \mu_2$ (with either $a_1=a_2$ or $a_1\neq a_2$).

We now define the notion of safe bisimulation. The idea is that, if $q_{1}$ and $q_{2}$ are bisimilar states, then every move from $q_{1}$ should be mimicked by a move from $q_{2}$ \emph{using the same (admissible) scheduler}.

\begin{definition}
	\label{def:safe_sch}
	Given a TPA $M = (\Qset,\Tset,\Lset,\qi,\vartheta)$, we say that a relation ${\mathcal R} \subseteq \Qset \times \Qset$ is a safe bisimulation if and only if, \revision{whenever $q_1 \,{\mathcal R}\, q_2$:}
	
	\begin{enumerate}
		\item $\avail(q_1) = \avail(q_2)$, and
		
		\item for all admissible global schedulers $\zeta$ for $M$ such that \revision{$\zeta(\sigma_1) \, \mathcal{R} \, \zeta(\sigma_2)$} whenever $\last(\sigma_1)=q_1$ and $\last(\sigma_2)=q_2$:
		
			\begin{itemize}
				\item if $q_1 \stackrel{a}{\longrightarrow}_\zeta \mu_1$, then there exists $\mu_2$ such that $q_2 \stackrel{a}{\longrightarrow}_\zeta \mu_2$ and $\mu_1 \,{\mathcal R}\, \mu_2$, and 
				\item if $q_2 \stackrel{a}{\longrightarrow}_\zeta \mu_2$, then there exists $\mu_1$ such that $q_1 \stackrel{a}{\longrightarrow}_\zeta \mu_1$ and $\mu_1 \,{\mathcal R}\, \mu_2$,
			\end{itemize}
		
	\end{enumerate}

where $\mu_1 \,{\mathcal R}\, \mu_2$ means that for all equivalence classes $X\in \Qset_{\hat{\mathcal R}}$, we have $\mu_1(X)= \mu_2(X)$, where $\hat{\mathcal R}$ is the smallest equivalence class induced by $\mathcal R$.

\end{definition}

It is possible to simplify Definition~\ref{def:safe_sch}, restricting the schedulers to be history-independent. In other words, to show that two distributed systems are bisimilar, it suffices to consider one-step computations and show that two states are equivalent by using only history-independent schedulers. The lemma bellow justifies this claim.

\begin{lemma}
	\label{lem:history}
	Let $M=(\Qset,\Tset,\Lset,\qi,\vartheta)$ be a TPA, and let $\mathcal{R}$ be an equivalence relation on the set of states $\Qset$. Consider $\zeta$ to be a global scheduler for $M$ such that, for every pair of states $q_{1}, q_{2} \in \Qset$, if $q_{1} = \last(\sigma_{1}) \mathcal{R} \last(\sigma_{2}) = q_{2}$ then $\zeta(\sigma_{1}) = \zeta(\sigma_{2})$. In that case $\zeta$ is history-independent, i.e. it depends only on the last state of a path $\sigma$.
\end{lemma}

\begin{proof}
	It is easy to see that the relation of having the same last state is an equivalence relation on paths, and therefore it determines a partition on the set of paths. Since the above $q_1$ and $q_2$ may be identical, the scheduler must give the same value on equivalent paths and it is, therefore, history-independent.
\end{proof}

Using the lemma above, in the following results about safe bisimulation we will usually write $\zeta(q)$ where $q$ is a state. Note however that this does not mean that in the computations of safely bisimilar systems the schedulers are necessarily history-independent: at each step of the computation we may change scheduler, and therefore we may change alternative when we pass by the same state $q$ at a later time.

The following result is analogous to the case of standard bisimulation. It implies that largest safe bisimulation exists, and coincides with the union of all safe bisimulations. We call it \emph{safe bisimilarity}, and we denote it by $\sim_s$.

\begin{proposition}
	\label{prop:union}
		The union of all the safe bisimulations is still a safe bisimulation.
\end{proposition}

\begin{proof}
	Assume that $q_1\sim_s q_2$. Then $q_1 \,{\mathcal R}\, q_2$ holds, for some safe bisimulation ${\mathcal R}$. Hence we have $\avail(q_1)=\avail(q_2)$, and for every global scheduler $\zeta$, if  $\zeta(q_1)=\zeta(q_2)$, and   $q_1 \stackrel{a}{\longrightarrow}_\zeta \mu_1$, then there exists $\mu_2$ such that $q_2 \stackrel{a}{\longrightarrow}_\zeta \mu_2$, and $\mu_1 \,{\mathcal R}\, \mu_2$. This implies that $\mu_1 \sim_s \mu_2$. In fact $\hat{\mathcal R}$ (the smallest equivalence class induced by $\mathcal R$) is a finer relation than $\hat\sim_s$, i.e. $q_1\, \hat{\mathcal R} \,q_2$ implies $q_1 \hat{\sim}_s q_2$. Also, $\hat{\mathcal R}$ is an equivalence relation, and therefore it induces a partition on each of the equivalence classes $X\in \Qset_{\hat\sim_s}$. Hence we have, for each   $X\in \Qset_{\hat{\sim}_s}$, $\mu_1(X) = \sum_{Y\in X_{\hat{\mathcal R}}} \mu_1(Y) =  \sum_{Y\in X_{\hat{\mathcal R}}} \mu_2(Y) = \mu_2(X)$.

We proceed analogously to show that, if $q_2 \stackrel{a}{\longrightarrow}_\zeta \mu_2$, then  there exists $\mu_1$ such that $q_1 \stackrel{a}{\longrightarrow}_\zeta \mu_1$ and $\mu_1\sim_s\mu_2$.

\end{proof}

Given two TPAs $M_1= (\Qset_1,\Tset,\Lset,\qi_1,\vartheta_1)$ and $M_2 = (\Qset_2,\Tset,\Lset,\qi_2,\vartheta_2)$ sharing the same set of tags $\Tset$ and actions $\Lset$, we can define bisimulation and bisimilarity across their states, i.e. as relations on $(\Qset_1\cup \Qset_2)$, in the obvious way, by constructing the TPA $M$ with a new initial state $\qi$ with transitions to $\delta_{\qi_1}$ and to $\delta_{\qi_2}$, respectively.

Given two components or systems $q_1$ and $q_2$, we will say that $q_1$ and $q_2$ are safely bisimilar, denoted by $q_1\sim_s q_2$, if the initial states of the corresponding TPAs are safely bisimilar. Note that $q_1\sim_s q_2$ is possible only if $q_1$ and $q_2$ have the same number of active components, where \qm{active}, for a component, means that  during the execution of the system it will make at least one step.
Note that in the case of components, or of systems constituted by one component only, safe bisimulation and safe bisimilarity coincide with standard bisimulation and bisimilarity (denoted by $\sim$), respectively. This is not the case for systems, as shown by the following example:
%
%

\begin{example} Consider again the TPAs of Example~\ref{exa:TPAs}. As pointed out earlier in this chapter, we have  $S\left[ ^{a}/_{sec}\right]\sim S\left[ ^{b}/_{sec}\right]$. Yet $S\left[ ^{a}/_{sec}\right]\not\sim_s S\left[ ^{b}/_{sec}\right]$. To show this, let us construct a new TPA (as described before) with initial state $\qi$ such that $\smash{\qi\stackrel{t_{g}:\tau}{\longrightarrow}S\left[ ^{a}/_{sec}\right]}$ and $\smash{\qi\stackrel{t_{g}:\tau}{\longrightarrow} S\left[ ^{b}/_{sec}\right]}$. 
Now consider the (admissible) global scheduler $\zeta$ such that
\[\zeta(\sigma) \defsym
\begin{cases}
           t_{g} & {~if~} \sigma=\qi,\\
           \{1,4\} & {~if~} \sigma=\qi\stackrel{t_{g}:\tau}{\longrightarrow}S\left[ ^{a}/_{sec}\right],\\
           2 & {~if~} \sigma=\qi\stackrel{t_{g}:\tau}{\longrightarrow}S\left[ ^{a}/_{sec}\right] \stackrel{1,2:\tau}\longrightarrow (\ -\  ||\ \overline{out}\langle a \rangle\ ||\ -\ ||\ - \ ),\\
           3 & {~if~} \sigma=\qi\stackrel{t_{g}:\tau}{\longrightarrow}S\left[ ^{a}/_{sec}\right] \stackrel{1,3:\tau}\longrightarrow (\ -\ ||\ -\ ||\ \overline{out}\langle b \rangle\ ||\ -\ ),\\
           4 & {~if~} \sigma=\qi\stackrel{t_{g}:\tau}{\longrightarrow}S\left[ ^{a}/_{sec}\right] \stackrel{1,4:\tau}\longrightarrow (\ -\ ||\ -\ ||\ -\ ||\ \overline{out}\langle a \rangle\ ),\\
           \{1,4\} & {~if~} \sigma=\qi\stackrel{t_{g}:\tau}{\longrightarrow}S\left[ ^{b}/_{sec}\right],\\
           2 & {~if~} \sigma=\qi\stackrel{t_{g}:\tau}{\longrightarrow}S\left[ ^{b}/_{sec}\right] \stackrel{1,2:\tau}\longrightarrow (\ -\  ||\ \overline{out}\langle a \rangle\ ||\ -\ ||\ - \ ),\\
           3 & {~if~} \sigma=\qi\stackrel{t_{g}:\tau}{\longrightarrow}S\left[ ^{b}/_{sec}\right] \stackrel{1,3:\tau}\longrightarrow (\ -\ ||\ -\ ||\ \overline{out}\langle b \rangle\ ||\ -\ ),\\
           4 & {~if~} \sigma=\qi\stackrel{t_{g}:\tau}{\longrightarrow}S\left[ ^{b}/_{sec}\right] \stackrel{1,4:\tau}\longrightarrow (\ -\ ||\ -\ ||\ -\ ||\ \overline{out}\langle b \rangle\ ),\\
           \bot & {~otherwise.}\\
\end{cases}
\]
\noindent It is easy to see that $S\left[ ^{b}/_{sec}\right]$ cannot mimic the transition $4\!:\!\overline{out}\langle a \rangle$ produced by $S\left[ ^{a}/_{sec}\right]$ using the same scheduler $\zeta$.
\end{example}

We now show that safe bisimulation is a congruence with respect to all the operators of our language.  In the following theorem, statements~\ref{item:compo-a} and \ref{item:compo-b} are just the standard compositionality result for probabilistic bisimulation.

\begin{theorem}
	\label{theo:compo} \ \ \ \ \ \
		\begin{enumerate}
		
			\item $\sim_s$ is an equivalence relation.
		
			\item Let $a\in \Lset$ be an action and $A,  B, B' \subseteq \Lset$ be sets of restrictions. Let $p_1,\ldots,p_n$ be probability values, and let $q, q_1, q_2, \ldots ,q_n$, $q'_1, q'_2, \ldots ,q'_n$ be components.
				\begin{enumerate}
					\item \label{item:compo-a} If $q_1\sim_s q_2$, \quad  then\quad  $a.q_1\sim_s a.q_2$,  \quad $q_1+ q\sim_s q_2 +q$, \quad $(a)q_1\sim_s (a)q_2$,  and \quad $q_1 \mid q\sim_s q_2 \mid q$.
					\item \label{item:compo-b} If $q_1\sim_s q'_1,\ldots, q_n\sim_s q'_n$ , \quad then \quad $\sum_i p_i : q_i \sim_s \sum_i p_i : q_i'$. 
					\item \label{item:compo-c} If $(B) \ q_1\parallel\ldots\parallel q_n \ \sim_s \ (B') \ q'_1\parallel \ldots \parallel q'_n$, \ \ and ${\mathit fn}(q)\not \in B\cup B'$, \quad then 
						\[(A \cup B)\  q_1\parallel \ldots\parallel q \parallel \ldots \parallel q_n \ \sim_s \ (A\cup B') \ q'_1\parallel\ldots\parallel q \parallel \ldots \parallel q'_n.\]
				\end{enumerate}
		\end{enumerate}
\end{theorem}

\begin{proof} \qquad
	
	\begin{enumerate}
		
		\item Although safe bisimulations are not equivalence relations in general, their union, i.e. safe bisimilarity, is an equivalence. In fact: 
			\begin{itemize}
				\item  It is easy to see that, if $\mathcal R$ is a safe bisimulation, then the smallest equivalence that includes $\mathcal R$, namely $\hat{\mathcal R}$, is also a safe bisimulation.
				\item From Proposition~\ref{prop:union} we know that $\sim_s$ is a safe bisimulation.
				\item Hence we derive that $\hat{\sim}_s$ is a safe bisimulation, and therefore $\hat{\sim}_s\subseteq \ \sim_s$. But since obviously $\sim_s\subseteq\  \hat\sim_s$, we conclude that $\sim_s = \ \hat{\sim}_s$, which means that $\sim_s$ is already an equivalence relation.\\
			\end{itemize}
	
		\item Assume that $a$, $A,  B, B' , p_1,\ldots,p_n, q, q_1, q_2, \ldots ,q_n$, $q'_1, q'_2, \ldots ,q'_n$ are of the type prescribed by the hypothesis of the theorem. 		
			\begin{enumerate}
				\item Assume  $q_1\sim_s q_2$.
					\begin{itemize}
						
						\item Let 
							\[{\mathcal R} = \{(a.q_1,a.q_2)\}\,\cup\, \sim_s.\]
						We show that ${\mathcal R}$ is a safe bisimulation, which is sufficient to prove that  $a.q_1\sim_s a.q_2$.  Note that,  since there is only one component in each of those states, and it is enabled,  we have $\avail(a.q_1)=\avail(a.q_2) = \{1\}$, and $\zeta(a.q_1)=\zeta(a.q_2)= 1$ for any global scheduler $\zeta$.  Given a global scheduler $\zeta$, there is exactly one transition from each of  $a.q_1$ and $a.q_2$: these are $a.q_1\stackrel{a}\rightarrow_\zeta \delta_{q_1}$ and $a.q_2\stackrel{a}\rightarrow_\zeta \delta_{q_2}$, respectively, which mimic each other in the action $a$. Finally, since $q_1 \sim_s q_2$, we have $\delta_{q_1} \sim_s \delta_{q_2}$ and therefore $\delta_{q_1} \, {\mathcal R}\, \delta_{q_2}$.\\[-2mm]

						\item Let
							\[{\mathcal R} = \{(q_1+q,q_2+q)\}\,\cup\, \sim_s.\]
						We show that ${\mathcal R}$ is a safe bisimulation, which is sufficient to prove that  $q_1+q\sim_s q_2+q$. We have that $\avail(q_1+q)=\avail(q_1)\cup\avail(q) = \avail(q_2)\cup\avail(q)=\avail(q_2+q)$, in fact $\avail(q_1) = \avail(q_2)$ since $q_1\sim_s q_2$. Correspondingly, given  a global scheduler $\zeta$, we have either $\zeta(q_1+q)=\zeta(q_2+q)=1$ or $\zeta(q_1+q)=\zeta(q_2+q)=\perp$, since there is only one component. Assume $q_1+q\stackrel{a}\rightarrow_\zeta \mu_1$. We have two cases: either  $q_1\stackrel{a}\rightarrow_\zeta \mu_1$, or $q \stackrel{a}\rightarrow_\zeta \mu_1$. The second case is obvious. In the first case, since $q_1\sim_s q_2$, we have that also $q_2\stackrel{a}\rightarrow_\zeta \mu_2$, with $\mu_1\sim_s\mu_2$. We derive that $\mu_1\,{\mathcal R}\, \mu_2$. For the transitions from $q_2+q$ we proceed in the analogous way. \\[-2mm]

						\item Let
 							\[{\mathcal R} = \{((a)q_1,(a)q_2)\mid q_1 \sim_s q_2\}.\]
 						We show that ${\mathcal R}$ is a safe bisimulation, which is sufficient to prove that, if $q_1\sim_s q_2$, then  $(a)q_1\sim_s (a)q_2$. First observe that $\avail((a)q_1)=\avail(q_1)=\{1\}$ if $q_1$ can make a transition with a label different from $a$, otherwise $\avail((a)q_1)=\emptyset$. The same holds for $(a)q_2$. Since $q_1\sim_s q_2$, we derive that $\avail((a)q_1)=\avail((a)q_2)$. Accordingly, given a  global scheduler $\zeta$, we have that either $\zeta((a)q_1)=\zeta((a)q_2)=1$, or $\zeta((a)q_1)=\zeta((a)q_2)=\perp$. Assume  \smash{$(a)q_1\stackrel{b}\rightarrow_\zeta \mu_1$}. Then we must have $b \neq a$ and $\mu_1= (a)\mu'_1$, where  \smash{$q_1\stackrel{b}\rightarrow_\zeta \mu'_1$}. Since $q_1\sim_s q_2$, we have also \smash{$q_2\stackrel{b}\rightarrow_\zeta \mu'_2$}, with $\mu'_1\sim_s \mu'_2$.
We derive \smash{$(a)q_2\stackrel{b}\rightarrow_\zeta (a) \mu'_2$}, and $(a)\mu'_1\,{\mathcal R}\, (a)\mu'_2$.

						We proceed in an analogous way for the transitions from $(a) q_2$.\\[-2mm]

						\item The case of the parallel operator in components is similar to the case of the parallel operator on systems (see the last item of this proof).\\
						
					\end{itemize}

				\item Assume $q_1\sim_s q'_1,\ldots, q_n\sim_s q'_n$. Let
					\[{\mathcal R} = \{(\sum_i p_i : q_i,\sum_i p_i : q_i')\}\,\cup\, \sim_s.\]
					We show that ${\mathcal R}$ is a safe bisimulation, which is sufficient to prove that  $\sum_i p_i : q_i \sim_s \sum_i p_i : q_i'$. Observe that both $\sum_i p_i : q_i$ and $\sum_i p_i : q_i'$ are enabled, and, since there is only one component, $\avail(\sum_i p_i : q_i)=\avail(\sum_i p_i : q_i')=\{1\}$. Accordingly, if $\zeta$ is a global scheduler, we have $\avail(\sum_i p_i : q_i)=\avail(\sum_i p_i : q_i')=1$. Given a global scheduler $\zeta$,  the only transitions from   $\sum_i p_i : q_i$ and $\sum_i p_i : q_i'$ are  $\sum_i p_i : q_i\stackrel{\tau}\rightarrow_\zeta \probsum{i} p_i \cdot \delta_{q_i}$ and $\sum_i p_i : q_i'\stackrel{\tau}\rightarrow_\zeta \probsum{i} p_i \cdot \delta_{q'_i}$ respectively, which mimic each other in the action $\tau$. It is easy to see that we have $(\sum_i p_i : q_i)\sim_s (\sum_i p_i : q'_i)$, and therefore $(\sum_i p_i : q_i)\,{\mathcal R}\, (\sum_i p_i : q'_i)$. \\

				\item Let
					\[{\mathcal R} = 
						\left\{ 
						\begin{array}{l}
							((A \cup B)\ q_1\parallel \ldots\parallel q \parallel \ldots \parallel q_n, \\
							(A\cup B') \ q'_1\parallel\ldots\parallel q \parallel \ldots \parallel q'_n) \mid \\ 
							(B) \ q_1\parallel\ldots\parallel q_n \ \sim_s \ (B') \ q'_1\parallel \ldots \parallel q'_n 
						\end{array}
	 					\right\}\]
					We show that ${\mathcal R}$ is a safe bisimulation, which is sufficient to prove that, if
					\[(B) \ q_1\parallel\ldots\parallel q_n \ \sim_s \ (B') \ q'_1\parallel \ldots \parallel q'_n\ ,\] 
					\noindent then
					\[(A \cup B)\  q_1\parallel \ldots\parallel q \parallel \ldots \parallel q_n \ \sim_s \ (A\cup B') \ q'_1\parallel\ldots\parallel q \parallel \ldots \parallel q'_n \ .\]
					Observe first that
					\begin{equation*}
						\begin{split}
							\avail((A \cup B)\  q_1\parallel \ldots\parallel q \parallel \ldots \parallel q_n) = \\ 
							\avail((A\cup B') \ q'_1\parallel\ldots\parallel q \parallel \ldots \parallel q'_n)
						\end{split}
					\end{equation*}
					In fact the enabled components  are the same as those of \linebreak $(B)\  q_1\parallel \ldots \parallel q_n$ and of $(B')\  q'_1\parallel \ldots \parallel q'_n$ (modulo the index shift), which are equal by the bisimilarity  hypothesis, plus possibly the component $q$, plus possibly the synchronizations with $q$, which again are equal by the bisimilarity hypothesis, minus the transitions with labels in $A$. Note that the hypothesis ${\mathit  fn}(q)\not\in B\cup B'$ is essential here to guarantee that the component  $q$ is enabled (or disabled) in both sides.

					Let us consider the synchronization case; the interleaving case is just a simplified variant. Given a global scheduler $\zeta$, assume 
					\[\zeta((A \cup B)\  q_1\parallel \ldots\parallel q \parallel \ldots \parallel q_n) = \zeta((A\cup B') \ q'_1\parallel\ldots\parallel q \parallel \ldots \parallel q'_n).\]
					Consider a move from the system in the left-hand side:
					\[(A \cup B)\ q_1\parallel \cdots \parallel q_i\parallel\cdots\parallel q_j \parallel \cdots \parallel q_n \stackrel{i,j:\tau}{\longrightarrow}\delta_{(A)  q_1\parallel \cdots \parallel r_{i} \parallel\cdots\parallel r_j \parallel \cdots \parallel q_n}.\]
					Then we must have
					\[q_i \stackrel{a}\rightarrow \delta_{ r_i}\quad ,  \quad  q_j \stackrel{\bar{a}}\rightarrow \delta_{r_j} \quad, \]
					where one of the $q_i, q_j$ could be $q$, and 
					\[\zeta((A \cup B)\ q_1\parallel \cdots \parallel q_i\parallel\cdots\parallel q_j \parallel \cdots \parallel q_n)=\{i,j\}.\]
					Since $q_i\sim_s q'_i$ and $q_j\sim_s q'_j$ (in case $q_i= q$ then $q'_i=q$ and therefore $q_i\sim_s q'_i$ because $\sim_q$ is reflexive, and analogously for $q_j$),
we must have
					\[q'_i \stackrel{a}\rightarrow \delta_{ r'_i}\quad ,  \quad  q'_j \stackrel{\bar{a}}\rightarrow \delta_{r'_j} \quad, \]
					\noindent for some $r'_i,r'_j$ such that $ \delta_{ r_i}\sim_s  \delta_{ r'_i}$ and $ \delta_{ r_j}\sim_s  \delta_{ r'_j}$.
We derive that
					\[(A \cup B)\ q'_1\parallel \cdots \parallel q'_i\parallel\cdots\parallel q'_j \parallel \cdots \parallel q'_n \stackrel{i,j:\tau}{\longrightarrow}\delta_{(A)  q'_1\parallel \cdots \parallel r_{i}^\prime \parallel\cdots\parallel r_j^\prime \parallel \cdots \parallel q'_n}\ , \]					
					\noindent and, since $\delta_{ r_i}\sim_s  \delta_{ r'_i}$, $ \delta_{ r_j}\sim_s  \delta_{ r'_j}$ imply $ { r_i}\sim_s   { r'_i}$, $  { r_j}\sim_s   { r'_j}$, and by the definition of $\mathcal R$, we conclude
					\[(\delta_{(A)  q_1\parallel \cdots \parallel r_{i} \parallel\cdots\parallel r_j \parallel \cdots \parallel q_n}) \ {\mathcal R}\ (\delta_{(A)  q'_1\parallel \cdots \parallel r_{i}^\prime \parallel\cdots\parallel r_j^\prime \parallel \cdots \parallel q'_n}).\]
					We proceed in an analogous way for the transitions from the right-hand side.
					
			 \end{enumerate}

	\end{enumerate}

\end{proof}

The following property shows that bisimulation is stronger than safe-trace equivalence, like in the standard case.

\begin{proposition}\label{prop:simsimeq}
	If $q_1\sim_s q_2$ then $q_1 \simeq_s q_2$.
\end{proposition}

\begin{proof}

For this proof, it is convenient to consider a coinductive approximation of safe-trace equivalence. We start with a coinductive characterization of the safe traces. This in itself is not a key notion of the proof, but will help understanding the definition of the approximation.

Given a TPA $M=(\Qset,\Tset,\Lset,\qi,\vartheta)$, consider the operator
	\[{\mathscr T}_{\textrm Tr}: (\Qset\rightarrow {\mathcal P}(\cpaths(M)\rightarrow [0,1]))\rightarrow (\Qset\rightarrow {\mathcal P}(\cpaths(M)\rightarrow [0,1]))\]
	\noindent defined as:
\[{\mathscr T}_{\textrm Tr}(F)(q) = \begin{array}[t]{l}\{\ f: (\Tset\times\Lset)^\infty \rightarrow [0,1] \mid
\\ \qquad\text{if } q\not\rightarrow \text{ then } f(\epsilon) =1,
 \text{ else }  f(\epsilon) = 0 \text{ and, } \\[1mm]
 	\qquad \text{ for all } t_{g}\in \Tset, a\in \Lset,\\[1mm]
	\qquad\bullet \text{ if there exists } \mu \text{ s.t. } q\stackrel{t_{g}:a}{\longrightarrow}\mu, \text{ then for each } q'\in \Qset \\[1mm]
	\qquad \text{ there exists } f'_{q'}\in F(q') \text{ s.t. for every } t\in  (\Tset\times\Lset)^\infty, \\[1mm]
	\qquad f (t_{g}:a\cdot t) = \sum_{q'} \mu(q') f'_{q'}(t) \\[1mm]
	\qquad \bullet \text{ if }  q\ \ \ \not\stackrel{\!\!\!t_{g}:a}{\!\!\!\longrightarrow}, \text{ then } f(q)(t_{g}:a\cdot t) =0 \  \  \ \}	
	\end{array}
	\]
	\noindent where $q\not\rightarrow$ means that for all $t_{g}\in \Tset, a\in \Lset$, we have $q\ \ \ \not\stackrel{\!\!\!\!\!\!\!t_{g}:a}{\!\!\!\longrightarrow}$.

Consider the ordering $\sqsubseteq$ on $\Qset\rightarrow {\mathcal P}(\cpaths(M)\rightarrow [0,1])$ given by
	\[ F\sqsubseteq F'\qquad \mbox{if and only if} \qquad \text{for all }q\in \Qset, F(q)\subseteq F'(q)\]
Clearly $(\cpaths(M)\rightarrow [0,1]),\sqsubseteq)$ is a complete lattice and ${\mathscr T}_{\textrm Tr}$ is monotonic, so by the theorem of Knaster-Tarski it has a greatest fixed point, which coincides with $\Traces_{s}$.

Following the definition of ${\mathscr T}_{\textrm Tr}$, we  now give a coinductive approximation of the equivalence relation induced by $\Traces_{s}$. Given a TPA $M=(\Qset,\Tset,\Lset,\qi,\vartheta)$, consider the operator
\[{\mathscr T}_{\textrm Treq}: (\cpaths(M)\rightarrow \Qset\times \Qset)\rightarrow (\cpaths(M)\rightarrow \Qset\times \Qset)\]
\noindent defined as:
\[q_1\  {\mathscr T}_{\textrm Treq} ({\mathscr R})(\epsilon) \ q_2 \quad \stackrel{\textrm def}{\Leftrightarrow} \quad (q_1\not\rightarrow \ \ \Leftrightarrow \ q_2 \not \rightarrow)
\]
and 	
\begin{equation*}
	\begin{split}
		q_1\  {\mathscr T}_{\textrm Treq} ({\mathscr R})(t_{g}:a\cdot t) \ q_2 \quad \stackrel{\textrm def}{\Leftrightarrow} \quad \quad \quad \quad \quad \quad \quad \quad \quad \quad \quad \quad \quad \quad \quad \quad \quad \quad \quad \quad \quad \\
		\left (
		\begin{array}{c}
			q_1\stackrel{t_{g}:a}{\longrightarrow}\mu_1 \Rightarrow \exists \mu_2. (q_2\stackrel{t_{g}:a}{\longrightarrow}\mu_2\ \wedge\  \mu_1\ {\mathscr R}(t) \ \mu_2)\\[1mm]
			\wedge\\[1mm]
			q_2 \stackrel{t_{g}:a}{\longrightarrow}\mu_2 \Rightarrow \exists \mu_1. (q_1\stackrel{t_{g}:a}{\longrightarrow}\mu_1\ \wedge\  \mu_1 \ {\mathscr R}(t)\  \mu_2)\\[1mm]
		\end{array}
		\right )
	\end{split}
\end{equation*}
	Consider the ordering $\preceq$ on $\cpaths(M)\rightarrow \Qset\times \Qset$ given by
	\[ {\mathscr R}\preceq {\mathscr R}' \qquad \mbox{if and only if} \qquad \text{for all }t\in \cpaths(M), \ {\mathscr R}(t)\subseteq {\mathscr R}'(t)\]
Clearly ($\cpaths(M)\rightarrow \Qset\times \Qset, \preceq)$ is a complete lattice and ${\mathscr T}_{\textrm Treq}$ is monotonic, hence by the Knaster-Tarski theorem it has a greatest fixed point, which also coincides with the greatest pre-fixed point, i.e. the greatest relation ${\mathscr R}$ such that ${\mathscr R}\preceq{\mathscr T}_{\textrm Treq} ({\mathscr R})$. Using the definition of ${\mathscr T}_{\textrm Tr}$ it is easy to see that, if ${\mathscr R}$ is a pre-fixed point, and $q_1\ {\mathscr R}(t) \ q_2$ for all $t\in\cpaths(M)$,  then $\Traces_{s}(q_1)=\Traces_{s}(q_2)$, i.e. $q_1\simeq_{s} q_2$. In fact, if $F(q_1)=F(q_2)$, and $q_1\ {\mathscr R}(t)\  q_2$ for all $t\in\cpaths(M)$, and ${\mathscr R}$ is a pre-fixed point of  ${\mathscr T}_{\textrm Treq}$, then ${\mathscr T}_{\textrm Tr} (F)(q_1) = {\mathscr T}_{\textrm Tr} (F)(q_2)$\footnote{Note that the condition is only sufficient, because $\sum_{q'} \mu_1(q') f'_{q' 1}(t) = \sum_{q'} \mu_2(q') f'_{q' 2}(t)$ may hold even if $ \mu_1$ and $\mu_2$ assign different probability to some equivalence class of $\hat{{\mathscr R}(t)}$.}. Consider now a safe bisimulation $\mathcal R$, and let us lift it to a constant function ${\mathscr R}: \cpaths(M)\rightarrow \Qset\times \Qset$ defined as ${\mathscr R}(t) = {\mathcal R}$. It is easy to see that ${\mathscr R}$ is a pre-fixed point of ${\mathscr T}_{\textrm Treq}$\footnote{Note that the converse does not hold, i.e. ${\mathscr R}$ could be a pre-fixpoint of ${\mathscr T}_{\textrm Treq}$ even if $\mathcal R$ is not a bisimulation. This is because $\mathcal R$ is sensitive to the (nondeterministic) branching structure, while ${\mathscr R}$ is not.}.

\commentMS{Comment on the footnote: I am not sure anymore that ${\mathscr R}$ is not sensitive to the branching structure. Need to check.}Assume now  $q_1\ {\mathcal R} \ q_2$. We trivially derive that $q_1\ {\mathscr R}(t) \ q_2$ for all $t\in\cpaths(M)$, from which we conclude $q_1\simeq q_2$.

\end{proof}

Like in the standard case, the vice-versa does not hold, and safe-trace equivalence is not a congruence\footnote{This is because we are considering the \emph{complete} traces.}.

\section{Safe nondeterministic information hiding}
\label{section:nd-ih}

In this section we define the notion of information hiding under the most general hypothesis that the nondeterminism
is handled partly in a demonic way and partly in an angelic way.
We assume that the demonic part is in the realm of the global scheduler, while the angelic part is controlled by the local scheduler.
The motivation is that in a protocol the local components can be thought of as programs running locally in a single machine, and
locally predictable and controllable, while
the network can be subject to attacks that make the interactions unpredictable.

We recall that, in a purely  probabilistic setting, the absence of leakage, such as noninterference and strong anonymity, is expressed as follows (see for instance \cite{Bhargava:05:CONCUR}).
Given a purely probabilistic automaton $M$, and a sequence $\tilde{a}=a_1a_2\ldots a_n$, let $\psp_M([\tilde{a}])$ represent the probability measure of all complete paths with trace $\tilde{a}$ in $M$. Let $S$ be a protocol containing a variable action $\mathit{secr}$, and let $s$ be secret actions. Let $M_s$ be the automaton corresponding  to $S[^s/_\mathit{secr}]$. Define $\mathit{Pr}(\tilde{a}\mid s)$ as $\psp_{M_s}([\tilde{a}])$. Then $S$ is leakage-free if for every observable trace $\tilde{a}$ , and for every secret $s_1$ and $s_2$, we have $\mathit{Pr}(\tilde{a}\mid s_1) =\mathit{Pr}(\tilde{a}\mid s_2)$.

In a purely nondeterministic setting, on the other hand, the absence of leakage has  been characterized  in the literature by the property
$S[^{s_1}/_\mathit{secr}] \cong S[^{s_2}/_\mathit{secr}]$,
where $\cong$ is an equivalence relation like trace equivalence, or bisimulation. As we have argued in the introduction, this definition assumes an angelic interpretation of nondeterminism.

We want to combine the above notions so to cope with both probability and nondeterminism. Furthermore, we want to extend it to the case in which part of the nondeterminism is interpreted demonically.
Let us first introduce some notation.

 Let $S$ be a system containing  a variable action $ \mathit{secr}$. Let $s$ be a secret action. Let $M_s$ be the TPA associated to $S[^s/_\mathit{secr}]$ and let $(\zeta,\xi)$ be a compatible pair of global and local schedulers for $M_s$.   The probability of an observable  trace $\tilde{a}$ given $s$ is defined as
$ \mathit{Pr}_{\zeta,\xi}(\tilde{a} \mid s) = \psp_{M_s,\zeta,\xi}([\tilde{a}]).$

The global nondeterminism is interpreted demonically, and therefore we need to ensure that the conditional of an observable, given the two secrets, are calculated with respect to the same global scheduler.
On the other hand, the local scheduler is interpreted angelically, and therefore we can compare  the conditional probabilities generated by the two secrets as sets under different schedulers. In other words, we have the freedom to  match conditional probability  from the first set with one of the other set, without requiring the local scheduler to be the same.

Either angelic or demonic, we want to avoid the  clairvoyant schedulers, i.e. a scheduler  should not be able to use the secret information  to achieve  its  goals.  For this purpose, we require both the global and the local scheduler to be admissible.

\begin{definition} A system  is leakage-free if, for every \revision{pair of} secrets $s_1$ and $s_2$, every admissible global scheduler $\zeta$, and every observable trace $\tilde a$,
\commentMS{There is a bit of an issue: how we ensure that we can have the same global scheduler on both sides -- something to think for the final version}
\begin{equation*}
	\begin{split}
		\{\mathit{Pr}_{\zeta,\xi}(\tilde{a} \mid s_1)\mid \xi \text{ is admissible and compatible with } \zeta\} = \\
   	\{\mathit{Pr}_{\zeta,\xi}(\tilde{a} \mid s_2)\mid \xi \text{ is admissible and compatible with } \zeta\}
  \end{split}
\end{equation*}
  \end{definition}

The safe equivalences defined in Section~\ref{section:equivalences} imply the absence of leakage:

\begin{theorem}\label{theo:abs_leak}
	Let $S$ be a system with a variable action $\mathit{secr}$ and assume $S[^{s_1}/_\mathit{secr}] \simeq_s  S[^{s_2}/_\mathit{secr}]$ for every pair of secrets $s_1$ and $s_2$. Then $S$ is leakage-free.
\end{theorem}

\begin{proof}

Consider the abstraction operator $\beta$ from safe traces to pairs of the form (tagged observable trace, probability) defined as:
	\[
	({\tilde a}, p)\in  \beta(F) \quad \stackrel{\textrm def}{\Leftrightarrow} \quad
	p \ \ =
	\sum_{\begin{array}{c}
	f\in F\\
	t_{\upharpoonright \Tset\times{\mathcal O}}= \tilde{a}
	\end{array}
	} f(t)
	\]
It is easy to see that $\beta$ is an abstraction, i.e. if $F_1 = F_2$ then $\beta(F_1)=\beta(F_2)$. Therefore, $S[^{s_1}/_\mathit{secr}] \simeq_{s}  S[^{s_2}/_\mathit{secr}] $ implies $\beta(\Traces_{s}(S[^{s_1}/_\mathit{secr}])=\beta(\Traces_{s}(S[^{s_2}/_\mathit{secr}])$. Finally, the latter holds (for every pair of secrets $s_1$, $s_2$) if and only if $S$ is leakage-free. 

\end{proof}

Note that the vice versa is not true, i.e. it is not the case that the leakage-freedom of $S$ implies $S[^{s_1}/_\mathit{secr}] \simeq_s S[^{s_2}/_\mathit{secr}]$. This is because in the definition of safe trace equivalence we compare the set of probability functions (determined by the schedulers) on traces,  while in the definition of leakage-freedom we compare the set of probabilities of each trace, which may come from different functions. This additional degree of freedom generated by the local scheduler  helps the system to obfuscate the secret, and  provides further  justification for  the adjective \qm{angelic}
for  the local nondeterminism.

From the above theorem and  from  Proposition~\ref{prop:simsimeq}, we also have the following corollary (with the same  premises as the previous theorem):
\begin{corollary}\label{coro:bisim}
If $S[^{s_1}/_\mathit{secr}] \sim_s S[^{s_2}/_\mathit{secr}] $ for every pair of secrets $s_1$ and $s_2$, then $S$ is leakage-free.
\end{corollary}

\section{Related work}
\label{section:related-work-equivalences}

The problem of deriving correct implementations from secrecy specifications has received a lot of attention already. One of the first works to address the problem was ~\cite{Jacob:89:SSP}, which showed that the fact that an implementation is a consistent refinement with respect to a specification  does not imply that the (information-flow) security properties  are preserved. More recently,~\cite{Alur:06:ICALP} has proposed a notion of secrecy-preserving refinement, and a simulation-based technique for proving that a system is the refinement of another. \cite{Clarkson:08:CSF} argues that important classes of security policies such as noninterference and average response time cannot be expressed by traditional notion of \emph{properties}, which consist  of sets of traces, and proposes to use \emph{hyperproperties} (sets of properties) instead. ~\cite{Dubreil:09:TAC} addresses the problem of supervisory control, i.e. given a critical system $G$ that may leak confidential information, how to design a controller $C$ so that the system $G|C$ dos not leak. An effective algorithm is presented to compute the most permissible controller such that the system is still opaque with respect to a secret.

Concerning angelic and demonic nondeterminism, there are various works which investigate their relation and possible combination. In ~\cite{Back:92:TCS} it is shown that angelic and demonic nondeterminism are dual. ~\cite{Martin:2007:SCP} uses multi-relations to express specifications involving both angelic and demonic nondeterminism. There are two kinds of agents, demonic and angelic ones, and there is the point of view of the internal system and the one of the external adversary. 

\cite{Morgan:09:SCP} considers the problem of refining specifications while preserving ignorance. While the focus is on the reduction of demonic nondeterminism of the specification, the hidden values are treated essentially in a angelic way.

The problem of the leakage caused by full-information schedulers has also been investigated  in the literature. \cite{Canetti:06:WODES} and \cite{Canetti:06:DISC} work in the framework of probabilistic automata and introduce a restriction on the scheduler to the purpose of making them suitable to applications in security protocols. Their approach is based on dividing the actions of each component of the system in equivalence classes (\emph{tasks}). The order of execution of different tasks is decided in advance by a so-called \emph{task scheduler}, which is history-independent and therefore  much more restricted than our notion of global scheduler. \cite{Andres:10:TechRep} proposes a notion of system and admissible scheduler very similar to our notion of system and admissible global scheduler. The main difference is that in that work the components are deterministic and therefore there is no notion of local scheduler.

The work in \cite{Chatzikokolakis:07:CONCUR,Chatzikokolakis:09:FOSSACS} is similar to ours in spirit, but in a sense  \emph{dual}  from a technical point of view. Instead of defining a restriction on the class of schedulers,  the authors a way to specify that a choice is transparent to the scheduler. They achieve this by introducing labels in process terms, used to represent both the states of the execution tree  and the next action or step to be scheduled. They make two states indistinguishable to schedulers, and hence the choice between them private, by associating to them the same label. We believe that every scheduler in our formalism can be expressed in theirs, too. In \cite{Chatzikokolakis:09:FOSSACS} the authors consider the problem of defining a safe version of bisimulation for expressing security properties. They call it \emph{demonic bisimulation}. The main difference with our work is that we consider a combination of angelic and demonic  nondeterminism, and this affects also the definition of bisimulation. Similarly, our definition of leakage-freedom reflects this combination. In  \cite{Chatzikokolakis:09:FOSSACS} the aspect of angelicity is not considered, although they may be able to simulate it with an appropriate labeling.

The fact that full-information  schedulers are unrealistic has also been observed in fields other than security.  First attempts used restricted schedulers in order to obtain rules for compositional reasoning~\cite{deAlfaro:01:CONCUR}. The justification for those restricted schedulers is the same as for ours, namely, that not all information is available to all entities in the system. That work considers a synchronous parallel composition, however, so the setting is rather different from ours. Later on, it was shown that model checking is unfeasible in its general form for the restricted schedulers in~\cite{deAlfaro:01:CONCUR} (see~\cite{Giro:07:FORMATS} and, more recently,~\cite{Giro:09:SBMF}). Despite of undecidability, not all results concerning such schedulers have been negative as, for instance, the technique of partial-order reduction can be improved by assuming that schedulers can only use partial information~\cite{Giro:09:CONCUR}.

\section{Chapter summary and discussion}
\label{section:conclusion-safe-equivalences}

In this chapter we have observed that some  definitions of security properties based on process equivalences may be too naive, in the sense that they assume the scheduler to be angelic, and, worse yet, to achieve its angelic strategy by peeking at the secrets. We have presented a formalism allowing us to specify a demonic constituent of the scheduler, possibly in collusion with the attacker, and an angelic one, under the control of the system. We have also considered restrictions on the schedulers to limit the power of what they can see, and extended to our nondeterministic  framework the   (probabilistic) information-hiding properties like non interference and strong anonymity. We then have defined \qm{safe} equivalences. In particular we have defined the notions of safe trace equivalence and safe bisimilarity, and we have shown that the latter is still a congruence. Finally, we have shown that the safe equivalences can be used to prove information-hiding properties.

For the future, we plan to extend our framework to quantitative notions of information leakage, possibly based on information theory. We also plan to implement model checking techniques to verify information hiding properties for our kind of systems.

\chapter{Conclusion}
\label{chapter:conclusion}
\mscite{To succeed, jump as quickly at opportunities as you do at conclusions.}{Benjamin Franklin}
In this thesis we concentrated on the problem of information hiding in the scenarios of interactive systems, statistical disclosure control, and the refinement of specifications. We started by giving a general overview of the field of information hiding, including a brief description of its historical development. We then discussed the main differences between the qualitative and the quantitative approaches to information hiding, and we introduced the background for the three main topics covered in this thesis: information flow (exemplified by anonymity), statistical disclosure control, and the refinement of specifications into implementations.

Having adopted the quantitative approach, we then continued to discuss the rationale of the use of information theory for quantitative information flow. We reviewed several formulations of entropy, with a special focus on Shannon entropy and min-entropy, and the related concept of mutual information and its interpretation in terms of attacks and information leakage.

We then proceeded to present the technical contributions of the thesis. We started with the scenario of interactive systems, i.e systems where secrets and observables can alternate and influence each other during the computation. In this type of systems the traditional information theoretical approach that makes use of classic memoryless channels, and the related concepts of mutual information and classical capacity, no longer works. We proposed to model interactive systems with a richer notion of channels, namely channels with memory and feedback. In this more general model it is possible to split the statistical correlation between secrets and observables (that correspond to the input and the output of the channel, respectively) into two causal components: the \emph{directed information from input to output} represents the flow of information through the channel, and the \emph{directed information from output to input} corresponds to the way the input is influenced by the output via feedback. We showed that the directed information is the correct measure of leakage in interactive systems, and so is the concept of directed capacity if we are interested in the worst case leakage. We also proved that our model is a proper extension of the classic one: in the absence of feedback (i.e interaction) our model collapses into the simpler classic model. Finally, we showed that the capacity of channels with memory and feedback is a continuous function of a pseudometric based on the Kantorovich metric. 

With respect to interactive systems, as future work we want to explore algorithms to calculate the leakage and the maximum leakage using our model. This is a rather challenging problem, given the exponential growth of reaction functions (a technical aspect of our model) and the quantification of possibly infinite many reactors (also another technicality of our model). We also want to explore other notions of entropy as a measure of leakage, as for instance the min-entropy and the corresponding notion of one-try attack.

In the sequence we moved to the problem of statistical disclosure control. We considered the problem of preserving the privacy of individuals participating in a database that allows statistical queries to be posed by users. Using differential privacy, databases that are similar, i.e differ by the contents of at most one row, should give statistically \qm{similar} answers to the same query. This is achieved by introducing noise in the query mechanism to blur the link between the reported answer and the data about individuals. We proposed a model where the differential privacy mechanism can be split into two channels in cascade, in the case the randomization mechanism is oblivious (i.e it only depends on the real answer to the query, and not on the database itself). The first channel corresponds to the query, and it maps the database to the real answer to the query. The second channel corresponds to the oblivious randomization mechanism, and it takes the real answer and maps it to a randomized answer to be reported to the user. In this scenario we see the \emph{leakage} as the correlation between the reported answer and the database, and the \emph{utility} as the correlation between the real answer and the reported one. We used this model to derive bounds for the leakage and utility based on the level of differential privacy designed for the system (namely the parameter $\epsilon$). As a measure of leakage we adopted the min-entropy leakage, and for utility we used the notion of gain functions, focusing on the binary gain function, which is strictly related to min-entropy leakage and Bayes risk. We used the graph structure on the input domain derived from the adjacency relation on databases to derive bounds for the maximum min-entropy leakage of channels. We showed that if the graph structure is distance-regular or $\vtt$ (which is always the case for the database domain), then we can derive bounds for the maximum min-entropy leakage associated to the channel. Finally, we found a way of constructing a utility-maximizing randomization function that respects differential privacy for a special class of graph structures.

In relation to statistical databases, as future work we intend to extend our results to other types of gain functions 
\revision{than} the binary one, namely gain functions that take into consideration a notion of distance between answers. We also want to investigate whether or not non-oblivious randomization mechanisms can be used to improve utility while still preserving differential privacy.

The last scenario we investigated in the thesis was the use of equivalence relations to specify security guarantees, which is a common approach when refining implementations into specifications. Under this perspective, two systems (e.g a specification and its implementation) are considered equivalently secure if they respect some equivalence relation defined to capture the intended security guarantee. Such equivalences include, for instance, trace-equivalence and bisimilarity. We showed that a naive use of these equivalences can lead to unrealistic assumptions about the scheduler: (i) that the scheduler is angelic, i.e that it will help to keep the secret information from the attacker; and (ii) that the scheduler can peek at the secrets to make its choices. Those assumptions are not safe in practical cases and, therefore, we proposed a model that deals with the problem. We introduced a formalism that explicitly separates the demonic and angelic parts of the scheduler, and we imposed restrictions to limit the power of the scheduler with respect to what it can see. Namely, the scheduler cannot peek at the secrets to make its choices. We then defined notions of safe-equivalences (safe trace equivalence and safe bisimilarity) and we showed that the latter is a congruence. Finally, we showed that safe equivalences can be used to prove information hiding properties.

As future work regarding safe equivalences, we want to extend our model to quantitative notions based on information theory, and we want to use model checking to certify information hiding properties for our systems.

As final remark, we believe that information hiding is a very promising field of research, and we are excited and thrilled by the promising challenges that lie ahead.

\bibliographystyle{alpha}
\bibliography{msalvim}

\end{document}